\pdfoutput=1
\documentclass[a4paper,11pt]{article}

\usepackage[a4paper, left=2.5cm, right=2.5cm, top=2cm, bottom=2.5cm]{geometry}

\usepackage[english]{babel}
\usepackage[utf8]{inputenc}
\usepackage[T1]{fontenc}
\usepackage{lmodern}
\usepackage{amsmath,amsthm,mathdots,amssymb,bm,bbm,mathrsfs,mathtools,stmaryrd}
\usepackage{multirow}
\usepackage{makecell} 
\usepackage{booktabs} 
\usepackage[final]{graphicx} 
\usepackage{tikz,pgfplots}\pgfplotsset{compat=1.16}
\usepackage[bottom]{footmisc}

\usepackage{ifpdf}
\ifpdf 
\else 
\usepackage[hyphenbreaks]{breakurl}
\fi

\makeatletter
\def\@fnsymbol#1{\ensuremath{\ifcase#1\or *\or \mathsection\or **\or \mathparagraph\or \|\or \ddagger \else\@ctrerr\fi}}
\makeatother

\usepackage{comment}

\usepackage{mleftright}\mleftright 

\usepackage{enumitem}
\usepackage{microtype,color} 
\usepackage{xcolor}
\usepackage[colorlinks=true, urlcolor=violet, linkcolor=blue, citecolor=red, hyperindex=true, linktocpage=true, pagebackref=true, draft=false]{hyperref}
\usepackage[nameinlink]{cleveref}
\usepackage{stackengine}
\newcommand\circpr{\ensurestackMath{\stackinset{c}{}{c}{0mm}{\cdot}{\circ}}}

\usepackage{thm-restate}
\newtheorem{theorem}{Theorem}
\newtheorem{definition}{Definition}
\newtheorem{fact}{Fact}
\newtheorem{lemma}{Lemma}
\newtheorem{proposition}{Proposition}

\newtheorem{corollary}{Corollary}
\newtheorem{example}{Example}
 
\renewcommand{\vec}{\bm}

\newcommand{\BC}{\mathbb{C}}
\newcommand{\CD}{\mathcal{D}}

\newcommand{\BE}{\mathbb{E}}
\newcommand{\CF}{\mathcal{F}}

\newcommand{\CH}{\mathcal{H}}
\newcommand{\CI}{\mathcal{I}}

\newcommand{\CL}{\mathcal{L}}
\newcommand{\CM}{\mathcal{M}}
\newcommand{\CN}{\mathcal{N}}

\newcommand{\CO}{\mathcal{O}}

\newcommand{\CQ}{\mathcal{Q}}

\newcommand{\CR}{\mathcal{R}}
\newcommand{\BR}{\mathbb{R}}
\newcommand{\CS}{\mathcal{S}}
\newcommand{\BS}{\mathbb{S}}
\newcommand{\CT}{\mathcal{T}}

\newcommand{\vA}{\bm{A}}

\newcommand{\vB}{\bm{B}}

\newcommand{\vC}{\bm{C}}
\newcommand{\vD}{\bm{D}}

\newcommand{\vG}{\bm{G}}

\newcommand{\vH}{\bm{H}}
\newcommand{\vI}{\bm{I}}

\newcommand{\vK}{\bm{K}}

\newcommand{\vL}{\bm{L}}

\newcommand{\vM}{\bm{M}}
\newcommand{\vN}{\bm{N}}
\newcommand{\vO}{\bm{O}}
\newcommand{\vP}{\bm{P}}

\newcommand{\vPi}{\bm{\Pi}}
\newcommand{\vQ}{\bm{Q}}
\newcommand{\vR}{\bm{R}}

\newcommand{\vT}{\bm{T}}

\newcommand{\vU}{\bm{U}}
\newcommand{\vV}{\bm{V}}
\newcommand{\vW}{\bm{W}}
\newcommand{\vX}{\bm{X}}
\newcommand{\vY }{\bm{Y }}
\newcommand{\vZ}{\bm{Z}}
\newcommand{\vsigma}{\bm{ \sigma}}

\newcommand{\vrho}{\bm{ \rho}}
\renewcommand{\L}{\left}
\newcommand{\R}{\right}
\newcommand{\gammaprob}{g}

\newcommand{\energy}{E}
\newcommand{\ftime}{f}
\newcommand{\fenergy}{\hat{f}}

\SetSymbolFont{stmry}{bold}{U}{stmry}{m}{n}
\newcommand{\vCD}{\bm{\CD}}

\newcommand{\vCM}{\bm{\CM}}
\newcommand{\vCS}{\bm{\CS}}
\newcommand{\vCT}{\bm{\CT}}
\newcommand{\vCQ}{\bm{\CQ}}

\newcommand{\tOmega}{\tilde{\Omega}}

\newcommand{\eps}{\varepsilon}

\newcommand{\dagg}{\dagger}

\newcommand{\vertiii}[1]{{\big\vert\kern-0.25ex\big\vert\kern-0.25ex\big\vert #1 \big\vert\kern-0.25ex\big\vert\kern-0.25ex\big\vert}}
\newcommand{\norm}[1]{\Vert {#1} \Vert}

\newcommand{\normp}[2]{\norm{#1}_{#2}}

\newcommand{\labs}[1]{\left\vert {#1} \right\vert}
\newcommand{\lnorm}[1]{\left\Vert {#1} \right\Vert}
\newcommand{\e}{\mathrm{e}}

\newcommand{\ri}{\mathrm{i}}
\newcommand{\rd}{\mathrm{d}}
\newcommand*{\diag}{\operatorname{diag}}
\newcommand*{\tr}{\operatorname{Tr}}
\newcommand*{\poly}{\operatorname{Poly}}

\newcommand{\Lword}[1]{\text{Lindbladian}}

\newcommand{\nrm}[1]{\left\| #1 \right\|}
\newcommand{\ipc}[2]{\left\langle#1,#2\right\rangle}
\newcommand{\undersetbrace}[2]{ \underset{#1}{\underbrace{#2}}}

\DeclareMathOperator{\sinc}{sinc}
\DeclareMathOperator{\sinhc}{sinhc}

\DeclareMathOperator{\erfc}{erfc}
\DeclarePairedDelimiter\ket{\lvert}{\rangle}
\DeclarePairedDelimiter\bra{\langle}{\rvert}
\DeclarePairedDelimiterX{\braket}[1]{\langle}{\rangle}{#1}
\DeclarePairedDelimiterX\ketbra[2]{| }{|}{#1 \delimsize\rangle\!\delimsize\langle #2}	
\DeclarePairedDelimiterX\dotp[2]{\langle}{\rangle}{#1, #2}
\newcommand{\bigO}[1]{\mathcal{O}\left( #1 \right)}
\newcommand{\bigOt}[1]{\widetilde{\mathcal{O}}\left( #1 \right)}

\DeclareMathAlphabet{\dutchcal}{U}{dutchcal}{m}{n}
\SetMathAlphabet{\dutchcal}{bold}{U}{dutchcal}{b}{n}
\DeclareMathAlphabet{\dutchbcal} {U}{dutchcal}{b}{n}

\makeatletter
\DeclareRobustCommand*{\pmzerodot}{%
	\nfss@text{%
		\sbox0{$\vcenter{}$}
		\sbox2{0}%
		\sbox4{0\/}%
		\ooalign{%
			0\cr
			\hidewidth
			\kern\dimexpr\wd4-\wd2\relax 
			\raise\dimexpr(\ht2-\dp2)/2-\ht0\relax\hbox{%
				\if b\expandafter\@car\f@series\@nil\relax
				\mathversion{bold}%
				\fi
				$\cdot\m@th$%
			}%
			\hidewidth
			\cr
			\vphantom{0}
		}%
	}%
}

\usepackage{mmacells}

\mmaSet{
	morefv={gobble=2},
	linklocaluri=mma/symbol/definition:#1,
	morecellgraphics={yoffset=0mm},
	leftmargin=11.5mm,
	labelsep=1mm,
}

\usepackage{ifdraft}
\ifdraft{\pdfcompresslevel=0\pdfobjcompresslevel=0}{\pdfminorversion=5 \pdfcompresslevel=9 \pdfobjcompresslevel=3} 
\ifdraft{
	\newcommand{\authnote}[3]{{\color{#3} {\bf  #1:} #2}}	
}{
	\newcommand{\authnote}[3]{}
}
\newcommand{\AC}[1]{\authnote{ Anthony}{#1}{blue}}
\newcommand{\anote}[1]{\authnote{ András}{#1}{teal}}

\newcommand{\jnote}[1]{\authnote{ João}{#1}{green}}

\usetikzlibrary{quantikz2}
\makeatletter
\def\l@subsection#1#2{}
\def\l@subsubsection#1#2{}
\makeatother

\title{\vskip-5mm Quantum generalizations of Glauber and Metropolis dynamics}

\author{András Gilyén\thanks{HUN-REN Alfréd Rényi Institute of Mathematics, Budapest, Hungary. 
		Funded by the AWS Center for Quantum Computing and the EU's Marie Skłodowska-Curie program QuantOrder-891889.
		{\tt gilyen@renyi.hu}}
	\and 
	Chi-Fang Chen\thanks{Institute for Quantum Information and Matter,
		California Institute of Technology, Pasadena, CA, USA;
University of California, Berkeley, CA, USA, supported by a Simons-CIQC postdoctoral fellowship through NSF QLCI Grant No. 2016245;
Massachusetts Institute of Technology, Cambridge, Massachusetts, USA.
		{\tt achifchen@gmail.com}}
	\and 
	Joao F.\ Doriguello\thanks{HUN-REN Alfréd Rényi Institute of Mathematics, Budapest, Hungary. 
		Funded by ERC grant No.\ 810115-DYNASNET. {\tt doriguello@renyi.hu}}
	\and 	
	Michael J.\ Kastoryano\thanks{AWS Center for Quantum Computing, Pasadena, CA \& University of Copenhagen, Denmark.} 
}

\date{}

\begin{document}
\maketitle
\begin{abstract}
Classical Markov Chain Monte Carlo methods have been essential for simulating statistical physical systems and have proven well applicable to other systems with many degrees of freedom. 
Motivated by the statistical physics origins, Chen, Kastoryano, and Gilyén~\cite{chen2023ExactQGibbsSampler} proposed a continuous-time quantum thermodynamic analogue to Glauber dynamics that is (i) exactly detailed balanced, (ii) efficiently implementable, and (iii) quasi-local for geometrically local systems. Physically, their construction resembles the dissipative dynamics arising from weak system-bath interaction. 
In this work, we give an efficiently implementable discrete-time counterpart to any continuous-time quantum Gibbs sampler. Our construction preserves the desirable features (i)-(iii) while does not decrease the spectral gap. Also, we give an alternative highly coherent quantum generalization of detailed balanced dynamics that resembles another physically derived master equation, and propose a smooth interpolation between this and earlier constructions. Moreover, we show how to make earlier Metropolis-style Gibbs samplers (which estimate energies both before and after jumps) exactly detailed balanced. We study generic properties of all constructions, including the uniqueness of the fixed point, the \mbox{(quasi-)}locality of the resulting operators. Finally, we prove that the spectral gap of our new highly coherent Gibbs sampler is constant at high temperatures, thereby it mixes fast. We hope that our systematic approach to quantum Glauber and Metropolis dynamics will lead to widespread applications in various domains. 
\end{abstract}

\tableofcontents

\newpage
\section{Introduction}
Markov Chain Monte Carlo (MCMC) methods are an essential tool in classical algorithm design. Their discovery marked a significant leap in our ability to simulate and understand systems characterized by a vast number of interacting components, from the intricate interplay of atoms and molecules in materials science to the unpredictable dynamics of financial markets. The elegance of, e.g., the Metropolis sampling algorithm~\cite{metropolis1953StateCalculationsByComputers}, lies in its simplicity and versatility, enabling a wide range of complex applications. In the realm of scientific computation, particularly in physics and chemistry simulation, the Glauber dynamics~\cite{glauber1963TimeDepIsing} and its numerous variants have drastically enhanced our understanding of material properties, reaction dynamics, phase transitions, and thermodynamics. 

However, the equations governing many-body quantum mechanics are notoriously difficult to address by the most basic classical MCMC methods. To tackle this issue, an arsenal of classical simulation methods, such as density functional theory~\cite{hohenberg1964InhomElectronGas}, tensor network representations~\cite{white1992DMRG}, and quantum Monte Carlo methods~\cite{ceperley1986QMonteCarlo}, has been developed and employed for understanding chemicals, materials, and nuclei. However,  all these clever methods can break down for large, strongly correlated quantum systems featuring highly entangled states. 

Given the widespread belief that quantum computing will play a revolutionary role in quantum simulation, quantum inference, and possibly also in optimization, it is natural to speculate that quantum MCMC algorithms will have an equally important role as their classical counterpart. Perhaps driven by this hope, there has been a new wave of quantum MCMC algorithms~\cite{moussa2019LowDepthQMetropolis,wocjan2021SzegedyWalksForQuantumMaps,shtanko2021PreparingQThermalStatesNoieslessy,rall2023ThermalStatePrepRounding} that draws inspiration from the cooling process in Nature to design continuous-time Quantum Markov chains (i.e., Lindbladians) satisfying detailed balance. These constructions are quantum algorithmic generalizations of the classical Glauber dynamics --- the reference classical model for continuous-time thermalization of lattice systems.
Nevertheless, the quantum analog of detailed balance, which has been central to classical Markov chain design and analysis, has posed a challenge to quantum algorithms design and has only recently been achieved exactly~\cite{chen2023ExactQGibbsSampler} in an algorithmically accessible way. This construction provably leads to an efficient Gibbs state preparation method in the high-temperature regime~\cite{rouze2024EffThermalizationGibbsSamp,rouze2024optimal}.

On the other hand, despite the prominent role the Metropolis sampling algorithm has played in the classical world, there have only been very few quantum generalizations~\cite{temme2009QuantumMetropolis,yung2010QuantumQuantumMetropolis}. While the continuous-time Glauber dynamics enjoys a closer connection to Nature, discrete-time Markov chains are sometimes more favorable in applications, as they are often more straightforward to implement on a digital computer. However, quantum detailed balance in discrete-time quantum channels poses additional challenges not present in the continuous-time setting. In a nutshell, a naive quantum generalization of the rejection step in the Metropolis algorithm requires cloning, which was earlier circumvented~\cite{temme2009QuantumMetropolis} by ``rewinding'' the quantum state by the sophisticated Marriott-Watrous~\cite{marriott2005QAMGames} technique. This appears to be an obstacle to achieving detailed balance in discrete-time, which motivates the following question:
\begin{align*}
   \text{\emph{Can we design a discrete-time Quantum Metropolis algorithm with exact detailed balance?}}
\end{align*} 
We give precisely such a construction: given a Hamiltonian and a set of jumps, there is a quantum channel (completely positive and trace-preserving map)  that satisfies the desirable features of a discrete-time quantum Gibbs sampler:
\begin{itemize}
   \item (Exact detailed balance) The channel is exactly detailed balance.
   \item (Quasi-locality) The channel inherits the locality of the jumps and the Hamiltonian. 
   \item (Efficient simulation) The channel  can be simulated efficiently.
\end{itemize}
%
In addition, we propose alternative constructions of exact quantum detailed-balanced Lindbladians, which depart from the earlier constructions that largely resemble the Davies-generator~\cite{davies1974MasterEquationI,davies1976MasterEquationII}.

\subsection{Unifying perspective on classical and quantum detailed balanced dynamics}\label{sec:BuildUpIntro}

In addition to presenting new constructions, a motivation for this work is to provide a systematic approach to the various old and new constructions. 
Both classical Glauber and Metropolis dynamics can be stated abstractly, which then enables us to find natural concepts along which one can make sense of quantum generalizations. We present a unifying perspective of quantum and classical detailed balanced dynamics, and thoroughly compare our new results to earlier quantum and classical constructions, see \Cref{table:Constructions}.

\begin{table}[!ht]
	\setlength{\tabcolsep}{3pt} 
	\renewcommand{\arraystretch}{1.2} 
	\centering
	\begin{tabular}{|c|c|c|c|c|c|l|}
		\hline
		& Method & \makecell{Efficient \\ implement.\ } & Local & \makecell{Energy \\ uncertainty} & \makecell{Ergodicity \\ preserving} & \makecell{Kraus \\ rank} \\ 
		\hline
		\multirow{1}{*}{Classical} & Glauber \& Metropolis & \checkmark & \checkmark & $0$ & \checkmark & $\,\leq \text{orig.}$\\ 
		\hline
		\hline
		\multirow{5}{*}{Quantum} & Davies generator & - & - & $0$ & \checkmark & $\,\leq \max$\\ 
		\cline{2-7}
		\cline{2-7}
		& Phase estimation & \checkmark & - & $\CO(1)$ (\textbf{new}) & \checkmark \!(\textbf{new})\! & $\,\leq \max$\\
		\cline{2-7}  
		& Operator Fourier transf.\ & \checkmark & \!quasi & $\CO(1)$ & \checkmark (\textbf{new})\! & $\,\leq \max$\\  
		\cline{2-7}
		& Interpolated (\textbf{new}) & \checkmark & quasi & $\sigma$ & almost surely & $\,\leq \max$\\
  		\cline{2-7}
		& Coherent reweighing\! (\textbf{new}) & \checkmark & \!quasi & $\infty$ & almost surely & $\,\leq \text{orig.}$\\

		\hline
	\end{tabular}
	\caption{\label{table:Constructions}
		Comparison of classical (see~\Cref{sec:MetropolisGlauber}) and quantum~(see \Cref{sec:dbcp}) methods for detailed-balanced sampling; the table is valid for both the discrete- and continuous-time variants. 
		We compare the various methods and their properties: efficient implementation on (quantum) computers, (quasi-) locality of the dynamics, energy uncertainty of the transitions, preservation of the uniqueness of the fixed point (i.e., ergodicity), and the Kraus rank. 
		All prior quantum constructions can induce maximal Kraus rank, even starting from a single jump operator.
		This is interesting because the Kraus rank can be thought of as a proxy for how much coherence-preserving the construction is. 
		Another measure of coherence is related to how much one (or in general the environment) learns about the energy change while the transitions happen, which is detailed in the energy uncertainty column. As we show here, even the phase-estimation-based method can be made exactly detailed balanced even with $\bigO{1}$ uncertainty in the energy estimation.
		The proof of ergodicity of the existing construction~\cite{chen2023ExactQGibbsSampler} is also our contribution. Our new interpolated construction (\Cref{sec:interpolation}) enables choosing the energy uncertainty $\sigma$ while earlier construction required resolving the energy with at most about $\CO(1)$ uncertainty.
		Our new efficient construction, similarly to the earlier one by~\cite{chen2023ExactQGibbsSampler}, can be implemented by utilizing merely $\bigOt{1}$ (controlled) Hamiltonian simulation time,\textsuperscript{\ref{foot:beta}} while the interpolated construction requires $\bigOt{1+1/\sigma}$ time.
	    Thus, our new smooth reweighing and interpolated methods stand out with their highly coherent sampling procedure which, in the coherent case, does not even increase the Kraus rank. 
		A potential caveat is that they might accidentally violate ergodicity; albeit on a zero-measure set of CP maps. 
  }
\end{table}

We view the classical Metropolis rule as a map transforming a symmetric ergodic stochastic process described by a matrix $\vT$ into another symmetric ergodic stochastic process described by a matrix $\vT'$. If the matrix elements $\vT_{ij}$ describe the probability of transition to state $i$ from some state $j$, then the modified process $\vT'$ using the Metropolis rule has matrix elements $\vT'_{ij}=\min(1,\pi_i/\pi_j)\vT_{ij}$. This transformation is specially nice in that it ensures detailed balance with respect to the distribution $\pi$, i.e.,
\begin{align}
    \vT'_{ij}\pi_j = \vT'_{ji} \pi_i.
\end{align}
As a result, $\pi$ is the unique fixed point of the modified stochastic process $\vT'$. As we discuss below, there is always a straightforward and unique way of modifying the diagonal elements of $\vT'$ in order to restore stochasticity in both the discrete and continuous-time settings, see \eqref{eq:ClRejIntro}.
Glauber dynamics is similar with a different transformation $\gammaprob_G(r)=1/(1+1/r)$ of the matrix elements $\vT'_{ij}=\gammaprob_G(\pi_i/\pi_j)\vT_{ij}$. There are other possible variants that fall into this category, and they can all be described systematically as a \textit{superoperator} $\CS^{[g]}$ that maps the symmetric matrix $\vT$ to the $\pi$-detailed balanced matrix (or operator) 
\begin{align}\label{eq:SOpGlaub}
	\vT'=\CS^{[g]}[\vT] \quad \text{ where } \vT'\text{ is $\pi$-detailed balanced for any symmetric } \vT.
\end{align}
This conceptually clear picture lends itself to quantum generalizations. The natural quantum counterpart of $\pi$ is the mixed state \vskip-7mm
\begin{align*}
\vrho\propto \exp(-\vH)  
\end{align*}
corresponding to the Hamiltonian $\vH$, where the inverse temperature is absorbed into the Hamiltonian $\beta \vH \rightarrow \vH$,\footnote{\label{foot:beta}We work with dimensionless quantities; considering dimensions the uncertainty should scale with $1/\beta$~\cite{chen2023ExactQGibbsSampler}. Similarly, if we work with $\beta\vH$ instead of $\vH$, then the Hamiltonian simulation time increases by the factor $\beta$.} while the quantum counterpart of a symmetric transition matrix is a self-adjoint CP map. Quantum detailed balance takes a similar form
\begin{align}
\CT'[\sqrt{\vrho}\cdot\sqrt{\vrho}] = \sqrt{\vrho} \CT^{'\dagger}[\cdot]\sqrt{\vrho},
\end{align}
see~\Cref{def:KMSDetBalance} for more details. Accordingly, the quantum counterpart of the superoperator $\CS^{[g]}$ describing a generalized Glauber or Metropolis rule should be a linear super-superoperator $\BS^{[\gamma]}$ that maps a self-adjoint completely positive map $\CT$ to another $\vrho$-detailed balanced completely positive~map
\begin{align*}
	\CT'=\BS^{[\gamma]}\llbracket\CT\rrbracket \quad \text{ where } \CT'\text{ is $\vrho$-detailed balanced for any self-adjoint } \CT.
\end{align*}

Remarkably, once the transition part $\CT'$ is detailed balanced, there is an essentially \textit{unique} and well-behaved recipe for ensuring trace preservation (i.e., the quantum analogue of stochasticity) for both continuous and discrete-time quantum dynamics (see \Cref{table:DiscCont} and \Cref{sec:FromCPToDyn}).

Indeed, we identify the linear super-superoperator $\BS$ as the key component in all known constructions for quantum detailed balance. Linearity in particular enables arguing about the transformation of $\CT[\cdot]=\sum_{a\in A}\vA^{a}[\cdot]\vA^{a\dagger}$ on the level of its (non-unique) Kraus operators, without needing to worry about differing but equivalent Kraus decompositions $\sum_{a\in A}\vA^{a}[\cdot]\vA^{a\dagger}=\sum_{a\in A}\vB^a[\cdot]\vB^{a\dagger}$. More importantly, linearity leads to an efficient implementation of $\BS$ on a quantum computer by describing transformation rules at the level of individual jump operators. Hence, all our quantum algorithms can be efficiently implemented if $\CT$ is provided via a unitary block-encoding $\vU$ of a dilation $\vG=\sum_{a\in A} \ket{a} \otimes \vA^a$.


\subsection{The new discrete-time constructions}

The starting point of our discrete-time construction is to revisit the classical Metropolis sampling and decompose it into an `Accept' part $\vT'$ and a `Reject' part $\vR$, where $\vT'$ is given by the Metropolis rule applied to the symmetric transitions $\vT$. The ultimate dynamics is then given by 
\begin{align*}
    \vM = \vT' + \vR.
\end{align*}
 Here, the two parts play distinct roles: the accept part $\vT'$ is where the transitions are designed to be detailed balanced (but not probability preserving, i.e., stochastic) and the reject part $\vR$ is a diagonal matrix that ensures probability preservation
\begin{align}\label{eq:ClRejIntro}
    \vR_{jj}=1-\sum_{i\neq j}\vT'_{ij}.
\end{align}
This decomposition inspires the quantum generalizations mentioned in the previous section. In the classical setting, given a detailed-balanced accept part $\vT'$, obtaining the corresponding rejection part $\vR$ is straightforward. The same is not true in the quantum setting if one is given a detailed-balanced ``accept'' CP map $\CT'$ that describes the quantum transition acting on a given density operator. Our construction is a recipe for prescribing a rejection term so that the trace-preservation is restored without breaking detailed balance.
\begin{theorem}[Prescription for the rejection part]
    Consider any $\vrho$-detailed balanced CP map $\CT'$ such that $\CT'^\dagger[\vI] \preceq \vI$. Then, the following CP map 
\begin{align*}
    \CM = \CT' + \CR\quad \text{where}\quad \CR[\cdot] = \vK[\cdot]\vK^\dagger
\end{align*}    
with 
\begin{align}
    \vK:=\sqrt{\sqrt{\vrho}(\vI-\CT'^\dagger[\vI])\sqrt{\vrho}}\vrho^{-\frac12}\label{eq:main_K}
\end{align}
is $\vrho$-detailed balanced and trace-preserving.
\end{theorem}

\anote{Note that in the classical setting the continuous-time generator $\vL'$ and the discrete-time map $\vP'$ can be trivially related: $\vP'=\vI+\vL'$. This, however breaks down in the quantum case, as $\CI[\cdot]+\CL[\cdot]$ is in general not completely positive. Indeed, consider the self-adjoint Lindbladian $\CL[\cdot]:=(\vZ\otimes\ketbra{0}{0})[\cdot](\vZ\otimes\ketbra{0}{0})-\frac12((\vI\otimes\ketbra{0}{0})[\cdot]+[\cdot](\vI\otimes\ketbra{0}{0}))$ on two qubits. Then for all $\varepsilon>0$ we can see that $\Phi_{\varepsilon}:=\CI+\varepsilon\CL$ is not completely positive, since $\Phi_{\varepsilon}[\ketbra{+}{+}\otimes\ketbra{+}{+}]=\ketbra{+}{+}\otimes\ketbra{+}{+}+\frac{\varepsilon}{2}\ketbra{-}{-}\otimes\ketbra{0}{0}-\frac{\varepsilon}{2\sqrt{2}}\ketbra{+}{+}\otimes(\ketbra{0}{+}+\ketbra{+}{0})$, which is not positive definite, e.g., its matrix element at $\ket{+}(\varepsilon\ket{+}+2\ket{-})$ is $\varepsilon^2- \frac{\varepsilon}{2\sqrt{2}}(2\varepsilon(\frac{\varepsilon}{\sqrt{2}}+\sqrt{2}))=-\frac{\varepsilon^3}{2}$.}

The key insight from Chen, Kastoryano, and Gily\'en in the continuous-time setting~\cite{chen2023ExactQGibbsSampler} (see~\Cref{lem:FindingCoherenceTerm}) gives essentially a unique way to restore detailed balance using an additional \textit{coherent} term for continuous-time quantum dynamics; this fix also turns out to be a linear transformation of $\CT'$. In the discrete-time setting, if the transition part $\CT'$ is detailed balanced (and properly normalized $\CT'^\dag[\vI]\preceq \vI$), then there is also an essentially \textit{unique} and well-behaved prescription $\vK$ to make it trace-preserving. However, this operator $\vK$ appears to be a \textit{nonlinear} function of $\CT'$ due to the matrix square-root. The non-linearity is consistent with the fact that Lindbladian simulation algorithms also explicitly introduce higher-order terms~\cite{cleve2016EffLindbladianSim,li2022SimMarkOpen,chen2023QThermalStatePrep,ding2023SimOpenQSysUsingHamSim}. Nonetheless, our new discrete-time variant does not go through a black-box Lindbladian simulation algorithm and might lead to simpler implementation on a quantum computer.

Even though the matrix square-root and inversion of Gibbs state may appear daunting, we give an efficient implementation through linear-combination-of-unitaries (LCU) and quantum singular value transformation (QSVT).
\begin{theorem}[Time integral representation and LCU implementation]
	Let $\vO$ be a Hermitian matrix. Then, for small enough $\nrm{\vO} = \mathcal{O}(1/\log^2(\|\vH\|+2))$, the operator
	\begin{align}\label{thm:informalOForm}
		\sqrt{\sqrt{\vrho}(\vI-\vO)\sqrt{\vrho}}\vrho^{-\frac12}
	\end{align} 
	can be $\eps$-approximated by an integral over a product of short-time Heisenberg-evolutions
	\begin{align}\label{thm:informalIntForm}
		\underset{I^k}{\idotsint}f(t_1,\ldots,t_k)\prod_{i=1}^{k}\vO(t_i)\rd t_1\ldots \rd t_k,
	\end{align}
	where $k=\bigO{\log(1/\eps)}$ and $I$ is a symmetric interval of length $\bigOt{1}$. 
	Consequently, if $\vO$ is quasi-local, and $\vH$ is a geometrically local Hamiltonian, then the operator \eqref{thm:informalOForm} is also quasi-local.
\end{theorem}

Writing \eqref{thm:informalOForm} in the time-domain is technically very involved. In particular, we derive its matrix-valued Taylor series (\Cref{lem:disTaylor}) in the variable $\vO=\CT'^\dagger[\vI]$ and expose a non-obvious time-integral representation over products of Heisenberg evolutions of the jumps~$\vA$. A key object in the derivation is the extensive use of the following superoperator (which is also central to constructing the coherent correction term in the continuous-time setting~\cite{chen2023ExactQGibbsSampler}):
\begin{align}\label{eq:SDefIntro2}
	\CS \colon \!\vA \mapsto \!\int_{\BR} \frac{\ri}{\sinh(2\pi t)}\left(\e^{\ri\vH t} \vA \e^{-\ri\vH t} -\vA\right) \rd t
    = \frac{1}{2\pi}\int_{-\infty}^{\infty} \!\frac{1}{\sinhc(2\pi t)} \int_0^1 \!\e^{\ri \vH s t }[\vA,\vH]\e^{-\ri \vH s t}\rd s \rd t,
\end{align}
where the identity comes from the matrix-valued function $\frac{\rd}{\rd t}(\e^{\ri \vH t} \vA \e^{-\ri \vH t})=\ri\e^{\ri \vH t} [\vH, \vA]\e^{-\ri \vH t}$. Our analysis of $\CS$ and the matrix Taylor series reveals that if $\CT'^\dagger[\vI]$ is quasi-local and normalized by $\|\CT'^\dagger[\vI]\|\leq \bigOt{1}$, then $\vK$ is quasi-local for geometrically local Hamiltonians (using standard Lieb-Robinson bound arguments, see, e.g.,~\cite{chen2023SpeedLimitsAndLocalityQDyn}).

Once we obtain the above time-domain expression, it can be naturally implemented by LCU. It only requires about $\bigOt{1}$ Hamiltonian simulation time, but the LCU-state preparation cost (which is independent of the Hamiltonian) grows rather quickly for the higher-order terms. This is similar to the difficulties arising in higher-order Lindbladian simulation methods~\cite{ding2023SimOpenQSysUsingHamSim}. Alternatively, one may also directly implement $\vK$ by QSVT.

\begin{table}[!ht]
\setlength{\tabcolsep}{4pt} 
\renewcommand{\arraystretch}{2} 
\centering
\begin{tabular}{|c|c|c|}
\hline
 & Continuous-time & Discrete-time \\ 
\hline
Classical $\vR$ & $\vR_{jj}=-\sum_{i\neq j}\vT'_{ij}$ & $\vR_{jj}=1-\sum_{i\neq j}\vT'_{ij}$ \\ 
\hline
Quantum $\CR$ & $-\frac12  \{\CT'^\dagger[\vI], \cdot \}+ [\CS[\CT'^\dagger[\vI]], \cdot]$ & $ \vK[\cdot]\vK^\dagger$,\quad $\vK\!=\sqrt{\sqrt{\vrho}(\vI-\CT'^\dagger[\vI])\sqrt{\vrho}}\vrho^{-\frac12}$\\
\hline
\end{tabular}
\caption{\label{table:DiscCont}
Comparison of classical and quantum constructions for detailed-balanced sampling. Assuming that $\vT'$ and $\CT'$ describes the transitions of the continuous or discrete-time stochastic process, the table summarizes how to ``complete the map'' to ensure both probability-preserving and detailed balance. In quantum continuous-time (see \Cref{lem:FindingCoherenceTerm}), the additional term in the Lindbladian depends linearly on $\CT'$ through the superoperator $\CS$ from \eqref{eq:SDefIntro2}; in discrete-time, the additional term in the quantum channel features non-linear behaviour, see~\Cref{lem:FindingDiscDecayTerm,lem:disTaylor} in \Cref{sec:QDB}. Here $\CS^{\pm}:=\frac{\mathcal{I}}{2}\pm\CS$. Both of our discrete-time constructions agree with the continuous-time construction up to first order, but higher-order terms takes different forms.
As we show, if $\|\CT'^\dagger[\vI]\|\leq \frac12$, the first discrete-time construction $\vK$ inherits the quasi-local nature of the continuous-time construction.
}
\end{table}

\begin{proposition}[QSVT implementation]
The Kraus operator $\vK$ can be implemented using quantum singular value transformation when $\|\CT'^\dagger[\vI]\|\leq \frac{1}{2}$ and $\CT'^\dagger[\vI]$ is energy quasi-local with radius $\sim\frac{1}{\log(1/\eps)}$, which can be readily achieved using an $\bigOt{1}$ overhead by slightly enhancing the resolution of the operator Fourier transform based construction of~{\rm \cite{chen2023ExactQGibbsSampler}}. In order to algorithmically implement $\vK$, we only need access to the operator $\CT'^\dagger[\vI]$ and $\bigOt{1}$ Hamiltonian simulation time.
\end{proposition}

\subsubsection{An alternative recursive construction}

In the above, the rejection term takes an explicit form with an explicit series expansion. In this section, we provide a different way of achieving a detailed balance that may be more friendly to implement. Instead of achieving detailed balance in one shot, we give a recursive construction that uses only the much simpler \textit{first-order} approximation formula of \eqref{thm:informalIntForm}. By \textit{iteratively} fixing the second-order error, we also obtain a sequence of Kraus operators that rapidly converge to a detailed-balanced map.
The Kraus operators get complicated quickly at higher orders of the recursion, but fortunately, the recursion converges so fast that only $\bigO{\log\log(t/\eps)}$-th order matter when the channel is applied $t$ times with $\eps$ target accuracy.

The recursive construction can be written in the Metropolis form $\CT'[\cdot] + \CR'[\cdot]$, with the ``rejection map'' $\CR'$ being a infinite sum of Kraus operators.
\begin{theorem}[Recursively prescribing the reject term]
Consider any $\vrho$-detailed-balanced CP map $\CT'$ and let $\CS^{\pm}:=\frac{\mathcal{I}}{2}\pm\CS$. There is an absolute constant $c>0$ such that if $\|\CT'^{\dagger}[\vI]\| < \frac{1}{4(c+\ln(\nrm{\vH}+1))^2}$, then the following CP map 
\begin{align*}
\CR'[\cdot] &:= \sum_{k=1}^{\infty} \vK_k[\cdot] \vK_k^{\dagger},\qquad\text{where}\qquad \vK_k := 2^{-\frac{k}2}\left(\vI - 2^{2^{k}-1} \CS^-[\vB_k]\right),\\
\vB_1 &:= \CT'^{\dagger}[\vI],\qquad\text{and}\qquad
\vB_{k+1} := (\CS^-[\vB_k])^{\dagger} \CS^-[\vB_k] = \CS^+[\vB_k]\CS^-[\vB_k] \quad\text{for all}~ k\geq 1
\end{align*}
    is $\vrho$-detailed balanced and trace-preserving.
\end{theorem}
While the above recursion may be reminiscent of the infinite series of Kraus operator in the Mariott-Watrous rewinding scheme~\cite{temme2009QuantumMetropolis}, it is qualitatively different as it (i) rapidly converges for any input state and (ii) achieves exact detailed balance in the limit. In particular, as long as $\|\CT'^{\dagger}[\vI]\|\leq \frac{\lambda}{(c+\ln(\nrm{\vH}+1))^2}$ for any $\lambda>0$ and an absolute constant $c>0$, the high-order terms decay \textit{doubly exponentially} as (see~\Cref{apx:RecDisc})
\begin{align*}
\nrm{\vB_{k+1}} &\leq \frac{\lambda^{2^{k}}}{(c+\ln(\nrm{\vH}+1))^2},\\
\norm{\vK_{k+1} - 2^{-\frac{k+1}2}\vI} &\le \frac{(4\lambda)^{2^k}}{c+\ln(\nrm{\vH}+1)} \quad \text{for each}~ k\geq 1.
\end{align*}
Therefore, whenever we start with a transition part satisfying that $\|\CT'^{\dagger}[\vI]\| \leq \frac{\lambda}{(c+\ln(\nrm{\vH}+1))^2}$ with $\lambda < 1/4$, the above series of $\CR'[\cdot]$ is convergent.
Moreover, truncating\footnote{Note that all the resulting new identity Kraus operators can be ``merged''.} $\CR'[\cdot]$ by setting 
\begin{align}
    \vK_k\leftarrow 2^{-\frac{k}2}\vI\quad \text{for all}~ k> t
\end{align}
yields an error (order $\lambda^{2^t}$) decreasing in $t$ (see~\Cref{apx:RecDisc}, \eqref{eq:RestRClose}, and \eqref{eq:ChannelClose} for the ultimate error bounds and truncation scheme). For example, choosing $t=1$, $2$ or $3$ results in 2nd, 4th or 8th order error in $\lambda$, respectively.

Consider as an example the case of $\CT[\cdot] = \vA[\cdot] \vA$ for a single Hermitian jump $\vA=\vA^{\dagger}$ so that $\vB_0:=\vA$. Then, the ultimate discrete detailed-balanced channel construction has the following Kraus operators
\begin{align*}
\CQ[\cdot] &= \CT'[\cdot]+ \CR'[\cdot]=\sum_{k=0}^{\infty}\vK_k[\cdot] \vK_k^\dagger  \\ 
\text{where} \quad \vK_0&=\CS^-[\vA],\\
\vK_1&=\frac{\vI}{\sqrt{2}}-\sqrt{2}\CS^-[\CS^+[\vA]\CS^-[\vA]] \tag{order $\lambda^{2^0}$},\\
\vK_2&=\frac{\vI}{2}-4\CS^-[\CS^+[\CS^+[\vA]\CS^-[\vA]]\CS^-[\CS^+[\vA]\CS^-[\vA]]]\tag{order $\lambda^{2^1}$},\dots
\end{align*}
 The terms up to $\vK_2$ (which gives an order $\lambda^{4}$ residual error) seem reasonable from an implementation perspective since it only requires 5 segments of Heisenberg evolution, and the Kraus operators can be implemented via LCU-based integration over 7 real variables. In general, we prove the following proposition in \Cref{apx:RecDisc}. \AC{TODO: gate count and quasi-locality for the truncated map}
\anote{Oblivious amplitude amplification-based implementation.}\jnote{Done in \Cref{apx:RecDiscImpl}}
\begin{proposition}
    Assume access to a block-encoding of $4\CT'^\dagger[\vI]$, where $\|\CT'^{\dagger}[\vI]\| < \frac{\lambda}{(c+\ln(\nrm{\vH}+1))^2}$. It is possible to implement the truncated CP map $\widetilde{\CQ}_\ell + 2^{-\ell}\CI = \sum_{k=0}^{\ell}\widetilde{\vK}_k[\cdot]\widetilde{\vK}_k^\dagger + 2^{-\ell}\CI$ such that $\|\widetilde{\CQ}_\ell + 2^{-\ell}\CI - \CQ\|_{\Diamond} \leq \varepsilon + \frac{(4\lambda)^{2^\ell}}{c+\ln(\nrm{\vH}+1)}$ with
    \begin{align}
        \bigOt{\log^\ell\left(1 + \frac{\|\vH\|}{\varepsilon}\right)}
    \end{align}
    uses of the block-encoding of $4\CT'^\dagger[\vI]$ and (controlled) Hamiltonian simulation time.
\end{proposition}
 

\subsection{The spectral gap of the discrete-time quantum channel is always larger}

The following shows that all of our discrete-time quantum channels have at least as large spectral gap than the corresponding continuous-time generator under proper normalization.
\begin{theorem}\label{thr:discrete_vs_continuous_spectral_gap}
    Let $\vrho \succ 0$ be a full-rank density operator and $\CT^{(1)},\CT^{(2)}$ be trace-non-increasing CP maps, which satisfy detailed balanced with respect to $\vrho$. Suppose that $\CQ=\CT^{(1)}+\CT^{(2)}$ is a quantum channel and let the detailed-balanced Lindbladians $\CL^{(j)}$ be defined as
    \begin{align*}
        \CL^{(j)}[\cdot] &= \CT^{(j)}[\cdot] - \frac{1}{2}\{(\CT^{(j)})^\dagger[\vI],\cdot\} + [\CS[(\CT^{(j)})^\dagger[\vI]],\cdot].
    \end{align*}
    Then $\CI - \CQ \succeq -\CL^{(1)}$ (with respect to the KMS inner product). Consequently, the spectral gap of $\CQ$ is larger than the spectral gap of $\CL^{(1)}$.
\end{theorem}
\begin{proof}
    This follows from the observation that $\vI=\CQ^\dagger[\vI]=(\CT^{(1)})^\dagger[\vI]+(\CT^{(2)})^\dagger[\vI]$ and $\CS[\vI]=0$:
    \begin{align*}
        \CI - \CQ &= \CI - \CT^{(1)} - \CT^{(2)}
        = -\CL^{(1)} - \CL^{(2)}\succeq -\CL^{(1)}.\tag*{\qedhere}
    \end{align*}
\end{proof}

\subsection{A fully coherent reweighing scheme for achieving exact detailed balance}
Recall that our discrete-time constructions (and the existing continuous-time constructions) use a detailed balanced CP map $\CT'$ as black box. Another contribution of this work is a new way of turning any jumps into a detailed-balanced CP map $\CT'$, which immediately gives rise to a new family of Gibbs samplers. Actually, the observation is that superoperator $\CS$ from \eqref{eq:SDefIntro2} can simply be used for reweighing a self-adjoint CP map to achieve exact detailed balance. Indeed, if $\vA$ is a Hermitian jump (Kraus) operator, then the transformed jump operator 
\begin{align}\label{eq:SMDefIntro}
    \CS^-\colon\vA \mapsto \frac{\vA}{2}-\CS[\vA]
\end{align}
describes a detailed-balanced transition. Moreover, we also show that $\CS$ can be efficiently implemented on a quantum computer as follows.
\begin{restatable}[Approximating the superoperator $\CS$]{theorem}{SImp}\label{thm:SImp}
	Given a block-encoding $\vU\!$ of $\vA$ ($\norm{\vA}\le 1$), we can implement a block-encoding of an approximate, rescaled version of $\CS[\vA]$ given by
	\begin{align}\label{eq:LCUGoal}
		\frac{\pi\widetilde{\CS}[\vA]}{\ln\left(1+\frac{2\nrm{\vH}}{\pi \eps}\right)}\quad \text{such that}\quad \|\CS[\vA]-\widetilde{\CS}[\vA]\|\leq\eps
	\end{align}
	for every $\eps\in(0,\frac13]$, with a single use of $\vU\!$ and $\bigOt{1}\!$ (controlled) Hamiltonian simulation~time.
\end{restatable}

An efficient implementation can be achieved via a truncation of the integral \eqref{eq:SDefIntro2} using the methods of \Cref{apx:ApxInts} and applying LCU.
Similarly to~\cite{chen2023ExactQGibbsSampler}, it is also possible to efficiently block-encode the associated discriminant matrix, providing a frustration-free parent Hamiltonian of the purified Gibbs state, which
then could be used for preparing the purified Gibbs state via, e.g., an adiabatic path inspired by simulated annealing~\cite{wocjan2021SzegedyWalksForQuantumMaps,chen2023QThermalStatePrep,chen2023ExactQGibbsSampler}.
As we show in \Cref{sec:coherentDynDisc}, the associated parent Hamiltonian of the purified Gibbs state is
\begin{align*}
    \pmb{\CH}&=\sum_{\vA}\CS_c[\vA]\otimes(\CS_c[\vA])^*-\CS_c[\vD]\otimes\vI-\vI\otimes(\CS_c[\vD])^*, \quad\text{where}\\
    \CS_c[\vD]&=\frac12\CS_c[\vA^\dagger\vA]+\frac12\CS_c[\CS[\vA^\dagger]\vA]-\frac12\CS_c[\vA^\dagger\CS[\vA]]-\frac12 \CS_c[\vA^\dagger]\CS_c[\vA], \quad\text{and}\\
    \CS_c[\cdot]&=\vrho^{-\frac14}\CS^-[\cdot]\vrho^{\frac14}=\vrho^{\frac14}\CS^+[\cdot]\vrho^{-\frac14}=\int_{\BR} \frac{1}{\cosh(2\pi t)}\e^{\ri\vH t} [\cdot] \e^{-\ri\vH t} \,\rd t.
\end{align*}

A key feature of $\CS^-$ is that it preserves the number of Kraus operators, unlike the earlier constructions using the operator Fourier transform. Moreover, we show how this construction can be generalized by introducing a family of maps $\CS^{[\sqrt{\gamma}]}$, associated to functions $\gamma$ possessing the symmetry $\gamma(-\nu)=\e^{\nu}\gamma(\nu)$, that map self-adjoint jumps to detailed-balanced ones. 
Acting by $\CS^{[\sqrt{\gamma}]}$ on the Kraus operators transforms a self-adjoint CP map to a detailed-balanced one.
This is a very natural generalization of \eqref{eq:SOpGlaub}, as $\CS^{[\sqrt{\gamma}]}$ reduces to $\CS^{[\sqrt{g}]}$ in the classical case, or equivalently $\CS^{[g]}=\CS^{[\sqrt{\gamma}]}\circ \CS^{[\sqrt{\gamma}]}$ for the appropriate $g$. The reason for the square root is that $\CS^{[\sqrt{\gamma}]}$ acts ``twice'', since the jump operators come in pairs $\vA[\cdot]\vA^\dagger$.

Surprisingly, this new algorithmic construction closely resembles an existing physically-derived Lindbladian~\cite{nathan2020UniversalLindlad} (see~\Cref{sec:physMotivation} for an explicit comparison). Thus, both~\cite{chen2023ExactQGibbsSampler} and~\eqref{eq:SMDefIntro} give plausible physically motivated quantum Glauber and Metropolis dynamics.

\subsection{Mixing times}

An interesting question is whether there are qualitative differences between our discrete and continuous-time constructions in practice, e.g., in their mixing times. Towards this end, we prove spectral gap lower bounds for the discrete and continuous-time Glauber dynamics via our fully coherent reweighting scheme in \Cref{sec:spectral_gap}. More specifically, in the continuous-time setting, we start with the local CP map $\sum_{a\in\Lambda}\sum_{\alpha\in[3]}\vA^{a,\alpha}[\cdot]\vA^{a,\alpha}$ defined on a lattice $\Lambda$, where $\vA^{a,1} = \vsigma_{x}$, $\vA^{a,2} = \vsigma_{y}$, $\vA^{a,3} = \vsigma_{z}$ denote the $1$-local Pauli matrices on site $a\in\Lambda$, and turn it into a detailed-balanced CP map by employing the Glauber-inspired function $\gamma_G(\nu) = (\frac{1}{2} - \frac{1}{2}\tanh(\frac{\beta\nu}{4}))^2$ such that $\CS^{[\sqrt{\gamma_G}]} = \CS^-$. The resulting dynamics is converted into a Lindbladian for the continuous-time setting using the results from \Cref{table:DiscCont}. We then prove that such Lindbladian has a constant gap for high temperatures, i.e., up to constant $\beta^\ast$. Our results mirrors the one from~\cite{rouze2024EffThermalizationGibbsSamp}, who showed that the original construction from~\cite{chen2023ExactQGibbsSampler} is also gapped. 
\begin{theorem}\label{thr:spectral_gap_continuous_intro}
    Consider any local Hamiltonian $\vH$ defined on a lattice $\Lambda$ and the detailed-balanced Lindbladian
    \begin{align*}
        \CL_{\beta}[\cdot] &= \sum_{a\in\Lambda}\sum_{\alpha\in[3]} \CS^-[\vA^{a,\alpha}][ \cdot ]\CS^+[\vA^{a,\alpha}] -\frac{1}{2}\{\vD^{a,\alpha},\cdot\}-\ri[\ri\CS[\vD^{a,\alpha}],\cdot],
    \end{align*}
    where $\vD^{a,\alpha} = \CS^+[\vA^{a,\alpha}] \CS^-[\vA^{a,\alpha}]$ and $\vA^{a,1} = \vsigma_{x}$, $\vA^{a,2} = \vsigma_{y}$, $\vA^{a,3} = \vsigma_{z}$.
    There is a constant $\beta^\ast > 0$ independent of $|\Lambda|$ such that, for all inverse temperatures $\beta < \beta^\ast$, the spectral gap of $\CL_{\beta}^\dagger$ is lower bounded by $\frac{1}{2}$.
\end{theorem}

Regarding the discrete-time setting, on the other hand, we start with the local CP map $\frac{1}{3|\Lambda|}\sum_{a\in\Lambda}\sum_{\alpha\in[3]}\frac{\vA^{a,\alpha}}{s}[\cdot]\frac{\vA^{a,\alpha}}{s}$ defined on a lattice $\Lambda$, but now with $1$-local Pauli matrices normalized by $s=O\big(1 + \log(1+\beta\|\vH\|)\big)$. Note that $\|\CT^\dagger[\vI]\| < 1$. Again we turn it into a detailed-balanced CP map by employing the Glauber-inspired function $\gamma_G(\nu)$ and finally into a quantum channel for the discrete-time setting using the results from \Cref{table:DiscCont}. The spectral gap then follows immediately from \Cref{thr:discrete_vs_continuous_spectral_gap,thr:spectral_gap_continuous_intro} (up to the suitable normalization of the initial CP map, in this case, of $\frac{1}{3|\Lambda|s^2}$).
\begin{theorem}\label{thr:spectral_gap_discrete}
    Consider any local Hamiltonian $\vH$ defined on a lattice $\Lambda$ and the detailed-balanced quantum channel
    \begin{align*}
        \CQ_{\beta}[\cdot] = \frac{1}{3|\Lambda|}\sum_{a\in\Lambda}\sum_{\alpha\in[3]}\left(\frac{1}{s^2} \CS^-[\vA^{a,\alpha}][ \cdot ] \CS^+[\vA^{a,\alpha}] + \vK^{a,\alpha}[\cdot](\vK^{a,\alpha})^\dagger \right),
    \end{align*}
    where $\vK^{a,\alpha} = \sqrt{\sqrt{\vrho}(\vI-\vD^{a,\alpha})\sqrt{\vrho}}\vrho^{-\frac12}$ with $\vD^{a,\alpha} = \frac{1}{s^2}\CS^+[\vA^{a,\alpha}]\CS^-[\vA^{a,\alpha}]$ and $s=O\big(1 + \log(1+\beta\|\vH\|)\big)$. There is a constant $\beta^\ast > 0$ independent of $|\Lambda|$ such that, for all inverse temperatures $\beta < \beta^\ast$, the spectral gap of $\CQ_{\beta}$ is lower bounded by $\frac{1}{6|\Lambda|s^2}$.
\end{theorem}

By well-known results relating the spectral gap with mixing times, see, e.g.,~\cite[Proposition~II.3]{chen2023QThermalStatePrep} and~\cite[Corollary~B.2]{rouze2024EffThermalizationGibbsSamp}, the above results imply polynomial mixing times with respect to the lattice size.
\begin{corollary}
    In the setting of {\rm \Cref{thr:spectral_gap_continuous_intro,thr:spectral_gap_discrete}}, there is a constant $\beta^\ast > 0$ such that, for all $\beta < \beta^\ast$, the mixing times of $\CL_{\beta}$ and $\CQ_{\beta}$, i.e.,
    \begin{align*}
        \min\left\{t\in\mathbb{R}_+: \|\e^{t\CL_{\beta}}[\vsigma] - \vrho\|_1 \leq \frac{1}{3} ~\forall \vsigma\right\} \qquad\text{and}\qquad \min\left\{t\in\mathbb{R}_+: \|\CQ^t_\beta[\vsigma] - \vrho\|_1 \leq \frac{1}{3} ~\forall \vsigma\right\},
    \end{align*}
    are upper bounded by $O(|\Lambda|)$ and $O(|\Lambda|^2)$, respectively.
\end{corollary}

\subsection{Structure of the paper}\label{sec:Intro:Structure}
We begin with an introduction to classical and quantum detailed balance in \Cref{sec:MetropolisGlauber,sec:QDB}. Next, we describe and analyse our discrete-time construction in~\Cref{sec:FromCPToDyn}.
We then outline the different known and new constructions for detailed balance ensuring super-superoperators $\BS$ in \Cref{sec:dbcp}, study the uniqueness and convergence to the fixed point $\vrho$ in \Cref{sec:ergodic} and spectral gap lower bounds in \Cref{sec:spectral_gap}, and discuss their implementation details in \Cref{sec:imp}. 
In \Cref{sec:physMotivation}, we discuss the physical motivation of our new constructions. 
In~\Cref{sec:non-selfadjoint} we study some possible extensions to the situation when $\CT$ is non-self-adjoint, in analogy to the well-established classical Metropolis-Hastings algorithm~\cite{hastings1970MonteCarloSamplingUsingMCs}. In \Cref{apx:SchurBounds,apx:PEBasedDB,apx:DiscConst,app:spectral_gap_rapidly_decaying,apx:RecDisc}, we provide the technical details, including the proof of quasi-locality of the operator $\vK$ defined in \eqref{eq:main_K}, and the alternative recursive discrete-time construction. 

\subsection{Concurrent work}
As we were working on our manuscript, we became aware of the concurrent work of Ding, Li, and Lin~\cite{ding2024GibbsSamplingViaKMS}, who independently came to a detailed-balanced Lindbladian construction that is very similar to our coherent reweighing construction  in~\Cref{sec:coherentDyn}. Both works can be seen as generalizations and extensions of the earlier partially joint work~\cite{ding2024SingleAncillaGSPrepViaLindblad}. The core construction in~\cite{ding2024GibbsSamplingViaKMS} appears to be essentially the same as our reweighing superoperator~\eqref{eq:coherentReweighing}.\anote{According to Lin's e-mail, I removed this: "with the difference that they also work with complex weighing functions while we only consider real ones."} It is worth noting that~\cite{ding2024GibbsSamplingViaKMS} has exponentially better discretization error bounds for operator-valued numerical integration compared to ours, which in some cases could significantly reduce the ancilla and state preparation cost compared to our implementation methods.\anote{The following has been fixed in their August arXiv update without citing us after I explained Lin in May in Berkeley roughly how we got the right scaling. What should we do about it? Currently I changed the text as follows:}\AC{Looks good to me. It probably doesn't matter as much as the main novelty is the discrete-time proposal.} 
On the other hand, their bound on the operator norm of the induced coherent term $\vC$ was exponentially worse in early {\tt arXiv} versions of their work~\cite{ding2024GibbsSamplingViaKMSarXiv2}, which prior to the posting of our manuscript made their implementation cost bounds about a $\beta$ factor worse than ours. Fortunately, it turns out that the advantages of the papers can be easily combined.

While preparing the first {\tt arXiv} version of this manuscript, we also learned about the independent work of Jiang and Irani~\cite{jiang2024QMetropolisWeakMeas}, who devised a provably approximate detailed-balanced improvement of the original discrete-time quantum Metropolis algorithm~\cite{temme2009QuantumMetropolis}. Their construction differs from the main approach we follow here, and is conceptually more related to the weak-measurement-based Lindbladian simulation algorithm of~\cite{chen2023QThermalStatePrep} and the phase-estimation-based protocol of~\cite{temme2009QuantumMetropolis}; for a comparison of operator Fourier transform and phase estimation and a related subroutine see \cite[Section~III.C]{chen2023QThermalStatePrep}. We added \Cref{sec:PEVariant} in an update of our manuscript for the sake of completeness to show that the original Metropolis-style construction can be also made exactly detailed balanced.

\subsection{Discussion and open problems}\label{ssec:OpenProbs}
In this work, we set out to give a systematic approach to the quantum generalization of MCMC methods and describe alternative constructions for transforming self-adjoint completely positive maps into detailed-balanced ones. We also show how to naturally turn such maps into detailed-balanced discrete-time quantum channels, complementing the analogous generic continuous-time Lindbladian construction of~\cite{chen2023ExactQGibbsSampler}.
We also study analogous methods applicable to non-self-adjoint transitions $\CT$, however we do not know how to implement them efficiently in general, therefore finding an efficient quantum counterpart of the Metropolis-Hastings algorithm for non-self-adjoint transitions remains an open question.
While classical Metropolis and Glauber dynamics have already blossomed in the theoretical, algorithmic, and practical communities, quantum MCMC methods are still in their very early stages. 

A major open problem in quantum algorithms is to find a construction that is exactly detailed-balanced and features a Szegedy-type quadratic speed-up in the spectral gap~\cite{szegedy2004QMarkovChainSearch}. 
The quadratic speed-up in search problems~\cite{magniez2006SearchQuantumWalk} related to a $\pi$-detailed balanced random walk $\vP$ can be explained by the ability to ``quantum fast forward'' it~\cite{apers2018QFastForwardMarkovChains}, i.e., to mimic the action of $t$ time steps by only about $\propto\sqrt{t}$ quantum operations. 
It turns out that this speed-up relies on the ability to directly get a block-encoding of the so-called discriminant matrix $\diag(\pi^{-\frac12})\vP\diag(\pi^{\frac12})$ without any subnormalization factor.
A key issue for known efficient quantum constructions $\CQ$ is that the decay term needs a special construction in order to ensure (exact) detailed balance, as summarized in~\Cref{table:DiscCont}, which seems to prevent directly block-encoding the quantum discriminant $\big(\vrho^{\,\frac14}\otimes\vrho^{*\frac14}\big)\vCQ\big(\vrho^{-\frac14}\otimes\vrho^{*-\frac14}\big)$ without subnormalization.
Accordingly, the only known constructions that enable such Szegedy-type quadratic speed-ups are devised~\cite{wocjan2021SzegedyWalksForQuantumMaps} for the discriminant of approximately detailed-balanced maps~\cite{chen2023QThermalStatePrep}.

Aside from algorithmic development, perhaps the bigger question is the mixing time of these constructions. While we provided spectral gap lower bounds for our continuous-time construction, which readily imply fast mixing times, i.e., convergence to the Gibbs state in polynomial time on the system size, this was done using the Glauber-inspired function $\gamma_G$. We suspect that similar results can be extended to several other functions and also to our discrete-time constructions. Moreover, after this paper first appeared online, Rouz\'e, Stilck França, and Alhambra~\cite{rouze2024optimal} showed that the original construction from~\cite{chen2023ExactQGibbsSampler} mixes rapidly, i.e., in logarithmic time on the system size. We believe that their results can be extended to our constructions. 
Nonetheless, it remains to be seen which construction (and which heuristics variant) will be more practically favorable as improved quantum hardware is gradually becoming available. We expect that our newly introduced continuous and discrete-time constructions have features that make them more suitable in particular use-cases, and it is a priori unlikely that one of them strictly outperforms the others in all aspects. Therefore, we believe it is useful to study these variants in general, enabling one to choose the most fitting for a given particular application.

We also highlight the connection of MCMC methods to physical thermodynamic models of open quantum systems and find it conceptually pleasing that the continuous-time constructions have physical counterparts in the existing literature. Still, the master equations we quoted are limited by the current understanding of open system physics, particularly to Markovian, weak-coupling settings where long-time, non-Markovian effects are neglected. We do not know to what extent they quantitatively capture real-world physics with varying conditions, but we do expect the precise algorithmic constructions to lay new grounds for new open-system physics.

\section*{Acknowledgments}
We thank Pierfrancesco Dionigi, Daniel Stilck França, Balázs Kabella, 
Martin Hairer, Lin Lin, József Mák, Tibor Rakovszky, Frederik Nathan, Márió Szegedy, and Zoltán Zimborás for inspiring discussions. 
This work was done in part while the authors were visiting the Simons Institute for the Theory of Computing, supported by DOE QSA grant \#FP00010905. \anote{I removed it because it already appears on the front page: CFC is supported by a Simons-CIQC postdoctoral fellowship through NSF QLCI Grant No.\ 2016245.}

\section{Classical Metropolis sampling and Glauber dynamics}\label{sec:MetropolisGlauber}

To explain the analogy and connection between MCMC methods and open quantum system dynamics we first review the Metropolis algorithm and Glauber dynamics and their continuous-time counterparts. In the following, given $n\in\mathbb{N}:=\{1,2,\dots\}$, let $[n]:= \{1,\dots,n\}$.

\subsection{Discrete-time Metropolis sampling and Glauber dynamics}

Suppose that we are given a positive vector $v\in \BR_+^d$ over $d$ elements, and a symmetric stochastic matrix $\vP\in\BR^{d\times d}$. The algorithm devised by Nicholas Metropolis, Arianna W. Rosenbluth, Marshall N. Rosenbluth, Augusta H.
Teller, and Edward Teller~\cite{metropolis1953StateCalculationsByComputers} sets out to sample from the distribution $\pi\propto v$ using a time-independent Markov Chain with modified transition matrix $\vP'$ whose matrix elements are defined as
\begin{equation}\label{eq:ClDiscMetropolis}
	 \vP'_{ij}:=\vP_{ij}\min\bigg(1,\frac{v_i}{v_j}\bigg) \quad \text{for each}~ i\neq j,
\end{equation}
while $\vP'_{jj}$ is defined such that $\vP'$ remains (column) stochastic, i.e., $\sum_{i=1}^d \vP'_{ij} = 1$ for all $j\in[d]$. 

The stochastic process described by $\vP'$ can be intuitively understood as follows: first one makes a random transition from the current state $j$ according to $\vP$; if the new state $i$ is such that $\pi_i\geq\pi_j$ then the move is always accepted, otherwise, if $\pi_i<\pi_j$ then it is accepted only with probability $\frac{v_i}{v_j}=\frac{\pi_i}{\pi_j}$; when a move is rejected we go back to the state $j$ before the move. Crucially, in order to define and implement the process it suffices to know the unnormalized vector $v$.

The beauty of this construction is that if the current distribution is $\pi$, then in a step by $\vP'$ the amount of probability mass exchange between pairs $i,j$ is equal to
\begin{equation}\label{eq:ClDetBalance}
	\vP'_{ij}\pi_j = \vP'_{ji}\pi_i.
\end{equation}
This is known as the \emph{detailed balance} condition with respect to the distribution $\pi$; it is easy to see that the probability mass of any state is stationary and thus $\pi$ is a fixed point $\vP'\pi=\pi$. 

Generalizing the above construction, one can define a $\pi$-detailed balanced $\vP'$ as follows
\begin{equation}\label{eq:genMetropolis}
	\forall i\neq j\colon \vP'_{ij}:=\vP_{ij}\gammaprob(v_i/v_j),\quad \text{where }  \gammaprob\colon\BR_+\rightarrow [0,1] \text{ satisfies } \gammaprob(r)=r\cdot\gammaprob(1/r) \text{ for all } r>0.
\end{equation}
Choosing $\gammaprob_M(r)=\min(1,r)$ we obtain the Metropolis rule \eqref{eq:ClDiscMetropolis} while $\gammaprob_G(r)=1/(1+1/r)$ leads to the well-known Glauber weights often used in statistical physics.\footnote{Although Glauber originally introduced a continuous-time Markov chain~\cite{glauber1963TimeDepIsing}, the main idea applies equally well, and is indeed often used, in the discrete-time setting.}

We say that a transition matrix $\vM$ is \emph{ergodic}\footnote{In the literature this property is often called irreducible, but here we follow the naming convention of~\cite{burgarth2013ErgodicQChannels}.} iff for every non-empty proper subset $V\subset[d]$ there are some states $v\in V$ and $w\in [d]\setminus V$ such that there is a non-zero $v$ to $w$ transition. 
For a stochastic matrix $\vP$ this is also equivalent to saying that there is a unique stationary distribution with full support. 
It is easy to see that if $\pi$ is not the trivial uniform distribution, $g>0$, and $\vP$ is ergodic, then $\pi$ is the unique stationary distribution of $\vP'$ and the process $(\vP')^t$ converges to $\pi$ as $t$ goes to infinity. 
The speed of convergence is of course another question, which has to be addressed on a case-by-case basis, and has a large and well-developed literature in the classical case~\cite{levin2017MarkovChainsMixingTimes}.

Wilfred K.\ Hastings~\cite{hastings1970MonteCarloSamplingUsingMCs} extended the construction of~\cite{metropolis1953StateCalculationsByComputers} to the case of non-symmetric $\vP$. We review his generalized construction in \Cref{sec:non-selfadjoint}.

\subsection{Continuous-time Metropolis and Glauber dynamics}

So far we considered discrete-time Markov chains for sampling from a distribution $\pi\propto v\in \BR_+^d$. However, the discrete-time constructions have natural continuous-time counterparts which are more natural from the perspective of quantum generalizations. In this section we review the continuous-time Metropolis and Glauber dynamics.

A continuous-time Markov chain has transition matrix $\e^{t \vL}$, where $\vL\in \BR^{d\times d}$ is a \emph{Laplacian} matrix such that $\vL_{ij}\geq 0$ for all $i\neq j$ and $\vL_{ii}=-\sum_{j\neq i} \vL_{ji}$. As before we will assume that $\vL=\vL^T$ is symmetric, which then also implies that $\e^{t \vL}$ is symmetric. Analogously to \eqref{eq:ClDiscMetropolis}-\eqref{eq:genMetropolis} we can define the modified infinitesimal generator $\vL'$ for any $\gammaprob\colon\BR_+\rightarrow \BR_+$ satisfying \eqref{eq:genMetropolis} as
\begin{equation}\label{eq:ClContMetropolis}
	\vL'_{ij}:=\vL_{ij}g(v_i/v_j) \quad \text{for each}~ i\ne j
\end{equation}
and setting $\vL'_{ii}=-\sum_{j\neq i} \vL'_{ji}$ to get a valid Laplacian.

Once again, this construction ensures that if the current distribution is $\pi$, then in an infinitesimal step by $\e^{\delta t \vL'}=I-\delta t\vL'+\bigO{(\delta t)^2}$ the amount of probability mass exchange between pairs $i,j$ is equal because
\begin{equation*}
	\vL'_{ij}\pi_j = \vL'_{ji}\pi_i.
\end{equation*}
This detailed balance in the continuous-time setting once again ensures that $\pi$ is a stationary state. Similarly to the discrete case, if $g>0$ and $\vL$ is ergodic, then $\pi$ is the unique fixed point and the process $\e^{t \vL}$ converges to $\pi$ as $t$ goes to infinity.

\begin{example}\label{expl:Metropolis}
	Let us choose $\gammaprob:=\gammaprob_M(r)=\min(1,r)$, and consider a single-bit system with $\pi:=(\frac{\exp(\eps)}{2\cosh(\eps)},\frac{\exp(-\eps)}{2\cosh(\eps)})^T$ for some $\eps\geq 0$ together with a ``bit-flip generator''
	\begin{align*}
		\vL&:=\left(\begin{array}{rr}\!\!-1 & 1 \\ 1 & \!-1\end{array}\right) \text{, then }&
		\e^{t\vL}&=\frac{1}{1+\tanh(t)} \left(\begin{array}{cc} 1 & \tanh(t) \\ \tanh(t) & 1\end{array}\right)\text{, }\\
		\vL'& = \left(\begin{array}{rr}\!\!-\e^{-2\eps} & 1 \\ \e^{-2\eps} & \!-1\end{array}\right)
		\text{, and }&
		\!\e^{t\vL'}\!&=\frac{1}{2\cosh(\eps)}\!\left(\begin{array}{cc}\e^{\eps} & \e^{\eps} \\ \!\e^{-\eps} & \!\e^{-\eps}\!\end{array}\right)\! + \frac{\e^{-2t\cosh^2(\eps)(1-\tanh(\eps))}}{2\cosh(\eps)}\left(\begin{array}{rr}\e^{-\eps} & \!-\e^{\eps}\! \\ \!\!\!-\e^{-\eps}\! & \e^{\eps}\!\end{array}\right).\!
	\end{align*}
\end{example}

\section{Quantum dynamics and detailed balance}\label{sec:QDB}

In this section we describe quantum counterparts of discrete-time Markov chains (stochastic matrices) called quantum channels, and continuous-time Markov chains (and their Laplacian generators) called quantum dynamical semigroups generated by Lindbladians. We also define quantum detailed balance and show how it can be used for prescribing a desired fixed point.

\subsection{Quantum channels and Lindbladians}

In the classical case, stochastic operations act linearly on distributions $p\in[0,1]^d$, so they can be described by $\vM\in \BR^{d\times d}$ matrices. 
The quantum counterpart of a distribution on $d$ elements is called a \emph{density operator} $\vrho\in \BC^{d\times d}$ which is a positive semi-definite matrix of unit trace. Density operators generalize the notion of distributions, and a classical distribution $p$ corresponds to the diagonal density operator $\diag(p)$. The quantum counterparts of stochastic operations also act linearly, but on the space of density operators, so they can be described by \emph{superoperators} $\CM\colon \BC^{d\times d}\rightarrow \BC^{d\times d}$.

Stochastic transitions between states can only happen with non-negative probability, that is why stochastic matrices have non-negative elements and Laplacians have non-negative off-diagonal elements. The analogous concept for superoperators is captured by the notion of \emph{completely positive maps} which are defined as follows.

\begin{definition}[Completely positive maps]
	A superoperator $\CM\colon \BC^{d\times d}\rightarrow \BC^{d\times d}$ is called \emph{completely positive} (CP) if $\CM\otimes\CI$ maps positive semi-definite operators to positive semi-definite operators, where $\CI\colon \BC^{d\times d}\rightarrow \BC^{d\times d}$ is the identity map.
	Equivalently $\CM$ is CP iff there are Kraus operators $\{\vA^a\in \BC^{d\times d}\}_{a=1}^{d^2}$ such that $\CM[\cdot]=\sum_{a=1}^{d^2}\vA^a[ \cdot ]\vA^{a \dagger}$. This decomposition is not unique and the minimum number of non-zero Kraus operators in it is called the \emph{Kraus rank} of $\CM$.
\end{definition}

A stochastic matrix $\vP\in\BR^{d\times d}$ can describe any probabilistic transformation on $d$ elements. Its quantum counterpart, called a \emph{quantum channel} $\CQ\colon \BC^{d\times d}\rightarrow \BC^{d\times d}$, is a CP map such that $\CQ^\dagger[\vI]=\vI$, i.e., it is trace preserving.\footnote{The adjoint of a superoperator $\vA[\cdot] \vB^\dagger$ is simply $\vA^\dagger[\cdot ]\vB$ with respect to the Hilbert-Schmidt inner product $\ipc{\vM}{\vN}_{HS}=\tr(\vM^\dagger\vN)$. Then $\tr[\CQ[\vrho]]=\ipc{\vI}{\CQ[\vrho]}_{HS}=\langle\CQ^\dagger[\vI],\vrho\rangle_{\!HS}=\ipc{\vI}{\vrho}_{HS}=\tr[\vrho]$.} Similarly, a quantum channel can describe any quantum operation on a $d$-dimension state space.

A continuous-time counterpart to a quantum channel is a \emph{quantum dynamical semigroup} $\e^{t\CL}$ which describes a family of quantum channels parametrized by $t\geq 0$ generalizing continuous-time Markov chains. The generator $\CL\colon \BC^{d\times d}\rightarrow \BC^{d\times d}$ is called a \emph{Lindbladian} and can be written as 
\begin{equation}
	\CL[\cdot] =  \CT[\cdot]-\frac12\{\CT^\dagger[\vI],\cdot\} - \ri [\vC,\cdot], 
\end{equation}
where $\CT\colon \BC^{d\times d}\rightarrow \BC^{d\times d}$ is a CP map that describes ``transitions'' or ``jumps'', $\vD=\CT^\dagger[\vI]$ in the anti-commutator term $\frac12\{\CT^\dagger[\vI],\cdot\}$ describes decay and corresponds to the diagonal entries of Laplacian matrices, while $\vC\in \BC^{d\times d}$ is a Hermitian matrix that describes the coherent part of the evolution and has no classical counterpart. We say that $\CL$ is purely irreversible if $\vC=0$.
To connect this concept to its classical counterpart note that the purely irreversible Lindbladian corresponding to the Laplacian $\vL$ is described by the transition map 
\begin{equation}\label{eq:LapToLind}
	\CT[\cdot]=\sum_{i\neq j}\vL_{ij}\ketbra{i}{j}[ \cdot ]\ketbra{j}{i}.
\end{equation}
Lindbladians are widely used objects in the description of open quantum system dynamics where the system of interest is weakly coupled to the environment.

Throughout this paper $ 0 \prec \vrho\in \BC^{d\times d}$ is going to denote a full-rank target density operator, which we can also write as $\vrho\propto\exp(-\vH)$ where $\vH$ is Hermitian. We will use the following eigendecomposition throughout
\begin{align}\label{eq:eigendecomposition}
	\vH=\sum_{j=1}^d \energy_j \ketbra{\psi_j}{\psi_j}\quad\text{where}\quad \energy_i\geq\energy_j \quad\text{for each}\quad i\leq j,
\end{align} and call $B(\vH):=\operatorname{spec}(\vH)-\operatorname{spec}(\vH)$ the set of its \textit{Bohr frequencies}.
For a matrix $\vM\in \BC^{d\times d}$ we introduce the decomposition $\vM=\sum_{\nu\in B(\vH)} \vM_\nu$, where 
\begin{align}\label{eq:EnergyDiffDecomposition}
    \vM_\nu := \sum_{i,j\in[d] : \energy_i-\energy_j = \nu} \bra{\psi_i}\vM\ket{\psi_j}\cdot\ketbra{\psi_i}{\psi_j}.
\end{align}
Intuitively speaking $\vM_\nu$ is the part of  $\vM$ that induces $\nu$ energy increase according to $\vH$.

\subsection{Quantum detailed balance}
Detailed balance is at the heart of Metropolis and Glauber dynamics for ensuring convergence to the target distribution, and it plays an equally important role in our constructions.

In order to generalize the notion of detailed balance it is useful to express classical detailed balance~\eqref{eq:ClDetBalance} as the following matrix identity
\begin{equation}
	\vM\diag(\pi) = \diag(\pi)\vM^T \Leftrightarrow \diag(\pi^{-\frac12})\vM\diag(\pi^{\frac12}) = \diag(\pi^{\frac12})\vM^T \diag(\pi^{-\frac12}).
\end{equation}
Therefore, we can see that the operator $\vM\in \BR^{d\times d}$ is detailed balanced with respect to the classical distribution $\pi\in(0,1]^d$ iff $\diag(\pi^{-\frac12})\vM\diag(\pi^{\frac12})$ is symmetric. We can define the notion of \emph{quantum detailed balance} along the same lines.

\begin{definition}[Detailed balance condition]\label{def:KMSDetBalance}
	Given a full-rank density matrix $0\prec \vrho  \in\BC^{d\times d}$, we say that the superoperator $\CM$ is $\vrho$-detailed balanced (in the KMS sense) if
	\begin{align*}
		\vrho^{-\frac14}\CM[\vrho^{\frac14}\cdot \vrho^{\frac14}]\vrho^{-\frac14} = \vrho^{\frac14}\CM^\dagger[\vrho^{-\frac14}\cdot \vrho^{-\frac14}]\vrho^{\frac14}.
	\end{align*}
\end{definition}

Similarly to the classical case, it is easy to see that if a quantum channel $\CQ$ is $\vrho$-detailed balanced, then $\vrho$ is a fixed point
\begin{align*}
	\CQ[\vrho] 
	= \vrho^{\frac14}(\vrho^{-\frac14}\CQ[\vrho^{\frac14} \sqrt{\vrho} \vrho^{\frac14}]\vrho^{-\frac14})\vrho^{\frac14}
	&= \vrho^{\frac14}(\vrho^{\frac14}\CQ^\dagger[\vrho^{-\frac14} \sqrt{\vrho} \vrho^{-\frac14}]\vrho^{\frac14})\vrho^{\frac14}
	=\sqrt{\vrho}\CQ^\dagger[\vI]\sqrt{\vrho}
	=\vrho.
\end{align*}
The same argument applied to a $\vrho$-detailed-balanced Lindbladian $\CL$ shows that $\CL[\vrho]=0$ implying that $\vrho$ is a fixed point of $\e^{t\CL}=\CI + t\CL + \frac{t^2}{2}\CL\circ \CL + \sum_{k=3}^\infty \frac{t^k\CL^{\circ k}}{k!}$.

A related important concept is the discriminant of a $\vrho$-detailed-balanced CP map $\CM$, 
\begin{align*}
	\CD[\cdot] =  \vrho^{-\frac14}\CM[\vrho^{\frac14}\cdot\vrho^{\frac14}]\vrho^{-\frac14},
\end{align*}
whose matrix can be obtained by a similarity transformation from that of $\CM$ as
\begin{align}
	\vCD =  \big(\vrho^{\,\frac14}\otimes\vrho^{*\frac14}\big)\vCM\big(\vrho^{-\frac14}\otimes\vrho^{*-\frac14}\big).
\end{align}
Observe that $\vrho$-detailed balance implies that $\vCD$ is self-adjoint, and thus $\vCM$, which is co-spectral with $\vCD$, has a real spectrum.
A very nice feature of the discriminant matrix is that if $\CM$ is CP, then a top eigenvector of $\vCD$ is the purification $\sqrt{\vrho}= \sum_i \e^{- E_i/2} \ket{\psi_i} \otimes \ket{\psi_i^*}/\sqrt{\tr[\vrho]}$ (expressed in the eigenbasis \eqref{eq:eigendecomposition}). If $\CM$ is also ergodic (see~\Cref{sec:ergodic}), then $\sqrt{\vrho}$ is the unique top eigenvector.\footnote{This feature makes $\vCD$ a very nice tool, called a parent-Hamiltonian, for preparing the purification $\sqrt{\vrho}$ of the density operator $\vrho$ on a quantum computer using, e.g., simulated annealing~\cite{wocjan2021SzegedyWalksForQuantumMaps,chen2023QThermalStatePrep,chen2023ExactQGibbsSampler}.}

\subsection{Constructing detailed-balanced Lindbladians and quantum channels}\label{sec:FromCPToDyn}

We can follow the classical strategy for constructing detailed-balanced Lindbladians by focusing on ``balancing'' the CP transition $\CT$.\footnote{Note that essentially this is the only working strategy as $\CT$ has to be balanced due to \cite[Lemma 9]{ding2024GibbsSamplingViaKMS}.} However, the issue is that even if one finds a quantum analogue of the transformations \eqref{eq:ClDiscMetropolis}-\eqref{eq:ClContMetropolis} for constructing a $\vrho$-detailed-balanced $\CT'$, the resulting Lindbladian $\CT'[\cdot]-\frac12\{\CT'^\dagger[\vI],\cdot\}$ is $\vrho$-detailed-balanced iff $\vD=\CT'^\dagger[\vI]$ commutes with $\vrho$. This is a strong requirement that is hard to satisfy. 

The issue that $\vD=\CT'^\dagger[\vI]$ may not commute with $\vrho$ is a genuinely quantum problem, because in the classical case $\vrho=\diag(\pi)$ and $\vD=\frac12\sum_i \vL_{ii}\ketbra{i}{i}$ always commute. The key observation of~\cite[Lemma II.1]{chen2023ExactQGibbsSampler} was that this quantum issue has a genuinely quantum solution via an appropriate choice of the coherent term $\vC$.

\begin{lemma}[Prescribing the coherent term {\cite[Lemma II.1]{chen2023ExactQGibbsSampler}}]\label{lem:FindingCoherenceTerm}
For any full-rank density operator $0\prec \vrho \in\BC^{d\times d}$ and Hermitian operator $\vD \in\BC^{d\times d}$, there is a unique Hermitian operator $\vC \in\BC^{d\times d}$ (up to adding any scalar multiples of the identity $\vI$) such that the superoperator
\begin{align}
	-\frac12\{\vD,\cdot\} -\ri[\vC,\cdot] 
\end{align}
satisfies $\vrho$-detailed balance. For a Gibbs state $\vrho\propto\exp(-\vH)$ we can express $\vC$ as
\begin{align}\label{eq:CoherentTermRecipe}
	\vC=\sum_{\nu\in B(\vH)}\frac{\ri}{2}\tanh\left(\frac{\nu}{4}\right) \vD_\nu.
\end{align}
\end{lemma}

At the core of the above lemma are the superoperators
\begin{align}\label{eq:CS}
\CS[\cdot] \colon \vA\mapsto \sum_{\nu \in B(\vH)} \frac12\tanh\left(\frac{\nu}{4}\right)\vA_\nu \qquad\text{and}\qquad \CS^{\pm}[\cdot]\colon\vA \mapsto \frac{\vA}{2}\pm\CS[\vA],   
\end{align}
which feature in our other constructions as well. To better understand the coherent term $\vC=\ri \CS[\vD]$ we study the map $\CS$ and prove the following result.
\begin{lemma}[Integral representation and general bounds for $\CS$, cf.\ {\Cref{prop:integral} \& \Cref{cor:concreteBounds}}]\label{lem:SIntRep}
The superoperator $\CS$ in~\eqref{eq:CS} can be written in the time-domain by
		\begin{align*}
			\CS[\vA] &=\lim_{\theta\rightarrow 0+} \int_{\BR\setminus[-\theta,\theta]} \frac{\ri}{\sinh(2\pi t)}\e^{\ri\vH t} \vA \e^{-\ri\vH t} \rd t\\
			&=
			\ri \int_{\BR} \left( \frac{1}{\sinh(2\pi t)}-\frac{\pmb{1}_{[-\theta,\theta]}(t)}{2\pi t}\right)\e^{\ri\vH t} \vA \e^{-\ri\vH t} \rd t\nonumber\\&
			~~~+\frac{\ri}{\pi} \int_{0}^\theta \cos(\vH t) \vA\vH \sinc(\vH t) - \sinc(\vH t) \vH\vA \cos(\vH t)  \rd t.
		\end{align*}
	Moreover, for every $\vA\in\BC^{d\times d}$ with $\nrm{\vA}=1$, we have that 
	\begin{align*}
		\nrm{\CS[\vA]}\leq \CO\L(\log(\nrm{\vH} + 2)+\log{d}\R). 
	\end{align*}
\end{lemma}

One of the nice implications of the above time-domain representation is that $\CS$ approximately preserves the geometrical quasi-locality of operators if $\vH$ is a geometrically local Hamiltonian due to standard Lieb-Robinson bound arguments.

Now we also prove a discrete-time analog of the above~\Cref{lem:FindingCoherenceTerm} and show how to turn a completely positive map into a $\vrho$-detailed-balanced quantum channel. One option is, of course, to take the unit-time evolution $\e^{\CL}$ or even $\frac{1}{k}\sum_{i=1}^k\e^{\CL_i}$ where $\CL=\frac{1}{k}\sum_{i=1}^k\CL_i$, but we would like a direct construction that preserves the structure of the transition part $\CT'$. 

Just as in the continuous setting \cite{chen2023ExactQGibbsSampler}, the construction turns out to be unique in some sense. In order to satisfy trace preservation, we need to ensure that $\CT'^\dagger[\vI]\preceq \vI$ (if this condition is not met, we can first rescale $\CT'$).
As a first attempt, let $\vK':=\sqrt{\vI-\CT'^\dagger[\vI]}$, then clearly $\CT'[\cdot]+\vK'[\cdot]\vK'^\dagger$ is a quantum channel, but similarly to the continuous-time case it is $\vrho$-detailed balanced iff $\vD=\CT'^\dagger[\vI]$ commutes with~$\vrho$. This is once again a genuinely quantum issue, as these operators always commute when they come from classical Markov chains. Fortunately, similarly to the continuous-time case this has an almost unique unitary fix (\Cref{apx:DiscConst}).

\begin{restatable}[Prescribing a Reject term]{lemma}{disKraus}\label{lem:FindingDiscDecayTerm}
	For any full-rank density operator $0\prec \vrho \in\BC^{d\times d}$ and $\vrho$-detailed-balanced trace-non-increasing completely positive map $\CT'\colon \BC^{d\times d}\rightarrow \BC^{d\times d}$, 
	there is a Kraus operator $\vK$ such that 
 \begin{align}
 \CQ[\cdot]:=\CT'[\cdot] + \vK[\cdot]\vK^\dagger    
 \end{align}
  is a $\vrho$-detailed-balanced quantum channel. Moreover, such a $\vK$ can always be written as 
  \begin{align}
      \vK=\e^{\ri\varphi}\vW\vU \vV^\dagger\sqrt{\vI-\CT'^\dagger[\vI]},
  \end{align}
  where $\vV\Sigma \vW^\dagger = \sqrt{\vI-\CT'^\dagger[\vI]}\sqrt{\vrho}$ is a singular value decomposition. The $\e^{\ri \varphi}$ is an arbitrary phase ($\varphi \in \BR$) and $\vU$ is an arbitrary Hermitian unitary that commutes with $\Sigma$.
\end{restatable}

It is possible to interpret the final quantum channel $\CT'[\cdot] + \vK[\cdot]\vK^\dagger$ as the result of a map $\mathbb{S}$ from trace-non-increasing CP superoperators to quantum channels, i.e., $\BS\llbracket\CT'\rrbracket[\cdot] = \CT'[\cdot] + \vK[\cdot]\vK^\dagger$. The map $\mathbb{S}$ is thus a \emph{super-superoperator}.

Clearly, in the trivial case when $\CT'=0$ and thus $\CT'^\dagger[\vI]=0$, any Hermitian unitary $\vU$ is an artificial choice other than $\vI$, which leads to $\vK = \vI$ if $\varphi = 0$. If we also demand the choice of $\vK$ to be a continuous function of $\CT'$, then this uniquely pinpoints the natural choice to be $\vU=\vI$ for $\vK$ (while the global phase $\varphi$ is irrelevant).

We deem the other choices of $\vU$ unphysical as we expect that by taking $\CT'$ to $\delta \CT'$ with $\delta$ tending to $0$ the corresponding continuous-time evolution should be recovered up to lower order error, i.e., $\e^{\delta \CL_{\CT'}}= \delta \CT' + \vK_{(\delta\CT')}[\cdot]\vK_{(\delta\CT')}^\dagger +  \bigO{\delta^2}$, where $\CL_{\CT'}$ is the Lindbladian prescribed by \Cref{lem:FindingCoherenceTerm}.
Once again it is not difficult to see that this only holds if we choose $\vU=\vI$, cf.\ \Cref{table:DiscCont} and \Cref{lem:disTaylor}.
Another argument for the non-physical nature of the other choices of $\vU$ is that there does not seem to be any particular reason for the resulting Kraus operator to preserve quasi-locality unless $\vU=\vI$.  

Thus, the natural choice in the above lemma is $\varphi=0$, $\vU=\vI$, and the corresponding Kraus operator is $\vK=\vW\vV^\dagger\sqrt{\vI-\CT'^\dagger[\vI]}=\sqrt{\sqrt{\vrho}(\vI-\CT'^\dagger[\vI])\sqrt{\vrho}}\vrho^{-\frac12}$,
for which we also derive a convergent power series representation in \Cref{apx:DiscreteKrausPower}. In particular, $\vK$ remains quasi-local as long as $\CT'^\dagger[\vI]$ is quasi-local and $\|\CT'^\dagger[\vI]\|$ is small enough. 
\begin{restatable}[A matrix Taylor expansion with square-roots]{lemma}{disKrausTaylor}\label{lem:disTaylor}
    If a Hermitian operator $\vO$ satisfies $\|\vO\|\leq \frac{1}{4\|\CS^-\|_{\infty\to\infty}^2} \leq 1$, then
	\begin{align}\label{eq:optimisticSeries}
		\sqrt{\sqrt{\vrho}(\vI-\vO)\sqrt{\vrho}}\vrho^{-\frac12} = \vI - \sum_{k=1}^{\infty} (\CS^{-}\circ \nabla^{(k)})[\vO],
	\end{align}
	where $\nabla^{(1)}:=\CI$ is the identity map and the higher-order terms $\nabla^{(k)}$ (which are homogeneous degree-$k$ maps)\footnote{As the notation suggests, these are related to the higher-order Fréchet derivatives, but we multiplied them by the $\frac{1}{n!}$ factor to make it convenient for our power series computations.} can be obtained for $k\geq 2$ recursively by the formula
	\begin{align}\label{eq:DDef}
		\nabla^{(k)}[\vO]&:= \sum_{p=1}^{k-1} (\CS^{+}\circ \nabla^{(p)})[\vO] \cdot (\CS^{-}\circ \nabla^{(k-p)})[\vO],\quad \text{where} \quad
       \CS^{\pm}:=\frac{\CI}{2}\pm\CS,
	\end{align}	
 with $\CS$ as in~{\rm\autoref{lem:SIntRep}}.
\end{restatable}

\Cref{lem:FindingDiscDecayTerm} only tells how to choose a single Kraus operator to complement the CP map $\CT'$. When adding a collection of Kraus operators one has much more flexibility. 
One such construction is to take a decomposition $\CT'=\sum_{i=1}^{k}\CT'_i$, construct $\CQ_i$ for each CP map $k\CT'_i$ individually using \Cref{lem:FindingDiscDecayTerm}, and then define $\CQ=\frac{1}{k}\sum_{i=1}^{k}\CQ_i$. Depending on the decomposition one might need to introduce some subnormalization to satisfy that $k\CT_i'^\dagger[\vI]\preceq \vI$.

Indeed, a nice and practical choice for $\CT=\BS\llbracket\CT_0\rrbracket$ is when $\CT_0$ is a uniform convex combination of few-site Pauli jumps. In that case, we get a nice detailed-balanced channel by first picking one of the jumps randomly and then implementing the detailed-balanced discrete quantum channel corresponding to the single jump's detailed-balanced version obtained through $\BS$. In case we have a Gibbs state corresponding to a geometrically local Hamiltonian, then since everything is quasi local, to implement such a step it suffices to only act on a local patch. This approximately recovers the nice local update feature of the discrete-time classical Metropolis algorithm. 

\section{Constructing detailed-balanced completely positive maps}\label{sec:dbcp}

As discussed in \Cref{sec:FromCPToDyn}, once we have a $\vrho$-detailed-balanced completely positive map we can turn it into a $\vrho$-detailed-balanced Lindbladian or quantum channel according to the recipes of \Cref{lem:FindingCoherenceTerm}-\ref{lem:FindingDiscDecayTerm} in a more-and-less unique way. Thus, in order to get quantum analogues of Glauber and Metropolis dynamics it suffices to describe transformation for turning a completely positive map $\CT$ to one that is also $\vrho$-detailed balanced. Similarly to the classical we will assume that $\CT$ is self-adjoint to begin with.
We can call these transformations ``\emph{super-superoperators}'', as they map self-adjoint superoperators $\CT$ to $\vrho$-detailed-balanced ones $\CT'$ linearly.

Although we can always express any superoperator $\CT\colon \BC^{d\times d}\rightarrow \BC^{d\times d}$ using at most $|A| \leq d^2$ Kraus operators as $\CT[\cdot]=\sum_{a\in A}\vA^a[ \cdot ]\vA^{a\dagger}$ \cite[Theorem 2.1]{wolf2012QChannelsOpsLectureNotes}, this representation is generally not unique. 
For this reason it is often more convenient to work with the vectorization $\vCT\in\BC^{d^2\times d^2}$ of $\CT$ where $\vCT_{ij,k\ell}=\bra{i}\CT[\ketbra{k}{\ell}]\ket{j}=\sum_{a\in A}\bra{i}\vA^a\ketbra{k}{\ell}\vA^{a\dagger}\ket{j}$, 
i.e.,\footnote{This is closely related to the Choi-Jamiołkowski representation of maps~\cite[Proposition 2.1]{wolf2012QChannelsOpsLectureNotes}, which is commonly defined using an unnormalized maximally entangled state $\ket{\Omega}=\sum_{i=1}^d\ket{i}\ket{i}$ as $\left(\CT\otimes \CI\right)[\ketbra{\Omega}{\Omega}]$ whose $ik,j\ell$ matrix element is $\bra{ik}\left(\CT\otimes \CI\right)[\ketbra{\Omega}{\Omega}]\ket{j\ell}=\sum_{k',\ell'=1}^d\bra{ik}\left(\CT\otimes \CI\right)[\ketbra{k'k'}{\ell'\ell'}]\ket{j\ell}=\bra{i}\CT[\ketbra{k}{\ell}]\ket{j}=\vCT_{ij,k\ell}$.\label{foot:Choi}} $\vCT=\sum_{a\in A}\vA^a\otimes\vA^{a*}$, where $\cdot^*$ denotes complex conjugation.\footnote{
In general, for a superoperator $\CM[\cdot]\!=\!\sum_{a\in A}\!\vA^a[\cdot]\vB^a$, $\vCM\!=\!\sum_{a\in A}\!\vA^a \otimes\vB^{aT}$ and the matrix of $\CM^\dagger$ is $\vCM^\dagger$.\label{foot:superAdjoint}}
Due to linearity clearly $\CT$ and $\vCT$ uniquely determine each other.

Analogously to \eqref{eq:EnergyDiffDecomposition}, for a CP superoperator $\CT[\cdot]$ we define $\CT_{\nu}[\cdot]:=\sum_{a\in A}\vA^a_{\nu}[ \cdot ](\vA^a_{\nu})^\dagger$, which is well defined, as can be seen using its matrix
\begin{align*}
	\vCT_{\nu}
	=\sum_{a\in A}\vA^a_{\nu}\otimes(\vA^a_{\nu})^*
	&=\sum_{a\in A} \sum_{\underset{\energy_i-\energy_j = \nu = \energy_k-\energy_\ell}{i,j,k,\ell\in[d] \colon }} ( \bra{\psi_i}\vA^a\ket{\psi_j}\ketbra{\psi_i}{\psi_j})\otimes(\bra{\psi_k^*}\vA^{a*}\ket{\psi_\ell^*}\ketbra{\psi_k^*}{\psi_\ell^*})\\&	
	= \sum_{\underset{\energy_i-\energy_j = \nu = \energy_k-\energy_\ell}{i,j,k,\ell\in[d] \colon }} \Bigg( \bra{\psi_i}\bra{\psi_k^*}\underset{\vCT}{\underbrace{\sum_{a\in A}\vA^a\otimes\vA^{a*}}}\ket{\psi_j}\ket{\psi_\ell^*}\Bigg)\ketbra{\psi_i}{\psi_j}\otimes\ketbra{\psi_k^*}{\psi_\ell^*}.
\end{align*}

\subsection{The Davies generator}

The first construction we describe comes from the seminal work of Davies~\cite{davies1974MasterEquationI,davies1976MasterEquationII} who studied open quantum systems and derived a so-called ``master equation'' under a Markovian bath assumption. The Davies generator is the Lindbladian $\CL_D[\cdot] \!=\! \CT'[\cdot]-\frac12\{\CT'^\dagger[\vI],\cdot\} $, where $\CT'=\BS_D^{[\gamma]}\llbracket\CT\rrbracket$ is defined as
\begin{subequations}
\begin{align}\label{eq:DavieSSO}
\BS_D^{[\gamma]}\llbracket\CT\rrbracket[\cdot] &=\sum_{\nu\in B(\vH)}\gamma(\nu)\CT_{\nu}[\cdot] =\sum_{\nu\in B(\vH)}\gamma(\nu)\sum_{a\in A}\vA^a_{\nu}[ \cdot ](\vA^a_{\nu})^\dagger,\\&
\text{where } \gamma\colon \BR\rightarrow[0,1] \text{ satisfies } \gamma(-\nu)=\e^{\nu}\gamma(\nu).\label{eq:Gammareq}
\end{align}
\end{subequations}
This is a direct generalization of the classical construction, because if $\vrho$ is a classical (i.e., diagonal) distribution, $\CT$ comes from a classical Laplacian as in \eqref{eq:LapToLind}, and if we define $g(r):=\gamma(-\ln{r})$, then this exactly recovers \eqref{eq:ClContMetropolis}.

A nice feature of $\BS_D^{[\gamma]}$ is that by construction $\CT'^\dagger[\vI]$ always commutes with $\vrho$ because $(\vM_{\nu})^\dagger\vM_{\nu}$ is block diagonal in the eigenbasis of $\vH$ for any $\nu\in B(\vH)$ and $\vM\in\BC^{d\times d}$. Therefore, we do not need to add any coherence term $\vC$.
To see that $\BS_D^{[\gamma]}\llbracket\CT\rrbracket$ is $\vrho$-detailed balanced, observe that the Kraus operator of $\vrho^{\frac14}\BS_D^{[\gamma]}\llbracket\CT\rrbracket^\dagger[\vrho^{-\frac14}\cdot \vrho^{-\frac14}]\vrho^{\frac14} = \vrho^{\frac14}\CT'^\dagger[\vrho^{-\frac14}\cdot \vrho^{-\frac14}]\vrho^{\frac14}$ indexed by $(a,-\nu)$ is
\begin{align}
\sqrt{\gamma(-\nu)}\vrho^{\frac14}(\vA^a_{-\nu})^\dagger\vrho^{-\frac14}
&=\sqrt{\gamma(-\nu)}\e^{-\frac{\nu}{4}}(\vA^a_{-\nu})^\dagger \tag{due to \eqref{eq:EnergyDiffDecomposition}}\\&
=\sqrt{\gamma(\nu)}\e^{\frac{\nu}{4}}(\vA^a_{-\nu})^\dagger \tag{due to \eqref{eq:Gammareq}}\\&
=\sqrt{\gamma(\nu)}\e^{\frac{\nu}{4}}(\vA^{a\dagger})_{\nu} \tag{due to \eqref{eq:EnergyDiffDecomposition}}\\&
=\sqrt{\gamma(\nu)}\vrho^{-\frac14}(\vA^{a\dagger})_{\nu}\vrho^{\frac14}, \label{eq:KrausSymmetry}
\end{align}
which coincides with the Kraus operator of $\vrho^{-\frac14}\BS_D^{[\gamma]}\llbracket\CT^\dagger\rrbracket[\vrho^{\frac14}\cdot \vrho^{\frac14}]\vrho^{-\frac14}$ indexed by $(a,\nu)$. Since the set of Bohr-frequencies is symmetric $B(\vH)=-B(\vH)$, this implies that 
\begin{equation}\label{eq:DaviesSymmetry}
	\vrho^{-\frac14}\BS_D^{[\gamma]}\llbracket\CT^\dagger\rrbracket[\vrho^{\frac14}\cdot \vrho^{\frac14}]\vrho^{-\frac14} = \vrho^{\frac14}\BS_D^{[\gamma]}\llbracket\CT\rrbracket^\dagger[\vrho^{-\frac14}\cdot \vrho^{-\frac14}]\vrho^{\frac14}.
\end{equation}
By \Cref{def:KMSDetBalance} this means that $\CT' = \BS_D^{[\gamma]}\llbracket\CT\rrbracket$ is $\vrho$-detailed balanced since $\CT^\dagger=\CT$.

The Davies generator is obtained in a weak-coupling, infinite-time limit, which makes it behave unphysically in certain situations. This is especially problematic in many-body quantum systems, where the spectrum of $\vH$ can be very dense in the sense that the energy levels are exponentially close to each other. This implies that the geometrical locality of the term $\vA [\cdot] \vA^\dagger$ might not be preserved, which is undesirable and makes the implementation of the resulting Lindbladian generally very challenging. 
Mathematically speaking the undesirable feature of  $\BS_D^{[\gamma]}$ is that it is discontinuous with respect to\footnote{While the notation does not explicitly reveal it, the transformation $\BS_D^{[\gamma]}$ heavily depends on the target state $\vrho$ through its logarithm $-\vH$.} $\vrho$ as showcased by the following example:

\begin{example}\label{expl:Davies}
Consider a single-qubit system with $\vrho:=\frac{\exp(\eps)}{2\cosh(\eps)}\ketbra{0}{0}+\frac{\exp(-\eps)}{2\cosh(\eps)}\ketbra{1}{1}$, and $\CT[\cdot]:=\vX[\cdot] \vX$, where $\vX=\ketbra{1}{0}+\ketbra{0}{1}$ is the Pauli-$X$ operator, and choose $\gamma:=\gamma_M=\min(1,\exp(-\nu))$.
Then, for $\eps=0$ we have $\CT_0=\BS_D^{[\gamma]}\llbracket\CT\rrbracket = \CT$ and $\vCT_0=\vX\otimes\vX$, but for every $\eps>0$ we have 
$\CT_\eps=\BS_D^{[\gamma]}\llbracket\CT\rrbracket = \exp(-2\eps)\ketbra{1}{0}[\cdot]\ketbra{0}{1} + \ketbra{0}{1}[\cdot]\ketbra{1}{0}$ and $\vCT_\eps=\exp(-2\eps)\ketbra{11}{00}+\ketbra{00}{11}$.
\end{example}

The above example examines a coherent bit-flip, and can be thought of as a quantum analogue of \Cref{expl:Metropolis}. For $\eps=0$ the dynamics is highly coherent and there is only a single Kraus operator. However, for $\eps>0$ the dynamics becomes entirely classical, and in fact the corresponding Lindblad evolution is identical to the continuous-time Markov chain of \Cref{expl:Metropolis}.

We are not aware of any generic purpose continuous construction (in $\vrho$) that applies to non-commuting Hamiltonians and guarantees both $\vrho$-detailed balance and $[\CT'^\dagger[\vI],\vrho]=0$.

\subsection{Quantum Glauber and Metropolis dynamics via coherent reweighing}\label{sec:coherentDyn}

We now describe our new construction that leads to a highly coherent (discrete and continuous-time) quantum analogue of Glauber and Metropolis dynamics. The construction gives rise to a map that depends continuously on $\vrho$, moreover enables efficient implementation on a quantum computer. 
The new construction is highly coherent in the sense that it does not increase the Kraus rank (i.e., the minimal number of Kraus operators required to describe a CP map), which can be thought of as a rudimentary measure of coherence. Indeed, in a quantum channel, coherent evolution can be represented by a single Kraus operator, while measurements and noise typically increase the Kraus rank. 

In order to describe our construction, let us introduce the following superoperator 
\begin{align}\label{eq:SmoothMapDef}
	\CS^{[f]}[\vM]:=\sum_{\nu \in B(\vH)}f(\nu)\vM_{\nu},
\end{align}
which leads to our new super-superoperator acting as
\begin{align}\label{eq:coherentReweighing}
	\BS_C^{[f]}\llbracket\CT\rrbracket[\cdot] =\sum_{a\in A}\CS^{[f]}[\vA^a][ \cdot ]\left(\CS^{[f]}[\vA^a]\right)^\dagger.
\end{align}
To see that $\BS_C^{[f]}\llbracket\cdot\rrbracket$ is well defined, we can observe that its matrix is  $\pmb{\BS}_C^{[f]}=\vCS^{[f]}\otimes\bar{\vCS}^{[f]}$, where $\bar{\vCS}^{[f]}$ is the ``conjugate'' of $\vCS^{[f]}$, which we define through its action as
$\bar{\CS}^{[f]}[\vM]:=(\CS^{[f]}[\vM^*])^*$.

\begin{theorem}\label{thm:CoherentRule}
    If $\CT$ is self-adjoint and  analogously to \eqref{eq:Gammareq} $f$ satisfies $f(-\nu)=\e^{\frac\nu2}f(\nu)^*$, 
    then $\CT'=\BS_C^{[f]}\llbracket\CT\rrbracket$ is $\vrho$-detailed balanced.
\end{theorem}
\begin{proof}
	We proceed similarly as in \eqref{eq:KrausSymmetry}.
    The Kraus operator of $\vrho^{\frac14}\BS_C^{[f]}\llbracket\CT\rrbracket^\dagger[\vrho^{-\frac14}\cdot\vrho^{-\frac14}]\vrho^{\frac14} = \vrho^{\frac14}\CT'^\dagger[\vrho^{-\frac14}\cdot\vrho^{-\frac14}]\vrho^{\frac14}$ labelled by $a$ is
	 \begin{align*}
		\vrho^{\frac14}(\CS^{[f]}[\vA^a])^\dagger\vrho^{-\frac14} 
		&=\vrho^{\frac14}\bigg(\sum_{\nu\in B(\vH)}f(-\nu)\vA^a_{-\nu}\bigg)^{\!\dagger}\vrho^{-\frac14} \tag{due to \eqref{eq:SmoothMapDef} and $B(\vH)=-B(\vH)$}\\&
		=\bigg(\sum_{\nu\in B(\vH)}f(-\nu)\e^{-\frac{\nu}{4}}\vA^a_{-\nu}\bigg)^{\!\dagger} \tag{due to \eqref{eq:EnergyDiffDecomposition}}\\&
		=\bigg(\sum_{\nu\in B(\vH)}f(\nu)^*\e^{\frac{\nu}{4}}\vA^a_{-\nu}\bigg)^{\!\dagger} \tag{due to $f(-\nu)=\e^{\frac\nu2}f(\nu)^*$}\\&
		=\sum_{\nu\in B(\vH)}f(\nu)\e^{\frac{\nu}{4}}(\vA^{a\dagger})_{\nu} \tag{due to \eqref{eq:EnergyDiffDecomposition}}\\&	
		=\vrho^{-\frac14}\bigg(\sum_{\nu\in B(\vH)}f(\nu)(\vA^{a\dagger})_{\nu}\bigg)\vrho^{\frac14} \tag{due to \eqref{eq:EnergyDiffDecomposition}}\\&		
		=\vrho^{-\frac14}(\CS^{[f]}[\vA^{a\dagger}])\vrho^{\frac14}  \tag{due to \eqref{eq:SmoothMapDef}},
	\end{align*}  
which coincides with the $a$-labelled Kraus operator of $\vrho^{-\frac14}\BS_C^{[f]}\llbracket\CT^\dagger\rrbracket[\vrho^{\frac14}\cdot \vrho^{\frac14}]\vrho^{-\frac14}$, implying that
\begin{equation*}
	\vrho^{-\frac14}\BS_C^{[f]}\llbracket\CT^\dagger\rrbracket[\vrho^{\frac14}\cdot \vrho^{\frac14}]\vrho^{-\frac14}
	=\vrho^{\frac14}\BS_C^{[f]}\llbracket\CT\rrbracket^\dagger[\vrho^{-\frac14}\cdot \vrho^{-\frac14}]\vrho^{\frac14}.
\end{equation*}
By \Cref{def:KMSDetBalance} this means that $\CT' = \BS_C^{[f]}\llbracket\CT\rrbracket$ is $\vrho$-detailed balanced since $\CT^\dagger=\CT$.
\end{proof}

This is another direct generalization of the classical construction, because if $\vrho$ is a classical (i.e., diagonal) distribution, $\CT$ comes from a classical Laplacian as in \eqref{eq:LapToLind}, and we define $g(r):=|f(-\ln{r})|^2$, then $\BS_C^{[f]}$ exactly recovers  \eqref{eq:ClContMetropolis}. However, in contrast to the analogous Davies map $\BS_D^{[|f|^2]}$, it depends continuously on $\vrho$. To illustrate the difference we revisit \Cref{expl:Davies}.

\begin{example}\label{expl:Coherent}
	Consider the same single-qubit system as in {\rm \Cref{expl:Davies}}.
	Then, for every $\eps\geq 0$ we have 
	$\CT_\eps=\BS_C^{[\sqrt{\gamma_M}]}\llbracket\CT\rrbracket = \vX_\eps[\cdot](\vX_\eps)^\dagger$ where $\vX_\eps=\exp(-\eps)\ketbra{1}{0}+\ketbra{0}{1}$.
\end{example}

Due to quantum interference there are some new phenomenon that we should consider. 
In the classical setting, if we use the Glauber and Metropolis dynamics corresponding to $\gamma_G(\nu)=1/(1+\exp(\nu))=\frac12-\frac12\tanh(\frac{\nu}{2})$ and $\gamma_M(\nu)=\min(1,\exp(-\nu))$ respectively, then some transitions are reduced while no transition is enhanced, so overall the transition strength cannot increase. 
However, in the quantum setting due to interference the weakening of a transition might actually strengthen another one, so it is important to understand how much the maximum transition strength can increase. It is natural to measure the transition strength of $\CT[\cdot]$ by $\|\CT^\dagger[\vI]\|$, as it equals the maximum of $\tr[\CT[\vsigma]]/\tr[\vsigma]$. Motivated by this we define the following super-superoperator norm
\begin{align}\label{eq:SSONorm}
	\vertiii{\BS}:=\sup_{\CT}\frac{\big\|(\BS\llbracket\CT\rrbracket)^\dagger[\vI]\big\|}{\nrm{\CT^\dagger[\vI]}},
\end{align}
where the supremum is over the set of completely positive maps $\CT\colon \BC^{d\times d}\rightarrow \BC^{d\times d}$.

It turns out that if $f$ is sufficiently smooth, then $\vertiii{\BS_C^{[f]}}=\mathcal{O}\big({\log^2}(\nrm{\vH}+1)\big)$, which is also tight as we prove in \Cref{apx:SchurBounds}. The proof techniques also enable efficient implementation, and in particular imply that the map $\vrho \rightarrow \BS_C^{[f]}$ is continuous for smooth-enough~$f$.

Probably the natural choice is $f(\nu)=f_G(\nu):=1/(1+\exp(\nu/2))=\frac12-\frac12\tanh(\frac{\nu}{4})$. This is especially nice because we get that $\CS^{[f]}=\CS^-$, thus up to a shift by identity we can apply the same transformation to the Kraus operators as we use for constructing the coherent term. 

\subsubsection{The discriminant of the corresponding continuous-time construction}\label{sec:coherentDynDisc}

We now explicitly compute the discriminant $\CD$ (whose ``vectorization'' is the parent Hamiltonian) for the special case of $\CS^-$ which leads to simplified formulas due to its symmetries, but note that the general case of $\CS^{[\sqrt{\gamma}]}$ can be handled similarly.
When considering the discriminant, the superoperator $\CS_c[\cdot]:=\vrho^{-\frac14}\CS^-[\cdot]\vrho^{\frac14}=\vrho^{\frac14}\CS^+[\cdot]\vrho^{-\frac14}$ plays an important role. A small calculation shows that 
\begin{align*}
\CS_c[\vA]=\sum_{\nu\in B(\vH)}\frac{1}{2\cosh(\nu/4)}\vA_\nu=\int_{\BR} \frac{1}{\cosh(2\pi t)}\vA(t) \rd t,
\end{align*}
thus notably the singularity of \eqref{eq:SDefIntro2} at $0$ is removed. 
Due to linearity it suffices to consider a single Kraus operator $\CT[\cdot]=\vA[ \cdot ]\vA^\dagger$ whose associated Lindbladian is
\begin{align*}
    \CL[\cdot]=\CS^-[\vA][ \cdot ](\CS^-[\vA])^\dagger -\frac{1}{2}\{\vD,\cdot\}-\ri[\ri\CS[\vD],\cdot] \qquad\text{where}\qquad\vD=(\CS^-[\vA])^\dagger \CS^-[\vA].
\end{align*}
Considering that $\CS[\vA]^\dagger=-\CS[\vA^\dagger]$ we further obtain that 
\begin{align*}
    \CL[\cdot]=\CS^-[\vA][ \cdot ]\CS^+[\vA^\dagger]-\vM[\cdot]-[\cdot]\vM^\dagger \qquad\text{where}\qquad\vM=\CS^-[\vD]=\CS^-[\CS^+[\vA^\dagger]\CS^-[\vA]],
\end{align*}
using that $\vD^\dagger = \vD$. Consequently we get that $\CD[\cdot]=\CS_c[\vA][\cdot]\CS_c[\vA^\dagger]-\CS_c[\vD][\cdot]-[\cdot]\CS_c[\vD]$. To understand this better we use that $\CI=\CS^+ + \CS^-$ to expand
\begin{align}
    \CS_c[\vD]&=\CS_c[\CS^+[\vA^\dagger]\CS^-[\vA]]
    =\frac12\CS_c[\CS^+[\vA^\dagger](\vA-\CS^+[\vA])+(\vA^\dagger-\CS^-[\vA^\dagger])\CS^-[\vA]]\nonumber\\
    &=\frac12\CS_c[\underset{=\vA^\dagger\vA+\CS[\vA^\dagger]\vA-\vA^\dagger\CS[\vA]}{\underbrace{\CS^+[\vA^\dagger]\vA+\vA^\dagger\CS^-[\vA]}}]-\frac12\underset{= \CS_c[\vA^\dagger]\CS_c[\vA]}{\underbrace{\CS_c[\CS^+[\vA^\dagger]\CS^+[\vA]+\CS^-[\vA^\dagger]\CS^-[\vA]]}}.\label{eq:ScD}
\end{align} 
In the last term we removed the apparent triple superoperator by observing that 
\begin{align*}
        \CS_c[\CS^+[\vA^\dagger]\CS^+[\vA]+\CS^-[\vA^\dagger]\CS^-[\vA]]&=
		\frac{1}{2}\CS_c[\vA^\dagger\vA+4\CS[\vA^\dagger]\CS[\vA]]\\&
		=\frac12\CS_c\left[\sum_{\nu,\nu'\in B(\vH)}(\vA^\dagger)_\nu\vA_{\nu'}\left(1+\tanh\left(\frac\nu4\right)\tanh\left(\frac{\nu'}4\right)\right)\right]\\&
		=\frac12\sum_{\nu,\nu'\in B(\vH)}(\vA^\dagger)_\nu\vA_{\nu'}\frac{1+\tanh(\frac\nu4)\tanh(\frac{\nu'}4)}{2\cosh(\frac{\nu+\nu'}{4})}\\&
		=\frac14\sum_{\nu,\nu'\in B(\vH)}(\vA^\dagger)_\nu\vA_{\nu'}\frac{1}{\cosh(\frac\nu4)}\frac{1}{\cosh(\frac{\nu'}4)}\\&	
		=\CS_c[\vA^\dagger]\CS_c[\vA].
	\end{align*}

In the above, \eqref{eq:ScD} still contains an apparent ``singularity'' due to the one application of the operator $\CS^\pm$ that remains after simplifying the expression. This suggests that the operator norm might be as large as $\log(\nrm{\vH})\nrm{\vA}$. However, in \Cref{apx:DiscNormBound} we show that in fact $\nrm{\vCD}=\nrm{\vA}$ as the final application of $\CS_c$ regularizes the norm.

\subsection{Detailed balance via Fourier-transformed Heisenberg evolution}\label{sec:OFT}

Now we describe the construction of~\cite{chen2023QThermalStatePrep,chen2023ExactQGibbsSampler} which can be thought of as a smoothened variant of the Davies generator.
Their constructions proceed through a decomposition of the Kraus operators following the Bohr frequencies with Gaussian noise\footnote{In~\cite{chen2023QThermalStatePrep,chen2023ExactQGibbsSampler}, this is denoted by $\hat{\vA}(\omega)$. Here we introduce the superoperator to feed into the superoperator formulation.}
\begin{align}
    \CF^{[\sigma]}(\omega)[\vA]: = \sum_{\nu \in B(\vH)}  \hat{f}_{\sigma}(\omega-\nu) \vA_{\nu},
\end{align}	
where $\hat{f}_{\sigma}(\omega)=\frac{1}{\sqrt{\sigma\sqrt{2\pi}}} \e^{- \frac{\omega^2}{4\sigma^2}}$ is the Fourier transform of $f_{\sigma}(t)=\e^{-\sigma^2 t^2}\sqrt{\sigma \sqrt{2/\pi}}$
and $\sigma$ is the relative energy resolution.
This is a continuous-parameter decomposition in the sense that 
\begin{align}\label{eq:ForierDecomposition}
    \int_{\BR}\CF^{[\sigma]}(\omega)[\vA]\rd\omega
    =\int_{\BR} \sum_{\nu \in B(\vH)}\!\! \hat{f}_{\sigma}(\omega-\nu) \vA_{\nu}\rd\omega
    = \sum_{\nu \in B(\vH)}\int_{\BR} \hat{f}_{\sigma}(\omega-\nu) \vA_{\nu}\rd\omega
    =\vA \int_{\BR} \hat{f}_{\sigma}(\omega) \rd\omega.
\end{align}	
The motivation behind this decomposition is that, due to Parseval's identity (see~\cite[Appendix~A]{chen2023QThermalStatePrep}), it can be expressed in the alternative from of
\begin{align}\label{eq:OpOFT}
    \CF^{[\sigma]}(\omega)[\vA]=  \frac{1}{\sqrt{2\pi}}\int_{-\infty}^{\infty} \e^{\ri \vH t} \vA \e^{-\ri \vH t} \e^{-\ri \omega t} f_{\sigma}(t)\rd t,
\end{align}
which is a Fourier transform of the Heisenberg evolution of $\vA$ with Gaussian damping $f_{\sigma}(t)$. This representation in turn enables efficient (approximate) implementation using the quantum Fourier transform~\cite{chen2023QThermalStatePrep}.

The above decomposition leads to a continuous-parameter ``decomposition'' of any CP map $\CT[\cdot]=\sum_{a\in A}\vA^a[ \cdot ]\vA^{a\dagger}$ as follows:
\begin{align*}
    \BS_H^{[\sigma]}(\omega)\llbracket\CT\rrbracket[\cdot] := \sum_{a\in A} \left(\CF^{[\sigma]}(\omega)[\vA^a]\right) [\cdot]\left(\CF^{[\sigma]}(\omega)[\vA^a]\right)^\dagger,
\end{align*}
which is well defined as can be seen through, e.g., the integral representation and the matrix $\vCT$:\footnote{With a slight abuse of notation here we write $\BS_H^{[\sigma]}(\omega)\llbracket\vCT\rrbracket$ for the vectorization of $\BS_H^{[\sigma]}(\omega)\llbracket\CT\rrbracket$.}
\begin{align*}
    \BS_H^{[\sigma]}(\omega)\llbracket\vCT\rrbracket = \frac{1}{2\pi}\int_{\BR^2} \left(\e^{\ri \vH t_1}\otimes \e^{-\ri \vH^* t_2}\right) \vCT \left(\e^{-\ri \vH t_1}\otimes \e^{\ri \vH^* t_2}\right) \e^{-\ri \omega (t_1-t_2)} f_\sigma(t_1)f_\sigma(t_2)\rd t_1\rd t_2.
\end{align*}
Integrating over $\omega$ we get a $\sigma$-energy resolved version of $\CT$ via the operator Fourier transform
\begin{align}\label{eq:OFTSUnweightedDef}
    \BS_H^{[\sigma]}\llbracket\CT\rrbracket:=\int_{\BR}\BS_H^{[\sigma]}(\omega)\llbracket\CT\rrbracket \rd \omega.
\end{align}	
In~\cite{chen2023QThermalStatePrep} it is also proven that this map does not increase the ``intensity'', more precisely
\begin{align*}
    \vertiii{\BS_H^{[\sigma]}}=1.
\end{align*}
Ultimately, the construction described in~\cite{chen2023ExactQGibbsSampler} is
\begin{align}\label{eq:OFTSDef}
    \BS_H^{[\sigma,g]}\llbracket\CT\rrbracket := \int_{\BR} \gamma^{(g)}(\omega)\BS_H^{[\sigma]}(\omega)\llbracket\CT\rrbracket  \rd \omega, 
    \qquad\text{where}~ \gamma^{(g)}(\omega)=&\int_{\frac{\sigma^2}{2}}^{\infty}g(x)\e^{-\frac{(\omega + x)^2}{4x-2\sigma^2}}\rd x
\end{align}	
for some integrable function $g\colon [\frac{\sigma^2}{2},\infty)\rightarrow \BR_+$. In particular the authors of~\cite{chen2023ExactQGibbsSampler} argue that the most natural choice for $\gamma^{(g)}$ is $\tilde{\gamma}^{(\sigma)}_M(\omega)=\exp\left(-\max\left(\omega +\sigma^2/2,0\right)\right)$, which is a shifted version of the original Metropolis function $\tilde{\gamma}^{(\sigma)}_M(\omega)=\gamma_M(\omega +\sigma^2/2)$. We refer to~\cite[Proposition~II.2]{chen2023ExactQGibbsSampler} for the proof that $\BS_H^{[\sigma,g]}\llbracket\CT\rrbracket$ is $\vrho$-detailed balanced whenever $\CT$ is self-adjoint.\footnote{In~\cite{chen2023ExactQGibbsSampler} the result is stated under the condition that in the decomposition $\CT[\cdot]=\sum_{a \in A}\vA^a[ \cdot ]\vA^{a\dagger}$ for each Kraus operator $\vA$ their adjoint is also a Kraus operator. It is easy to see that this is in fact equivalent to $\CT$ being self-adjoint. Indeed, consider the natural decomposition $\frac12(\CT+\CT^\dagger)$ of $\CT$. Also note that we are working with dimensionless quantities, but adding back $\beta$ resolves the apparent dimension mismatch, see \cite[Proposition~II.4]{chen2023ExactQGibbsSampler}.\label{foot:adjointKrauses}}

\subsubsection{A variant closer related to phase estimation}\label{sec:PEVariant}

Another construction outlined in~\cite[Section III.C]{chen2023QThermalStatePrep} that more closely resembles the phase-estimation-based Metropolis variants~\cite{wocjan2021SzegedyWalksForQuantumMaps,temme2009QuantumMetropolis,jiang2024QMetropolisWeakMeas} is related to sampling from an arbitrary weight function $\sum_i p(E_i) |\psi_i\rangle\langle \psi_i|$. Here we prove that this can also be made exactly detailed balanced in a similar fashion to \cite{chen2023ExactQGibbsSampler}. This construction is based on a two-sided Fourier transform
\begin{align}
    \CF^{[\sigma_1,\sigma_2]}(E_f,E_i)[\vA] := \frac{1}{2\pi}\int_{\mathbb{R}^2} \e^{\ri (\vH - E_f) t_2} \vA \e^{-\ri (\vH - E_i) t_1} f_{\sigma_1}(t_1)f_{\sigma_2}(t_2)\rd t_1 \rd t_2,
\end{align}
which can also be expressed as
\begin{align*}
    \CF^{[\sigma_1,\sigma_2]}(E_f,E_i)[\vA] &= \frac{1}{2\pi}\sum_{E_1,E_2\in\operatorname{spec}(\vH)} \vPi_{E_2} \vA\vPi_{E_1}\int_{\mathbb{R}^2} \e^{\ri (E_2 - E_f) t_2 - \ri (E_1 - E_i)t_1} f_{\sigma_1}(t_1)f_{\sigma_2}(t_2)\rd t_1 \rd t_2\\
    &= \sum_{E_1,E_2\in\operatorname{spec}(\vH)} \hat{f}_{\sigma_1}(E_i - E_1) \hat{f}_{\sigma_2}(E_f - E_2) \vPi_{E_2}\vA\vPi_{E_1},
\end{align*}
where $\vPi_E$ is the projection onto the $E$-energy subspace.
The two-sided Fourier transform does not directly take into account the energy difference $\omega = E_f - E_i$, but both the initial and final energy levels $E_i$ and $E_f$, respectively. Once again, it is possible to obtain a continuous-parameter decomposition of any CP map $\CT[\cdot]=\sum_{a\in A}\vA^a[\cdot]\vA^{a\dagger}$ through the two-sided Fourier transform as
\begin{align*}
    \BS_H^{[\sigma,\sigma]}(E_f,E_i)\llbracket\CT\rrbracket[\cdot] := \sum_{a\in A} \left(\CF^{[\sigma,\sigma]}(E_f,E_i)[\vA^a]\right) [\cdot]\left(\CF^{[\sigma,\sigma]}(E_f,E_i)[\vA^a]\right)^\dagger.
\end{align*}
Consider then the construction
\begin{align}\label{eq:two-sided_fourier_construction_0}
    \BS_H^{[\sigma,\sigma,g]}\llbracket\CT\rrbracket :=& \int_{\BR^2} \gamma^{(g)}(E_f,E_i)\BS_H^{[\sigma,\sigma]}(E_f,E_i)\llbracket\CT\rrbracket  \rd E_i \rd E_f, \quad\text{where}\\
    \gamma^{(g)}(E_f,E_i) =& \int_{\sigma^2}^\infty g(x)\e^{\frac{(E_f-E_i + x)^2}{4(x-\sigma^2)}}\rd x \quad\text{for an integrable function}~ g:[\sigma^2,\infty)\to\mathbb{R}_+.\nonumber
\end{align}
Unlike the previous construction in \eqref{eq:OFTSDef}, the energy levels \emph{before} and \emph{after} a jump are taken into consideration, and not simply their difference. 

We prove in \Cref{apx:PEBasedDB} that $\BS_H^{[\sigma,\sigma,g]}\llbracket\CT\rrbracket$ is $\vrho$-detailed balanced in a similar fashion to~\cite[Proposition~II.2]{chen2023ExactQGibbsSampler}:
\begin{theorem}\label{thr:PEBasedDB}
    $\BS_H^{[\sigma,\sigma,g]}\llbracket\CT\rrbracket$ is $\vrho$-detailed balanced for any self-adjoint $\CT[\cdot] = \sum_{a\in A} \vA^a[\cdot]\vA^{a\dagger}$.
\end{theorem}
 
\subsubsection{Classical Metropolis sampling via Gaussian uncertainty}

It is worthwhile to consider the classical analogue of this construction. In the classical setting it is implicitly assumed that we can perfectly estimate  $\pi_i$ and $\pi_j$ and therefore can implement the reject/accept decision perfectly. This does not hold in general in the quantum case due to fundamental uncertainty principles, also making the Davies generator unphysical in general.

Now let us consider another classical model where we wish to sample from $\pi\propto \exp(-h)$, where $h$ is a vector of ``energies'', however every time we make a transition we only get an estimate $\omega=(h_j-h_i)+\CN(0,\sigma)$ of the energy difference with Gaussian noise. The question is whether we can devise an accept/reject rule that still ensures $\pi$-detailed despite the noise. Our construction can be seen as a solution to this classical problem as well.

Consider the following rule: make a transition according to the symmetric transition matrix $\vM$ and get an estimate of the induced energy difference $\omega=(h_j-h_i)+\CN(0,\sigma)$. Accept the move with probability $\tilde{\gamma}^{(\sigma)}_M(\omega)$ and upon rejection go back to the state $i$ before the transition. What is the probability that a transition $i \rightarrow j$ is accepted?
It is easy to see that it is the convolution of $\tilde{\gamma}^{(\sigma)}_M$ and $\CN(0,\sigma)$, i.e., for $\nu=h_j-h_i$ we get
\begin{align}\label{eq:ClassicalUncertainBalance}
    \gamma(\nu)
    =\int_{\BR}\frac{\e^{\frac{-x^2}{2\sigma^2}}}{\sigma\sqrt{2\pi}}\gamma_M(\nu + x +\sigma^2/2)\rd x
    =\frac12\left(\!\e^{-\nu } \erfc\left(\frac{\sigma^2-2 \nu }{2 \sqrt{2} \sigma }\right)\!+\erfc\left(\frac{\sigma^2+2 \nu }{2
        \sqrt{2} \sigma }\right)\right),
\end{align}
which indeed features the symmetry $\gamma(-\nu)=\e^{\nu}\gamma(\nu)$ required by detailed balance, thereby this process results in a $\pi$-detailed-balanced transition matrix.

The ability to deal with energy estimation featuring Gaussian noise has been previously noticed by~\cite{ceperley1999PenaltyRWsUncertainEnergies}. However, it seems unlikely that other types of noises can be exactly counteracted, as the Gaussian shape seems to play a special role, see~\cite[Appendix~D]{chen2023ExactQGibbsSampler}.

\subsection{Interpolating between the Davies and the coherent constructions}\label{sec:interpolation}

Motivated by physical intuition, we consider an interpolated version of the previous two constructions. 
As we discuss in~\Cref{sec:physMotivation}, the coherent construction of \Cref{sec:coherentDyn} is more closely related to the Redfield and universal Lindblad equations~\eqref{ULE}, while the Operator Fourier transform variant of \Cref{sec:OFT} is more closely related to the weak coupling induced coarse-grained master equation~\eqref{eq:CGME}. It is worth considering a construction that reflects the freedom in the choice of the coupling strength.

One could simply consider the Operator Fourier transform variant~\Cref{sec:OFT} with different choices for $\sigma$, however when $\sigma \gg 1$ (or $\sigma \gg 1/\beta$ when working with physical dimensions), the dynamics potentially becomes stalled, as transitions between neighbouring energies are blocked, because the function $\tilde{\gamma}^{(\sigma)}_M(\omega)=\gamma_M(\omega +\sigma^2/2)$ effectively overestimates energy increase by about $\sigma^2/2$. To remedy this, we propose a combination of the two schemes 
\begin{align}\label{eq:interpolatedS}
    \BS^{[\sigma,\gamma]}_{I}:=\BS_C^{[\gamma]}\circ\BS_H^{[\sigma]},
\end{align}
where $\BS_H^{[\sigma]}$ is the Operator Fourier transform without reweighing \eqref{eq:OFTSUnweightedDef} that one gets by replacing $\gamma^{(g)}$ by the constant $\equiv 1$ function in~\eqref{eq:OFTSDef}. This is essentially equivalent to saying that we apply the construction of \Cref{sec:OFT} with the maximally mixed state as the desired fixed-point, which then shows that the Operator Fourier transform $\BS_H^{[\sigma]}$ maps self-adjoint superoperators to self-adjoint ones. This in turn reduces to correctness of~\eqref{eq:interpolatedS} to that of \Cref{sec:coherentDyn,sec:OFT}.

It is possible to show that the $\sigma \rightarrow \infty$ limit recovers the coherent construction
\begin{align}\label{eq:interpolation_sigma_infty}
    \lim_{\sigma\rightarrow \infty}\BS^{[\sigma,\gamma]}_{I}=\BS^{[\gamma]}_{C},
\end{align}
while the $\sigma \rightarrow 0$ limit recovers the Davies generator
\begin{align}\label{eq:interpolation_sigma_0}
    \lim_{\sigma\rightarrow 0+}\BS^{[\sigma,\gamma]}_{I}=\BS^{[\gamma]}_{D}.
\end{align}
To see \eqref{eq:interpolation_sigma_infty}, first write $\BS_H^{[\sigma]}$ from \eqref{eq:OFTSUnweightedDef} as
\begin{align*}
    \BS_H^{[\sigma]}\llbracket\CT\rrbracket[\cdot] &= \frac{1}{2\pi}\sum_{a\in A}\int_{\mathbb{R}^2}  \e^{\ri \vH t_1} \vA^a \e^{-\ri \vH t_1}[\cdot] \e^{\ri \vH t_2} \vA^{a\dagger} \e^{-\ri \vH t_2} f_{\sigma}(t_1)f_{\sigma}(t_2)\int_{\mathbb{R}} \e^{\ri \omega (t_2 - t_1)}\rd \omega \rd t_1 \rd t_2\\
    &= \sum_{a\in A}\int_{\mathbb{R}}  \e^{\ri \vH t} \vA^a \e^{-\ri \vH t}[\cdot] \e^{\ri \vH t} \vA^{a\dagger} \e^{-\ri \vH t} f_{\sigma}(t)^2 \rd t,
\end{align*}
where we used that $\frac{1}{2\pi}\int_{\mathbb{R}} \e^{\ri \omega (t_2 - t_1)}\rd \omega = \delta(t_2-t_1)$ is a Dirac delta. From here, we see that $\lim_{\sigma\to \infty} f_\sigma(t)^2 = \lim_{\sigma\to \infty}\sqrt{2/\pi}\sigma\e^{-2\sigma^2 t^2} = \delta(t)$ is a Dirac delta. Therefore,
\begin{align*}
    \lim_{\sigma\to\infty}\BS_H^{[\sigma]}\llbracket\CT\rrbracket[\cdot] = \sum_{a\in A}\int_{\mathbb{R}}  \e^{\ri \vH t} \vA^a \e^{-\ri \vH t}[\cdot] \e^{\ri \vH t} \vA^{a\dagger} \e^{-\ri \vH t} \delta(t) \rd t = \sum_{a\in A} \vA^a[\cdot]\vA^{a\dagger} = \CT[\cdot],
\end{align*}
meaning that $\lim_{\sigma\to\infty}\BS_H^{[\sigma]} = \mathcal{I}$, which yields \eqref{eq:interpolation_sigma_infty}. On the other hand, to obtain \eqref{eq:interpolation_sigma_0}, write $\BS_H^{[\sigma]}$ from \eqref{eq:OFTSUnweightedDef} as
\begin{align*}
    \BS_H^{[\sigma]}\llbracket\CT\rrbracket[\cdot] &= \sum_{a\in A}\sum_{\nu_1,\nu_2\in B(\vH)} \vA^a_{\nu_1} [\cdot] \vA_{\nu_2}^{a\dagger}\int_{\mathbb{R}} \hat{f}_{\sigma}(\omega - \nu_1)\hat{f}_{\sigma}(\omega - \nu_2)\rd \omega \\
    &= \frac{1}{2}\sum_{a\in A}\sum_{\nu_1,\nu_2\in B(\vH)} \vA^a_{\nu_1} [\cdot] \vA_{\nu_2}^{a\dagger} \e^{-\frac{(\nu_1 - \nu_2)^2}{8\sigma^2}}.
\end{align*}
From here we see that  
\begin{align*}
    \lim_{\sigma\to 0}\BS_H^{[\sigma]}\llbracket\CT\rrbracket[\cdot] = \sum_{a\in A}\sum_{\nu\in B(\vH)} \vA^a_{\nu} [\cdot] \vA_{\nu}^{a\dagger} = \BS_D^{[1]}\llbracket\CT\rrbracket[\cdot]
\end{align*}
since $\lim_{\sigma\to 0}\e^{-\frac{(\nu_1 - \nu_2)^2}{8\sigma^2}} = 0$ if $\nu_1 \neq \nu_2$. In other words, the limit yields the Davies generator with constant $\equiv 1$ function. Applying $\BS_C^{[\gamma]}$ after $\lim_{\sigma\to 0}\BS_H^{[\sigma]}$ yields the Davies generator $\BS_D^{[\gamma]}$ with function $\gamma$, and thus \eqref{eq:interpolation_sigma_0}.

Therefore, we can see that $\BS^{[\sigma,\gamma]}_{I}$ provides a smooth interpolation between the Davies generator and our coherent construction. We do not know whether there is a physical derivation of the master equation that precisely matches this interpolation family, however intuitively speaking the effect of $\sigma$ in this construction resembles the strength of the coupling in the system-bath setting.

\section{Ergodicity in the presence of a full-rank fixed state}\label{sec:ergodic}

In the classical case we know that if the transition matrix $\vM$ is ergodic and $g>0$, then its generalized Glauber/Metropolis version $\vM'$ is also ergodic, thus the resulting dynamics always converges to the unique fixed point. In this section we study to what extent  this generalizes to our quantum constructions.

To study the analogous question for the quantum generalizations we define an appropriate notion of ergodicity. In general one could consider ergodic maps that do not have a full-rank fixed point, however as we focus on detailed-balanced maps with respect to Gibbs states, our maps always have a full-rank fixed point, and therefore we do not need to deal with the nuanced issues arising from the lack of full-rank fixed points. Exploiting this, we define a fitting notion of ergodicity for completely positive maps that simplifies our treatment. 

\begin{definition}[{Ergodic completely positive maps, cf.\ \cite[Definition~3 \& Theorem~1]{burgarth2013ErgodicQChannels}}]\label{def:ergodicity}
	We say that a completely positive map $\CM$ is \emph{ergodic}\footnote{In \cite{wolf2012QChannelsOpsLectureNotes} such maps are called \emph{irreducible}. Indeed, if $\CM$ is ergodic as in \Cref{def:ergodicity}, then by \cite[Theorem 6.2(1.)]{wolf2012QChannelsOpsLectureNotes} it is irreducible. On the other hand, according to~\cite[Theorem 6.2(4.)]{wolf2012QChannelsOpsLectureNotes}, if $\CM$ is irreducible, then for every orthogonal unit vectors $|\psi\rangle$, $|\varphi\rangle\in\BC^d$ there is some $t\in [d-1]$ such that $\tr[\ketbra{\psi}{\psi}\CM^{\circ t}[\ketbra{\varphi}{\varphi}]]>0$, which by~\cite[Theorem 8(iii)]{burgarth2013ErgodicQChannels} implies that $\CM$ is ergodic.} if it has no proper invariant subspaces, i.e., for any non-trivial Hermitian projector $\vPi\notin\{0,\vI\}$ we have that $(\vI-\vPi)\CM[\vPi](\vI-\vPi)\neq 0$.
\end{definition}

As the name suggests, ergodicity is equivalent to the uniqueness of the fixed point.

\begin{theorem}[Ergodic channels]\label{thm:ergodicChannel}
    A quantum channel $\CQ$ is ergodic iff it has a unique fixed state $\vrho_*\succ 0$. If $\CQ$ is ergodic and $\vrho$-detailed balanced, then ${\displaystyle\lim_{t\rightarrow \infty}}\big(\frac{\CQ^{\circ t}}{2}+\frac{\CQ^{\circ(t+1)}}{2}\big)[\cdot]=\vrho\tr[\cdot]$.
\end{theorem}
\begin{proof}
	The first part of the statement follows from \cite[Theorem 6.4]{wolf2012QChannelsOpsLectureNotes}-\cite[Theorem 7]{burgarth2013ErgodicQChannels}.
	If $\CQ$ is $\vrho$-detailed balanced then its spectrum is real and its peripheral spectrum is thus in $\{-1,1\}$ due to~\cite[Theorem~9]{burgarth2013ErgodicQChannels}. The uniqueness of the fixed state implies the $+1$ eigensubspace is one-dimensional, and since the $-1$ eigensubspace of $\CQ$ is in the kernel of $(\frac{\CI}{2}+\frac{\CQ}{2})$ we get the claimed result (note that $\tr[\vO] = 0$ for $\vO$ in the $-1$ eigensubspace or kernel of $\CQ$).
\end{proof}
\begin{theorem}[{Ergodic Lindbladians \cite[Proposition 7.5 \& Theorem 7.2]{wolf2012QChannelsOpsLectureNotes}}]\label{thm:ergodicLindblad}
	A self-adjoint purely irreversible Lindbladian $\CT[\cdot]-\frac12\{\CT^\dagger[\vI],\cdot\}$ has a unique fixed state $\vrho_*=\frac{1}{d}\vI$ iff $\CT$ is ergodic. Moreover, if $\CT$ is ergodic and the Lindbladian $\CL[\cdot]=\CT[\cdot]-\frac12\{\CT^\dagger[\vI],\cdot\}-i[\vC,\cdot]$ has a full-rank fixed state $\vrho \succ 0$, then $\lim_{t\rightarrow \infty}\e^{t\CL}[\cdot]=\vrho\tr[\cdot]$ and $\vrho$ is the unique fixed state of $\CL[\cdot]$.
\end{theorem}
\begin{proof}
	If $\CT$ is self-adjoint, then $\frac{1}{d}\vI$ is a fixed point. Then, by \cite[Theorem 7.2]{wolf2012QChannelsOpsLectureNotes} the kernel of $\CT[\cdot]-\frac12\{\CT^\dagger[\vI],\cdot\}$ equals the commutant of the Kraus operators of $\CT$ (here we used that we can assume without loss of generality that the set of Kraus operators is self-adjoint, see \Cref{foot:adjointKrauses}). Thus, the uniqueness of the fixed point is equivalent to the commutant being trivial.
	So it suffices to show that commutant being trivial is equivalent to $\CT$ being ergodic.
	
	Suppose the commutant is trivial. Take a Hermitian projector $\vPi$ such that $(\vI-\vPi) \vA \vPi=0$ for all  Kraus operators $\vA$ of $\CT$. Due to the self-adjointness of the set of Kraus operators this then also implies that $ \vPi \vA (\vI-\vPi) =0$ for all  Kraus operators $\vA$ implying that $\vI \vA \vI = \vPi \vA \vPi + (\vI-\vPi) \vA (\vI-\vPi)$. This in turn means that $\vA\vPi = \vPi \vA \vPi=\vPi \vA$ for all $\vA$, i.e., $\vPi$ is in the commutant implying that $\vPi\in\{0,\vI\}$ and in turn that $\CT$ is ergodic.
	
	On the other hand, if $\CT$ is ergodic then the matrix algebra generated by the Kraus operators is irreducible \cite[Theorem 8(ii)]{burgarth2013ErgodicQChannels}, and thus the commutant of the Kraus operators is trivial.
	
	Since $\CL[\vrho]=0$ and $\vrho\succ 0$, once again by \cite[Theorem 7.2]{wolf2012QChannelsOpsLectureNotes} the kernel of $\CL^\dagger$ equals the commutant $\{\vC,\vA^a,\vA^{a\dagger}\}'$ for the Kraus operators $\vA^a$ of $\CT$. As before, $\CT$ being ergodic implies that the commutant is trivial, meaning that the kernel of $\CL^\dagger$ is one-dimensional. Due to \cite[Theorem 6.11 \& Proposition 6.2]{wolf2012QChannelsOpsLectureNotes} we know that every eigenvalue $0$ has a trivial Jordan block, meaning that the kernel of $\CL$ is also one-dimensional. Due to \cite[Proposition 7.5(4.-5.)]{wolf2012QChannelsOpsLectureNotes} this implies that $\lim_{t\rightarrow \infty}\e^{t\CL}[\vsigma]=\vrho$ for all quantum state $\vsigma$.
\end{proof}

The above theorems tell us that understanding whether our detailed-balanced maps converge to the Gibbs state or not can be reduced to studying the invariant subspaces of their completely positive part. Therefore, in order to understand this property we shall study the ergodicity-preserving property of our super-superoperators $\BS$ on self-adjoint completely positive maps~$\CT$. 

\begin{definition}
	We say that the super-superoperator $\BS$ is ergodicity preserving if $\BS\llbracket\CT\rrbracket$ is ergodic for every ergodic self-adjoint completely positive map $\CT$.
\end{definition}

If $\BS$ is ergodicity preserving then the constructed Lindbladian converges to the Gibbs state due to \Cref{thm:ergodicLindblad}. To see the analogous implication in the discrete case, observe that if $\CT_1$ is ergodic and $\CT_2$ is completely positive then 
\begin{align*}
	(\vI-\vPi)(\CT_1+\CT_2)[\vPi](\vI-\vPi)\succeq  (\vI-\vPi)\CT_1[\vPi](\vI-\vPi)\neq 0,
\end{align*}
implying that by adding completely positive maps, such as a Kraus term for decay (\Cref{lem:FindingDiscDecayTerm}), one cannot remove ergodicity. Therefore, once again if $\BS$ is ergodicity preserving, then our construction yields a quantum channel that converges to the Gibbs state in the sense\footnote{This result also implies that the ``lazy'' version of the channel $(\frac{\CI}{2}+\frac{\CQ}{2})^{\circ t}$ converges in the usual sense.}~of~\Cref{thm:ergodicChannel}.

The above considerations in particular imply that, e.g., if we have an $n$-qubit system and $\CT$ includes all single-site Pauli
$\vX$ and $\vZ$ as Kraus operators, then $\CT$ is ergodic. Consequently if $\BS$ is ergodicity preserving then the resulting discrete and continuous-time quantum dynamics will certainly converge to the Gibbs state, and only the rate of convergence is the question.

\subsection{Ergodicity preservation of the described super-superoperators}

Examining \Cref{def:ergodicity}, one can see that decomposing the Kraus operators to finitely many pieces, or weighing them by strictly positive weights, can never create new invariant subspaces.
This immediately implies that the Davies super-superoperator is ergodicity preserving whenever $\gamma(\nu)>0$ for all $\nu\in B(\vH)$.

\begin{theorem}
	$\BS_D^{[\gamma]}$ is ergodicity preserving whenever $\gamma(\nu)>0$ for all $\nu\in B(\vH)$.
\end{theorem}
\begin{proof}
Let $m = \min_{\nu \in B(\vH)} \gamma(\nu) >0$. Then, $\gamma  = (\gamma -m) + m $, and 
\begin{align}
    \BS_D^{[\gamma]} = \BS_D^{[m]} +\BS_D^{[\gamma - m]}.
\end{align}
Therefore, it suffices to show the unweight Davies generator (through $\BS_D^{[m]}$) remains ergodic.
For any nontrivial projector $\vec{\Pi}$, 
\begin{align*}
     (\vI-\vec{\Pi}) \sum_{\nu_1\in B(\vH)} \vA^a_{\nu_1} \vec{\Pi} \sum_{\nu_2\in B(\vH)} \vA^{a\dagger}_{\nu_2}(\vI-\vec{\Pi})\ne 0 &\implies (\vI-\vec{\Pi}) \sum_{\nu_1\in B(\vH)} \vA^a_{\nu_1} \vec{\Pi} \ne0\\
     &\implies \exists \nu\in B(\vH) ~\text{s.t.}~ \tr[(\vI-\vec{\Pi}) \vA^a_{\nu} \vec{\Pi} \vA^{a\dagger}_{\nu}] > 0\\
     &\iff \sum_{\nu\in B(\vH)} (\vI-\vec{\Pi}) \vA^a_{\nu} \vec{\Pi} \vA^{a\dagger}_{\nu}(\vI-\vec{\Pi}) \ne 0. \qedhere
\end{align*}
\end{proof}

We port this observation to the continuous setting by exploiting the observation that the operator Fourier transform decomposes the Kraus operators as in \eqref{eq:ForierDecomposition} (a similar argument works for the construction described in \Cref{sec:PEVariant}).

\begin{theorem}
	$\BS_H^{[\sigma,g]}$ is ergodicity preserving if $\nrm{g}_1>0$.
\end{theorem}
\begin{proof}	
	If $\CT[\cdot]=\sum_{a\in A}\vA^a[ \cdot ]\vA^{a\dagger}$ is ergodic, then for every non-trivial orthogonal projector $\vPi\notin\{0,\vI\}$ we have that $(\vI-\vPi)\CT[\vPi](\vI-\vPi)\neq 0$, implying the existence of a label $a\in A$ such that $(\vI-\vPi)\vA^a\vPi\neq 0$. Due to \eqref{eq:ForierDecomposition} this then implies that $\int_{\BR}(\vI-\vPi)\CF^{[\sigma]}(\omega)[\vA^a]\vPi\rd\omega\neq 0$. 
	In particular, we get that $\int_{\BR}\|(\vI-\vPi)\CF^{[\sigma]}(\omega)[\vA^a]\vPi\|_2^2\rd\omega > 0$. Since $\nrm{g}_1>0$ we also have that $\gamma^{(g)}(\omega)$ is strictly greater than $0$ for every $\omega\in\BR$, and therefore
	\begin{align*}
		0&<\sum_{a\in A}\int_{\BR}\gamma^{(g)}(\omega)\nrm{(\vI-\vPi)\CF^{[\sigma]}(\omega)[\vA^a]\vPi}_2^2\rd\omega\\&
		=\sum_{a\in A} \int_{\BR}\gamma^{(g)}(\omega)\tr\left[(\vI-\vPi)\left(\CF^{[\sigma]}(\omega)[\vA^a]\right)\vPi^2\left(\CF^{[\sigma]}(\omega)[\vA^a]\right)^\dagger (\vI-\vPi)\right]\rd\omega\\&
		=\tr\left[(\vI-\vPi)\sum_{a\in A}\int_{\BR}\gamma^{(g)}(\omega)\left(\CF^{[\sigma]}(\omega)[\vA^a]\right)\vPi\left(\CF^{[\sigma]}(\omega)[\vA^a]\right)^\dagger \rd\omega(\vI-\vPi)\right]\\&
		=\tr\left[(\vI-\vPi)\left(\BS_H^{[\sigma,g]}\llbracket\CT\rrbracket\right)[\vPi](\vI-\vPi)\right]. \tag*{\qedhere}
	\end{align*}
\end{proof}

The coherent construction described in \Cref{sec:coherentDyn} might destroy ergodicity, but that can only happen ``accidentally'' on a zero-measure set of the self-adjoint CP maps, as we prove below. 

	First observe that the set of self-adjoint maps $\CM\colon \BC^{d\times d}\rightarrow \BC^{d\times d}$ is isomorphic to the set of self-adjoint $\vCM\in\BC^{d^2\times d^2}$ matrices according to \Cref{foot:superAdjoint}, which is in turn isomorphic to $\BR^{d^4}$ by the standard real parametrisation of Hermitian matrices. 
	Due to the convex structure of completely positive maps we get that they form a convex cone which is also full rank, i.e., $d^4$-dimensional over the reals. 
	This is because for any self-adjoint map $\CM$, for a large enough $\lambda>0$, the map $\CM[\cdot]+\lambda\sum_{i,j=1}^d\ketbra{i}{j}[\cdot]\ketbra{j}{i}$ is completely positive, which can be seen using the Choi-Jamiołkowski representation~\cite[Proposition 2.1]{wolf2012QChannelsOpsLectureNotes} and \Cref{foot:Choi}. 
	This in particular means that the Lebesgue measure of the cone of all self-adjoint CP maps is $\infty$.

\begin{theorem}
	If $\gamma(\nu)>0$ for all $\nu\in B(\vH)$, then $\BS_C^{[\gamma]}\llbracket\CT\rrbracket$ is ergodic for almost all self-adjoint CP map $\CT$, apart from a zero-measure set, where the measure is defined through the CP map's matrix and the $d^4$-dimensional (real) Lebesgue measure on the $d^2\times d^2$ Hermitian matrices.
\end{theorem}
\begin{proof}
	Observe that the $\vrho$-detailed-balanced CP map $\CT'$ has the state $\vrho\succ 0$ as a unique fixed point iff the Hermitian matrix $\big(\vrho^{\,\frac14}\otimes\vrho^{*\frac14}\big)\vCT'\big(\vrho^{-\frac14}\otimes\vrho^{*-\frac14}\big)+\ketbra{\sqrt{\vrho}}{\sqrt{\vrho}}$ has non-zero determinant.\footnote{Here we used the definition $\ket{\sqrt{\vrho}}\propto\sum_{i=1}^{d}\e^{-\frac{E_i}{2}}\ket{\psi_i}\ket{\psi_i^*}$ of the purified Gibbs state.}
	The squared absolute value of this determinant can be expressed as a polynomial $p$ of the matrix elements of $\vCT'=\BS_C^{[\gamma]}\llbracket\vCT\rrbracket$ and, in turn, due to linearity as a real polynomial of the $d^4$ independent real numbers underlying $\vCT$. 
	We also know that there are maps $\CT$ that lead to a unique fixed point (e.g., Davies because $\gamma(\nu)>0$), therefore $p\neq 0$. 
	Since the zero-set of any non-zero polynomial over $\BR^{d^4}$ has measure $0$, we get that the set of completely positive maps $\CT$ for which $\vrho\succ 0$ is not the unique fixed point of $\BS_C^{[\gamma]}\llbracket\vCT\rrbracket$ has measure zero.
\end{proof}

\section{Spectral gap for high-temperature quantum Gibbs samplers}
\label{sec:spectral_gap}

In the previous section, we showed that our quantum Gibbs constructions are (almost surely) ergodic, meaning that there is a unique stationary state. In this section, we take one step further and show that our quantum constructions are also gapped in the high-temperature regime. This immediately implies that our constructions converge fast to the Gibbs state above a certain constant threshold temperature for Hamiltonians obeying a Lieb-Robinson bound~\cite{lieb1972FiniteGroupVelocity}. For simplicity, we will focus on the coherent reweighing construction $\BS_C^{[f]}\llbracket\CT\rrbracket$ from \Cref{sec:coherentDyn} with the Glauber function $f_G(\nu) = \frac{1}{2} - \frac{1}{2}\tanh(\frac{\beta \nu}{4})$ and on local Hamiltonians defined on a lattice. Our proof is inspired by~\cite{rouze2024EffThermalizationGibbsSamp}, which shows this property for the earlier construction of~\cite{chen2023ExactQGibbsSampler} (\Cref{sec:OFT}).

The physical setting in this section is a local Hamiltonian $\vH$ defined on the lattice $\Lambda = [1,L]^D$. The local Hamiltonian assigns to every finite sets $Z\subseteq \Lambda$ a self-adjoint ``interaction term'' $\mathbf{h}_Z$ supported on $Z$ such that $\max_{Z\subseteq \Lambda}\|\mathbf{h}_Z\| \leq h$ (typically $\mathbf{h}_Z=0$ whenever $|Z|>2$). The Hamiltonian of the system in a region $R\subseteq\Lambda$ is thus defined as the sum of terms belonging to $R$
\begin{align*}
    \vH_R := \sum_{Z\subseteq R} \mathbf{h}_Z.
\end{align*}
Moreover, we assume that $\vH$ is a $(k,\ell)$-local Hamiltonian, meaning that each interaction $\mathbf{h}_Z$ has support on at most $k$ sites and each site $a\in\Lambda$ appears on at most $\ell$ non-zero interactions $\mathbf{h}_Z$. We assume that $h, k, \ell$ are independent of the system size $|\Lambda|$. Finally, we define the so-called Lieb-Robinson velocity $J := \max_{a\in\Lambda}\sum_{Z \ni a}|Z|\|\mathbf{h}_Z\| \leq hk\ell$.

\paragraph{Continuous-time construction.} 
For the continuous-time setting, consider the local CP map $\CT[\cdot] = \sum_{a\in\Lambda}\sum_{\alpha\in[3]}\vA^{a,\alpha}[\cdot]\vA^{a,\alpha}$ defined on $\Lambda$, where $\vA^{a,1} = \vsigma_{x}$, $\vA^{a,2} = \vsigma_{y}$, and $\vA^{a,3} = \vsigma_{z}$ denote the $1$-local Pauli matrices on site $a\in\Lambda$. Then the detailed-balanced Lindbladian we are interested in is
\begin{align}\label{eq:lindbladian_mixing_time}
    \CL_{\beta}[\cdot] &= \sum_{a\in\Lambda}\sum_{\alpha\in[3]} \CS^-[\vA^{a,\alpha}][ \cdot ]\CS^+[\vA^{a,\alpha}] -\frac{1}{2}\{\vD^{a,\alpha},\cdot\} + [\CS[\vD^{a,\alpha}],\cdot] \\
    \text{where}\quad\vD^{a,\alpha} &= \CS^+[\vA^{a,\alpha}] \CS^-[\vA^{a,\alpha}]. \nonumber
\end{align}
Here we made explicit the dependence on $\beta$. It will be more convenient to work with the discriminant $\CD_{\beta}[\cdot]$ of the corresponding construction, which has the same spectrum as $\CL_{\beta}[\cdot]$ since they are related by a similarity transformation. Moreover, $\CD_{\beta}[\cdot]$ is self-adjoint with respect to the Hilbert-Schmidt inner product (since $\CL_{\beta}[\cdot]$ is detailed balanced), and can thus be interpreted as a Hamiltonian. The discriminant was already computed in \Cref{sec:coherentDynDisc} and is
\begin{align*}
    \CD_{\beta}[\cdot]  &= \sum_{a\in\Lambda}\sum_{\alpha\in[3]} \CD_{\beta}^{a,\alpha}[\cdot] = \sum_{a\in\Lambda}\sum_{\alpha\in[3]} \CS_c[\vA^{a,\alpha}][\cdot]\CS_c[\vA^{a,\alpha}] - \CS_c[\vD^{a,\alpha}][\cdot]-[\cdot]\CS_c[\vD^{a,\alpha}],\\
    \text{where}\quad \CS_c[\vD^{a,\alpha}] &= \frac{1}{2}\CS_c\big[\vI + [\CS[\vA^{a,\alpha}],\vA^{a,\alpha}]\big] - \frac{1}{2}\CS_c[\vA^{a,\alpha}]\CS_c[\vA^{a,\alpha}].
\end{align*}
Recall that $\vA_\beta^{a,\alpha}(t) := \e^{\ri \beta t \vH}\vA^{a,\alpha}\e^{-\ri \beta t\vH}$. The terms $\CS_c[\vI]$ and $\CS_c[\vA^{a,\alpha}]\CS_c[\vA^{a,\alpha}]$ are straightforward to compute, while
\begin{align*}
    \CS_c[[\CS[\vA^{a,\alpha}], \vA^{a,\alpha}]] &= \int_{-\infty}^\infty \frac{1}{\cosh(2\pi t)}\e^{\ri \beta t \vH}[\CS[\vA^{a,\alpha}], \vA^{a,\alpha}] \e^{-\ri \beta t \vH} \rd t\\
    &= \iint_{-\infty}^\infty \! \frac{1}{2\pi\cosh(2\pi t)\sinhc(2\pi t')} \!\int_0^1 \! \e^{\ri \beta t \vH}[[\vA_\beta^{a,\alpha}(st'), \beta\vH],\vA^{a,\alpha}]\e^{-\ri \beta t \vH}\rd s \rd t \rd t'\\
    &= \iint_{-\infty}^\infty \! \frac{1}{2\pi\cosh(2\pi t)\sinhc(2\pi t')}\int_0^1[[\vA^{a,\alpha}_\beta( t+ st'), \beta\vH],\vA_\beta^{a,\alpha}(t)]\rd s \rd t \rd t'.
\end{align*}
Therefore, each term $\CD_{\beta}^{a,\alpha}$ of the discriminant can be rewritten as
\begin{align*}
    \CD_{\beta}^{a,\alpha}[\cdot] = \vT^{a,\alpha}_{\beta}[\cdot]\vT^{a,\alpha}_{\beta} - (\vM^{a,\alpha}_{\beta} + \vN^{a,\alpha}_{\beta})[\cdot] - [\cdot](\vM^{a,\alpha}_{\beta} + \vN^{a,\alpha}_{\beta}),
\end{align*}
where
\begin{align*}
    \vT^{a,\alpha}_{\beta} &= \int_{-\infty}^\infty \frac{1}{\cosh(2\pi t)}  \vA_\beta^{a,\alpha}(t) \rd t, 
    \\
    \vM^{a,\alpha}_{\beta} &= \iint_{-\infty}^\infty \frac{1}{4\pi\cosh(2\pi t)\sinhc(2\pi t')}\int_0^1[[\vA_\beta^{a,\alpha}(t+ st'), \beta\vH],\vA_\beta^{a,\alpha}(t)]\rd s \rd t \rd t',
    \\
    \vN^{a,\alpha}_{\beta} &= \frac{\vI}{4} - \iint_{-\infty}^\infty \frac{1}{2\cosh(2\pi t)\cosh(2\pi t')} \vA_\beta^{a,\alpha}(t)  \vA_\beta^{a,\alpha}(t') \rd t\rd t'.
\end{align*}

\paragraph*{Discrete-time construction.} For the discrete-time setting, first let $s := 1+\frac{1}{\pi}\ln\big(1+\frac{\beta\nrm{\vH}}{2}\big)$. Once again we consider the local CP map $\CT[\cdot] = \frac{1}{3|\Lambda|}\sum_{a\in\Lambda}\sum_{\alpha\in[3]}\frac{\vA^{a,\alpha}}{s}[\cdot]\frac{\vA^{a,\alpha}}{s}$ defined on $\Lambda$, but now the $1$-local Pauli matrices are normalized by $s$. Applying our coherent reweighing construction to $\CT$ plus our discrete-time renormalization from \Cref{lem:FindingDiscDecayTerm} yields the detailed-balanced local channel
\begin{align*}
    \CQ_{\beta}[\cdot] &= \frac{1}{3|\Lambda|}\sum_{a\in\Lambda}\sum_{\alpha\in[3]}\left( \frac{1}{s^2}\CS^-[\vA^{a,\alpha}][ \cdot ] \CS^+[\vA^{a,\alpha}] + \vK^{a,\alpha}[\cdot](\vK^{a,\alpha})^\dagger \right),\\
    \text{where}~\vK^{a,\alpha} &= \sqrt{\sqrt{\vrho}(\vI-\vD^{a,\alpha})\sqrt{\vrho}}\vrho^{-\frac12} \quad\text{and}\quad \vD^{a,\alpha} = \frac{1}{s^2}\CS^+[\vA^{a,\alpha}]\CS^-[\vA^{a,\alpha}].
\end{align*}
We note that the Kraus operators $\vK^{a,\alpha}$ are well defined since $\|\vD^{a,\alpha}\| \leq 1$ due to $\|\CS^+\|_{\infty\to\infty} = \|\CS^-\|_{\infty\to\infty} \leq 1 + \frac{1}{\pi}\ln\left(1+\frac{\beta\|\vH\|}{2}\right) \leq s$ according to \Cref{prop:TruncatedIntegral}. 

    As shown in \Cref{thr:discrete_vs_continuous_spectral_gap}, the spectral gap of $\CQ_{\beta}$ is always larger than the spectral gap of the continuous-time dynamics $\CD_{\beta}$, therefore we need only to focus on computing the spectral gap of $\CD_{\beta}$, for which we employ well-known results regarding the stability of gapped Hamiltonians under small quasi-local perturbations~\cite{bravyi2010topological,michalakis2013stability}.

\subsection{Stability of frustration-free gapped Hamiltonians}

The main idea for proving that $\CD_{\beta}$ is gapped is by viewing it as a perturbation of its infinite-temperature counterpart $\CD_{0}$. If the infinite-temperature discriminant $\CD_{0}$ is gapped and possesses some properties described below, then the quasi-local perturbation $\CD_{\beta} - \CD_{0}$ will preserve the original gap up to some constant factor according to a result by Michalakis and Zwolak~\cite{michalakis2013stability}, which we now review.

Consider a Hamiltonian $\vH_0$ defined on $\Lambda$ satisfying:\footnote{Do not confuse $\vH_0$ with $\vH$ from the previous section. $\vH_0$ will actually be $\CD_{0}$ when applying the result of Michalakis and Zwolak to our case.}
\begin{enumerate}
    \item (\textbf{Spatially-local}) The unperturbed Hamiltonian can be written as $\vH_0 = \sum_{a\in\Lambda} \vQ_a$ where each interaction $\vQ_a$ acts non-trivially on the Hilbert space supported on $B_a(1)$, where $B_a(r)$ is the ball of radius $r$ centered at $a\in\Lambda$;
    \item (\textbf{Periodic-boundary}) $\vH_0$ satisfies periodic boundary conditions;
    \item (\textbf{Frustration-free}) $\vQ_a \vP_0 = 0$ for all $a\in\Lambda$, where $\vP_0$ is the projector onto the groundstate subspace of $\vH_0$;
    \item (\textbf{Gapped}) $\vH_0$ has spectral gap $\gamma > 0$ above the groundstate subspace for all $L\geq 2$. 
\end{enumerate}

Apart from the above properties, we also review the \emph{Local Topological Quantum Order (Local-TQO)} condition, which corresponds to local-indistinguishability of the groundstates of local Hamiltonians, and the \emph{Local-Gap} condition. In the following, let $\vP_{Z}(\eps)$ be the projection onto the subspace of eigenstates of $\vH_Z$ with energy at most $\eps \ge 0$ and $\vP_{Z} := \vP_{Z}(0)$.
\begin{definition}[Local-TQO]
    Let $r\leq L^\ast < L$ for some cut-off parameter $L^\ast$ and $\vH_0$ be a Hamiltonian. For any operator $\vO_{B_a(r)}$ supported on $B_a(r)$, let
    \begin{align*}
        c_s(\vO_{B_a(r)}) := \frac{\tr[\vP_{B_a(r+s)}\vO_{B_a(r)}]}{\tr[\vP_{B_a(r+s)}]} \qquad\text{for}~s\in\mathbb{N}.
    \end{align*}
    The Hamiltonian $\vH_0$ satisfies the \emph{Local-TQO} condition if, for each fixed $1\leq s \leq L-r$,
    \begin{align*}
        \| \vP_{B_a(r+s)} \vO_{B_a(r)} \vP_{B_a(r+s)} - c_s(\vO_{B_a(r)}) \vP_{B_a(r+s)}\| \leq \|\vO_{B_a(r)}\| \Delta_0(s),
    \end{align*}
    where $\Delta_0(s)$ is a decaying function of $s$.
\end{definition}

\begin{definition}[Local-Gap]
    A Hamiltonian $\vH_0$ is locally gapped with gap $\gamma(r) > 0$ if and only if, for each $a\in\Lambda$ and $r\geq 0$, $\vP_{B_a(r)}(\gamma(r)) = \vP_{B_a(r)}$. When $\gamma(r)$ decays at most polynomially in $r$, we say that $\vH_0$ satisfies the \emph{Local-Gap} condition.
\end{definition}

The Local-TQO condition, introduced by Michalakis and Zwolak~\cite{michalakis2013stability}, generalises previous stability assumptions for the commuting case~\cite{bravyi2010topological,bravyi2011ShortProofStabilityTopologicalOrderLocalPert} and guarantees that the gap of $\vH_0$ does not collapse when a quasi-local perturbation is added. The notion of quasi-locality for such a perturbation is made precise in the next definition.
\begin{definition}\label{def:R_f_perturbation}
    A perturbation $\vV$ is said to have strength $R$ and decay $g(r)$, with $g(r) \leq 1$, $r\geq 0$, if and only if it can be written as $\vV = \sum_{a\in\Lambda}\sum_{r=0}^L \vV_a(r)$, such that $\vV_a(r)$ has support on $B_a(r)$ and satisfies $\|\vV_a(r)\| \leq R g(r)$, for some rapidly decaying function $g(r)$, faster than $(1+r)^{-(D+2)}$. We write that $\vV$ is a $(R,g)$-perturbation.
\end{definition}
Michalakis and Zwolak~\cite{michalakis2013stability} proved the following result on the spectral gap of $\vH_0 + \vV$ for a $(R,g)$-perturbation $\vV$.
\begin{fact}[{\cite[Theorem~1]{michalakis2013stability}}]\label{fact:spectral_gap}
    Let $\vH_0$ be a frustration-free Hamiltonian with spectral gap $\gamma$ and satisfying the Local-TQO and Local-Gap conditions. For a $(R,g)$-perturbation $\vV$, there are constants $R_0>0$ and $L_0 \geq 2$ such that, for $R\leq R_0$ and $L\geq L_0$, the spectral gap of $\vH_0 + \vV$ is at least $\gamma/2$.
\end{fact}
The constants $R_0$ and $L_0$ depend on the unperturbed Hamiltonian $\vH_0$, the form of the decay $g(r)$, and on the dimension $D$, and can be determined from~\cite[Proposition~2]{michalakis2013stability}.

\subsection{Spectral gap in the high-temperature regime}

\subsubsection{Continuous-time construction}

We now apply \Cref{fact:spectral_gap} to the discriminant $\CD_{\beta}$ in order to bound its spectral gap. We consider the infinite-temperature term $\CD_{0}$ as an unperturbed Hamiltonian and the difference $\CD_{\beta} - \CD_{0}$ as a perturbation. The first step is to check whether $\CD_{0}$ is gapped and possesses all the aforementioned properties. It is not hard to see that
\begin{align*}
    \lim_{\beta \to 0} \vT_{\beta}^{a,\alpha} = \frac{\vA^{a,\alpha}}{2}, \qquad \lim_{\beta \to 0} \vM_{\beta}^{a,\alpha} = 0, \quad\text{and}\quad    \lim_{\beta \to 0} \vN_{\beta}^{a,\alpha} = \frac{\vI}{8}. 
\end{align*}
Therefore, $\CD_{0}[\vX] = \sum_{a\in\Lambda}\sum_{\alpha\in[3]} \frac{1}{4}(\vA^{a,\alpha}\vX\vA^{a,\alpha} - \vX) = \sum_{a\in\Lambda}(\frac{1}{2}\tr_a[\vX] - \vX)$, which is the generator of the local depolarizing channel with spectral gap $\gamma_{0} = 1$. Moreover, it is frustration-free and satisfies the Local-TQO (with sharp cutoff) and Local-Gap conditions. Hence, it only remains to show that $\CD_{\beta} - \CD_{0}$ is a $(R,g)$-perturbation for some $R>0$ and a rapidly decaying function $g(r)$. For such, we decompose the perturbation $\CD_{\beta}^{a,\alpha} - \CD_{0}^{a,\alpha}$ into a telescopic sum as
\begin{align*}
    \CD_{\beta}^{a,\alpha} - \CD_{0}^{a,\alpha} = (\CD_{\beta,0}^{a,\alpha} - \CD_{0,0}^{a,\alpha}) + \sum_{r=0}^\infty (\CD_{\beta,r+1}^{a,\alpha} - \CD_{\beta,r}^{a,\alpha}),
\end{align*}
where for any $r\in\mathbb{N}$, $\CD_{\beta,r}^{a,\alpha}$ is the discriminant defined by replacing $\vH$ with the reduced Hamiltonian $\vH_{B_a(r)}$. This means that the perturbation $\CD_{\beta,r+1}^{a,\alpha} - \CD_{\beta,r}^{a,\alpha}$ is supported on $B_a(r+1)$, as required by \Cref{def:R_f_perturbation}. The main idea is then to control the norms $\|\CD_{\beta,0}^{a,\alpha} - \CD_{0,0}^{a,\alpha}\|_{2\to 2}$ and $\|\CD_{\beta,r+1}^{a,\alpha} - \CD_{\beta,r}^{a,\alpha}\|_{2\to 2}$ by showing that they decay sufficiently fast with the radius $r$. We prove in \Cref{app:spectral_gap_rapidly_decaying} the following result.

\begin{restatable}{theorem}{decayingfunctions}\label{thr:decaying_functions_continuous}
    Consider a $(k,\ell)$-local Hamiltonian $\vH$ with Lieb-Robinson velocity $J \leq hk\ell$. Then
    \begin{align*}
        \|\CD_{\beta,0}^{a,\alpha} - \CD_{0,0}^{a,\alpha}\|_{2\to 2} \leq 0.38 \beta h \ell.      
    \end{align*}
    Moreover, let $\chi(x) := \frac{x}{1+x}$. There are constants $c_1,c_2>0$ such that, for all $r\in\mathbb{N}$, 
    \begin{align*}
        \|\CD_{\beta,r+1}^{a,\alpha} - \CD_{\beta,r}^{a,\alpha}\|_{2\to 2} \leq c_1 (1 + \beta J \ell )\frac{\chi(\beta J)^r}{\sqrt{r}} + (\beta h + \beta^2 h^2\ell)\left(\frac{c_2 r}{\sqrt{D}}\right)^{D+1}\e^{-\Omega\big(r \frac{\chi(\beta J)}{\beta J}\big)}.
    \end{align*}
\end{restatable}
According to the above result, $\|\CD_{\beta,r+1}^{a,\alpha} - \CD_{\beta,r}^{a,\alpha}\|_{2\to 2}$ is upper bounded by an exponentially decaying function in $r$, which is sufficient for applying \Cref{fact:spectral_gap}. This yields the sought-after result from this section.
\begin{theorem}\label{thr:spectral_gap_continuous}
    For any $(k,\ell)$-local Hamiltonian $\vH$ defined on a $D$-dimensional lattice $\Lambda$, there is a constant $\beta^\ast > 0$ independent of $|\Lambda|$ such that, for all $\beta < \beta^\ast$, the spectral gap of $\CD_{\beta}$ is lower bounded by $\frac{1}{2}$.
\end{theorem}

\subsubsection{Discrete-time construction}

As mentioned above, it is not necessary to directly compute the spectral gap of the discrete-time dynamics $\CQ_\beta$ like we did with $\CD_\beta$ in the previous section, as it follows immediately from \Cref{thr:discrete_vs_continuous_spectral_gap} together with \Cref{thr:spectral_gap_continuous}.
\begin{corollary}
    Consider any $(k,\ell)$-local Hamiltonian $\vH$ defined on a $D$-dimensional lattice $\Lambda$ and let $s = 1 + \frac{1}{\pi}\ln\big(1 + \frac{\beta\|\vH\|}{2} \big)$. There is $\beta^\ast>0$ such that, for all $\beta < \beta^\ast$, the spectral gap of $\CQ_\beta$ is lower bounded by $\frac{1}{6|\Lambda|s^2}$.
\end{corollary}
\begin{proof}
    First decompose the Lindbladian from \eqref{eq:lindbladian_mixing_time} as $\CL_\beta = \sum_{a\in\Lambda}\sum_{\alpha\in[3]} \CL_\beta^{a,\alpha}$ and for each $(a,\alpha)\in \Lambda\times [3]$, write $\CQ_\beta^{a,\alpha}[\cdot] = \frac{1}{3|\Lambda|}(\frac{1}{s^2}\CT_\beta^{a,\alpha}[\cdot] + \vK^{a,\alpha}[\cdot](\vK^{a,\alpha})^\dagger)$, where $\CT^{a,\alpha}_\beta[\cdot] := \CS^-[\vA^{a,\alpha}][\cdot]\CS^+[\vA^{a,\alpha}]$. According to \Cref{thr:discrete_vs_continuous_spectral_gap}, then $\frac{1}{s^2}\CL_\beta^{a,\alpha} \succeq 3|\Lambda|\CQ_\beta^{a,\alpha} - \CI$. Summing over all $(a,\alpha)\in \Lambda\times [3]$, $\CQ_\beta - \CI \preceq \frac{1}{3|\Lambda|s^2}\CL_\beta$. According to \Cref{thr:spectral_gap_continuous}, there is a constant $\beta^\ast>0$ independent of $|\Lambda|$ such that, for all $\beta<\beta^\ast$, the spectral gap of $\CL_\beta$ is larger than $\frac{1}{2}$. Therefore, the spectral gap of $\CQ_\beta$ is larger than $\frac{1}{6|\Lambda|s^2}$ for all $\beta<\beta^\ast$. 
\end{proof}

\section{Efficient implementation on a quantum computer}\label{sec:imp}

Our algorithms work under a very general and natural input model inspired by the success of the block-encoding framework in quantum linear algebra algorithms. Indeed, we only assume that $\CT$ is given via a block-encoding of a dilation $\vG=\sum_{a\in A} \ket{a} \otimes \vA^a$. This enables us to extend the transformation devised for a single jump operator $\vA^a$ to the first register trivially by tensoring with $\vI$. 
This way we reduce the (approximate) implementation of a block-encoded dilation $\vG':=\sum_{a\in A} \ket{a} \otimes \CS^-[\vA^a]$ of $\CT'$ to the (approximate) implementation of $\CS^-=\frac{\mathcal{I}}{2}-\CS$, and thus in turn to that of $\CS$. This then also directly enables the implementation of the coherent term $\vC=\CS[\vG'^\dagger \vG']$. Having access to (approximate) block-encodings of $\vG'$ and $\vC$ we can use the Lindbladian simulation algorithm of~\cite{chen2023QThermalStatePrep} to efficiently implement the continuous-time evolution $\e^{t\CL}$ with $\bigOt{t}$ gate complexity and $\bigOt{t}$ uses of the block-encodings of $\vG'$ and $\vC$.

Thus, we begin by introducing our efficient approximate implementation of $\CS$.
\SImp*
\begin{proof}
    Let $\widetilde{\CS}$ be the $\frac{\eps}{2}$-approximation of $\CS$ given by the truncated integral in \Cref{prop:TruncatedIntegral},
	\begin{align}\label{eq:LCUIntTrunc}
		\widetilde{\CS}[\vA]\! :=\! \int_{\Delta} \frac{\ri}{\sinh(2\pi t)}\e^{\ri \vH t} \vA \e^{-\ri \vH t} \rd t ~\text{where}~ 
		\Delta:=\!\left[-\frac{\ln(4/\eps)}{2\pi},\frac{\ln(4/\eps)}{2\pi}\right]\setminus\left(-\frac{\eps}{2\nrm{\vH}},\frac{\eps}{2\nrm{\vH}}\right).
	\end{align}
    It is possible to prove that (see, e.g.,~\eqref{eq:sinhcInt})
	\begin{align*}
		  \nrm{\frac{\ \pmb{1}_\Delta(t)}{\sinh(2\pi t)}}_1 \leq \frac1\pi\ln\left(1+\frac{2\nrm{\vH}}{\pi \eps}\right).
	\end{align*}
	Dividing both sides of \eqref{eq:LCUIntTrunc} we can see that \eqref{eq:LCUGoal} is defined by weighted Heisenberg evolution for some weight function $f$ such that $\nrm{f}_1\leq 1$. Since Heisenberg evolution up to time $t$ is $t\nrm{\vH}$-Lipschitz continuous, and $\frac{\ri}{\sinh(2\pi t)}$ is $\frac{2\|\vH\|^2}{\pi\varepsilon^2}$-Lipschitz continuous on $\Delta$, it suffices to discretize the integral with $\bigOt{\frac{\eps^3}{\nrm{\vH}^2}}$ equidistant evaluation points on $\Delta$ in order to achieve an $\frac{\eps}{2}$-approximation. Then by applying the linear combination of unitaries technique~\cite{childs2012HamSimLCU,gilyen2018QSingValTransf} we can obtain an $\frac{\eps}{2}$-approximation of the integral \eqref{eq:LCUIntTrunc} by sandwiching $\vU$ (and thus $\vA$) by $\bigO{1}$ controlled Hamiltonian simulation unitaries. The associated state-preparation problem can be solved efficiently with $\bigOt{1}$ time-complexity using, e.g., \cite{mcArdle2022QStatePreparationWOArithm} because $\frac{\ri}{\sinh(2\pi t)}$ is a nice analytic function whose poles are on the imaginary axis, meaning that we can efficiently approximate it by a $\bigOt{1}$-degree polynomial on $\Delta$.
\end{proof}

For the discrete-time case we can combine the implementation of the dilation $\vG'$ of $\CT'$ with the implementation of the single Kraus $\vK$ operator for the reject part to get a $\bigOt{1}$ subnormalized dilation of the channel $\CQ$.
This can then be converted to an implementation of the channel by $\bigOt{1}$ uses of oblivious amplitude amplification~\cite{cleve2016EffLindbladianSim,gilyen2018QSingValTransf}. Thus the problem is reduced to implementing an (approximate) block-encoding of $\vK$ which can be done via \Cref{lem:QSVTShortcut} or \Cref{cor:recursive_int_rep} in combination with LCU. When using \Cref{cor:recursive_int_rep} in combination with, e.g., \eqref{eq:localitiyParams}, the implementation still uses only $\bigOt{1}$ controlled Hamiltonian simulation time, but we do not bound the gate complexity of the LCU-state-preparation cost for higher-order terms. We leave it as an open problem whether it can be done efficiently for arbitrary high order.

\section{Physical origin of our new continuous-time construction}\label{sec:physMotivation}
We have considered a number of quantum extensions of the classical Metropolis algorithm and Glauber dynamics. A priori, these constructions are purely algorithmic: we wanted to prepare the Gibbs state, so we impose exact KMS detailed balance and hope for an efficient digital implementation on a quantum computer. However, classical MCMC algorithms, especially the Glauber dynamics, sometimes serve as a mathematically succinct \textit{model} of open-system thermodynamics. Here, we consider the relationship between our abstract constructions and the naturally occuring dynamics of a system weakly coupled to a large thermal bath.
Remarkably, both of the continuous-time construction from~\cite{chen2023ExactQGibbsSampler} and in this work locally match physically derived master equation. Therefore, our algorithmic constructions may serve a similar role to Glauber dynamics as a clean mathematical model of thermalization. 

We recall  the standard setup in open systems textbooks considers a system governed by a Hamiltonian $\vH_S$ that interacts weakly with a large thermal bath. The total Hamiltonian is then modelled by $\vH_{\rm tot}=\vH_S + \vH_B + \vH_{\rm int}$. The system bath interaction can be decomposed as $\vH_{\rm int}= \sqrt{\eta}\sum_{a\in A} \vA^a\otimes \vB^a$, where $\{\vA^a\}_{a\in A}$ are \textit{jump operators} on the system, $\{\vB^a\}_{a\in A}$ are operators on the bath, and $\eta$ is the coupling constant.

While the joint system-bath evolution appears only harder to analyze, one can obtain an \textit{effective} master equation, or a Lindbladian, governing the system dynamic alone, under two assumptions: (i) the Markovian assumption and (ii) the weak-coupling assumption. Roughly, the Markovian assumption states that the state of the bath remains unentangled with the system for all time, and is in a thermal state (of the bath). This assumption is widely considered valid when the system and bath have diverging time scales. The weak coupling assumption assumes that the interaction between the system and the bath is weak. The weaker the coupling is, the longer the master equation remains valid. This weak-coupling assumption is also manifested in the fact that most master equations are second-order in the jumps $\vA^a$ since the fourth-order terms are subleading in the weak-coupling limit. 

There exist a number of derivations of master equations, offering varying degrees of theoretical guarantees. 
For our purposes, let us recall two modern derivations of the master equation with rigorous error bounds.

(I) The \textit{Coarse-grained Master Equation}~\cite{mozgunov2020CPMaterEquationSmallLevSpacing} takes the following form
\begin{align}
\CL_{(CGME)}[\vrho] &:= - \ri[\vH_{LS}, \vrho] + \sum_{a \in A}\int_{-\infty}^{\infty}  \gamma(\omega) \left( \hat{\vA}^a(\omega)\vrho\hat{\vA^{a}}(\omega)^\dagg -\frac12\{\hat{\vA^{a}}(\omega)^\dagg\hat{\vA}^a(\omega),\vrho \} \right)\rd\omega,\label{eq:CGME}
\end{align}
where the operator Fourier transforms have uniform weight
\begin{align}
    \hat{\vA}^a(\omega) :=\frac{1}{\sqrt{2\pi T}}\int_{-T}^{T} \e^{\ri \vH t} \vA^a \e^{-\ri \vH t} \e^{-\ri \omega t}\rd t
\end{align}
and the time-scale $T$ depends on the coupling strength, on the bath memory time, and on the correlation function $\gamma(\omega)$ satisfying the symmetry $\gamma(\omega)/\gamma(-\omega)=\e^{-\beta\omega}$. The Lamb-shift term $\vH_{LS}$ is a complicated operator depending on the Heisenberg evolution of $\vA^a$. In Ref.~\cite{chen2023ExactQGibbsSampler}, instead of truncating the integral at a specific value of $T$, we impose a Gaussian damping term of width roughly $T$ (together with a careful choice of bath correlation function). In addition, rather than a Lamb-shift coherent term, we impose a particular coherent term as prescribed in \Cref{lem:FindingCoherenceTerm}.

Nicely, the above master equation~\eqref{eq:CGME} (and the algorithmic counter part~\cite{chen2023ExactQGibbsSampler}) can also be regarded as a smooth variant of the Davies' generator. By sending $T\rightarrow \infty,$ the operator Fourier transform $\hat{\vA}^{a}(\omega)$ recovers the exact Bohr frequencies $\nu$, as in the Davies' form. However, this $T\rightarrow\infty$ limit requires the coupling strength to be ultra-weak (e.g., weaker than all other energy scales, including the level spacings in the spectrum of $\vH_S$).
While this ultra-weak assumption is easily satisfied for commuting many-body Hamiltonians, it becomes physically infeasible for general non-commuting many-body Hamiltonians. At finite values of $T$, the Lindbladian does not generally have the Gibbs state as its stationary state, but only approximately~\cite{chen2023QThermalStatePrep}. 

(II) The \textit{universal Lindblad equation} takes the following form (upon rescaling from~\cite{nathan2020UniversalLindlad} and dropping the system Hamiltonian)
\begin{equation}\label{ULE}
\CL_{(ULE)}[\vrho] = - i[\vH_{LS},\vrho] + \sum_{a\in A} \vL^a\vrho \vL^{a\dag} - \frac12\{\vL^{a\dag} \vL^a,\vrho\},
\end{equation}
where $\vH_{LS}$ is a Lamb-shift term and the Lindblad operators
\begin{align*} 
\vL^a = \sum_{\nu\in B(\vH)}\sqrt{\gamma(\nu)} \vA_{\nu}^a
= \int_{-\infty}^{\infty} g(t) \e^{\ri \vH t} \vA^a \e^{-\ri \vH t} \rd t,
\end{align*}
where the time-domain function is the inverse Fourier transform (up to constant prefactors)
\begin{align}
    g(t) := \int_{-\infty}^{\infty}\sqrt{\gamma(\omega)} \e^{\ri\omega t} \rd \omega .
\end{align}
Compared with~\eqref{eq:CGME}, the universal Lindblad equation has fewer jumps as the sum over frequencies is within each jump. We have adjusted the notation from~\cite{nathan2020UniversalLindlad} such that the transition part appears identical to our algorithmic construction. 

Both master equations follow from a separation of time scales governed by the quantity $\Gamma \tau$, with
\begin{equation}
\Gamma = 2 \sqrt{\eta}  \int_{-\infty}^\infty |g(t)| \rd t \qquad\text{and}\qquad \tau = \frac{\int_{-\infty}^\infty |t g(t)| \rd t}{\int_{-\infty}^\infty |g(t)| \rd t}. 
\end{equation}

The quantities $\tau$ and $\Gamma$ depend only on the bath correlation function and on the coupling strength. $\tau$ can be seen as the characteristic time scale of the bath, while $\Gamma$ is an energy cut-off, and sets an upper bound for the rate of bath-induced evolution on the system. The approximation scale set by $\Gamma \tau$ also provides a system-size independent way of setting the coupling strength $\eta$ so that $\tau \Gamma \ll 1$. This limit is sufficient for the derivation of the universal Lindblad equation, as well as the related Redfield equation \cite{gardiner2004QuantumNoiseHandbookMarkovianAndNon}. For the coarse grained master equation, we also need a separating timescale for the local refreshment of the bath that must satisfy\footnote{The correct bounds invoke a slightly different definition of the energy and time cutoff, but capture essentially the same physics.} $\tau\ll T_{\rm refresh}\ll \Gamma^{-1}$. 

It is worth pausing for a second to consider the Lindbladian in \eqref{ULE}. 
The Lindblad operators are essentially identical in form to the ones from \Cref{sec:coherentDyn}, and can be considered a coherent or joint version of the Davies jump operators. In the weak-coupling limit, the off-diagonal terms $2 \pi \eta \sqrt{\gamma(\nu_1)\gamma(\nu_2)} (\vA_{\nu_1}^{a})^\dagger [\cdot]\vA_{\nu_2}^{a}$ vanish, leaving only the diagonal terms $\nu_1=\nu_2$, and recovering the Davies master equation. The intermediate regimes can be related to the interpolated construction of \Cref{sec:interpolation}. Unless the weak-coupling limit is considered, the Lindbladian does not generally have the Gibbs state as its stationary state.

In general, there is no reason to believe that the global stationary state of the master equation in \eqref{ULE} is a Gibbs state. Even if there is a derivation of the master equation with exact detailed balance, the master equation itself is already an approximation to the true system-bath dynamics where subleading correction (e.g., backreaction and non-Markovianity) are ignored. We do not know whether there are qualitative differences between approximately and exactly detailed balance Lindbladians.

In summary, we observe that the two leading derivations of the Markovian master equation closely resemble the algorithmic constructions in \Cref{sec:dbcp}, in that they give rise to a similar form of jump operators, under realistic physical assumptions. 
Strictly speaking, for many-body quantum systems, the physical derivations are only valid for relatively short time scales (if fixing the local parameters and taking the thermodynamics limit), and non-Markovian or nonpertubative effects may play a role at later times. We leave it for future work to study the physical connection to late-time open system dynamics.

\section{Constructions for non-self-adjoint completely positive maps}\label{sec:non-selfadjoint}
Wilfred K.\ Hastings~\cite{hastings1970MonteCarloSamplingUsingMCs} extended the construction of~\cite{metropolis1953StateCalculationsByComputers} to the case of non-symmetric Markov chains $\vP$. The Metropolis-Hastings rule is a direct generalization of \eqref{eq:ClDiscMetropolis} as follows
\begin{align}\label{eq:ClDiscMetropolisHastings}
	\forall i\neq j\colon \vP'_{ij}=\vP_{ij}\min\bigg(1,\frac{\vP_{ji}v_i}{\vP_{ij}v_j}\bigg),
\end{align}
which by construction always leads to detailed balance $\vP'_{ij}\pi_j = \vP'_{ji}\pi_i$ since $\pi\propto v$. As before, the same rule can also be applied to Laplacian generators of continuous-time Markov chains, and the Metropolis function $g_M=\min(1,r)$ can be replaced by any function $\gammaprob\colon\BR_+\rightarrow [0,1]$ that satisfies $\gammaprob(r)=r\cdot\gammaprob(1/r)$ for all $r>0$ (defining $\gammaprob(1/0):=\gammaprob(\infty):=\lim_{r\rightarrow \infty}\gammaprob(r)=0$).

One can immediately see that this non-symmetric construction behaves very differently compared to the symmetric case. A major difference is that the transformation $\vP\rightarrow \vP'$ is no longer linear. This means that in order to find quantum generalizations for non-self-adjoint completely positive maps we should probably look beyond linear super-superoperators. This seems to be a barrier for efficient implementation. While in the classical situation the non-linear transformation is not really an issue, in the quantum case it is problematic due to the adverse interplay of quantum uncertainty and the natural limitation to linear operations in quantum mechanics. In particular, it seems that, in general, implementing the resulting maps becomes costly when the access to $\CT$ is only provided through a block-encoded dilation $\vG$, and it is in general less clear what would be a natural quantum construction in the non-self-adjoint case.

A first approach could be to try and balance each Kraus operator one by one, for which we already have a construction as laid out in \Cref{lem:FindingDiscDecayTerm,lem:KrausBalance}. This is not satisfactory for several reasons, first of all it turns out that this leads to completely stalled dynamics. Another problem is that the Kraus representation is not unique, and therefore the mapping is not even well defined. This is in fact an issue in general, which might mean that it is not sufficient to think in terms of Kraus decompositions.

Another approach is simply to symmetrize the CP map $\CT$, and then apply one of the self-adjoint constructions. The issue with this approach is that it might lead to a huge increase in the strength of the map, since in general the best bound on $\|\CT[\vI]\|$ is the trivial $\tr[\CT^\dagger[\vI]]$, which comes from considering trace $\|\CT[\vI]\|\leq \tr[\CT[\vI]]=\ipc{\vI}{\CT[\vI]}_{HS}=\langle\CT^\dagger[\vI],\vI\rangle_{HS}= \tr[\CT^\dagger[\vI]]$.

In this section we describe two other approaches, which are a somewhat more satisfactory, at least from a theoretical perspective.
Both constructions can be viewed as natural generalizations of the Metropolis-Hastings rule. The first construction is more related to our smooth quantum Glauber construction, and could be efficiently implemented in some cases. The second construction is more closely related to the Davies generator, and thus seems less relevant from the practical implementation perspective.

\subsection{Reducing to the self-adjoint case via inverse Metropolis rule}

In \Cref{sec:coherentDyn}, we introduced the super-superoperator $\BS_{C}^{[\gamma]}\llbracket\cdot\rrbracket$, which we will denote here by $\BS_{\vrho}^{[\gamma]}\llbracket\cdot\rrbracket$ in order to explicitly display the $\vrho$ dependence. We proved that it maps a superoperator $\CT$ onto a $\vrho$-detailed-balanced superoperator $\CT'$ if $\CT$ is self-adjoint to start with and $\gamma$ satisfies $\gamma(-\nu) = \e^\nu \gamma(\nu)$. Here we overcome the self-adjointness requirement by introducing a fine-tuned symmetrizatrion procedure for the CP map $\CT$ before applying the super-superoperator $\BS_{\vrho}^{[\gamma]}\llbracket\cdot\rrbracket$. Our symmetrization procedure is based on the super-superoperator
\begin{align*}
    \mathbb{T}_{\vsigma}\llbracket \CT \rrbracket[\cdot] = \frac12\CT[\cdot] + \frac12 \sqrt{\vsigma}\CT^\dagger[\sqrt{\vsigma}^{-1}\cdot \sqrt{\vsigma}^{-1}]\sqrt{\vsigma}, \qquad\text{where}~\vsigma \succ 0,
\end{align*}
which symmetrizes a superoperator $\CT$ with the aid of a density matrix $\vsigma \succ 0$. With respect to the usual Hilbert-Schmidt norm, the adjoint of $\mathbb{T}_{\vsigma}\llbracket \CT \rrbracket$ is
\begin{align*}
    \mathbb{T}_{\vsigma}\llbracket \CT \rrbracket^\dagger[\cdot] = \frac12\CT^\dagger[\cdot] + \frac12\sqrt{\vsigma}^{-1}\CT[\sqrt{\vsigma}\cdot \sqrt{\vsigma}]\sqrt{\vsigma}^{-1},
\end{align*}
meaning that $\mathbb{T}_{\vsigma}\llbracket \CT \rrbracket$ is not necessarily self-adjoint. However, $\mathbb{T}_{\vsigma}\llbracket \CT \rrbracket$ is self-adjoint with respect to the $\vsigma$-weighted inner product $\langle \vY, \vX\rangle_{\vsigma} := \operatorname{Tr}[\vY^\dagger \sqrt{\vsigma}^{-1}\vX \sqrt{\vsigma}^{-1}]$, as can be seen by
\begin{align*}
    \langle \vY, \CT[\vX]\rangle_{\vsigma} &= {\operatorname{Tr}}\big[\vY^\dagger \sqrt{\vsigma}^{-1} \CT[\vX] \sqrt{\vsigma}^{-1}\big]\\
    &= {\operatorname{Tr}}\big[\CT^\dagger[\sqrt{\vsigma}^{-1}\vY^\dagger \sqrt{\vsigma}^{-1}] \vX \big]\\
    &= {\operatorname{Tr}}\big[\sqrt{\vsigma}^{-1}(\sqrt{\vsigma}\CT^\dagger[\sqrt{\vsigma}^{-1}\vY^\dagger \sqrt{\vsigma}^{-1}]\sqrt{\vsigma}) \sqrt{\vsigma}^{-1} \vX \big]\\
    &= \langle \sqrt{\vsigma}\CT^\dagger[\sqrt{\vsigma}^{-1}\vY \sqrt{\vsigma}^{-1}]\sqrt{\vsigma}, \vX\rangle_{\vsigma}.
\end{align*}
Being self-adjoint with respect to the $\vsigma$-weighted inner product is equivalent to $\vsigma$-detailed balance (\Cref{def:KMSDetBalance}) as can be seen below
\begin{align}\label{eq:weightSelfAdjointDB}
    \CT[\cdot]=\sqrt{\vsigma}\CT^\dagger[\sqrt{\vsigma}^{-1}[\cdot] \sqrt{\vsigma}^{-1}]\sqrt{\vsigma}
    \quad\Longleftrightarrow\quad
    \vsigma^{-\frac14}\CT[\vsigma^{\frac14}\cdot \vsigma^{\frac14}]\vsigma^{-\frac14} = \vsigma^{\frac14}\CT^\dagger[\vsigma^{-\frac14}\cdot \vsigma^{-\frac14}]\vsigma^{\frac14}.
\end{align}

In order to describe our overall construction, we first need the following well-known fact.
\begin{fact}[{\cite[Theorem~6.5]{wolf2012QChannelsOpsLectureNotes}}]
    Let $\CT$ be a positive map with spectral radius $r$. Then there is a positive semi-definite matrix $\vsigma$ such that $\CT[\vsigma] = r\vsigma$.
    Moreover, if $\CT$ is ergodic, then $\vsigma$ is unique and positive definite.
\end{fact}

As a weaker version of the above result, if $\CT$ is also trace-preserving, then $r=1$ and there is a density matrix $\vsigma$ such that $\CT[\vsigma] = \vsigma$~\cite[Theorem~6.11]{wolf2012QChannelsOpsLectureNotes}.

Given a positive map $\CT$, consider one of its positive semidefinite eigenstates $\vsigma$, so that $\CT[\vsigma] = r\vsigma$. If $\vsigma \succ 0$ and we symmetrize $\CT$ via $\mathbb{T}_{\vsigma}$ then 
\begin{align*}
    (\mathbb{T}_{\vsigma}\llbracket \CT \rrbracket)^\dagger[\vI] = \frac12\CT^\dagger[\vI] + \frac12\sqrt{\vsigma}^{-1}\CT[\sqrt{\vsigma}\vI \sqrt{\vsigma}]\sqrt{\vsigma}^{-1} = \frac12\CT^\dagger[\vI] + \frac{r}{2}\vI.
\end{align*}
Therefore, if $\CT$ is trace-preserving, then so is $\mathbb{T}_{\vsigma}\llbracket \CT \rrbracket$.

The high-level idea of our construction is to symmetrize $\CT$ via $\mathbb{T}_{\vsigma}$ using one of its eigenstates $\vsigma$, which results into a $\vsigma$-detailed-balanced superoperator. We can then use the results of \Cref{sec:coherentDyn} (in particularly \Cref{thm:CoherentRule}) \emph{in reverse} in order to go from detailed balance to self-adjointness using the map $\BS_{\vsigma}^{[\bar{\gamma}]}$, where $\bar{\gamma}:\BR\to[0,1]$ satisfies $\bar{\gamma}(-\nu) = \e^{-\nu}\bar{\gamma}(\nu)$. After that, we can use any of the constructions in \Cref{sec:dbcp} to ensure the $\vrho$-detailed-balanced property. We formalize the overall procedure in the next result.
\begin{theorem}
    Let $\bar{\gamma}:\BR \to [0,1]$ be a function such that $\bar{\gamma}(-\nu) = \e^{-\nu}\bar{\gamma}(\nu)$, and $\BS$ any super-superoperator that maps self-adjoint CP maps to $\vrho$-detailed-balanced ones. Given a completely positive superoperator $\CT:\BC^{d\times d} \to \BC^{d\times d}$ with positive definite eigenstate $\vsigma=\e^{-\vH_{\vsigma}}\succ 0$, then
    \begin{align*}
        (\BS \circ \BS_{\vsigma}^{[\bar{\gamma}]} \circ \mathbb{T}_{\vsigma})\llbracket \CT \rrbracket
    \end{align*}
    is $\vrho$-detailed balanced.
\end{theorem}
\begin{proof}
    Essentially the same argument as in the proof of \Cref{thm:CoherentRule} shows that $\BS_{\vsigma}^{[\bar{\gamma}]}\llbracket \mathbb{T}_{\vsigma}\llbracket \CT \rrbracket \rrbracket$ is self-adjoint. Write $ \mathbb{T}_{\vsigma}\llbracket \CT \rrbracket[\cdot] := \CT'[\cdot]:= \sum_{a\in A}\vA^a[\cdot]\vA^{a\dagger}$. 
    Due to \eqref{eq:weightSelfAdjointDB} we have that 
    \begin{align}\label{eq:TPrimeEquivKraus}
        \sum_{a\in A}\vA^a[\cdot]\vA^{a\dagger} = \CT'[\cdot] 
        = \sqrt{\vsigma}\CT'^\dagger[\sqrt{\vsigma}^{-1}[\cdot] \sqrt{\vsigma}^{-1}]\sqrt{\vsigma} 
        = \sum_{a\in A}\left(\sqrt{\vsigma}\vA^{a\dagger}\sqrt{\vsigma}^{-1}\right)[\cdot]\left(\sqrt{\vsigma}^{-1}\vA^a\sqrt{\vsigma}\right).
    \end{align}
    The $a$-labelled Kraus operator of $\BS_{\vsigma}^{[\bar{\gamma}]}\llbracket\CT' \rrbracket[\cdot] = \BS_{\vsigma}^{[\bar{\gamma}]}\llbracket \sqrt{\vsigma}\CT'^\dagger[\sqrt{\vsigma}^{-1}\cdot\sqrt{\vsigma}^{-1}]\sqrt{\vsigma} \rrbracket$ in \eqref{eq:TPrimeEquivKraus} is
    \begin{align*}
        \CS^{[\sqrt{\bar{\gamma}}]}[\sqrt{\vsigma}\vA^{a\dagger}\sqrt{\vsigma}^{-1}] &
        = \sum_{\nu \in B(\vH_{\vsigma})} \sqrt{\bar{\gamma}(\nu)} \left(\sqrt{\vsigma}\vA^{a\dagger}\sqrt{\vsigma}^{-1}\right)_{\nu} \tag{due to \eqref{eq:SmoothMapDef}}\\&
        = \sum_{\nu \in B(\vH_{\vsigma})} \sqrt{\bar{\gamma}(\nu)} \e^{-\frac\nu2}(\vA^a_{-\nu})^\dagger \tag{due to \eqref{eq:EnergyDiffDecomposition}}\\&
        = \sum_{\nu \in B(\vH_{\vsigma})} \sqrt{\bar{\gamma}(-\nu)} \left(\vA^a_{-\nu}\right)^\dagger \tag{since $\bar{\gamma}(-\nu) = \e^{-\nu}\bar{\gamma}(\nu)$}\\&
        = \sum_{\nu \in B(\vH_{\vsigma})} \sqrt{\bar{\gamma}(\nu)} (\vA^a_{\nu})^\dagger \tag{since $B(\vH_{\vsigma}) = -B(\vH_{\vsigma})$}\\&
        = \left(\CS^{[\sqrt{\bar{\gamma}}]}[\vA^a]\right)^\dagger, \tag{due to \eqref{eq:SmoothMapDef}}
    \end{align*}
    which coincides with the $a$-labelled Kraus operator of $\big(\BS_{\vsigma}^{[\bar{\gamma}]}\llbracket \CT' \rrbracket\big)^\dagger$ implying due to \eqref{eq:TPrimeEquivKraus} that $\BS_{\vsigma}^{[\bar{\gamma}]}\llbracket \CT' \rrbracket$ is indeed self-adjoint, and therefore $(\BS \circ \BS_{\vsigma}^{[\bar{\gamma}]} \circ \mathbb{T}_{\vsigma})\llbracket \CT \rrbracket$ is $\vrho$-detailed balanced.
\end{proof}

From an implementation perspective we already studied how to implement $\BS$ and $\BS_{\vsigma}^{[\bar{\gamma}]}$, so the question remains how to implement $\mathbb{T}_{\vsigma}$, i.e., $\sqrt{\vsigma}^{-1}\CT[\sqrt{\vsigma}\vI \sqrt{\vsigma}]\sqrt{\vsigma}^{-1}$. If we can implement $\vH_{\vsigma}$ efficiently, and the Kraus operators of $\CT$ are localized in energy, then this transformation can be implemented via QSVT similarly to \Cref{lem:QSVTShortcut}. However, we do not know how to implement it in a more general setting.

\subsection{Reduction to the classical case via Haar-random resolution of degeneracies}

As our second construction, we reduce the problem of turning a CP map into a detailed-balanced one to its classical counterpart. We proceed by completely decohering a given CP map in an eigenbasis of $\vrho$ in order to obtain a classical stochastic transition matrix, after which we can immediately apply, for example, the classical Metropolis-Hastings rule~\eqref{eq:ClDiscMetropolisHastings}. 

Consider the eigendecomposition $\vH=\sum_{j=1}^d \energy_j \ketbra{\psi_j}{\psi_j}$ from \eqref{eq:eigendecomposition} and assume for now that there is no degeneracy, meaning that $E_i \neq E_j$ for all $i\neq j$. Given a CP map $\CT : \BC^{d\times d} \to \BC^{d\times d}$ with Kraus decomposition $\CT[\cdot] = \sum_{a\in A} \vA^a[\cdot]\vA^{a\dagger}$, we decohere it similarly to the Davies generator via $\BS_D^{[\gamma]}$ as in \eqref{eq:DavieSSO} but using the constant function $\nu\to 1$ instead of $\gamma$, i.e.,
\begin{align}
    \CT'[\cdot] &= \sum_{a\in A} \sum_{i,j\in[d]} |\psi_i\rangle \langle\psi_i|\vA^a|\psi_j\rangle\langle \psi_j| [\cdot] |\psi_j\rangle \langle\psi_j|\vA^{a\dagger}|\psi_i\rangle\langle \psi_i|\nonumber\\
    &= \sum_{i,j\in[d]} \big(\langle\psi_i|\CT[|\psi_j\rangle\langle\psi_j|]|\psi_i\rangle\big) |\psi_i\rangle  \langle\psi_j| [\cdot] |\psi_j\rangle \langle \psi_i|\nonumber\\
    &= \sum_{i,j\in[d]} \vCT_{ii,jj} |\psi_i\rangle  \langle\psi_j| [\cdot] |\psi_j\rangle \langle \psi_i|. \label{eq:decoherence}
\end{align}
We call this a decohering transformation because it just zeroes out the ``coherent entries'' $\vCT_{ik,j\ell}$ where $k\neq i$ or $\ell \neq j$. Equivalently, it is the channel we get by composing a measurement in the eigenbasis, followed by applying $\CT$ and then applying another measurement in the eigenbasis.
By arranging the entries $\vCT_{ii,jj}$ in a stochastic matrix $\vP$ with matrix elements $\vP_{ij} := \vCT_{ii,jj}$, we can readily apply the Metropolis-Hastings rule from~\eqref{eq:ClDiscMetropolisHastings} to $\vP$ in order to turn it into a detailed-balanced stochastic matrix $\vP'$ and, consequently, into a detailed-balanced CP map $\CT'$.

Unfortunately, the above approach is not well defined when there are degeneracies in $\operatorname{spec}(\vH)$, as the choice for the eigenbasis with respect to which to decohere $\CT$ is not unique. 
To tackle this issue, we pick a Haar-random basis in every degenerate eigensubspace of $\vH$ and proceed as in~\eqref{eq:decoherence}. 
More specifically, let $S(E) := \{|\psi\rangle\in \BC^d: \vH|\psi\rangle = E|\psi\rangle\}$ be the eigensubspace of $\vH$ with energy $E$. We can get a Haar-random eigenbasis of $\vH$ by independently sampling a Haar random unitary $\vU_E$ on each energy-eigensubspace $S(E)$, and applying $\vU := \bigoplus_{E\in\operatorname{spec}(\vH)}\vU_E$ to the fixed eigenbasis in~\eqref{eq:eigendecomposition}.
Once an eigenbasis is picked, we can apply the procedure as described above to get the resulting $\vrho$-detailed-balanced CP map $\CT'_{\vU}$. The final CP maps is given by $\CT':=\BE_{\vU}[\CT'_{\vU}]$, where the expectation is over the Haar-random eigenbasis choice or, equivalently, over the Haar-random unitary $\vU$. Due to the Haar random choices it is easy to see that $\CT'$ is well defined, i.e., independent from the initial choice of eigenbasis~\eqref{eq:eigendecomposition}. Also, due to the convexity of the set of CP maps, $\CT'$ is a CP map. Moreover if $\CT^\dagger[\vI]=\vI$, then $\CT'^\dagger[\vI]=\vI$ as inherited form the classical construction.

This construction is analogous to the Davies generator, and inherits some of its issues: it is discontinuous with respect to $\vrho$ and we do not know how to efficiently implement it on a quantum computer.

\bibliographystyle{alphaUrlePrint.bst}
\bibliography{qc_gily}

\newpage
\appendix
\section*{Nomenclature}\label{sec:recap_notation}
This section recapitulates notations, and the reader may skim through this and return as needed.
We write scalars, functions, and vectors in normal font, matrices in bold font $\vO$, superoperators in curly font~$\CT$, and super-superoperators by $\BS$.
The matrix of a (super-)superoperator is typeset using a bold version of the original symbol.
Natural constants $\e, \ri, \pi$ are denoted in Roman font.\pagebreak[0]
\begin{align*}
	\vrho&:= \frac{\e^{-\vH }}{\tr[ \e^{-\vH }]} \quad &\text{the Gibbs state associated with Hamiltonian $\vH$}\\	
	\vH &= \sum_i E_i \ketbra{\psi_i}{\psi_i}&\text{the eigen-decomposition of the Hermitian $\vH$}\\
	\text{spec}(\vH) &:= \{ E_i \}_i & \text{the spectrum of the Hamiltonian $\vH$}\\
	\nu \in B = B(\vH)&:= \text{Spec}(\vH) - \text{Spec}(\vH) &\text{the set of Bohr frequencies}\\
	\vA_\nu &:= \sum_{i,j\in[d] \colon \energy_i-\energy_j = \nu} \bra{\psi_i}\vA\ket{\psi_j}\cdot\ketbra{\psi_i}{\psi_j}\kern-11mm &\text{operator $\vA$ at exact Bohr frequency $\nu$}\\
	\vA(t)&:= \e^{\ri \vH t}\vA\e^{-\ri \vH t} & \text{time-$t$ Heisenberg evolution of $\vA$}\\		
	\ket{\sqrt{\vrho}} &:= \sum_i \e^{- E_i/2} \ket{\psi_i} \otimes \ket{\psi_i^*}/\sqrt{\tr[\vrho]}\kern-10mm &\text{the purified Gibbs state}\\	
	\{\vA^a\}_{a \in A}: & &\text{set of jump operators}\\
	\labs{A}: & & \text{cardinality of the set of jumps}\\
	\vI:& &\text{the identity operator}\\
	\bigOt{\cdot},\tOmega (\cdot) :& &\kern-30mm\text{complexity expressions ignoring (poly)logarithmic factors}
\end{align*}
Functions and norms:
\begin{align*}
	\pmb{1}_S(x)&: = 1\text{ if }x\in S\text{ and 0 otherwise} \quad &  \text{the indicator function of }S\\	
	\norm{f(x)}_p&: = \L(\int \labs{f(x)}^p \mathrm{d} x \R)^{\! 1/p} \quad &  \text{the $p$-norm of a function $f\colon\BR\rightarrow\BC$ for $p\in[1,\infty]$}\\
	\norm{f(x)}&: = \norm{f(x)}_2 = \sqrt{\int \labs{f(x)}^2 \mathrm{d} x} \quad & \text{the 2-norm of a function $f\colon\BR\rightarrow\BC$}\\	
	\norm{\ket{\psi}}&: \quad &\text{the Euclidean norm of a vector $\ket{\psi}$}\\
	\norm{\vO}&:= \sup_{\ket{\psi},\ket{\phi}} \frac{\bra{\phi} \vO \ket{\psi}}{\norm{\ket{\psi}}~\norm{\ket{\phi}}} \quad &\text{the operator norm of a matrix $\vO$}\\
	\norm{\vO}_p&:= (\tr \labs{\vO}^p)^{1/p}\quad&\text{the Schatten p-norm of a matrix $\vO$}\\
	\norm{\CL}_{p-p} &:= \sup_{\vO} \frac{\normp{\CL[\vO]}{p}}{\normp{\vO}{p}}\quad&\text{the induced $p-p$ norm of a superoperator $\CL$}\\
	\vertiii{\BS}&:=\sup_{\CT}\frac{\big\|(\BS\llbracket\CT\rrbracket)^\dagger[\vI]\big\|}{\nrm{\CT^\dagger[\vI]}}\quad&\kern-3mm\text{super-superoperator norm for trace-increasingness \eqref{eq:SSONorm}}
\end{align*}
Linear algebra:
\begin{align*}
	\vO^\dagger&: \quad & \text{the Hermitian conjugate of a matrix $\vO$}\\
	\vO^*&: \quad & \text{the entry-wise complex conjugate of a matrix $\vO$}\\
	\ket{\psi^*}&: \quad&\text{ entry-wise complex conjugate of a vector $\ket{\psi}$}\\
	\vA\circpr \vB&: \quad&\text{ entry-wise (Hadamard) product of the matrices in a specified basis}	
\end{align*}
\newpage

\section{Bounds on the norm of our reweighing super-superoperator}\label{apx:SchurBounds}
In this appendix, we give general norm bounds on the superoperator $\CS^{[\fenergy]}[\vA] := \sum_{\nu \in B(\vH)}\hat{f}(\nu)\vA_{\nu}$ from properties of $\fenergy$ and the Hamiltonian $\vH$. 
The following proposition is meant to be applied with $\fenergy(\nu)=\frac12\tanh(-\nu/4)$ as in~\eqref{eq:CoherentTermRecipe} or in the context of \Cref{thm:CoherentRule} by setting $\fenergy(\nu)=\sqrt{\gamma(\nu)}-\frac12$.
\begin{proposition}[Regularizing $\CS$]\label{prop:integral} Let $\fenergy\colon \BR\rightarrow \BR$.
	If 
 \begin{align}
 \lim_{\theta\rightarrow 0+} \frac{1}{\sqrt{2\pi}} \int_{\BR\setminus[-\theta,\theta]} \ftime(t)\e^{\ri\omega t} \rd t = \fenergy(\omega),    
 \end{align}
 then \begin{align}\label{eq:IntRepLim}
		\CS^{[\fenergy]}[\vA]=\lim_{\theta\rightarrow 0+} \frac{1}{\sqrt{2\pi}} \int_{\BR\setminus[-\theta,\theta]} \ftime(t)\e^{\ri\vH t} \vA \e^{-\ri\vH t} \rd t.
	\end{align}
	In particular if $\big\|\ftime- \pmb{1}_{[-1,1]}(t)\frac{\ri}{\sqrt{2\pi} t}\big\|_1<\infty$, then for all $\theta\in\BR_+$ we have
	\begin{align}
		\CS^{[\fenergy]}[\vA]&=
		\frac{1}{\sqrt{2\pi}} \int_{\BR} \left(\ftime(t)-\pmb{1}_{[-\theta,\theta]}(t)\frac{\ri}{\sqrt{2\pi} t}\right)\e^{\ri\vH t} \vA \e^{-\ri\vH t} \rd t\nonumber\\&\label{eq:intRep}
		+\frac{1}{\pi} \int_{0}^\theta \cos(\vH t) \vA\vH \sinc(\vH t) - \sinc(\vH t) \vH\vA \cos(\vH t)  \rd t.
	\end{align}
\end{proposition}
Note that the above also gives an efficient implementation method via Linear-Combination-of-Unitaries, and also shows that if $\vA$ and $\vH$ are quasi-local, then so is $\CS^{[\fenergy]}[\vA]$ (cf.\ \Cref{lem:QuasiLocalApxS}).
\begin{proof}
We calculate
	\begin{align*}
		\lim_{\theta\rightarrow 0+} \frac{1}{\sqrt{2\pi}} \int_{\BR\setminus[-\theta,\theta]} \ftime(t)\e^{\ri\vH t} \vA \e^{-\ri\vH t} \rd t&=
		\sum_{\nu\in B(\vH)}\lim_{\theta\rightarrow 0+} \frac{1}{\sqrt{2\pi}} \int_{\BR\setminus[-\theta,\theta]} \ftime(t)\e^{\ri\vH t} \vA_{\nu} \e^{-\ri\vH t} \rd t\\
		&=\sum_{\nu\in B(\vH)}\lim_{\theta\rightarrow 0+} \frac{1}{\sqrt{2\pi}} \int_{\BR\setminus[-\theta,\theta]} \ftime(t)\e^{\ri\nu t} \vA_{\nu} \rd t\\
        &=\sum_{\nu\in B(\vH)}\fenergy(\nu)\vA_{\nu} \tag{by assumption}\\
		&=\CS^{[\fenergy]}[\vA]. \tag{by definition of $\CS^{[\fenergy]}$}
	\end{align*}
	\noindent On the other hand, we also have
	\begin{align*}
		&\lim_{\tau\rightarrow 0+} \frac{1}{\sqrt{2\pi}} \int_{\BR\setminus[-\tau,\tau]} \ftime(t)\e^{\ri\vH t} \vA \e^{-\ri\vH t} \rd t-
		\frac{1}{\sqrt{2\pi}} \int_{\BR} \left(\ftime(t)-\pmb{1}_{[-\theta,\theta]}(t)\frac{\ri}{\sqrt{2\pi} t}\right) \e^{\ri\vH t} \vA \e^{-\ri\vH t} \rd t\\&=
		\lim_{\tau\rightarrow 0+} \frac{1}{\sqrt{2\pi}} \int_{\BR\setminus[-\tau,\tau]} \pmb{1}_{[-\theta,\theta]}(t)\frac{\ri}{\sqrt{2\pi} t}\e^{\ri\vH t} \vA \e^{-\ri\vH t} \rd t \tag{since $\big\|\ftime- \pmb{1}_{[-1,1]}(t)\frac{\ri}{\sqrt{2\pi} t}\big\|_1<\infty$}\\
		&=\frac{1}{\pi}\lim_{\tau\rightarrow 0+}  \int_{[\tau,\theta]} \frac{1}{t}\left(\cos(\vH t) \vA \sin(\vH t)- \sin(\vH t) \vA \cos(\vH t)\right)\rd t\tag{due to parity}\\
		&=\frac{1}{\pi}\int_{0}^\theta \left(\cos(\vH t) \vA\vH \sinc(\vH t)- \sinc(\vH t) \vH\vA \cos(\vH t)\right)\rd t.\tag*{\qedhere}
	\end{align*}      
\end{proof}

For $g\colon [a,b] \rightarrow \BC$, define its total variation as 
\begin{align}
\labs{g}_{[a,b]} \!:=\!{\displaystyle\lim_{k\rightarrow \infty}}\sup_{a\leq x_0\leq \ldots \leq x_k \leq b}\sum_{i=1}^k\!|g(x_i)\!-\!g(x_{i-1})|.    
\end{align}

\begin{theorem}[Two bounds on $\CS$]\label{thm:SStrengthBound}
	For all $\theta\in\mathbb{R}_+$ such that $\big\|\ftime- \pmb{1}_{[-\theta,\theta]}(t)\frac{\ri}{\sqrt{2\pi} t}\big\|_1<\infty$, then
        \begin{align}
		\|\CS^{[\fenergy]}\|_{\infty\to\infty}< \frac{2\theta}{\pi}\|\vH\| + \nrm{\frac{\ftime}{\sqrt{2\pi}}- \pmb{1}_{[-\theta,\theta]}(t)\frac{\ri}{2\pi t}}_1.
	\end{align}
        In particular, if $\big\|\ftime- \pmb{1}_{[-1,1]}(t)\frac{\ri}{\sqrt{2\pi} t}\big\|_1<\infty$, then
	\begin{align}\label{eq:HamBounds}
		\|\CS^{[\fenergy]}\|_{\infty\to\infty}< \frac1\pi + \frac{1}{\pi}\ln(1+2\nrm{\vH})+\nrm{\frac{\ftime}{\sqrt{2\pi}}- \pmb{1}_{[-1,1]}(t)\frac{\ri}{2\pi t}}_1.
	\end{align}
	We can also get a dimension-dependent bound using\footnote{Note that if we only assume that the range of $\fenergy$ is bounded, say $[0,1]$, then we cannot give a better than $\mathcal{O}(\sqrt{d})$ bound. Indeed, if the Hamiltonian $\vH$ is diagonal and its eigenvalues are algebraically independent, then $|B(\vH)|=d(d-1)+1$ and thus we can ``address'' each off-diagonal element of the matrices with a different $\nu$ value. Let us consider the $n$-fold tensor product of the $2\times 2$ Hadamard matrix $\vU:=\vH^{\otimes n}$. If we set $\gamma(\nu)=0$ for each $\nu\in B(\vH)$ for which $\vU_{\nu}$ has a negative matrix element, and set $\gamma(\nu)=1$ for all other $\nu\in B(\vH)$, then it is easy to see that $\pmb{1}^T \CS^{[\gamma]}[\vU]\pmb{1}=\Theta(2^{\frac{3}{2}n})$, where $\pmb{1}$ is the all-$1$ vector, showing that $\|{\CS^{[\gamma]}[\vU]}\|=\Omega{(\sqrt{d})}$.} the total variation of $\fenergy$ as
	\begin{align}\label{eq:DimBounds}
		\|\CS^{[\fenergy]}\|_{\infty\to\infty}\leq \left(|\fenergy(2\nrm{\vH})| + |\fenergy|_{[-2\nrm{\vH},2\nrm{\vH}]}\right)(\lceil\log_2{d}\rceil + 1).
	\end{align}
\end{theorem}
Note that to get a good bound on $\CS$ by \eqref{eq:HamBounds} requires the function $\ftime$ to decay fast.
\begin{proof} In order to bound
	\begin{align*}
		\|\CS^{[\fenergy]}\|_{\infty\to\infty}
		=\sup_{\vA\in \BC^{d\times d}}\frac{\|\CS^{[\fenergy]}[\vA]\|}{\nrm{\vA}},
	\end{align*}
	we can use the integral representation \eqref{eq:intRep} of \Cref{prop:integral} yielding  
	\begin{align*}
		\|\CS^{[\fenergy]}[\vA]\| &\leq \frac{1}{\sqrt{2\pi}}\int_{\BR}\left|\ftime(t) - \mathbf{1}_{[-\theta,\theta]}(t)\frac{\ri}{\sqrt{2\pi}t}\right|\norm{\e^{\ri\vH t}\vA \e^{-\ri\vH t}}\rd t \\
		&\quad+\frac{2}{\pi}\int_{0}^{\theta} \nrm{\vA}\nrm{\vH}\nrm{\cos(\vH t)}\nrm{\sinc(\vH t)}\rd t\\
		&\leq \nrm{\vA}\left(\frac{1}{\sqrt{2\pi}}\nrm{\ftime(t) - \mathbf{1}_{[-\theta,\theta]}(t)\frac{\ri}{\sqrt{2\pi}t}}_1 + \frac{2\theta}{\pi}\nrm{\vH}\right) \tag{since $\nrm{\cos(\vH t)}, \nrm{\sinc(\vH t)} \leq 1$}  \\
		&\leq \nrm{\vA}\left(\frac{1}{\sqrt{2\pi}}\nrm{\ftime(t) - \mathbf{1}_{[-1,1]}(t)\frac{\ri}{\sqrt{2\pi}t}}_1 +  \frac{1}{\pi}\nrm{\frac{\ri}{t}\mathbf{1}_{[\theta,1]}(t)}_1 + \frac{2\theta}{\pi}\nrm{\vH}\right),   \tag{if $\theta\leq 1$}  
	\end{align*}
	where we used that $\|\e^{\ri\vH t}\vA \e^{-\ri\vH t}\| = \nrm{\vA}$. Setting $\theta = 1/(2\nrm{\vH}+1)$ proves \eqref{eq:HamBounds}.
	
	We prove \eqref{eq:DimBounds} using the following decomposition. Let us enumerate and sort the Bohr frequencies $\nu_1>\nu_2>\ldots >\nu_{|B(\vH)|}$. Define the operators $\vI^{(\leq \nu)}:=\sum_{i,j\in[d]: \energy_i-\energy_j \leq \nu} \ketbra{\psi_i}{\psi_j}$ and $\vI^{(= \nu)}:=\sum_{i,j\in[d]: \energy_i-\energy_j = \nu} \ketbra{\psi_i}{\psi_j}$ for $\nu\in B(\vH)$. Notice that
    \begin{align*}
        \CS^{[\fenergy]}[\vA] = \sum_{i=1}^{|B(\vH)|}\hat{f}(\nu_i) \vA\circpr \vI^{(=\nu_i)} = \fenergy(\nu_1)\vA\circpr\vI^{(\leq \nu_1)}+\sum_{i=2}^{|B(\vH)|} (\fenergy(\nu_i)-\fenergy(\nu_{i-1}))\vA\circpr\vI^{(\leq \nu_i)}
    \end{align*}
    where $\circpr$ denotes the Hadamard multiplication. Thus
    \begin{align*}
        \|\CS^{[\fenergy]}[\vA]\| &\leq \Bigg(|\fenergy(\nu_1)|+\underset{\leq |{\fenergy}|_{[\nu_{|B(\vH)|},\nu_1]}}{\underbrace{\sum_{i=2}^{|B(\vH)|} |\fenergy(\nu_i)-\fenergy(\nu_{i-1})| }}\Bigg) \max_{\nu\in B(\vH)}\|{\vA\circpr\vI^{(\leq \nu)}}\|\\&
		\leq\left(\fenergy(2\nrm{\vH}) + |{\fenergy}|_{[-2\nrm{\vH},2\nrm{\vH}]}\right)\max_{\nu\in B(\vH)}\|{\vA\circpr\vI^{(\leq \nu)}}\|\\&
		\leq\left(\fenergy(2\nrm{\vH}) + |{\fenergy}|_{[-2\nrm{\vH},2\nrm{\vH}]}\right)\max_{\nu\in B(\vH)}\|{\CS^{[\pmb{1}_{[-\infty,\nu]}]}}\|_{\infty\to\infty}\nrm{\vA}.     
    \end{align*}
	So it suffices to show that $\|\CS^{[\pmb{1}_{[-\infty, \nu]}]}\|_{\infty\to\infty}\leq 1+\lceil\log_2{d}\rceil$, which follows from \Cref{lem:GenTriangle}.
\end{proof} 

In the next auxiliary result, we bound the norm of the superoperator $[\cdot]\circpr \vM$ for specific matrices $\vM$, which in turn is used to bound $\|\CS^{[\pmb{1}_{[-\infty, \nu]}]}\|_{\infty\to\infty} = \|[\cdot]\circpr\vI^{(\leq \nu)}\|_{\infty\to\infty}$ in \Cref{thm:SStrengthBound}.

\begin{lemma}\label{lem:GenTriangle}
    Consider a matrix $\vM\in \{0,1\}^{n\times m}$ such that whenever $\vM_{ij}=1$, then $\vM_{i'j'}=1$ for all $i'\geq i$ and $j'\leq j$. Let $\CS_{\vM}[\cdot]:=[\cdot]\circpr \vM$ denote the Hadamard multiplication by $\vM$ (i.e., element-wise multiplication of matrices in the standard basis), then $\nrm{\CS_{\vM}}_{\infty\to\infty}\leq 1+\lceil\log_2(\min(n,m))\rceil$.
\end{lemma}
\begin{proof}
    If $m=1$, let $\vPi:=\sum_{i=1}^{n} \vM_{i1}\ketbra{i}{i}$, then $\vM \circpr \vA=\vPi \vA$, therefore $\nrm{\CS_{\vM}}_{\infty\to\infty}= \nrm{\vPi}\leq 1$. The $n=1$ case can be proven similarly. We prove the $\min(n,m)>1$ cases by induction. 
		
First assume that $n\leq m$. We decompose $\vM$ into a block-diagonal $\vM''$ and a lower-block part $\vM'$,
  \begin{align*}
  \vM =:\vM''+\vM'\quad\text{where}\quad 
    \vM':=\sum_{i=\lceil \frac{n}{2}\rceil}^n \sum_{j=1}^k\ketbra{i}{j}\quad \text{and}\quad k:=\max\{j\in[m]\colon \vM_{\lceil \frac{n}{2} \rceil, j}=1\},
  \end{align*}
and define the two associated projectors $\vPi^{\rm left}:=\sum_{i=\lceil \frac{n}{2}\rceil}^n\ketbra{i}{i}$ and $\vPi^{\rm right}:=\sum_{j=1}^{k}\ketbra{j}{j}$. By triangle inequality and linearity of $\CS_{\vM}$ in $\vM$, then $\nrm{\CS_{\vM}}_{\infty\to\infty}\leq\nrm{\CS_{\vM'}}_{\infty\to\infty}+\nrm{\CS_{\vM''}}_{\infty\to\infty}$
and
\begin{align*}
    \nrm{\CS_{\vM'}}_{\infty\to\infty}\leq\|\vPi^{\rm left}\|\|{\vPi^{\rm right}}\|\leq 1 \quad \text{since}\quad \vM' \circpr \vA=\vPi^{\rm left} \vA \vPi^{\rm right}.
\end{align*}
On the other hand, $\vM''$ is block-diagonal, which means that $\vM''=\vM^{(1)}+\vM^{(2)}$ where $\vM^{(1)}=(\vI-\vPi^{\rm left})\vM''\vPi^{\rm right}$ and $\vM^{(2)}=\vPi^{\rm left}\vM''(\vI-\vPi^{\rm right})$, thus 
\begin{align*}
\nrm{\CS_{\vM''}}_{\infty\to\infty}\!\leq\max(\nrm{\CS_{\vM^{(1)}}}_{\infty\to\infty},\nrm{\CS_{\vM^{(2)}}}_{\infty\to\infty})  \leq 1+\lceil\log_2(n/2)\rceil = \lceil\log_2{n}\rceil,
\end{align*}
where we used induction and the fact that, for permutations $\pi, \pi'$, $\nrm{\CS_{\pi\vM\pi'}}_{\infty\to\infty}=\nrm{\CS_{\vM}}_{\infty\to\infty}$ because 
$(\pi \vM\pi') \circpr \vA=\pi( \vM \circpr (\pi^{-1}\vA\pi'^{-1}))\pi'$ and permutations (being unitaries) do not
change the operator norm. The $m< n$ case is completely analogous.
\end{proof}    

\Cref{lem:GenTriangle} is a generalization of the norm bound on the so-called ``main triangle projection''~\cite{kwapien1970MainTriangleProjection}. We can readily apply \Cref{thm:SStrengthBound} to the case of Metropolis or Glauber dynamics.%
\begin{corollary}\label{cor:concreteBounds}
    Let $\gamma_M(\nu) := \min(1,\e^{-\nu})$ and $\gamma'_G(\nu) := 1/(1+\exp(\nu/2))^2=\big(\frac12-\frac12\tanh(\frac{\nu}{4})\big)^2$. For $\fenergy_M(\nu):=\sqrt{\gamma_M(\nu)}-\frac12 = \min(\frac{1}{2},\e^{-\frac{\nu}{2}}-\frac{1}{2})$ and $\fenergy_G(\nu):=\sqrt{\gamma_G(\nu)}-\frac12=-\frac12\tanh\left(\frac{\nu}{4}\right)$,
    \begin{align*}
        \|\CS^{[\hat{f}]}\|_{\infty\to\infty}&\leq 1+\min\left(\frac{1}{\pi}\ln(1 + \|\vH\|), \lceil\log_2{d}\rceil\right), \qquad\text{where}~ \hat{f}=\hat{f}_G, \hat{f}_M.
    \end{align*}
\end{corollary}
\begin{proof}
    Note that $\ftime_G(t)=\sqrt{2\pi}\frac{\ri}{\sinh(2\pi t)}$ for $\fenergy_G(\nu)=-\frac12\tanh\left(\frac{\nu}{4}\right)$, and $\ftime_M(t)=\frac{1}{\sqrt{2\pi}}\frac{1}{t(2t-\ri)}$ for $\fenergy_M(\nu) = \min(\frac{1}{2},\e^{-\frac{\nu}{2}}-\frac{1}{2})$. In order to bound $\|\CS^{[\hat{f}_G]}\|_{\infty\to\infty}$, we first compute
    \begin{align*}
        \nrm{\frac{f_G}{\sqrt{2\pi}}- \pmb{1}_{[-\theta,\theta]}(t)\frac{\ri}{2\pi t}}_1 &= 2\int_{0}^\theta \left(\frac{1}{2\pi t} - \frac{1}{\sinh(2\pi t)} \right)\rd t + 2\int_\theta^\infty \frac{1}{\sinh(2\pi t)}\rd t \\
        &= \frac{1}{\pi}\ln\theta - \frac{2}{\pi}\ln(\tanh(\pi \theta)) + \frac{1}{\pi}\ln\pi\\
        &\leq \frac{1}{\pi}\ln(\pi\theta) + \frac{2}{\pi}\ln\left(1 + \frac{1}{\pi \theta}\right),
    \end{align*}
    where we used that $\int \frac{1}{\sinh(2\pi t)}\rd t = \frac{1}{2\pi}\ln(\tanh(\pi t)) + \text{constant}$ and\footnote{For the last inequality note that for $x>0$ we have $\frac{\e^{2x} + 1}{\e^{2x} - 1} < 1 + \frac{1}{x}\Leftrightarrow 2 < \frac{\e^{2x} - 1}{x}\Leftrightarrow 1+2x < \e^{2x}$.}
    \begin{align}\label{eq:tanh_inequality}
        \frac{1}{\tanh(\pi\theta)} = \frac{1 + \exp(-2\pi\theta)}{1 - \exp(-2\pi \theta)} \leq 1 + \frac{1}{\pi\theta}.
    \end{align}
    Therefore, according to \Cref{thm:SStrengthBound}, and picking $\theta = \frac{1}{c(1+\|\vH\|)} \leq 1$ with $c=6.16$,
    \begin{align*}
        \|\CS^{[\hat{f}_G]}\|_{\infty\to\infty} &\leq \frac{2\theta}{\pi}\|\vH\| + \frac{1}{\pi}\ln(\pi\theta) + \frac{2}{\pi}\ln\left(1 + \frac{1}{\pi\theta}\right)\\
        &= \frac{2}{c\pi}\frac{\|\vH\|}{1+\|\vH\| } - \frac{1}{\pi}\ln\left(\frac{c}{\pi}(1+\|\vH\|)\right) + \frac{2}{\pi}\ln\left(1+\frac{c}{\pi}(1+\|\vH\|)\right) \\
        &\leq \frac{2}{c\pi} - \frac{1}{\pi}\ln\frac{c}{\pi} + \frac{2}{\pi}\ln\left(1+\frac{c}{\pi}\right) + \frac{1}{\pi}\ln(1 + \|\vH\|)\\
        &\leq 0.59 + \frac{1}{\pi}\ln(1 + \|\vH\|).
    \end{align*}

    Regarding $\|\CS^{[\hat{f}_M]}\|_{\infty\to\infty}$, we have that
    \begin{align*}
        \nrm{\frac{f_M}{\sqrt{2\pi}}- \pmb{1}_{[-\theta,\theta]}(t)\frac{\ri}{2\pi t}}_1 &= \frac{1}{2\pi}\int_{-\theta}^\theta\left|\frac{1}{t(2t-\ri)} - \frac{\ri}{t}\right|\rd t + \frac{1}{\pi}\int_{\theta}^\infty \frac{1}{t\sqrt{4t^2 + 1}}\rd t\\
        &= \frac{2}{\pi}\int_{0}^\theta\frac{1}{\sqrt{4t^2+1}}\rd t + \frac{1}{\pi}\ln(1+\sqrt{4\theta^2+1}) - \frac{1}{\pi}\ln{\theta} - \frac{1}{\pi}\ln{2}\\
        &= \frac{1}{\pi}\ln(\theta+\sqrt{4\theta^2+1}) + \frac{1}{\pi}\ln(1+\sqrt{4\theta^2+1}) - \frac{1}{\pi}\ln(2\theta)\\
        &\leq \frac{1}{\pi}\ln(1+\theta + 2\theta^2) + \frac{1}{\pi}\ln(2+2\theta^2) - \frac{1}{\pi}\ln(2\theta),
    \end{align*}
    using $\int\frac{1}{t\sqrt{4t^2 + 1}}\rd t = \ln{t} - \ln(1+\sqrt{4t^2+1}) + \text{constant}$ and $\int\frac{1}{\sqrt{4t^2 + 1}}\rd t = \frac{1}{2}\ln(t + \sqrt{4t^2 + 1}) + \text{constant}$. Therefore, according to \Cref{thm:SStrengthBound}, and picking $\theta = \frac{1}{c(1+\|\vH\|)} \leq 1$ with $c=3.94$,
    \begin{align*}
        \|\CS^{[\hat{f}_M]}\|_{\infty\to\infty} 
        &\leq \frac{2\theta}{\pi}\|\vH\| + \frac{1}{\pi}\ln(1+\theta+2\theta^2) + \frac{1}{\pi}\ln(2+2\theta^2) - \frac{1}{\pi}\ln(2\theta)\\
        &\leq \frac{2}{c\pi}\frac{\|\vH\|}{1 \!+\! \|\vH\|} + \frac{1}{\pi}\ln\left(1 + \frac{\frac1c+\frac2{c^2}}{1\!+\!\|\vH\|}\right) + \frac{1}{\pi}\ln\left(2 + \frac{2/c^2}{1\!+\!\|\vH\|}\right) + \frac{1}{\pi}\ln\left(\frac{c}{2}(1\!+\!\|\vH\|)\right) \\
        &\leq \frac{2}{c\pi} + \frac{1}{\pi}\ln\left(1+\frac{1}{c}+\frac{2}{c^2}\right) + \frac{1}{\pi}\ln\left(2+\frac{2}{c^2}\right) + \frac{1}{\pi}\ln\frac{c}{2} + \frac{1}{\pi}\ln(1 + \|\vH\|)\\
        &\leq 0.73 + \frac{1}{\pi}\ln(1 + \|\vH\|).
    \end{align*}

    Finally, since $\fenergy_G,\hat{f}_M\colon \BR\rightarrow [0,1]$ are both monotone decreasing, we get $|\hat{f}_G|_{[-\tau,\tau]}+|\fenergy_G(\tau)|\leq 1$ and $|\hat{f}_M|_{[-\tau,\tau]}+|\fenergy_M(\tau)|\leq 1$ for all $\tau\geq 0$, and the final inequality in \Cref{thm:SStrengthBound} yields the dimension-dependent bounds.
\end{proof}

We can use the related explicit example in~\cite{kwapien1970MainTriangleProjection} to show that the bounds in \Cref{thm:SStrengthBound} are tight up to constant factors, noting that $\sqrt{\|(\BS_C^{[\fenergy]}\llbracket\CT\rrbracket)^\dagger[\vI]\|/\nrm{\CT^\dagger[\vI]}}=\|\CS^{[\fenergy]}[\vA]\|/\nrm{\vA}$ for $\CT[\cdot]=\vA[\cdot] \vA^\dagger$.
\begin{proposition}
    {\rm \Cref{cor:concreteBounds}}, and thus also {\rm \Cref{thm:SStrengthBound}} are tight up to constant factors. Let $\vH=\sum_{k=-2^{n}}^{2^n}k\ketbra{k}{k}$ and $\vA=\sum_{-2^{n}\leq i < j \leq 2^n}\frac{1}{i-j}(\ketbra{i}{j}-\ketbra{j}{i})$, then $\|\CS^{[\sqrt{\gamma}]}[\vA]\|/\nrm{\vA}=\Omega(n)$ for both $\gamma=\gamma_G,\gamma_M$. 
\end{proposition}
\begin{proof}
	First we note that $\nrm{\vA}\leq \pi$, see for example \cite{titchmarsh1926ReciprocalFormulaeSeries}.
	On the other hand we have 
	\begin{align*}
		\|\CS^{[\sqrt{\gamma}]}[\vA]\|\geq\frac{1}{2^{n+1}+1}\sum_{i,j=-2^{n}}^{2^n}\bra{i}\CS^{[\sqrt{\gamma}]}[\vA]\ket{j}=\Omega(n).\tag*{\qedhere}
	\end{align*}
\end{proof}

\begin{theorem}\label{thm:SSOStrength}
	For every $f\colon\BR\rightarrow [0,1]$  we have 
	\begin{align}\label{eq:SSOBound}
		\vertiii{\BS_C^{[f]}}\leq\|\CS^{[f]}\|_{\infty\to\infty}^2.
	\end{align}
\end{theorem}
\begin{proof}
	We can think about the map $\CS^{[f]}$ as matrix-element-wise multiplication by the Schur-multiplier $\bm{f}:=\sum_{i,j\in[d]: \energy_i-\energy_j = \nu} f(\nu)\ketbra{\psi_i}{\psi_j}$, i.e., $\CS^{[f]}[\vA]=\vA\circpr \bm{f}$, if we treat $\vA$ as a matrix expressed in the eigenbasis of $\vH$. Due to a result by Grothendieck~\cite[Theorem 5.1]{pisier1996CompletelyBoundedMaps} we know that the superoperator-norm $\|\CS^{[f]}\|_{\infty\to\infty}$ always coincides with the corresponding completely bounded norm $\|\CS^{[f]}\|_{\infty\to\infty}^{c.b.}$ of $\CS^{[f]}$. Considering the Stinespring dilation $\vG$ of the map $\CT[\cdot]$, e.g., $\vG = \sum_{a\in\Sigma} |e_a\rangle\otimes \vA^a$ if $\CT[\cdot]$ has the Kraus representation $\sum_{a\in A}\vA^a[\cdot]\vA^{a\dagger}$, we see that 
	\begin{align*}
		\vG^\dagger \vG = \left(\sum_{a\in A} \langle e_a|\otimes \vA^{a\dagger} \right)\left(\sum_{b\in A} |e_b\rangle\otimes \vA^b \right) = \sum_{a\in A} \vA^{a\dagger} \vA^a = \CT^\dagger[\vI].
	\end{align*}
	Moreover, it is also easy to see that $\vG':=(\CI\otimes \CS^{[f]})[\vG]$ is a dilation of $\CT'=\BS_C^{[f]}\llbracket\CT\rrbracket$, implying
	\begin{align*}
		\frac{\big\|(\BS_C^{[f]}\llbracket\CT\rrbracket)^\dagger[\vI]\big\|}{\nrm{\CT^\dagger[\vI]}}=\frac{\nrm{\vG'}^2}{\nrm{\vG}^2}\leq \|\CI\otimes \CS^{[f]}\|_{\infty\to\infty}^2 \leq (\|\CS^{[f]}\|_{\infty\to\infty}^{c.b.})^2 = \|\CS^{[f]}\|_{\infty\to\infty}^2.\tag*{\qedhere}
	\end{align*}
\end{proof} 

\subsection{Truncating the integral}\label{apx:ApxInts}

We show how to truncate the integral \eqref{eq:IntRepLim} in order to get a better behaved approximate expression.%
\begin{proposition}[Truncated time integrals]\label{prop:TruncatedIntegral}
	For $\eps\in(0,\frac{1}{5}]$, let 
	\begin{align*}
		\widetilde{\CS}[\vA]\! :=\! \int_{\Delta_{\eps}} \frac{\ri}{\sinh(2\pi t)}\e^{\ri \vH t} \vA \e^{-\ri \vH t} \rd t\quad \text{where}\quad
		\Delta_{\eps}:=\left[-\frac{\ln(4/\eps)}{2\pi},\frac{\ln(4/\eps)}{2\pi}\right]\setminus\left(-\frac{\eps}{2\nrm{\vH}},\frac{\eps}{2\nrm{\vH}}\right).
	\end{align*}
 Then $\|\CS - \widetilde{\CS}\|_{\infty\to\infty}\leq \frac{\eps}{2}$, $\|\CS^{\pm} - \widetilde{\CS}^{\pm}\|_{\infty\to\infty}\leq \eps \|\CS^{\pm}\|_{\infty\to\infty}$,
	\begin{align*}
		\|\CS^{\pm}\|_{\infty\to\infty}\in \left[\frac12,1+\frac1\pi\ln\left(1+\frac{\nrm{\vH}}{2}\right)\right], \qquad\text{and}\qquad \|\widetilde{\CS}^{\pm}\|_{\infty\to\infty}\leq \frac{1}{2}+\frac1\pi\ln\left(1+\frac{2\nrm{\vH}}{\varepsilon\pi}\right).
	\end{align*}
\end{proposition}
Note that $\Delta_{\eps} = \emptyset$ if $\frac{\eps}{\|\vH\|} > \frac{\ln(4/\eps)}{\pi}$, in which case $\widetilde{\CS} = 0$ and consequently $\|\CS\|_{\infty\to\infty}\leq \frac{\eps}{2}$.
\begin{proof}
	In order to handle the apparent singularity near $0$, we use the following integral identity\footnote{The identity simply comes from the matrix valued primitive function $\frac{\rd}{\rd t}(\e^{\ri \vH t} \vA \e^{-\ri \vH t})=\ri\e^{\ri \vH t} [\vH, \vA]\e^{-\ri \vH t}$.}
	\begin{align}
		\e^{\ri \vH t} \vA \e^{-\ri \vH t} - \vA =\ri \int_{0}^{t}\e^{\ri \vH s} [\vH, \vA]\e^{-\ri \vH s}\rd s.
	\end{align}
	The key observation is that the leading order term cancels since the denominator is odd, implying
	\begin{align*}
		\lim_{\theta\rightarrow 0+} \int_{-\frac{\eps}{2\nrm{\vH}}}^{\frac{\eps}{2\nrm{\vH}}} \frac{\ri\pmb{1}_{[-\theta,\theta]}(t)}{\sinh(2\pi t)}\e^{\ri\vH t} \vA \e^{-\ri\vH t} \rd t
		&=
		\int_{-\frac{\eps}{2\nrm{\vH}}}^{\frac{\eps}{2\nrm{\vH}}} \frac{-1}{\sinh(2\pi t)}\left(\int_{0}^{t}\e^{\ri \vH s} [\vH, \vA]\e^{-\ri \vH s}\rd s\right) \rd t.
	\end{align*}
	Thus, the truncation error around zero can be upper bounded as follows
	\begin{align}
		\lnorm{\int_{-\frac{\eps}{2\nrm{\vH}}}^{\frac{\eps}{2\nrm{\vH}}} \frac{-1}{\sinh(2\pi t)}\left(\int_{0}^{t}\e^{\ri \vH s} [\vH, \vA]\e^{-\ri \vH s}\rd s\right) \rd t} &
		\le \int_{-\frac{\eps}{2\nrm{\vH}}}^{\frac{\eps}{2\nrm{\vH}}} \frac{t\nrm{[\vH, \vA]}}{\sinh(2\pi t)} \rd t\nonumber\\&
		\le \frac{\eps}{2\pi\norm{\vH}}\nrm{[\vH, \vA]}\tag{since $\labs{\sinh(x)}\ge \labs{x}$}\nonumber\\&
		\le \frac{1}{\pi}\eps\norm{\vA}.\label{eq:divTrunc}
	\end{align}	
	
	In order to bound the tail truncation error in \eqref{eq:IntRepLim} we use that
	\begin{align}\label{eq:sinhcInt}
		\int_{b}^\infty \frac{1}{\sinh(2\pi t)} \rd t
		= -\frac{1}{2\pi}\ln\left(\tanh\left(\pi b\right)\right)
		=\frac{1}{2\pi}\ln\left(\frac{1+\exp(-2\pi b)}{1-\exp(-2\pi b)}\right)
        \leq\frac{1}{2\pi}\ln\left(1+\frac{1}{\pi b}\right),
	\end{align}
	which then gives the following bound
	\begin{align*}
		\lnorm{\int_{\labs{t} \ge \frac{\ln(4/\eps)}{2\pi}} \frac{1}{\sinh(2\pi t)}\e^{\ri \vH t} \vA \e^{-\ri \vH t} \rd t} &\le 2\norm{\vA} \int_{\frac{\ln(4/\eps)}{2\pi}}^\infty \frac{1}{\sinh(2\pi t)} \rd t\\
		&= \frac{\|\vA\|}{\pi}\ln\left(\frac{1+\eps/4}{1-\eps/4}\right)\tag{by \eqref{eq:sinhcInt}}\\
		&\le \frac{5}{9\pi}\eps\norm{\vA}. \tag*{(since $\ln{x} \leq x-1$ and $\eps \leq \frac{1}{5}$)}
	\end{align*}	
				
	Since $\frac{5}{9\pi}+\frac{1}{\pi}< \frac{1}{2}$, then $\|\CS - \widetilde{\CS}\|_{\infty\to\infty}\leq \frac{\eps}{2}$ follows. Moreover, $\|\CS^{\pm}\|_{\infty\to\infty}\geq \|\CS^{\pm}[\vI]\|= \|\vI/2\|=\frac{1}{2}$, so that $\|\CS^{\pm} - \widetilde{\CS}^{\pm}\|_{\infty\to\infty} = \|\CS - \widetilde{\CS}\|_{\infty\to\infty}\leq \eps \|\CS^{\pm}\|_{\infty\to\infty}$ also follows. Regarding the upper bound on $\|\CS^{\pm}\|_{\infty\to\infty}$, set $\eps=\frac\pi2$ in \eqref{eq:divTrunc} and $b=\frac{\pi}{4\nrm{\vH}}$ in \eqref{eq:sinhcInt}, so that $\|\CS\|_{\infty\to\infty}\leq \frac{1}{2}+\frac{1}{\pi}{\ln}\Big(1+\frac{\nrm{\vH}}{2}\Big)$. Finally, the bound on $\|\widetilde{\CS}^{\pm}\|_{\infty\to\infty}$ comes from \eqref{eq:sinhcInt} with $b = \frac{\varepsilon/2}{\|\vH\|}$.
\end{proof}

Another smooth approximation can be obtained by observing that
\begin{align*}
	-\frac12\tanh\left(\frac\nu4\right) \approx -\frac12\tanh\left(\frac\nu4\right) + \frac14\tanh\left(\frac\nu4+\theta\right)+\frac14\tanh\left(\frac\nu4-\theta\right),
\end{align*}
where the approximation error 
\begin{align*}
	 \frac14\tanh\left(\frac\nu4+\theta\right)+\frac14\tanh\left(\frac\nu4-\theta\right),
\end{align*}
is an odd monotone increasing function of $\nu$. The total variation of the above function on the interval $[-2\nrm{\vH},2\nrm{\vH}]$ is therefore at most 
\begin{align}
 	\frac12\tanh\left(\frac{\nrm{\vH}}{2}+\theta\right)+\frac12\tanh\left(\frac{\nrm{\vH}}{2}-\theta\right)&\leq
 	\frac12-\frac12\tanh\left(\theta-\frac{\nrm{\vH}}{2}\right)\nonumber\\&
 	=\frac{1}{1+\exp(2\theta-\nrm{\vH})}\nonumber\\&
 	\leq \exp(\nrm{\vH}-2\theta).\label{eq:TVBound}
\end{align}
From this we get the following.
\begin{proposition}\label{prop:SmoothIntegral}
	If $\eta \geq \frac12\nrm{\vH}+\ln\left(\frac{\!3\lceil\log_2(d)\rceil + 3\!}{2\eps}\right)$, then
	$\nrm{\CS[\cdot]\! -\! \int_{\BR}\! \frac{\ri(1-\cos(\eta t))}{\sinh(2\pi t)}\e^{\ri \vH t} [\cdot] \e^{-\ri \vH t} \rd t}_{\infty\to\infty}\!\!\leq \eps.$
\end{proposition}
\begin{proof}
	As noted in the proof of \Cref{cor:concreteBounds}, we have that 
	\begin{align*}
		-\frac{1}{2}\tanh\left(\frac{\nu}{4}\right)=\lim_{\theta\rightarrow 0+}  \int_{\BR\setminus[-\theta,\theta]} \frac{\ri}{\sinh(2\pi t)}\e^{\ri\nu t} \rd t,
	\end{align*}	
	which due to the ``shift to phase'' behaviour of the Fourier transform directly implies that
	\begin{align*}
		-\frac12\tanh\left(\frac\nu4\right) + \frac14\tanh\left(\frac\nu4+\eta\right)+\frac14\tanh\left(\frac\nu4-\eta\right)=\lim_{\theta\rightarrow 0+}  \int_{\BR\setminus[-\theta,\theta]} \frac{\ri(1-\cos(\eta t))}{\sinh(2\pi t)}\e^{\ri\nu t} \rd t,
	\end{align*}
	because $(1 - (\exp(-\ri\eta t)+\exp(\ri\eta t))/2)=1-\cos(\eta t)$.
	By \Cref{thm:SStrengthBound} and \eqref{eq:TVBound} we get that the induced error can be upper bounded by $\frac{3}{2}\exp(\nrm{\vH}-2\eta)(\lceil\log_2(d)\rceil + 1)$.

	Finally, note that the above integral approximation removes the singularity because
	\begin{align}
		\frac{(1-\cos(\eta t))}{\sinh(2\pi t)}=\frac{2\sin^2(\eta t/2)}{\sinh(2\pi t)}=\frac{\eta }{2\pi} \frac{\sin(\eta t/2)\, \mathrm{sinc}(\eta t/2)}{\mathrm{sinhc}(2\pi t)}.\tag*{\qedhere}
	\end{align}	
\end{proof}

\subsection{Bounds on the operator norm of the discriminant matrix}\label{apx:DiscNormBound}
For the detailed-balanced construction of \Cref{sec:coherentDyn} we have shown in \Cref{sec:coherentDynDisc} that the discriminant corresponding to the continuous-time constructions associated to $\CT[\cdot]=\vA[ \cdot ]\vA^\dagger$ is
\begin{align*}
    \vCD=\CS_c[\vA]\otimes\overline{\CS_c[\vA]}-\CS_c[\vD]\otimes\vI-\vI\otimes\overline{\CS_c[\vD]} \qquad\text{where}\qquad\vD=(\CS^-[\vA])^\dagger\CS^-[\vA],
\end{align*}
and $\overline{\vM}$ stands for the element-wise complex conjugation of the matrix $\vM$ in the computational basis.
Here we prove in general that $\nrm{\vCD}=\mathcal{O}(\nrm{\vCT}+\|\CT^\dagger[\vI]\|)$ through the following lemma.
\begin{lemma}\label{lem:ADiscBound}
For every matrix $\vA\in\mathbb{C}^{d\times d}$ it holds that
    $\nrm{\CS_c[(\CS^-[\vA])^\dagger\CS^-[\vA]]}=\mathcal{O}(\|\vA^\dagger\vA\|)$.
\end{lemma}
We first give a proof sketch: The main idea is to look at block-matrices, where we group eigenvectors together based on the integer part of their energy. Due to the outer $\CS_c$ superoperator, it suffices to show that every block of $(\CS^-[\vA])^\dagger\CS^-[\vA]$ has norm bounded by $\mathcal{O}(\|\vA^\dagger\vA\|)$. Indeed, we can decompose $(\CS^-[\vA])^\dagger\CS^-[\vA]$ into affine block-diagonal parts, and conclude that contributions of these are exponentially small in terms of how far they are from the main diagonal due to the fast decay of $1/\cosh(\nu/4)$. To see that every block of $(\CS^-[\vA])^\dagger\CS^-[\vA]$ has norm bounded by $\mathcal{O}(\|\vA^\dagger\vA\|)$, first consider $\widetilde{\vA}$ as a crude proxy for $\CS^-[\vA]$, such that $\widetilde{\vA}$ is strictly upper-block diagonal, and it agrees with the blocks of $\vA$ above the diagonal. The blocks of $(\widetilde{\vA})^\dagger\widetilde{\vA}$ are clearly bounded by $\mathcal{O}(\|\vA^\dagger\vA\|)$, and showing that $\|\widetilde{\vA}-\CS^-[\vA]\|\leq \bigO{\nrm{\vA}}$ completes the proof.

\begin{proof}
    In the proof we will heavily rely on the following identities borrowed from \Cref{sec:coherentDynDisc}
    \begin{align}\label{eq:ScConjs}
        \CS_c[\vA]:=\vrho^{-\frac14}\CS^-[\vA]\vrho^{\frac14}=\vrho^{\frac14}\CS^+[\vA]\vrho^{-\frac14}=\sum_{\nu\in B(\vH)}\frac{1}{2\cosh(\nu/4)}\vA_\nu=\int_{\BR} \frac{1}{\cosh(2\pi t)}\vA(t) \rd t
    \end{align}
    and 
    \begin{align}\label{eq:SpmComp}
        \CI=\CS^+ + \CS^-.
    \end{align}    

    The main idea is to look at block-matrices, where we group eigenvectors together based on the integer part of their energy. More formally let us define
    \begin{align*}
        \vPi_k:=\sum_{i:\lfloor E_i\rfloor=k}\ketbra{\psi_i}{\psi_i}\qquad \text{ and } \qquad \vM_{[k,j]}:=\vPi_k\vM\vPi_j.
    \end{align*}
    First we show that
    \begin{align}\label{eq:blockNormBound}
        \nrm{\CS_c[\vD]}\leq \bigO{\max_{k,j}\|\vD_{[k,j]}\|}.
    \end{align}
    The key idea behind the bound of \eqref{eq:blockNormBound} is that the blocks $\vD_{[k,j]}$ get exponential suppression in terms of $|k-j|$ due to the fast decay of $1/\cosh(\nu/4)$. To make this precise we employ a decomposition into ``affine block-diagonal parts'' as follows:
    \begin{align}
        \nrm{\CS_c[\vD]} &\leq \sum_\ell \nrm{\CS_c\left[\sum_k \vD_{[k,k-\ell]}\right]} \tag{by the triangle inequality}\\&
        = \sum_\ell \max_k \norm{\CS_c[\vD_{[k,k-\ell]}]}, \label{eq:AffineDiagBlocksTriangle}
    \end{align}
    where the last equality follows from the ``affine block-diagonal'' block-structure of $\sum_k \vD_{[k,k-\ell]}$.

    Now consider the modified Hamiltonian $\vH'_\lambda:=\vH_{[k,k]}+\vH_{[j,j]} -\lambda\vI$, which has norm at most $\max\{|k-\lambda|,|j-\lambda|\}+1$ by design. By \eqref{eq:ScConjs} it is clear that the matrices $\CS_c[\vD_{[k,j]}]$ and $\CS^\pm[\vD_{[k,j]}]$ do not change if we replace $\vH$ with $\vH'_\lambda$, so we may as well assume $\vH=\vH'_\lambda$. Then we get that 
    \begin{align*}
        \CS_c[\vD_{[k,j]}]&=\exp\left(\mp\frac14\vH'_\lambda\right)\CS^\pm[\vD_{[k,j]}]\exp\left(\pm\frac14\vH'_\lambda\right)\\&
        =\exp\left(\mp\frac14(\vH_{[k,k]}-\lambda\vI)\right)\CS^\pm[\vD_{[k,j]}]\exp\left(\pm\frac14(\vH_{[j,j]}-\lambda\vI)\right)\\&
        =\exp\left(\mp\frac{k\!-\!\lambda}{4}\right)\exp\left(\mp\frac14(\vH_{[k,k]}-\!k\vI)\right)\CS^\pm[\vD_{[k,j]}]\exp\left(\pm\frac14(\vH_{[j,j]}-\!j\vI)\right)\exp\left(\pm\frac{j\!-\!\lambda}{4}\right).  
    \end{align*}
    Choosing $\lambda=k$ and $\lambda=j$ respectively we get
    \begin{align*}
        \norm{\CS_c[\vD_{[k,j]}]}&\leq\norm{\CS^+[\vD_{[k,j]}]}\e^{\frac14}\exp\left(\frac{j-k}{4}\right),\\
        \norm{\CS_c[\vD_{[k,j]}]}&\leq\exp\left(\frac{k-j}{4}\right)\e^{\frac14}\norm{\CS^-[\vD_{[k,j]}]}.        
    \end{align*}
    Applying \Cref{cor:concreteBounds} with $\vH'_\lambda$ further implies that 
    \begin{align*}
        \norm{\CS_c[\vD_{[k,j]}]}&\leq\e^{-\frac{|k-j|}{4}}\e^{\frac14}(\ln(|k-j|+2)+c)\|\vD_{[k,j]}\|,
    \end{align*}
    which combined with \eqref{eq:AffineDiagBlocksTriangle} implies \eqref{eq:blockNormBound} as follows
    \begin{align*}
        \nrm{\CS_c[\vD]} &\leq \sum_\ell \e^{-\frac{|\ell|}{4}}\e^{\frac14}(\ln(|\ell|+2)+c)\max_k\|{\vD_{[k,k-\ell]}}\|\\& 
        \leq \sum_{\ell=-\infty}^\infty \e^{-\frac{|\ell|}{4}}\e^{\frac14}(\ln(|\ell|+2)+c)\max_{k,j} \|{\vD_{[k,j]}}\|.    
    \end{align*}

    It remains to bound the quantity $\max_{k,j} \|{\vD_{[k,j]}}\|$. Let us introduce the projector $\vPi_{<j}:=\sum_{k<j}\vPi_k$ and define $\widetilde{\vA}$ as a crude proxy for $\CS^-[\vA]$ defined block-by-block as $\widetilde{\vA}_{[i,j]}:=\vPi_{<j}\vA_{[i,j]}$.
    Since $\vD=(\widetilde{\vA}+(\CS^-[\vA]-\widetilde{\vA}))^\dagger(\widetilde{\vA}+(\CS^-[\vA]-\widetilde{\vA})) $ it suffices to bound
    \begin{align}
        \norm{(\widetilde{\vA}^\dagger \widetilde{\vA})_{[k,j]}}&+ \norm{(\widetilde{\vA}^\dagger (\CS^-[\vA]-\widetilde{\vA}))_{[k,j]}}+ \norm{((\CS^-[\vA]-\widetilde{\vA})^\dagger \widetilde{\vA})_{[k,j]}}\nonumber\\&+ \norm{((\CS^-[\vA]-\widetilde{\vA})^\dagger(\CS^-[\vA]-\widetilde{\vA}))_{[k,j]}}. \label{eq:ATildeDec} 
    \end{align}    
    Observe that 
    \begin{align*}
        (\widetilde{\vA}^\dagger \widetilde{\vA})_{[k,j]}=\vPi_k(\widetilde{\vA}^\dagger \widetilde{\vA})\vPi_j
        =\vPi_k\vA^\dagger\vPi_{<k} \vPi_{<j}\vA\vPi_j,
    \end{align*}
    and as the operator norm of a projector is at most $1$ we get 
    \begin{align*}
        \|{(\widetilde{\vA}^\dagger \widetilde{\vA})_{[k,j]}}\|\leq\|\vA^\dagger\|\nrm{\vA}=\|{\vA^\dagger\vA}\|.
    \end{align*}
    Following similar arguments for each term in \eqref{eq:ATildeDec} and using $\norm{\vM^\dagger} = \norm{\vM}$ for any matrix $\vM$,
    \begin{align*}
        \max_{k,j} \|{\vD_{[k,j]}}\|\leq (\nrm{\vA}+ \|{\CS^-[\vA]-\widetilde{\vA}}\|)^2.
    \end{align*}
    Finally, showing that $\|{\CS^-[\vA]-\widetilde{\vA}}\|=\bigO{\nrm{\vA}}$ completes the proof.
    We employ a similar decomposition as in \eqref{eq:AffineDiagBlocksTriangle}:
    \begin{align}
        \|{\CS^-[\vA]-\widetilde{\vA}}\| &\leq \sum_\ell \nrm{\sum_k (\CS^-[\vA]-\widetilde{\vA})_{[k+\ell,k]}} \tag{by the triangle inequality}\\&
        = \sum_\ell \nrm{\max_k (\CS^-[\vA]-\widetilde{\vA})_{[k+\ell,k]}} \tag{due to ``affine block-diagonal'' structure}\\& 
        = \sum_{\ell< 0} \Bigg\lVert\max_k  \underset{(\CS^+[\vA])_{[k+\ell,k]}=((\CS^-[\vA^\dagger])_{[k,k+\ell]})^\dagger}{\underbrace{(\CS^-[\vA]-\widetilde{\vA})_{[k+\ell,k]}}}\Bigg\rVert + \sum_{\ell\geq 0} \nrm{\max_k  (\CS^-[\vA])_{[k+\ell,k]}} \tag{by \eqref{eq:SpmComp}}\\&
        = \sum_{\ell\geq 0} \nrm{\max_k  (\CS^-[\vA^\dagger])_{[k+\ell+1,k]}} + \nrm{\max_k  (\CS^-[\vA])_{[k+\ell,k]}}. \label{eq:TildeADiffDecomp}
    \end{align}
    Now we use once again \eqref{eq:ScConjs}
    \begin{align*}
        (\CS^-[\vA])_{[k+\ell,k]}&=(\vrho^{\frac14}\CS_c[\vA]\vrho^{-\frac14})_{[k+\ell,k]}\\&
        =\left(\exp\left(-\frac14(\vH-k\vI)\right)\CS_c[\vA]\exp\left(\frac14(\vH-k\vI)\right)\right)_{[k+\ell,k]}\\&        
        =\e^{-\frac{\ell}{4}}\exp\left(-\frac14(\vH_{[k+\ell,k+\ell]}-(k+\ell)\vI)\right)\vPi_{k+\ell}\CS_c[\vA]\vPi_k\exp\left(\frac14(\vH_{[k,k]}-k\vI)\right),
    \end{align*}
    showing that $\norm{(\CS^-[\vA])_{[k+\ell,k]}}\leq\e^{\frac{1-\ell}{4}}\nrm{\CS_c[\vA]}$. 
    Due to the integral representation of \eqref{eq:ScConjs} and the triangle inequality this further implies $\norm{(\CS^-[\vA])_{[k+\ell,k]}}\leq\e^{\frac{1-\ell}{4}}\frac{1}{2}\nrm{\vA}$. Combining this with \eqref{eq:TildeADiffDecomp} we get that $\|{\CS^-[\vA]-\widetilde{\vA}}\|\leq \sum_{\ell= 0}^\infty \e^{\frac{1-\ell}{4}}\nrm{\vA}=\frac{\e^{\frac14}}{1-\e^{-\frac14}}\nrm{\vA}\leq 5.81\nrm{\vA}$, completing the proof.
\end{proof}
Note that a similar argument shows the same result for the Metropolis variant, i.e., when we replace $\CS^-$ in the above lemma by $\CS^{[\sqrt{\gamma_M}]}$.

We conjecture that \Cref{lem:ADiscBound} holds with the tight bound $\norm{\CS_c[(\CS^-[\vA])^\dagger\CS^-[\vA]]}\leq\frac{1}{2}\|\vA^\dagger\vA\|$.

\begin{corollary}
    For every CP-map $\CT[\cdot]=\sum_{a\in A} \vA^a[\cdot](\vA^a)^\dagger$ its continuous-time discriminant matrix $\vCD$ constructed as in {\rm \Cref{sec:coherentDynDisc}} satisfies $\norm{\vCD}=\mathcal{O}(\nrm{\vCT}+\|\CT^\dagger[\vI]\|)$.
\end{corollary}
\begin{proof}
Due to linearity we have that 
\begin{align*}
    \vCD=\sum_{a\in A}\CS_c[\vA^a]\otimes\overline{\CS_c[\vA^a]}-\CS_c[\vD]\otimes\vI-\vI\otimes\overline{\CS_c[\vD]} \qquad\text{where}\qquad\vD=\sum_{a\in A}(\CS^-[\vA^a])^\dagger\CS^-[\vA^a].
\end{align*}
By the integral representation of \eqref{eq:ScConjs} and the triangle inequality it follows that 
\begin{align*}
    \nrm{\sum_{a\in A}\CS_c[\vA^a]\otimes\overline{\CS_c[\vA^a]}}\leq \frac14 \nrm{\sum_{a\in A}\vA^a\otimes\overline{\vA^a}}= \frac14\nrm{\vCT}.
\end{align*}
In order to bound $\nrm{\CS_c[\vD]}$ let us consider the dilation matrix $\vG=\sum_{a\in A} \ketbra{a}{a_0} \otimes \vA^a$ (where $a_0\in A$ is fixed), and the expanded Hamiltonian $\vI_A\otimes \vH$. Since $\CT^\dagger[\vI]=(\bra{a_0}\otimes \vI)\vG^\dagger \vG(\ket{a_0}\otimes \vI)$ we get that $\|\vG^\dagger \vG\|=\|\CT^\dagger[\vI]\|$ and since $\vD=(\bra{a_0}\otimes \vI)((\CI_A\otimes\CS^-)[\vG])^\dagger (\CI_A\otimes\CS^-)[\vG](\ket{a_0}\otimes \vI)$ due to \Cref{lem:ADiscBound} we get $\nrm{\CS_c[\vD]}=\norm{\CS_c[((\CI_A\otimes\CS^-)[\vG])^\dagger (\CI_A\otimes\CS^-)[\vG]]}=\mathcal{O}(\|\vG^\dagger \vG\|)=\mathcal{O}(\|\CT^\dagger[\vI]\|)$.
\end{proof}

Finally, note that the integral representations of the discriminant derived in~\cite{chen2023ExactQGibbsSampler} together with the observations outlined in \eqref{eq:ScConjs} imply that the same result holds for the operator Fourier-transform-based construction of \eqref{eq:OFTSDef} in \Cref{sec:OFT}. \anote{Maybe expand this better, invoke ``Parseval identity'', and the quoted integral representation.}

\section{Exact detailed balance via damped phase estimation}\label{apx:PEBasedDB}

In this appendix, we prove \Cref{thr:PEBasedDB}, which shows that the phase-estimation-based construction from \Cref{sec:PEVariant} is $\vrho$-detailed balanced.

The two-sided Fourier transform does directly not take into account the energy difference $\omega = E_f - E_i$, but both the initial and final energy levels $E_i$ and $E_f$, respectively. Once again, it is possible to obtain a continuous-parameter decomposition of any CP map $\CT[\cdot]=\sum_{a\in A}\vA^a[\cdot]\vA^{a\dagger}$ through the two-sided Fourier transform as
\begin{align*}
    \BS_H^{[\sigma_1,\sigma_2]}(E_f,E_i)\llbracket\CT\rrbracket[\cdot] := \sum_{a\in A} \left(\CF^{[\sigma_1,\sigma_2]}(E_f,E_i)[\vA^a]\right) [\cdot]\left(\CF^{[\sigma_1,\sigma_2]}(E_f,E_i)[\vA^a]\right)^\dagger.
\end{align*}
Consider then the construction
\begin{align}\label{eq:two-sided_fourier_construction}
    \BS_H^{[\sigma_1,\sigma_2,g]}\llbracket\CT\rrbracket :=& \int_{\BR^2} \gamma^{(g)}(E_f,E_i)\BS_H^{[\sigma_1,\sigma_2]}(E_f,E_i)\llbracket\CT\rrbracket  \rd E_i \rd E_f, \quad\text{where}\\
    \gamma^{(g)}(E_f,E_i) =& \int_{\frac{\sigma_1^2+\sigma_2^2}{2}}^\infty g(x)\e^{\frac{(E_f-E_i + x)^2}{4x-2\sigma_1^2 - 2\sigma_2^2}}\rd x \quad\text{for an integrable function}~ g:\Big[\frac{\sigma_1^2+\sigma_2^2}{2},\infty\Big)\to\mathbb{R}_+.\nonumber
\end{align}
Let us rewrite the super-superoperator $\BS_H^{[\sigma_1,\sigma_2,g]}\llbracket\CT\rrbracket$ in terms of the energy levels from $\operatorname{spec}(\vH)$,
\begin{align*}
    \BS_H^{[\sigma_1,\sigma_2,g]}\llbracket\CT\rrbracket[\cdot] &= \int_{\BR^2} \gamma^{(g)}(E_f,E_i)\BS_H^{[\sigma_1,\sigma_2]}(E_f,E_i)\llbracket\CT\rrbracket[\cdot]  \rd E_i \rd E_f \\
    &= \sum_{a\in A}\sum_{E_1,E_2,E_3,E_4\in\operatorname{spec}(\vH)} \alpha_{E_1,E_2,E_3,E_4}^{(g,\sigma_1,\sigma_2)} \vPi_{E_2}\vA^a\vPi_{E_1}[\cdot]\vPi_{E_3}\vA^{a\dagger}\vPi_{E_4},
\end{align*}
where
\begin{align}\label{eq:alpha_coefficients}
    \alpha_{E_1,E_2,E_3,E_4}^{(g,\sigma_1,\sigma_2)} := \int_{\mathbb{R}^2} \gamma^{(g)}(E_f,E_i)\hat{f}_{\sigma_1}(E_i - E_1) \hat{f}_{\sigma_2}(E_f - E_2)\hat{f}_{\sigma_1}(E_i - E_3) \hat{f}_{\sigma_2}(E_f - E_4) \rd E_i \rd E_f.
\end{align}
In the next lemma, we prove that $\BS_H^{[\sigma_1,\sigma_2,g]}\llbracket\CT\rrbracket$ is $\vrho$-detailed balanced if the coefficient tensor $\alpha_{E_1,E_2,E_3,E_4}^{(g,\sigma_1,\sigma_2)}$ possesses certain symmetries. The subsequent lemma then shows that the phase-estimation-based construction from \Cref{sec:PEVariant} possesses such symmetry under the condition that the initial and final energies have the same variance $\sigma_1=\sigma_2=\sigma$.

\begin{lemma}\label{thr:two-sided_fourier_construction}
    Let $\CT[\cdot] = \sum_{a\in A} \vA^a[\cdot]\vA^{a\dagger}$ be a self-adjoint superoperator and consider
    \begin{align*}
        \CT'[\cdot] = \sum_{a\in A}\sum_{E_1,E_2,E_3,E_4\in\operatorname{spec}(\vH)} \alpha_{E_1,E_2,E_3,E_4} \vPi_{E_2}\vA^a\vPi_{E_1}[\cdot]\vPi_{E_3}\vA^{a\dagger}\vPi_{E_4}
    \end{align*}
    such that $\alpha_{E_1,E_2,E_3,E_4} \e^{\frac{E_2- E_1 + E_4 - E_3}{2}} = \alpha_{E_2,E_1,E_4,E_3}$ for all $E_1,E_2,E_3,E_4\in\operatorname{spec}(\vH)$. Then $\CT'$ is $\vrho$-detailed balanced.
\end{lemma}
\begin{proof}
    A direct calculation yields
    \begin{align*}
        \vrho^{-\frac{1}{2}}&\CT'[\vrho^{\frac{1}{2}}\cdot\vrho^{\frac{1}{2}}]\vrho^{-\frac{1}{2}} \\
        &= \sum_{a\in A}\sum_{E_1,\dots,E_4\in\operatorname{spec}(\vH)} \alpha_{E_1,E_2,E_3,E_4} \vrho^{-\frac{1}{2}} \vPi_{E_2}\vA^a\vPi_{E_1}\vrho^{\frac{1}{2}}[\cdot]\vrho^{\frac{1}{2}}\vPi_{E_3}\vA^{a\dagger}\vPi_{E_4} \vrho^{-\frac{1}{2}} \\
        &= \sum_{a\in A}\sum_{E_1,\dots,E_4\in\operatorname{spec}(\vH)} \alpha_{E_1,E_2,E_3,E_4} \e^{\frac{E_2- E_1 + E_4 - E_3}{2}}  \vPi_{E_2}\vA^a\vPi_{E_1}[\cdot]\vPi_{E_3}\vA^{a\dagger}\vPi_{E_4} \tag{since $\vrho \propto \e^{-\vH}$} \\
        &= \sum_{a\in A}\sum_{E_1,\dots,E_4\in\operatorname{spec}(\vH)}\alpha_{E_2,E_1,E_4,E_3}   \vPi_{E_2}\vA^a\vPi_{E_1}[\cdot]\vPi_{E_3}\vA^{a\dagger}\vPi_{E_4} \tag{since $\alpha_{E_1,E_2,E_3,E_4} \e^{\frac{E_2- E_1 + E_4 - E_3}{2}} = \alpha_{E_2,E_1,E_4,E_3}$}\\
        &= \sum_{a\in A}\sum_{E_1,\dots,E_4\in\operatorname{spec}(\vH)}\alpha_{E_2,E_1,E_4,E_3}   \vPi_{E_2}\vA^{a\dagger}\vPi_{E_1}[\cdot]\vPi_{E_3}\vA^{a}\vPi_{E_4} \tag{since $\CT = \CT^\dagger$} \\
        &= \sum_{a\in A}\sum_{E_1,\dots,E_4\in\operatorname{spec}(\vH)}\alpha_{E_2,E_1,E_4,E_3}   (\vPi_{E_1}\vA^{a}\vPi_{E_2})^\dagger[\cdot](\vPi_{E_4}\vA^{a\dagger}\vPi_{E_3})^\dagger \\
        &= \sum_{a\in A}\sum_{E_1,\dots,E_4\in\operatorname{spec}(\vH)}\alpha_{E_1,E_2,E_3,E_4}   (\vPi_{E_2}\vA^{a}\vPi_{E_1})^\dagger[\cdot](\vPi_{E_3}\vA^{a\dagger}\vPi_{E_4})^\dagger \tag{by $E_1 \leftrightarrow E_2$, $E_3 \leftrightarrow E_4$} \\
        &= \CT'^\dagger[\cdot]. \tag*{\qedhere}
    \end{align*}
\end{proof}

It only remains to show that our choice of coefficients $\alpha_{E_1,E_2,E_3,E_4}^{(g,\sigma_1,\sigma_2)}$ from \eqref{eq:alpha_coefficients} satisfies the symmetry required in \Cref{thr:two-sided_fourier_construction}. The proof follows the same reasoning as~\cite[Lemma~II.2]{chen2023ExactQGibbsSampler}.%
\begin{lemma}
    Let the coefficients $\alpha_{E_1,E_2,E_3,E_4}^{(g,\sigma_1,\sigma_2)}$ defined in \eqref{eq:alpha_coefficients},
    \begin{align*}
        \alpha_{E_1,E_2,E_3,E_4}^{(g,\sigma_1,\sigma_2)} = \frac{1}{4\pi\sigma^2} \int_{\mathbb{R}^2}\gamma^{(g)}(E_f,E_i) \e^{-\frac{(E_i-E_1)^2 +  (E_i - E_3)^2}{4\sigma^2_1} - \frac{(E_f - E_2)^2 + (E_f - E_4)^2}{4\sigma^2_2}} \rd E_i \rd E_f,
    \end{align*}
    where $\gamma^{(g)}(E_f,E_i) = \int_{\frac{\sigma_1^2 + \sigma_2^2}{2}}^\infty g(x) \e^{-\frac{(E_f - E_i + x)^2}{4x-2\sigma^2_1 - 2\sigma^2_2}}\rd x $ as in \eqref{eq:two-sided_fourier_construction}.
    Then $\alpha_{E_1,E_2,E_3,E_4}^{(g,\sigma_1,\sigma_2)}$ can be written as
    \begin{align}\label{eq:alpha_coefficients_exp2}
        \alpha_{E_1,E_2,E_3,E_4}^{(g,\sigma_1,\sigma_2)} = \frac{1}{2}\e^{-\frac{(E_1 - E_3)^2}{8\sigma^2_1} - \frac{(E_2 - E_4)^2}{8\sigma^2_2}}\int_{\frac{\sigma_1^2+\sigma_2^2}{2}}^\infty g(x) \sqrt{1-\frac{\sigma_1^2+\sigma_2^2}{2x}}\cdot\e^{-\frac{(E_1 + E_3 - E_2 - E_4 - 2x)^2}{16x}} \rd x.
    \end{align}
    Moreover, if $\sigma_1=\sigma_2=\sigma$, then $\alpha_{E_1,E_2,E_3,E_4}^{(g,\sigma,\sigma)}$ has a ``skew-symmetry'' under transposition,
    \begin{align}\label{eq:skew_symmetry}
        \alpha^{(g,\sigma,\sigma)}_{E_1,E_2,E_3,E_4} = \alpha^{(g,\sigma,\sigma)}_{E_2,E_1,E_4,E_3}\e^{\frac{E_1- E_2 + E_3 - E_4}{2}} \quad\text{for all}~ E_1,E_2,E_3,E_4\in\operatorname{spec}(\vH).
    \end{align}
\end{lemma}
\begin{proof}
    Let us first resolve the integration over $E_i$ and $E_f$ and ignore the integration over $x$. Consider the integration over $E_i$ and $E_f$,
    \begin{align*}
        &\iint_{-\infty}^\infty \e^{-\frac{(E_f - E_i + x)^2}{4x-2\sigma_1^2 - 2\sigma_2^2}-\frac{(E_i-E_1)^2 + (E_i - E_3)^2}{4\sigma^2_1} - \frac{(E_f-E_2)^2 + (E_f - E_4)^2}{4\sigma^2_2}} \frac{\rd E_i \rd E_f}{4\pi\sigma_1\sigma_2}  \\
        &= \iint_{-\infty}^\infty \e^{-\big(\frac{1}{4x-2\sigma_1^2 - 2\sigma_2^2} + \frac{1}{2\sigma_1^2}\big) E_i^2 + \big(\frac{x+E_f}{2x-\sigma^2_1-\sigma_2^2} + \frac{E_1 + E_3}{2\sigma_1^2}\big)E_i - \frac{(x+E_f)^2}{4x-2\sigma^2_1-2\sigma_2^2} - \frac{E_1^2 + E_3^2}{4\sigma^2_1} - \frac{(E_f-E_2)^2 + (E_f - E_4)^2}{4\sigma^2_2}}\frac{\rd E_i\rd E_f}{4\pi\sigma_1\sigma_2}  \\
        &\begin{multlined}[b][\textwidth]
            = \frac{1}{4\sqrt{\pi}\sigma_1\sigma_2}\frac{1}{\sqrt{\frac{1}{4x-2\sigma_1^2-2\sigma_2^2} + \frac{1}{2\sigma^2_1}}} \\ \times\int_{-\infty}^\infty \e^{\big(\frac{x+E_f}{2x-\sigma_1^2 - \sigma_2^2} + \frac{E_1 + E_3}{2\sigma^2_1}\big)^2\big/\big(\frac{2}{2x-\sigma_1^2-\sigma_2^2} + \frac{2}{\sigma_1^2}\big) - \frac{(x+E_f)^2}{4x-2\sigma_1^2-2\sigma_2^2} - \frac{E_1^2 + E_3^2}{4\sigma^2_1} - \frac{(E_f-E_2)^2 + (E_f - E_4)^2}{4\sigma^2_2}} \rd E_f
        \end{multlined}\\
        &= \frac{1}{2\sqrt{2\pi}\sigma_2}\frac{\e^{-\frac{(E_1 - E_3)^2}{8\sigma^2_1}}}{\sqrt{1+\frac{\sigma^2_1}{2x-\sigma^2_1-\sigma_2^2}}} \int_{-\infty}^\infty\e^{-\frac{(E_1 + E_3 - 2x - 2E_f)^2}{8(2x - \sigma_2^2)} - \frac{(E_f-E_2)^2 + (E_f - E_4)^2}{4\sigma^2_2}} \rd E_f\\
        &= \frac{1}{2\sqrt{2\pi}\sigma_2}\frac{\e^{-\frac{(E_1 - E_3)^2}{8\sigma^2_1}}}{\sqrt{1+\frac{\sigma^2_1}{2x-\sigma^2_1-\sigma_2^2}}} \int_{-\infty}^\infty\e^{-\big(\frac{1}{2(2x - \sigma^2_2)} + \frac{1}{2\sigma^2_2}\big)E_f^2 + \big(\frac{E_1+E_3-2x}{2(2x - \sigma^2_2)} + \frac{E_2 + E_4}{2\sigma^2_2}\big)E_f - \frac{(E_1 + E_3 - 2x)^2}{8(2x - \sigma^2_2)} - \frac{E_2^2 + E_4^2}{4\sigma^2_2}} \rd E_f \\
        &= \frac{1}{2\sqrt{2}\sigma_2}\frac{\e^{-\frac{(E_1 - E_3)^2}{8\sigma^2_1}}}{\sqrt{1+\frac{\sigma^2_1}{2x-\sigma^2_1-\sigma_2^2}}} \frac{1}{\sqrt{\frac{1}{2(2x - \sigma^2_2)} + \frac{1}{2\sigma^2_2}}} \cdot\e^{\big(\frac{E_1+E_3-2x}{2(2x - \sigma^2_2)} + \frac{E_2 + E_4}{2\sigma^2_2}\big)^2\big/\big(\frac{2}{2x - \sigma^2_2} + \frac{2}{\sigma^2_2}\big) - \frac{(E_1 + E_3 - 2x)^2}{8(2x - \sigma^2_2)} - \frac{E_2^2 + E_4^2}{4\sigma^2_2} } \\
        &= \frac{1}{2}\sqrt{1- \frac{\sigma_1^2+\sigma_2^2}{2x}}\cdot \e^{-\frac{(E_1 - E_3)^2}{8\sigma^2_1} - \frac{(E_2 - E_4)^2}{8\sigma^2_2} - \frac{(E_1 + E_3 - E_2 - E_4 - 2x)^2}{16x}}.
    \end{align*}
    Putting the above expression together with the integral over $x$ yields \eqref{eq:alpha_coefficients_exp2}. Finally, notice that
    \begin{align*}
        \e^{- \frac{(E_1 + E_3 - E_2 - E_4 - 2x)^2}{16x}} = \e^{\frac{E_1 + E_3 - E_2 - E_4}{2}}\cdot\e^{- \frac{(E_2 + E_4 - E_1 - E_3 - 2x)^2}{16x}},
    \end{align*}
    which implies \eqref{eq:skew_symmetry} if $\sigma_1 = \sigma_2$.
\end{proof}

The above two lemmas imply that $\BS_H^{[\sigma,\sigma,g]}$ is indeed $\vrho$-detailed balanced, proving \Cref{thr:PEBasedDB}.

\section{Analysis of our new discrete constructions}\label{apx:DiscConst}
In this section, we prove various technical lemmas that solve for detailed balance in the discrete-time setting. Unlike the continuous-time cases, the CPTP constraints for quantum channels seem to complicate the recipe and require higher-order corrections.

\begin{lemma}[Balancing a single Kraus operator]\label{lem:KrausBalance}
	Given a Kraus operator $\vK'\in\BC^{d\times d}$ and a full-rank density operator $0\prec \vrho \in\BC^{d\times d}$, consider the singular value decomposition $\vK'\sqrt{\vrho}=\vV\Sigma \vW^\dagger$. 
	Then an operator $\vK$ satisfies
\begin{align}
    \vK^\dagger\vK=\vK'^\dagger\vK'\quad \text{and}\quad \vK[\cdot]\vK^\dagger ~\text{is $\vrho$-detailed balanced} \quad \text{iff}\quad \vK=\e^{\ri\varphi}\vW\vU \vV^\dagger\vK'\label{eq:KK'}
\end{align}
for some $\varphi\in\BR$ and a Hermitian unitary $\vU$ that commutes with $\Sigma$.
\end{lemma}
\begin{proof}

First, the condition $\vK^\dagger\vK=\vK'^\dagger\vK'$ is equivalent to the existence of a unitary $\vU'$ such that $\vK=\vU'\vK'$~\cite[Theorem 7.3.11]{horn1990MatrixAnalysis}.
	Next, the condition that $\vK[\cdot]\vK^\dagger$ is $\vrho$-detailed balanced is equivalent to the existence of some $\varphi\in \BR$ such that
	\begin{align}\label{eq:KrausBalanced}
		\vrho^{-\frac14}\vK\vrho^{\frac14}&=\e^{2\ri\varphi}\vrho^{\frac14}\vK^\dagger \vrho^{-\frac14}\quad \text{or equivalently,}\quad   
		\e^{-\ri\varphi}\vK\!\sqrt{\vrho}=\e^{\ri\varphi}\sqrt{\vrho}\vK^\dagger.
	\end{align}
	Let $\vK'\sqrt{\vrho}=\vV\Sigma \vW^\dagger$ be a singular value decomposition where $\vV,\vW\in \BC^{d\times d}$ are unitaries. Then
 \begin{align}\label{eq:UDef}
 \e^{-\ri\varphi}\vK\!\sqrt{\vrho}=\e^{-\ri\varphi}\vU'\vK'\sqrt{\vrho}=\vW (\undersetbrace{=:\vU}{\e^{-\ri\varphi}\vW^\dagger\vU'\vV})\Sigma \vW^\dagger    
 \end{align}
 and thus $\e^{-\ri\varphi}\vK\!\sqrt{\vrho}$ is self-adjoint if and only if $\vU\Sigma$ is self-adjoint.
 
First consider the trivial direction: if $\vU$ is self-adjoint and commutes with $\Sigma$, then we have $\vU\Sigma=\Sigma \vU=\Sigma \vU^\dagger$, which is indeed self-adjoint.

On the other hand, if $\vU\Sigma = \Sigma\vU^{\dagger}$, then conjugating by $\vU^{\dagger}[\cdot]\vU$, we also get $\Sigma \vU = \vU^{\dagger}\Sigma$, so
	\begin{align}
		 \Sigma^2\vU = \Sigma \vU^{\dagger} \Sigma = \vU \Sigma^2 ,
	\end{align}
	meaning that $\Sigma^2$ commutes with $\vU$. Since $\Sigma\succeq 0$ it implies that $\vU$ also commutes with $\Sigma$,  i.e.,
	\begin{align}\label{eq:USigmaCommute}
		\Sigma\vU = \vU\Sigma.
	\end{align}	
	This allows us to block-diagonalize $\vU = \vU_{+}\oplus \vU_0$ by the positive singular value subspace and the $0$ singular value subspace of $\Sigma$. Therefore,
	\begin{align}
		\Sigma\vU^{\dagger}_+  =  \Sigma\vU^{\dagger} = \vU\Sigma = \Sigma\vU=\Sigma\vU_+,
	\end{align}
	which implies $\vU_+=\vU_+^{\dagger}$. Since \eqref{eq:UDef} imposes no constraint on $\vU_0$, we may choose $\vU_0 = \vI$ without changing the value of $\vK$. Therefore, we can assume without loss of generality that $\vU$ is Hermitian and commutes with $\Sigma$. 
\end{proof}

\disKraus*	
\begin{proof}
		First, observe that 
	$\vI=\CQ^\dagger[\vI]=\CT'^\dagger[\vI]+\vK^\dagger\vK$, thus $\vK$ must satisfy
	\begin{equation}\label{eq:RConstr}
		\vK^\dagger\vK=  \vI - \CT'^\dagger[\vI].
	\end{equation}
	Since $\CT'$ is trace non-increasing we have $\CT'^\dagger[\vI]\preceq \vI$, thus \eqref{eq:RConstr} can be satisfied by $\sqrt{\vI-\CT'^\dagger[\vI]}$.
	
	Since $\CT'$ is $\vrho$-detailed balanced, it is easy to see that $\CQ$ is $\vrho$-detailed balanced iff $\vK[\cdot]\vK^\dagger$ is $\vrho$-detailed balanced, therefore by \Cref{lem:KrausBalance} we get that $\vK=\e^{\ri\varphi}\vW\vU \vV^\dagger\sqrt{\vI-\CT'^\dagger[\vI]}$.
\end{proof}

\subsection{Power series expansion of the Kraus operator for decay}\label{apx:DiscreteKrausPower}
In this section, we evaluate the abstract prescription of~\autoref{lem:FindingDiscDecayTerm}. Suppose $\CT'^\dagger[\vI]\prec \vI$, we set $\vK'=\sqrt{\vI-\CT'^\dagger[\vI]} \succ 0$ and consider the singular value decomposition $\vK'\sqrt{\vrho}=\vV\Sigma \vW^\dagger$. We can then write $\vW\vV^\dagger=\vW\Sigma\vW^\dagger\vW \Sigma^{-1} \vV^\dagger=\sqrt{\sqrt{\vrho}\vK'^{\dagger}\vK' \sqrt{\vrho}}\vrho^{-\frac12}\vK'^{-1}$.
Therefore, we may choose $\varphi =0$ and $\vU = \vI$ in \autoref{lem:FindingDiscDecayTerm} to obtain a formula for $\vK$ in terms of $\vrho, \vK'$
\begin{align}\label{eq:unitaryCorrection}
	\vK = \vW\vV^\dagger\vK' = \sqrt{\sqrt{\vrho}\vK'^{\dagger}\vK' \sqrt{\vrho}} \vrho^{-\frac12} = \sqrt{\sqrt{\vrho}(\vI - \CT'^\dagger[\vI]) \sqrt{\vrho}} \vrho^{-\frac12}.
\end{align}
%

Note that the Gibbs state $\vrho$ generally has an exponentially large condition number $\e^{\Omega(\norm{\vH})}$ and the above expression may appear ill-conditioned. However, the powers of $\vrho$ ``sum to zero'' ($\frac14+\frac14-\frac12 =0$) and we may expect some cancellation. To expose the cancellation, we exploit a matrix-valued series expansion of the square root function \cite{delMoral2018TaylorSquareRootMat} to describe a convergent Taylor series for $\vK$ which can be written as short-time Heisenberg evolutions according to the Hamiltonian $\vH$. For our purposes, we will take $\vO\leftarrow \CT'^\dagger[\vI].$ Remarkably, the expansion is written in terms of the same $\CS$ that has appeared earlier, signaling the central role of the particular superoperator $\CS$ to quantum detailed balance. Further, through the integral formula for $\CS$ over short-time Heisenberg evolution $\vO(t)$, we expose the quasi-locality of the operator $\vK$.
\disKrausTaylor*

While the matrix-valued expansion may appear unwieldy, we give some intuition in the scalar case when $d=1$, $\vrho=1$, and $\vO=x\in[-1,1]$ is a number. Then, $\CS^\pm$ is simply the multiplication by the number $\frac12$, and therefore \eqref{eq:optimisticSeries} simply reduces to the Taylor series 
\begin{align}\label{eq:sqrt(1-x)}
    \sqrt{1-x}=\sum_{k=0}^\infty(-x)^k\binom{\frac12}{k} =:\sum_{k=0}^\infty -a_k x^k,
\end{align}
where $a_0 =-1$ and $a_1,a_2,\dots >0$ with $a_k = (-1)^{k+1}\binom{1/2}{k}$. To see this, it suffices to note that $(\CS^{\pm}\circ \nabla^{k})[x]=a_k x^k$. The $k=1$ case is trivial; for $k\geq 2$ the correctness of the recursive formula \eqref{eq:DDef} for $\vO=x$ can be proven by induction as follows
	\begin{align}\label{eq:Chu-Vandermonde}
		2a_k x^k=\sum_{p=1}^{k-1} (a_p x^p) (a_{k-p} x^{k-p})\quad \forall k\geq2 \qquad \Longleftrightarrow \qquad 0 = \sum_{p=0}^{k} a_p a_{k-p},
	\end{align}
	by comparing\footnote{In fact \eqref{eq:Chu-Vandermonde} is a special case of the so-called Chu-Vandermonde identity.} the $k$-th coefficient of the left and right-hand sides of 
	\begin{align*}
		1-x&=\sqrt{1-x}\sqrt{1-x}
		=\left(\sum_{p=0}^\infty a_p x^p\right)\left(\sum_{q=0}^\infty a_q x^q\right).
	\end{align*}
In order to generalize to the matrix-valued case (\autoref{lem:disTaylor}), we first establish that the series~\eqref{eq:optimisticSeries} is convergent whenever the input matrix is bounded by $\|\vO\|\leq \frac{1}{4\|\CS^-\|_{\infty\to\infty}^2}$. Since we will later need to truncate the operators $\CS$ for proving quasi-locality, we prove the following norm bound for general $\widetilde{\CS}^{\pm}$ superoperators.

\begin{lemma}[A recursive norm bound]\label{lem:recursive_norm}
        Let $\widetilde{\CS}^{\pm}$ be any superoperators, $\widetilde{\nabla}^{(1)}:=\CI$, and recursively define a family of homogeneous degree-$k$ maps for $k\geq 2$ as follows
		\begin{align}\label{eq:DDefApprox}
			\widetilde{\nabla}^{(k)}[\vM]:= \sum_{p=1}^{k-1} (\widetilde{\CS}^{+}\circ \widetilde{\nabla}^{(p)})[\vM] \cdot (\widetilde{\CS}^{-}\circ \widetilde{\nabla}^{(k-p)})[\vM] \quad \text{for all matrix } \vM.
		\end{align}	
		If $\max(\|\widetilde{\CS}^{+}\|_{\infty\to\infty},\|\widetilde{\CS}^{-}\|_{\infty\to\infty}) \leq s$ and $\nrm{\vM} \leq \frac{x}{4s^2}$ for some $s,x > 0$, then
		\begin{align}\label{eq:MatSQCoeffBound}
			\|\widetilde{\nabla}^{(k)}[\vM]\|\leq\frac{a_k}{2s^2}x^k \quad \text{for every}~ k\ge 1,
		\end{align}
  where $a_1,a_2,\dots > 0$ are the coefficients of the Taylor expansion of $\sqrt{1-x}$ as in~\eqref{eq:sqrt(1-x)}.
	\end{lemma}
	\begin{proof}
		For $k=1$, \eqref{eq:MatSQCoeffBound} holds due to the assumption $\nrm{\vM}\leq \frac{x}{4s^2}$.
		We prove the statement by induction for $k>1$,
		\begin{align}
			\|\widetilde{\nabla}^{(k)}[\vM]\| &= \nrm{\sum_{p=1}^{k-1} (\widetilde{\CS}^{+}\circ \widetilde{\nabla}^{(p)})[\vM] \cdot (\widetilde{\CS}^{-}\circ \widetilde{\nabla}^{(k-p)})[\vM]} \tag{by \eqref{eq:DDefApprox}}\\&
			\leq\sum_{p=1}^{k-1} (s\|\widetilde{\nabla}^{(p)}[\vM]\|) (s\|\widetilde{\nabla}^{(k-p)}[\vM]\|) \tag{by triangle inequality}\\&		
			\le \sum_{p=1}^{k-1}\frac{a_p}{2s}x^p\cdot \frac{a_{k-p}}{2s}x^{k-p}\tag{by inductive hypothesis and $a_1,\dots, a_{k-1} > 0$}\\&			
			= \frac{a_k}{2s^2}x^k. \tag*{(by \eqref{eq:Chu-Vandermonde})\quad\qedhere}
		\end{align}
	\end{proof}
Equipped with the norm bound, we can now prove the desired matrix-valued Taylor expansion \Cref{lem:disTaylor}. 
\begin{proof}[Proof of~{\rm \autoref{lem:disTaylor}}]	
By~\Cref{lem:recursive_norm}, the series \eqref{eq:optimisticSeries} is convergent whenever $\|\vO\|\leq \frac{1}{4\|\CS^-\|_{\infty\to\infty}^2}$, because 
	\begin{align}\label{eq:SqrtTaylorCoeffSum}
		\sum_{k=1}^\infty |a_k|=\sum_{k=1}^\infty a_k 1^k=1-\sqrt{1-1}=1,
	\end{align}
    and 
    \begin{align}
        \|\CS^{+}\|_{\infty\to\infty}=\|\CS^{-}\|_{\infty\to\infty} \quad \text{since}\quad \CS^{+}[\vM]=(\CS^{-}[\vM^\dagger])^\dagger \qquad \text{for each}~ \vM\in\BC^{d\times d} \label{eq:SpmSymmetry}. 
    \end{align}
From now on we assume that $\|\vO\| < \frac{1}{4\|\CS^-\|_{\infty\to\infty}^2}$, and denote by $P[\vO]$ the series \eqref{eq:optimisticSeries}. The next step in our proof is to show that $P[\vO]$ indeed behaves like a square root
	\begin{align*}
		(P[\vO])^\dagger P[\vO]&
		=\left(\vI - \sum_{p=1}^{\infty} (\CS^{+}\circ \nabla^{(p)})[\vO]\right)\left(\vI - \sum_{q=1}^{\infty} (\CS^{-}\circ \nabla^{(q)})[\vO]\right),\tag{by \eqref{eq:SpmSymmetry}}\\&
		=\vI - \sum_{p=1}^{\infty} (\underset{=\CI}{\underbrace{(\CS^{+}+\CS^{-})}}\circ \nabla^{(p)})[\vO]
		+\left(\sum_{p=1}^{\infty} (\CS^{+}\circ \nabla^{(p)})[\vO]\right)\left(\sum_{q=1}^{\infty} (\CS^{-}\circ \nabla^{(q)})[\vO]\right)\\&
		=\vI-\vO ,\tag{by \eqref{eq:DDef}}
	\end{align*}
	where the first equality uses Hermiticity for
 \begin{align}
  \nabla^{(k)}[\vO] =    \L(\nabla^{(k)}[\vO]\R)^{\dagger}\quad \text{and}\quad \left((\CS^{-}\circ \nabla^{(q)})[\vO]\right)^\dagger = (\CS^{+}\circ \nabla^{(q)})[\vO] \quad \text{for each Hermitian}~ \vO
 \end{align}
due to the symmetry \eqref{eq:SpmSymmetry}. The last equality follows by comparing the $k$-th order terms and using the uniqueness of the power series.

Finally, note that in the proof of \Cref{thm:CoherentRule} we have already shown that $\CS^-=\CS^{[\sqrt{\gamma'_G}]}$, where $\gamma'_G(\nu):=(\frac12-\frac12\tanh(\frac{\nu}{4}))^2$, maps Hermitian matrices to ones that satisfy detailed balance~\eqref{eq:KrausBalanced} (with $\varphi=0$). This implies that $P[\vO]$ always produces an operator $\vK''$ that satisfies~\eqref{eq:KrausBalanced} (with $\varphi=0$). Since we also have that $(P[\vO])^\dagger P[\vO]= \vI - \vO$, by following the argument of \Cref{lem:KrausBalance} we can see that $P[\vO]=\vW\vU \vV^\dagger\sqrt{\vI-\vO}$, where $\sqrt{\vI-\vO}\sqrt{\vrho}=\vV\Sigma \vW^\dagger$ is a singular value decomposition, and $\vU$ is some Hermitian unitary that commutes with $\Sigma$. To argue that $\vU=\vI$, we use a continuity argument. Indeed, $\vU=\vI$ at the particular value $\vO=0$. Further, the identity is isolated in the set of Hermitian unitary since for any $\vU\neq \vI$ we must have that $\nrm{\vU-\vI}=2$ as the spectrum takes values as $\pm1$. Therefore $\|\vW(\vU- \vI) \vV^\dagger\sqrt{\vI-\vO} \|>0$ since $\sqrt{\vI-\vO}\succ 0$ (assuming that $\|\vO\| < \frac{1}{4\|\CS^-\|_{\infty\to\infty}^2}\leq 1$). The continuity of $P[\vO]$ thus implies that $\vU=\vI$ whenever $\|\vO\| < \frac{1}{4\|\CS^-\|_{\infty\to\infty}^2}$, i.e., $P[\vO]=\vW\vV^\dagger\sqrt{\vI-\vO}$, and therefore~\eqref{eq:optimisticSeries} holds due to~\eqref{eq:unitaryCorrection}. Once again by continuity this implies the statement for $\|\vO\| = \frac{1}{4\|\CS^-\|_{\infty\to\infty}^2}$ as well.
\end{proof}

\subsection{Analyzing the quasi-locality of the discrete quantum channel}\label{apx:QuasiLocalDiscrete}
In this section we further expose the quasi-locality of $\vK$ using time-domain representation of $\CS^{\pm}$:
\begin{align}\label{eq:InRepApx}
\CS^{\pm}[\vO]&= \frac{\vO}{2} \pm \lim_{\theta\rightarrow 0+}\int_{\BR\setminus[-\theta,\theta]} \frac{\ri}{\sinh(2\pi t)}\e^{\ri \vH t} \vO \e^{-\ri \vH t} \rd t.
\end{align}
Thus, $\vK$ remains quasi-local as long as (i) $\CT'^\dagger[\vI]$ is quasi-local to begin with, and (ii) its operator norm is small enough. We now prove this rigorously using the truncated integral of \Cref{prop:TruncatedIntegral}.

In order to bound the accumulation of approximation errors we prove in the following result that the series expansion of $\vK$ can be truncated after $\mathcal{O}(\log(1/\varepsilon))$ terms and only incur in error~$\varepsilon$.%
\begin{lemma}[Approximation error bound]\label{lem:recursive_error_norm}
    Let $\eps \in (0,\frac12]$ and $s\geq \max(\|\CS^\pm\|_{\infty\to\infty},\|\widetilde{\CS}^\pm\|_{\infty\to\infty})$.
	If $\|\vO\|,\|\widetilde{\vO}\|\leq \frac{1}{16s^2}$, $\|\vO-\widetilde{\vO}\|\leq\frac{\eps}{16s^2}$, $\|\CS^\pm-\widetilde{\CS}^\pm\|_{\infty\to\infty}\leq\eps s$, then, using the notation of {\rm \Cref{lem:recursive_norm}},
	\begin{align}\label{eq:ApxErrorBound}
		\nrm{\vI - \sum_{k=1}^{\lfloor\log_2(1/\eps)\rfloor} (\widetilde{\CS}^{-}\circ \widetilde{\nabla}^{(k)})[\widetilde{\vO}] - \sqrt{\vrho-\sqrt{\vrho}\vO \sqrt{\vrho}}\vrho^{-\frac12}}\leq \eps.
	\end{align}	
\end{lemma}
\begin{proof}
	First we observe that the left-hand side of \eqref{eq:ApxErrorBound} can be written as
    \begin{align}\label{eq:ApxErrorBound_alt}
        \nrm{\sum_{k=\lceil\log_2(1/\eps)\rceil}^{\infty} (\CS^{-}\circ \nabla^{(k)})[\vO] + \sum_{k=1}^{\lfloor\log_2(1/\eps)\rfloor}  \left((\CS^{-}\circ \nabla^{(k)})[\vO] - (\widetilde{\CS}^{-}\circ \widetilde{\nabla}^{(k)})[\widetilde{\vO}]\right)}.
    \end{align}
    The first sum inside the norm of \eqref{eq:ApxErrorBound_alt} can be bounded as follows,
    \begin{align}
        \nrm{\sum_{k=\lceil\log_2(1/\eps)\rceil}^{\infty} (\CS^{-}\circ \nabla^{(k)})[\vO]}
        \leq\sum_{k=\lceil\log_2(1/\eps)\rceil}^{\infty}a_k\frac{\nrm{\vO}^k}{2\|\CS^-\|_{\infty\to\infty}}        \leq\nrm{\vO}^{\lceil\log_2(1/\eps)\rceil}\sum_{k=0}^{\infty}\frac{\nrm{\vO}^k}{8} \leq \frac{\eps}{2},
        \label{eq:truncApxBound}
    \end{align}
 since $0\le a_k\leq\frac{1}{8}$ for $k\geq2$, $\|\CS^-\|_{\infty\to\infty}\geq \frac12$, and $\|\vO\| \leq \frac{1}{2}$.
	By \Cref{lem:recursive_norm} we have that
	\begin{align}\label{eq:MatApxSQCoeffBound}
		\max\left(\|\nabla^{(k)}[\vO]\|,\|\widetilde{\nabla}^{(k)}[\widetilde{\vO}]\|\right)
		\leq a_k\frac{4^{-k}}{2s^2}
        \leq\frac{a_k}{8s^2}\quad \text{for each}\quad k \ge 1.
	\end{align}
	Now we prove similarly by induction that
	\begin{align}\label{eq:MatApxSQCoeffDiffBound}
		\nrm{\widetilde{\nabla}^{(k)}[\widetilde{\vO}]-\nabla^{(k)}[\vO]}\leq\frac{a_k}{8s^2}\eps.
	\end{align}
	Since $\widetilde{\nabla}^{(1)}[\widetilde{\vO}]=\widetilde{\vO}$, the premise \eqref{eq:MatApxSQCoeffDiffBound} holds for $k=1$ due to $\|\vO-\widetilde{\vO}\|\leq\frac{\eps}{16s^2}$ and $a_1 = \frac{1}{2}$. For $k>1$, it follows as
	\begin{align}
		\nrm{\widetilde{\nabla}^{(k)}\![\widetilde{\vO}]\!-\!\nabla^{(k)}\![\vO]} 
		&\kern-0.4mm\leq\nrm{\sum_{p=1}^{k-1} (\CS^{+}\!\circ \nabla^{(p)})[\vO] \cdot (\CS^{-}\!\circ \nabla^{(k-p)})[\vO]-(\CS^{+}\!\circ \nabla^{(p)})[\vO] \cdot (\CS^{-}\!\circ \widetilde{\nabla}^{(k-p)})[\widetilde{\vO}]} \nonumber\\&
		+\nrm{\sum_{p=1}^{k-1} (\CS^{+}\!\circ \nabla^{(p)})[\vO] \cdot (\CS^{-}\!\circ \widetilde{\nabla}^{(k-p)})[\widetilde{\vO}]-(\CS^{+}\!\circ \widetilde{\nabla}^{(p)})[\widetilde{\vO}] \cdot (\CS^{-}\!\circ \widetilde{\nabla}^{(k-p)})[\widetilde{\vO}]} \nonumber\\&		
		+ \nrm{\sum_{p=1}^{k-1}(\CS^{+}\!\circ \widetilde{\nabla}^{(p)})[\widetilde{\vO}] \cdot (\CS^{-}\!\circ \widetilde{\nabla}^{(k-p)})[\widetilde{\vO}] - (\CS^{+}\!\circ \widetilde{\nabla}^{(p)})[\widetilde{\vO}] \cdot (\widetilde{\CS}^{-}\!\circ \widetilde{\nabla}^{(k-p)})[\widetilde{\vO}]} \nonumber\\&
		+ \nrm{\sum_{p=1}^{k-1}(\CS^{+}\!\circ \widetilde{\nabla}^{(p)})[\widetilde{\vO}] \cdot (\widetilde{\CS}^{-}\!\circ \widetilde{\nabla}^{(k-p)})[\widetilde{\vO}] - (\widetilde{\CS}^{+}\!\circ \widetilde{\nabla}^{(p)})[\widetilde{\vO}] \cdot (\widetilde{\CS}^{-}\!\circ \widetilde{\nabla}^{(k-p)})[\widetilde{\vO}]} \tag*{(by \eqref{eq:DDef}, \eqref{eq:DDefApprox}, and the triangle inequality)}\\&		
		\le 4\eps s^2 \sum_{p=1}^{k-1}\frac{a_p}{8s^2}\cdot \frac{a_{k-p}}{8s^2}\tag*{(by \eqref{eq:MatApxSQCoeffBound} and the inductive hypothesis)}\\&			
		= \frac{a_k}{8s^2}\eps. \tag*{(by \eqref{eq:Chu-Vandermonde})}
	\end{align}
        Finally, the second sum inside the norm of \eqref{eq:ApxErrorBound_alt} can be bounded as follows, which proves the claim by the triangle inequality.
	\begin{align*}
		\nrm{\sum_{k=1}^{\lfloor\log_2(1/\eps)\rfloor}  (\CS^{-}\circ \nabla^{(k)})[\vO] - (\widetilde{\CS}^{-}\circ \widetilde{\nabla}^{(k)})[\widetilde{\vO}]}
		&\leq \nrm{\sum_{k=1}^{\lfloor\log_2(1/\eps)\rfloor}  (\CS^{-}\circ \nabla^{(k)})[\vO] - (\CS^{-}\circ \widetilde{\nabla}^{(k)})[\widetilde{\vO}]}\\&
		+ \nrm{\sum_{k=1}^{\lfloor\log_2(1/\eps)\rfloor}  (\CS^{-}\circ \widetilde{\nabla}^{(k)})[\widetilde{\vO}] - (\widetilde{\CS}^{-}\circ \widetilde{\nabla}^{(k)})[\widetilde{\vO}]}\\&		
		\leq (s\eps + s\eps)\sum_{k=1}^{\lfloor\log_2(1/\eps)\rfloor}\frac{a_k}{8s^2}\tag{by \eqref{eq:MatApxSQCoeffBound}-\eqref{eq:MatApxSQCoeffDiffBound}}\\&
		\leq \frac{\eps}{2}.\tag*{(by \eqref{eq:SqrtTaylorCoeffSum} and $s\geq\frac12$)\quad\qedhere}
	\end{align*}	
\end{proof}

Now we are ready to derive our main result about approximate integral representation.
We can apply the following result, e.g., setting
\begin{align}
    s=1+\frac{2}{\pi}\ln\left(1+\frac{2\nrm{\vH}}{\eps \pi}\right), \quad I:=\left[-\frac{\ln(\frac{4}{\eps})}{2\pi},\frac{\ln(\frac{4}{\eps})}{2\pi}\right],  \quad\text{and}\quad f^{\pm}_1(t) = \frac{\delta(t)}{2s} \pm \frac{\ri\cdot\pmb{1}_{\Delta_{\eps}}(t)}{s\cdot\sinh(2\pi t)}\label{eq:localitiyParams}
\end{align}
as show by \Cref{prop:TruncatedIntegral} and \eqref{eq:sinhcInt}. We could similarly get an alternative $f^\pm_1$ from \Cref{prop:SmoothIntegral} which could result in better discretization error bounds to it smoothness.
In particular we can also use a discretized version of the above functions meaning that they are replaced by a linear combination of Dirac delta distributions following some discretization method for the integral corresponding to $f^{\pm}_1(t)$.
\begin{corollary}[Approximate integral representation]\label{cor:recursive_int_rep}
    Given distributions $f^{\pm}_1$, let us recursively define for $k\geq 2$ the distributions\footnote{In fact we will only consider when $f^\pm_1$ is an integrable function plus a finite linear combination of Dirac deltas $\delta(t)$. Thus, we could also work with an approximate representation of $\delta(t)$ such as a narrow, truncated Gaussian. This way we can work with an integrable functions throughout the proof and recover the result as a limit.} $f^{\pm}_k$ over $\BR^k$ as
	\begin{align}\label{eq:inductiveConvolution}
		f^{\pm}_k(t_1,\dots,t_k)&:=\sum_{p=1}^{k-1} \int f^{\pm}_1(t)\cdot  f^{+}_p(t_1-t,\ldots,t_p-t)\cdot f^{-}_{k-p}(t_{p+1}-t,\ldots,t_{k}-t)\rd t.
	\end{align}
    If $\lVert f^{\pm}_1\rVert_1 \leq a_1 := \frac12$, then $\|f^\pm_k\|_1\leq a_k$ (with $a_k := (-1)^{k+1}\binom{1/2}{k} > 0$ defined in \eqref{eq:sqrt(1-x)}).
 Further, given $\eps \in (0,\frac12]$ and $s\geq 1+\frac{1}{\pi}\ln\big(1+\frac{\nrm{\vH}}{2}\big)$, 
	we have
	\begin{align*}
		\nrm{ \sqrt{\vrho- \!\sqrt{\vrho}\vO \sqrt{\vrho}}\vrho^{-\frac12} - \vI + \frac{1}{s} \!\sum_{k=1}^{\lfloor\log_2(1/\eps)\rfloor} \!\! \iiint \! f_k^-(t_1,\ldots,t_k)(s^2\widetilde{\vO}(t_1))\cdot\ldots\cdot(s^2\widetilde{\vO}(t_k))\rd t_1\ldots\rd t_k} \leq \eps
	\end{align*}	
 whenever $\big\|\CS^\pm[\cdot] - s\int_{\BR} f^{\pm}_1(t)\e^{\ri \vH t}[ \cdot ]\e^{-\ri \vH t} \rd t\big\|_{\infty\to\infty}\leq\eps s$, $\|\vO\|,\|\widetilde{\vO}\|\leq \frac{1}{16s^2}$, and $\|\vO-\widetilde{\vO}\|\leq\frac{\eps}{16s^2}$. Additionally, if the support of $f^{\pm}_1$ is contained in the symmetric interval $I=-I\subseteq \BR$, then the support of $f^\pm_k$ is contained in $(kI)^k\subseteq \BR^k$.
\end{corollary}
\begin{proof}
	Let us define $\widetilde{\CS}^\pm[\cdot]:=s\int_{\BR} f^{\pm}_1(t)\e^{\ri \vH t}[ \cdot ]\e^{-\ri \vH t} \rd t$ so that $\|\CS^\pm-\widetilde{\CS}^\pm\|_{\infty\to\infty}\leq\eps s$. Then by \Cref{lem:recursive_error_norm} and the assumption that $\eps \in (0,\frac12]$ and $s\geq1+\frac{1}{\pi}\ln\big(1+\frac{\nrm{\vH}}{2}\big)$, we have that
	\begin{align}\label{eq:ApxErrorBoundCor}
		\nrm{\vI - \sum_{k=1}^{\lfloor\log_2(1/\eps)\rfloor} (\widetilde{\CS}^{-}\circ \widetilde{\nabla}^{(k)})[\widetilde{\vO}] - \sqrt{\vrho-\sqrt{\vrho}\vO \sqrt{\vrho}}\vrho^{-\frac12}}\leq \eps.
	\end{align}	
Defining $\widetilde{\nabla}_s^{(k)}[\vM]:=s^2\widetilde{\nabla}^{(k)}[\vM]$, by \eqref{eq:DDefApprox} we have for all matrix  $\vM$ that
\begin{align}
	\widetilde{\nabla}^{(k)}_s[\vM]&= s^2\widetilde{\nabla}^{(k)}[\vM]\tag{by definition of $\widetilde{\nabla}_s^{(k)}[\vM]$}\\&
	=s^2\sum_{p=1}^{k-1} (\widetilde{\CS}^{+}\circ \widetilde{\nabla}^{(p)})[\vM] \cdot (\widetilde{\CS}^{-}\circ \widetilde{\nabla}^{(k-p)})[\vM]\tag{by \eqref{eq:DDefApprox}}\\&
	=\sum_{p=1}^{k-1} \left(\frac{\widetilde{\CS}^{+}}{s}\circ \widetilde{\nabla}_s^{(p)}\right)[\vM] \cdot \left(\frac{\widetilde{\CS}^{-}}{s}\circ \widetilde{\nabla}_s^{(k-p)}\right)[\vM].\label{eq:recRescaled}
\end{align}	
Using the above definitions we can rewrite \eqref{eq:ApxErrorBoundCor} as follows
	\begin{align*}
		\nrm{\vI - \frac{1}{s}\sum_{k=1}^{\lfloor\log_2(1/\eps)\rfloor} \left(\frac{\widetilde{\CS}^{\pm}}{s}\circ \widetilde{\nabla}_s^{(k)}\right)[\widetilde{\vO}] - \sqrt{\vrho-\sqrt{\vrho}\vO \sqrt{\vrho}}\vrho^{-\frac12}}\leq \eps.
	\end{align*}	

We now prove by induction that, for $k\geq 1$,
\begin{align*}
	\left(\frac{\widetilde{\CS}^{\pm}}{s}\circ \widetilde{\nabla}_s^{(k)}\right)[\widetilde{\vO}] = \iiint f^{\pm}_k(t_1,\ldots,t_k)(s^2\widetilde{\vO}(t_1))\cdot\ldots\cdot(s^2\widetilde{\vO}(t_k)) \rd t_1\ldots\rd t_k,
\end{align*}
where $f_k$ is supported on $(kI)^k$ and $\nrm{f_k}_1\leq a_k$. In this proof, all integrals are over the reals $\BR$, i.e., $\int = \int_{-\infty}^{\infty}$, and $\iiint$ denotes integration over several real variables. The case $k=1$ follows from definition. For $k\geq 2$ we calculate
\begin{align*}
    \widetilde{\nabla}_s^{(k)}[\widetilde{\vO}] &
	=\sum_{p=1}^{k-1} \left(\frac{\widetilde{\CS}^{+}}{s}\circ \widetilde{\nabla}_s^{(p)}\right)[\widetilde{\vO}] \cdot \left(\frac{\widetilde{\CS}^{-}}{s}\circ \widetilde{\nabla}_s^{(k-p)}\right)[\widetilde{\vO}] \tag{by \eqref{eq:recRescaled}}\\&
	= \sum_{p=1}^{k-1} \left(\iiint f^{+}_p(t_1,\ldots,t_p)(s^2\widetilde{\vO}(t_1))\cdot\ldots\cdot(s^2\widetilde{\vO}(t_p)) \rd t_1\ldots\rd t_p\right)\\& 
	\kern5mm\cdot \left(\iiint f^{-}_{k-p}(t_{p+1},\ldots,t_{k})(s^2\widetilde{\vO}(t_{p+1}))\cdot\ldots\cdot(s^2\widetilde{\vO}(t_k)) \rd t_{p+1}\ldots\rd t_k\right)\tag{by induction}\\
 &= \sum_{p=1}^{k-1} \iiint f^{+}_p(t_1,\ldots,t_p) f^{-}_{k-p}(t_{p+1},\ldots,t_{k}) \cdot (s^2\widetilde{\vO}(t_{1}))\cdot\ldots\cdot(s^2\widetilde{\vO}(t_k)) \rd t_{1}\ldots\rd t_k. 
\end{align*}

Therefore, 
\begin{align*}
	\left(\frac{\widetilde{\CS}^{\pm}}{s}\circ \widetilde{\nabla}_s^{(k)}\right)[\widetilde{\vO}]
 &= \int f^{\pm}_1(t)\e^{\ri \vH t} \widetilde{\nabla}_s^{(k)}[\widetilde{\vO}] \e^{-\ri \vH t} \rd t \tag{by definition of $\widetilde{\CS}^{\pm}$}\\
	&=\int f^{\pm}_1(t) \sum_{p=1}^{k-1} \left(\iiint f^{+}_p(t_1,\ldots,t_p)\right.\tag{telescopic product by $\e^{\ri \vH t}\e^{-\ri \vH t}$}\\&  \kern15mm\cdot\left.f^{-}_{k-p}(t_{p+1},\ldots,t_{k})(s^2\widetilde{\vO}(t+t_1))\cdot\ldots\cdot(s^2\widetilde{\vO}(t+t_k)) \rd t_1\ldots\rd t_k\phantom{\int}\!\!\!\!\right) \rd t \nonumber\\&
	=\int f^{\pm}_1(t) \sum_{p=1}^{k-1} \left(\iiint f^{+}_p(t_1-t,\ldots,t_p-t)\right.\tag{change of variables $t_i \rightarrow t_i-t$}\\&  \kern15mm\cdot\left.f^{-}_{k-p}(t_{p+1}-t,\ldots,t_{k}-t)(s^2\widetilde{\vO}(t_1))\cdot\ldots\cdot(s^2\widetilde{\vO}(t_k)) \rd t_1\ldots\rd t_k\phantom{\int}\!\!\!\!\right) \rd t \nonumber\\&
	=\iiint \left(\sum_{p=1}^{k-1} \int f^{\pm}_1(t) f^{+}_p(t_1-t,\ldots,t_p-t)\cdot f^{-}_{k-p}(t_{p+1}-t,\ldots,t_{k}-t)\rd t\right) \nonumber\\&  \kern15mm\cdot (s^2\widetilde{\vO}(t_1))\cdot\ldots\cdot(s^2\widetilde{\vO}(t_k)) \rd t_1\ldots\rd t_k.  \tag{by Fubini's theorem}
\end{align*}	
The convolution above is the source of the recursive definition \eqref{eq:inductiveConvolution},
leading to the norm bound
\begin{align*}
	\|f^{\pm}_k\|_1&\leq\sum_{p=1}^{k-1} \|f^{\pm}_1\|_1\cdot \|f^{+}_p\|_1\cdot \|f^{-}_{k-p}\|_1\\&
	\leq \sum_{p=1}^{k-1} \frac12 a_p\cdot a_{k-p}\tag{by induction}\\&
	= a_k. \tag{by \eqref{eq:Chu-Vandermonde}}
\end{align*}
By the inductive assumption, we also have that $f^{\pm}_p$ is supported on $(pI)^p$ for each $1\le p\le k-1$, and thus by \eqref{eq:inductiveConvolution} it follows that $f^{\pm}_k$ is supported on $(kI)^k$.%
\end{proof}

    Using \Cref{cor:recursive_int_rep} in combination with \Cref{prop:TruncatedIntegral} and \eqref{eq:localitiyParams} we get that the integrals in the time domain are supported on a box of size $\CO(\log^2(1/\eps))$. As a consequence of the Lieb-Robinson bound (\Cref{thm:lieb_robinson}) this implies that up to a global $\eps$-error we can treat each $\widetilde{\vO}(t)$ as $\CO(\beta\log^2(\nrm{\vH}/\eps))$-local time-evolution (using \Cref{lem:QuasiLocalApxS} recursively for each variable in the integrals and exploiting that $\sum_{k=1}^\infty \|f^{-}_k\|_1 \leq \sum_{k=1}^\infty a_k = 1$ for $f^{-}_k$ in \Cref{cor:recursive_int_rep}) as long as $\widetilde{\vO}$ is local. Here $\beta$ is a parameter from \Cref{lem:QuasiLocalApxS} that is intuitively the same as inverse temperature.
	We have at most $\log_2(1/\eps)$ such terms multiplied, meaning the overall integrand is also about $\CO(\beta\log^3(1/\eps))$ quasi-local up to $\eps$ approximation error. 
	Inverting this error-locality dependence we get spacial decay of the form $\e^{- \Omega(\sqrt[3]{\mathrm{distance}})}$. We suspect however, that this can be further strengthened by more carefully bounding the quasi-locality of intermediate terms.
	
	On top of quasi-locality, by discretising the integrals above, one gets an LCU-based implementation. Another consequence of \Cref{lem:recursive_error_norm} is that if $\vO$ is energy $\delta$-local, then up to error $\eps$ we also have that $\sqrt{\vrho-\sqrt{\vrho}\vO \sqrt{\vrho}}\vrho^{-\frac12}$ is energy $\delta\log(1/\eps)$-local.
Moreover, in the operator Fourier transform variant we can explicitly control energy locality via the variance $\sigma$. This makes it possible to use the following implementation:
\begin{lemma}[QSVT-based implementation of the discrete reject term]\label{lem:QSVTShortcut}
	Let $\eps \in (0,\frac14]$, $s\geq\|\CS^-\|_{\infty\to\infty}$, and $\alpha \geq \|\vH\|$.
	If $\nrm{\vO}\leq \frac{10^{-3}}{s^2}$ and $\vO$ is quasi-local in energy in the sense that there is an $\widetilde{\vO}$ such that $\|\vO-\widetilde{\vO}\| = \bigO{\frac{\eps}{s^2\poly\log(1/\eps)}}$ and $\|\widetilde{\vO}_\nu\|=0$ for $|\nu| \geq \frac{1}{10\lfloor\log_2(1/\eps)\rfloor}$, 
 then we can implement an $\eps$-approximate block-encoding of  
	\begin{align}\label{eq:desiredQSVTOutput}
		\frac{1}{2}\sqrt{\sqrt{\vrho}(\vI - \vO) \sqrt{\vrho}} \vrho^{-\frac12},
	\end{align}
	with $\bigOt{\alpha}$ uses of a block-encoding of $\vH/\alpha$ and $\bigOt{1}$ additional ancillary qubits.
\end{lemma}
\begin{proof}
    The strategy is the following. First assume strict energy locality and make $\vO$ block-diagonal by periodically filtering out energies with width $\delta$. Choose block-size such that $\vrho^{-\frac14}\sqrt{\vrho - \sqrt{\vrho}\vO\sqrt{\vrho}}\vrho^{-\frac14}$ can be efficiently implemented by QSVT directly on each block, and do the implementation in superposition on each block.
    Do this for a shifted mesh and combine the parts where the approximation is accurate (combining 2 parts suffices).
    
    Let us define for convenience $f(\vO):=\sqrt{\vrho-\sqrt{\vrho}\vO \sqrt{\vrho}}\vrho^{-\frac12}$.
    By \Cref{lem:recursive_norm,lem:recursive_error_norm} we have
    \begin{align}
        \|f(\vO) - f(\widetilde{\vO})\|\leq \frac\eps6.
    \end{align}	
    First we assume that we have a block-encoding of $\widetilde{\vO}$, and later we show that working instead with a block-encoding of $\vO$ does not introduce much error. Once again by \Cref{lem:recursive_error_norm}
    \begin{align}
        \nrm{\vI - \sum_{k=1}^{K} (\CS^{-}\circ \nabla^{(k)})[\widetilde{\vO}] - f(\widetilde{\vO})}\leq \frac\eps6,
    \end{align}	
    where $K:=\lfloor\log_2(1/\eps)\rfloor$. Let us define for $r\in[0,5K]$ the following periodic energy projectors
    \begin{align}
        \vPi_1^{(r)}:=\sum_{i\in[d] : |\{E_i\}|< \frac{r}{10K}}\ketbra{\psi_i}{\psi_i} \qquad\text{and}\qquad
        \vPi_2^{(r)}:=\sum_{i\in[d] : |\{E_i-\frac{1}{2}\}|\leq \frac{r}{10K}}\ketbra{\psi_i}{\psi_i},
    \end{align}
    where $\{x\}\in[-\frac12,\frac12)$ denotes the fractional part of the real number $x$. 
    Observe that
    \begin{align}
        \vPi_1^{(r)}+\vPi_2^{(5K-r)}=\vI \quad \text{if} \quad r\in[0,5K].
    \end{align}
    A key observation (see \Cref{lem:key_observation}) is that due to the energy quasi-locality we have, for all $k\leq K$ and $s\in [r+k,5K]$,%
    \begin{subequations}\label{eq:key_observation}\begin{align}
        (\CS^{-}\circ \nabla^{(k)})[\widetilde{\vO}]\vPi_1^{(r)}
        &=(\CS^{-}\circ \nabla^{(k)})[\vP_1\widetilde{\vO}\vP_1]\vPi_1^{(r)} \quad\text{whenever}\quad \vPi_1^{(s)}\preceq\vP_1\preceq \vI, \\
        (\CS^{-}\circ \nabla^{(k)})[\widetilde{\vO}]\vPi_2^{(r)}
        &=(\CS^{-}\circ \nabla^{(k)})[\vP_2\widetilde{\vO}\vP_2]\vPi_2^{(r)} \quad\text{whenever}\quad \vPi_2^{(s)}\preceq\vP_2\preceq \vI .
    \end{align}
    \end{subequations}
    The significance of \eqref{eq:key_observation} is that for all $\vP$ satisfying $\vPi_1^{(2K)}\preceq\vP\preceq \vPi_1^{(3K)}$,
    \begin{align*}
        \vI - \sum_{k=1}^{K} (\CS^{-}\circ \nabla^{(k)})[\widetilde{\vO}]&=
        \left(\vI - \sum_{k=1}^{K} (\CS^{-}\circ \nabla^{(k)})[\widetilde{\vO}]\right)\vP
        +\left(\vI - \sum_{k=1}^{K} (\CS^{-}\circ \nabla^{(k)})[\widetilde{\vO}]\right)(\vI-\vP)\\&
        =\left(\vI - \sum_{k=1}^{K} (\CS^{-}\circ \nabla^{(k)})[\vP_1\widetilde{\vO}\vP_1]\right)\vP & \tag{where $\vPi_1^{(4K)}\preceq\vP_1\preceq \vI$}\\
        &~~~+\left(\vI - \sum_{k=1}^{K} (\CS^{-}\circ \nabla^{(k)})[\vP_2\widetilde{\vO}\vP_2]\right)(\vI-\vP) \tag{where $\vPi_2^{(4K)}\preceq\vP_2\preceq \vI$},
    \end{align*}
    and thus once again applying \Cref{lem:recursive_error_norm} we get
        \begin{align}\label{eq:ApxPiecewiseGlued}
        \nrm{f(\widetilde{\vO}) - f(\vP_1\widetilde{\vO}\vP_1)\vP
        -f(\vP_2\widetilde{\vO}\vP_2)(\vI-\vP)
        }\leq \frac\eps6.
    \end{align}	
    This is useful, because if we also have that $\vP_1\preceq \vPi_1^{(4.5K)}$ and $\vP_2\preceq \vPi_2^{(4.5K)}$, then the operators $\vP_1\widetilde{\vO}\vP_1$, $\vP_2\widetilde{\vO}\vP_2$ are block-diagonal in the energy-basis, with the blocks having energy diameter $\leq 0.9$. 
    Within such a block, $\vrho$ has condition number $\bigO{1}$, meaning that we can efficiently implement $f$ directly using QSVT.

    We now describe now the algorithm works. The first step is to implement $\vP$, $\vP_1$, $\vP_2$ using smooth-function techniques applied to the block-encoding of $\vH/\alpha$. More specifically, let $I_1 := \bigcup_{j\in[d]:|\{E_j\}| \leq \frac{9}{20}}[E_j - \delta, E_j + \delta]$ for $\delta < \frac{1}{20}$. Then, by constructing a $\bigO{\alpha\log\alpha\log(\alpha/\varepsilon)}$-degree polynomial that approximates $\alpha$ rectangular functions~\cite[Lemma~29 \& Theorem~68]{gilyen2018QSingValTransfArXiv}, it is possible to implement $\widetilde{\vP}$ such that $\|\widetilde{\vP} - \vP\|_{I_1} \leq \varepsilon$ by employing $\bigO{\alpha\log\alpha\log(\alpha/\varepsilon)}$ uses of a block-encoding of $\vH/\alpha$ and one ancilla~\cite[Theorem~31 \& Theorem~68]{gilyen2018QSingValTransfArXiv}. A similar reasoning applies to $\vP_1$ and $\vP_2$ ($\vP_2$ is approximated on an interval $I_2 := \bigcup_{j\in[d]:|\{E_j-\frac{1}{2}\}| \leq \frac{9}{20}}[E_j - \delta, E_j + \delta]$ for $\delta < \frac{1}{20}$ instead).

    The second step is to use the block-encoding of $\vH/\alpha$ to implement, for $i=1,2$, a block-encoding of $\vM_i^\pm \propto \vrho^{\pm \frac{1}{2}}$ within each block of $\vP_i$ such that, on each energy-block, $\vM_i^+ \vM_i^- = c\vI$ for some $c=\Theta(1)$. The idea is similar to the implementation of $\vP$, $\vP_1$, $\vP_2$. Since $\vrho \propto \exp(-\vH)$, we approximate the exponential function on the interval $I_i$ by a $\bigO{\alpha\log\alpha\log(\alpha/\varepsilon)}$-degree polynomial such that $\|\vM_i^\pm - \vrho^{\pm\frac{1}{2}}\|_{I_i} \leq \varepsilon$ with $\bigO{\alpha\log\alpha\log(\alpha/\varepsilon)}$ uses of a block-encoding of $\vH/\alpha$ and one ancilla~\cite[Corollary~64 \& Theorem~68]{gilyen2018QSingValTransfArXiv}. It is possible to employ the same polynomial approximation for $\vM_i^+$ and $\vM_i^-$ but map $\vH \mapsto -\vH$ beforehand for $\vM_i^-$. This allows to obtain $\vM_i^+ \vM_i^- = c\vI$ on each energy-block for some $c=\Theta(1)$.

    Given $\widetilde{\vP}_i$ and $\vM_i^\pm$, we implement next $\sqrt{\vM_i^{+}(\vI-\widetilde{\vP}_i\widetilde{\vO}\widetilde{\vP}_i)\vM_i^{+}}$ via block-matrix multiplication and QSVT~\cite[Corollary 3.4.14]{gilyen2018QSingValTransfThesis}. More precisely, because $\vM_i^{+}(\vI-\widetilde{\vP}_i\widetilde{\vO}\widetilde{\vP}_i)\vM_i^{+}$ for $i=1,2$ have constant condition number $\leq\frac{2}{c^2}$ on each energy-block, we can apply a $\bigO{\log(1/\eps)}$-degree polynomial to approximate the square root function via QSVT~\cite[Corollary 3.4.14]{gilyen2018QSingValTransfThesis}. This employs the block-encoding of $\widetilde{\vO}$ a number of $\bigO{\log(1/\eps)}$ times. Finally, by using LCU we output
    \begin{align}\label{eq:final_output}
        \frac{1}{4}\sqrt{\vM_1^{+}(\vI-\widetilde{\vP}_1\widetilde{\vO}\widetilde{\vP}_1)\vM_1^{+}}\vM_1^{-}\widetilde{\vP} + \frac{1}{4}\sqrt{\vM_2^{+}(\vI-\widetilde{\vP}_2\widetilde{\vO}\widetilde{\vP}_2)\vM_2^{+}}\vM_2^{-}(\vI-\widetilde{\vP}).
    \end{align}
    Since $\|\widetilde{\vP} - \vP\|_{I_1}, \|\widetilde{\vP}_i - \vP_i\|_{I_i}, \|\vM_i^\pm - \vrho^{\pm\frac{1}{2}}\|_{I_i} \leq \varepsilon$, the above is $\poly(\varepsilon)$ away from the desired output
    \begin{align*}
        \frac{1}{4}\sqrt{\sqrt{\vrho}(\vI-\vP_1\widetilde{\vO}\vP_1)\sqrt{\vrho}}\vrho^{-\frac{1}{2}}\vP+ \frac{1}{4}\sqrt{\sqrt{\vrho}(\vI-\vP_2\widetilde{\vO}\vP_2)\sqrt{\vrho}}\vrho^{-\frac{1}{2}}(\vI-\vP).
    \end{align*}
    Due to $\vM_i^{+}\vM_i^{-} = c\vI$ and~\eqref{eq:ApxPiecewiseGlued}, the final output \eqref{eq:final_output} gives a $\poly(\varepsilon)$-approximate implementation of $\frac{c}{4} f(\widetilde{\vO})$. By the robustness of QSVT~\cite{gilyen2018QSingValTransf}, replacing the block-encoding of $\widetilde{\vO}$ by that of $\vO$ we make at most $\CO(\|{\widetilde{\vO}-\vO}\|\log(1/\eps))$-large change in the output. Since all the complexities depends logarithmically on $\varepsilon$, we can redefine $\varepsilon$ to obtain a final $\varepsilon$-error approximation with $\bigOt{\alpha}$ uses of the block-encoding of $\vH/\alpha$ and $\bigOt{1}$ additional ancillae.
    Finally, we can amplify the block-encoding to achieve subnormalization $\frac12$ as in~\eqref{eq:desiredQSVTOutput} using linear singular value amplification~\cite[Theorem~30]{gilyen2018QSingValTransf}.
\end{proof}

\begin{lemma}\label{lem:key_observation}
    Let $K\in\mathbb{N}$ and $\vO$ be an operator such that $\vO_\nu = 0$ for $|\nu| \geq \frac{1}{2K}$. Define the periodic energy projector
    \begin{align*}
         \vPi^{(r)}:=\sum_{i\in[d] : |\{E_i\}|< \frac{r}{2K}}\ketbra{\psi_i}{\psi_i} \qquad\text{for}~r\in[0,K].
    \end{align*}
    Consider the superoperators $\CS^\pm$ and $\nabla^{(k)}$ from {\rm \Cref{lem:disTaylor}}. Then, for all $r\in[0,K]$ and $k\in[K]$,
    \begin{align*}
        (\CS^{\pm}\circ \nabla^{(k)})[{\vO}]\vPi^{(r)} = \vPi^{(r+k)}(\CS^{\pm}\circ \nabla^{(k)})[{\vO}]\vPi^{(r)} =  (\CS^{\pm}\circ \nabla^{(k)})[\vP{\vO}\vP]\vPi^{(r)}
    \end{align*}
    for all $\vP$ such that $\vPi^{(r+k)}\preceq \vP \preceq \vI$.
\end{lemma}
\begin{proof}
    First note that $\vO_\nu\vPi^{(r)} = (\vO\vPi^{(r)})_\nu$ and $(\vPi^{(r)}{\vO})_\nu = \vPi^{(r)}{\vO}_\nu$. Moreover,
    \begin{align*}
        \vO_\nu\vPi^{(r)} = \sum_{\substack{i,j\in[d] \\ E_i - E_j = \nu \\ |\{E_j\}| < \frac{r}{2K}}} \langle \psi_i|\vO|\psi_j\rangle|\psi_i\rangle\langle\psi_j| = \sum_{\substack{i,j\in[d] \\ E_i - E_j = \nu \\ |\{E_j\}| < \frac{r}{2K}, |\{E_i\}| < \frac{r+1}{2K}}} \langle \psi_i|{\vO}|\psi_j\rangle|\psi_i\rangle\langle\psi_j| = \vPi^{(r+1)}{\vO}_\nu\vPi^{(r)},
    \end{align*}
    where we used that ${\vO}_\nu=0$ for $|\nu| \geq \frac{1}{2K}$, so $E_i = \nu + E_j$ and $|\{E_j\}| \leq \frac{r}{2K}$ imply that $|\{E_i\}| \leq \frac{r+1}{2K}$. This means that
    \begin{align*}
        {\vO}_\nu\vPi^{(r)} = \vPi^{(r+1)}\vO_\nu\vPi^{(r)} \preceq \vPi^{(r+1)}(\vP\vO\vP)_\nu\vPi^{(r)} \preceq \vPi^{(r+1)}\vO_\nu\vPi^{(r)}
    \end{align*}
    and so $\vO_\nu\vPi^{(r)} = \vPi^{(r+1)}\vO_\nu\vPi^{(r)} = (\vP\vO\vP)_\nu\vPi^{(r)}$. Thus 
    \begin{align}\label{eq:key_observation_2}
        \CS^{\pm}[{\vO}\vPi^{(r)}] = \CS^{\pm}[\vPi^{(r+1)}\vO\vPi^{(r)}] = \vPi^{(r+1)}\CS^{\pm}[{\vO}]\vPi^{(r)} = \CS^{\pm}[\vP{\vO}\vP]\vPi^{(r)} \quad \forall r\in[0,K].
    \end{align}
    We start by proving the first equality $(\CS^{\pm}\circ \nabla^{(k)})[{\vO}]\vPi^{(r)} = \vPi^{(r+k)}(\CS^{\pm}\circ \nabla^{(k)})[{\vO}]\vPi^{(r)}$ via induction on $k$. The case $k=1$ is \eqref{eq:key_observation_2} since $\nabla^{(1)} = \mathcal{I}$. For $k\geq 2$,
    \begin{align*}
        \nabla^{(k)}[{\vO}]\vPi^{(r)} &= \sum_{p=1}^{k-1} (\CS^{+}\circ \nabla^{(p)})[{\vO}] \cdot (\CS^{-}\circ \nabla^{(k-p)})[{\vO}]\vPi^{(r)} \\
        &= \sum_{p=1}^{k-1} (\CS^{+}\circ \nabla^{(p)})[{\vO}] \cdot \vPi^{(r+k-p)} (\CS^{-}\circ \nabla^{(k-p)})[{\vO}]\vPi^{(r)} \tag{induction hypothesis}\\
        &= \sum_{p=1}^{k-1} \vPi^{(r+k)}(\CS^{+}\circ \nabla^{(p)})[{\vO}] \cdot \vPi^{(r+k-p)}(\CS^{-}\circ \nabla^{(k-p)})[{\vO}]\vPi^{(r)} \tag{induction hypothesis}\\
        &= \vPi^{(r+k)}\nabla^{(k)}[{\vO}]\vPi^{(r)}.
    \end{align*} 
    By applying $\CS^-$ onto the above equality and using \eqref{eq:key_observation_2} we obtain $(\CS^{\pm}\circ \nabla^{(k)})[{\vO}]\vPi^{(r)} = \vPi^{(r+k)}(\CS^{\pm}\circ \nabla^{(k)})[{\vO}]\vPi^{(r)}$. We now move on to prove the second equality $(\CS^{\pm}\circ \nabla^{(k)})[{\vO}]\vPi^{(r)} = (\CS^{\pm}\circ \nabla^{(k)})[\vP{\vO}\vP]\vPi^{(r)}$. The case $k=1$ is~\eqref{eq:key_observation_2} since $\nabla^{(1)} = \mathcal{I}$. For $k\geq 2$ and $\vP$ such that $\vPi^{(r+k)}\preceq \vP \preceq \vI$,
    \begin{align*}
        \nabla^{(k)}[{\vO}]\vPi^{(r)} &= \sum_{p=1}^{k-1} (\CS^{+}\circ \nabla^{(p)})[{\vO}] \cdot (\CS^{-}\circ \nabla^{(k-p)})[{\vO}]\vPi^{(r)} \\
        &= \sum_{p=1}^{k-1} (\CS^{+}\circ \nabla^{(p)})[{\vO}] \cdot \vPi^{(r+k-p)} (\CS^{-}\circ \nabla^{(k-p)})[{\vO}]\vPi^{(r)} \tag{first equality}\\
        &= \sum_{p=1}^{k-1} (\CS^{+}\circ \nabla^{(p)})[\vP\vO\vP] \cdot \vPi^{(r+k-p)}(\CS^{-}\circ \nabla^{(k-p)})[\vP\vO\vP]\vPi^{(r)} \tag{induction hypothesis}\\
        &= \sum_{p=1}^{k-1} (\CS^{+}\circ \nabla^{(p)})[\vP\vO\vP] \cdot(\CS^{-}\circ \nabla^{(k-p)})[\vP\vO\vP]\vPi^{(r)} \tag{first equality}\\
        &= \nabla^{(k)}[\vP\vO\vP]\vPi^{(r)}.
    \end{align*}
    By applying $\CS^-$ onto the above equality and using \eqref{eq:key_observation_2} we obtain $(\CS^{\pm}\circ \nabla^{(k)})[{\vO}]\vPi^{(r)} = (\CS^{\pm}\circ \nabla^{(k)})[\vP{\vO}\vP]\vPi^{(r)}$.
\end{proof}

\subsection{Lieb-Robinson bound}

So far we saw that the Kraus operator for the decay term, $\vK = \sqrt{\sqrt{\vrho}(\vI - \CT'^\dagger[\vI]) \sqrt{\vrho}} \vrho^{-\frac12}$, can be approximated by a series expansion involving $\mathcal{O}(\log(1/\eps))$ terms, each of which involves an integral representation that, in turn, can be approximated by an integral representation over a well-behaved interval of size $\mathcal{O}(\log(1/\eps))$ if the spectral norm $\|\vH\|$ of the Gibbs Hamiltonian is small enough. The last missing step for the quasi-locality of $\vK$ is arguing that the integrand in this representation can be made quasi-local. More specifically, we want to argue that the Heisenberg evolution $\e^{\ri \vH t} \vO \e^{-\ri \vH t}$ in the superoperator $\widetilde{\mathcal{S}}$ is quasi-local. By assuming a geometrically local Hamiltonian $\vH$, this can be proven via a Lieb-Robinson bound~\cite{lieb1972FiniteGroupVelocity}.

Lieb-Robinson bounds are theorems that quantify the spread of information in geometrically local Hamiltonians. Informally, they tell us that no information is transmitted from one qubit to another far away qubit in a short amount of time. Lieb-Robinson bounds make this notion precise and reveal that enough information can travel between two qubits at a distance $d$ from each other after a $\Omega(d)$ amount of time. Lieb-Robinson bounds are thus somewhat similar to lightcone arguments in small-depth circuit unitaries with local $2$-qubit gates. The fact that a similar ``lightcone'' concept can be extended to geometrically local Hamiltonians is nontrivial. Here we shall use a Lieb-Robinson bound from~\cite{haah2018QAlgSimLatticeHam}.
\begin{fact}[{\cite[Lemma~5]{haah2018QAlgSimLatticeHam}}]\label{thm:lieb_robinson}
    Let $\Lambda$ be a lattice and consider the sets $\Omega, X$ such that $X \subseteq \Omega$. Let $\vH = \sum_{Z\subseteq \Lambda}\mathbf{h}_Z$ be a geometrically local time-independent Hamiltonian where $\mathbf{h}_Z = 0$ if $\operatorname{diam}(Z) > 1$ and let $\vO_X$ be any operator supported on $X$. Let $\ell = \lfloor\operatorname{dist}(X,\Lambda\setminus \Omega)\rfloor$. Then
    \begin{align*}
        \|\e^{\ri \vH t} \vO_X \e^{-\ri \vH t} - \e^{\ri \vH_{\Omega} t} \vO_X \e^{-\ri \vH_{\Omega} t}\| \leq |X|\|\vO_X\|\frac{(2J |t|)^\ell}{\ell!},
    \end{align*}
    where $\vH_{\Omega} = \sum_{Z\subseteq \Omega}\mathbf{h}_Z$ and $J := \max_{p\in \Lambda} \sum_{Z\ni p}|Z|\|\mathbf{h}_Z\|$, i.e., $J = \Theta(\max_{Z\subseteq \Omega} \|\mathbf{h}_Z\|)$.
\end{fact}

If we assume that an operator $\vO$ (which ultimately will be $\CT'^\dagger[\vI]$) is supported on a small region $X$ and pick a region $\Omega$ that contains $X$ and separates $X$ and $\Lambda\setminus \Omega$ by a distance $\mathcal{O}(\log(1/\eps))$, then we can use the above Lieb-Robinson bound to argue that $\e^{\ri \vH t} \vO \e^{-\ri \vH t}$ is $\eps$-far away from $\e^{\ri \vH_\Omega t} \vO \e^{-\ri \vH_{\Omega} t}$ in operator norm, which can be used to approximate the superoperator $\widetilde{S}$. We formalize this argument in the next result.
\begin{lemma}\label{lem:QuasiLocalApxS}
    Let $\eps\in(0,\frac{1}{2}]$ and $f$ be the sum of an integrable function and a finite liner combination of Dirac $\delta$ functions, so that $f$ is supported on the interval $I=[-c,c]$. Let $\Lambda$ be a lattice and consider the sets $X \subseteq \Omega$. Let $\vH = \sum_{Z\subseteq \Lambda}\mathbf{h}_Z$ be a geometrically local time-independent Hamiltonian where $\mathbf{h}_Z = 0$ if $\operatorname{diam}(Z) > 1$ and let $\vO_X$ be any operator supported on $X$. Assume that $\operatorname{diam}(\Omega) = \mathcal{O}(\operatorname{diam}(X)\log(1/\eps))$ such that 
    \begin{align*}
        \ell := \lfloor\operatorname{dist}(X,\Lambda\setminus \Omega)\rfloor \geq \max\left(4\e J c, \log_2\left( \frac{|X|}{\sqrt{2\pi}\eps}\right) , 1\right),
    \end{align*}
    where $J := \max_{p\in \Lambda} \sum_{Z\ni p}|Z|\|\mathbf{h}_Z\|$. Let
    \begin{align*}
        \widetilde{\CS}[\cdot] := \int_{\mathbb{R}} f(t) \e^{\ri \vH t} [\cdot] \e^{-\ri \vH t} \rd t \qquad\text{and}\qquad \widetilde{\CS}_\Omega[\cdot] := \int_{\mathbb{R}} f(t) \e^{\ri \vH_\Omega t} [\cdot]\e^{-\ri \vH_{\Omega} t} \rd t,
    \end{align*}
    where $\vH_{\Omega} = \sum_{Z\subseteq \Omega}\mathbf{h}_Z$. Then
    \begin{align*}
        \|\widetilde{\CS}[\vO_X] - \widetilde{\CS}_\Omega[\vO_X]\| \leq \eps\nrm{f}_1\|\vO_X\|.
    \end{align*}
\end{lemma}
\begin{proof}
    We first start with
    \begin{align*}
        \|\widetilde{\CS}[\vO_X] - \widetilde{\CS}_\Omega[\vO_X]\| &\leq \int_{\mathbb{R}} |f(t)| \|\e^{\ri \vH t} \vO_X \e^{-\ri \vH t} - \e^{\ri \vH_\Omega t} \vO_X \e^{-\ri \vH_{\Omega} t}\| \rd t \tag{by triangle inequality}\\
        &\leq |X|\|\vO_X\|\int_{\mathbb{R}} |f(t)| \frac{(2J |t|)^{\ell}}{\ell!}\rd t \tag{by \Cref{thm:lieb_robinson}} \\
        &\leq |X|\|\vO_X\| \frac{(2J c)^{\ell}}{\ell!} \int_{\mathbb{R}}|f(t)|\rd t \tag{due to $f$'s support}\\
        &\leq |X|\|\vO_X\| \left(\frac{2\e J  c}{\ell}\right)^{\ell}\frac{\|f\|_1}{\sqrt{2\pi \ell}}. \tag{due to $\ell! \geq \sqrt{2\pi \ell}(\ell/\e)^{\ell}$}
    \end{align*}
    By picking $\ell \geq 4\e J c$, then $(2\e J c/\ell)^\ell \leq 2^{-\ell}$, and by picking $\ell \geq {\log_2}\Big(\frac{|X|}{\sqrt{2\pi}\eps}\Big)$, then $2^{-\ell} \leq \frac{\sqrt{2\pi}\eps}{|X|}$. Therefore, we conclude that
    \[
        \|\widetilde{\CS}[\vO_X] - \widetilde{\CS}_\Omega[\vO_X]\| 
        \leq \frac{|X|\|\vO_X\|}{\sqrt{2\pi\ell}} 2^{-\ell} \leq \eps \|f\|_1\|\vO_X\|. \qedhere
    \]
\end{proof}

\section{Rapidly decaying functions for high-temperature spectral gap}
\label{app:spectral_gap_rapidly_decaying}

We prove the condition that the norms $\|\CD_{\beta,r+1}^{a,\alpha} - \CD_{\beta,r}^{a,\alpha}\|_{2\to 2}$ are fast decaying with the radius $r$ for the discriminant terms $\CD_{\beta}^{a,\alpha}$ from \Cref{sec:spectral_gap}. We recall the physical setting being a $(k,\ell)$-local Hamiltonian $\vH = \sum_{Z\subseteq \Lambda}\mathbf{h}_Z$ defined on a $D$-dimensional lattice $\Lambda$, meaning that each interaction $\mathbf{h}_Z$ has support on at most $k$ sites and each site $a\in\Lambda$ appears on at most $\ell$ non-zero interactions $\mathbf{h}_Z$. Moreover, $\max_{Z\subseteq\Lambda}\|\mathbf{h}_Z\| \leq h$ and $J \leq hk\ell$ is the Lieb-Robinson velocity. The reduced discriminant terms $\CD_{\beta, r}^{a,\alpha}$ are defined by replacing $\vH$ with the reduced Hamiltonian $\vH_{B_a(r)}$ in the expressions for $\CD_{\beta}^{a,\alpha}$. This means that
\begin{align*}
    \CD_{\beta,r}^{a,\alpha}[\cdot] = \vT^{a,\alpha}_{\beta,r}[\cdot]\vT^{a,\alpha}_{\beta,r} - (\vM^{a,\alpha}_{\beta,r} - \vN^{a,\alpha}_{\beta,r})[\cdot] - [\cdot](\vM^{a,\alpha}_{\beta,r} -\vN^{a,\alpha}_{\beta,r}),
\end{align*}
where (let $\vA^{a,\alpha}_{\beta,r}(t) := \e^{\ri \beta t\vH_{B_a(r)}} \vA^{a,\alpha} \e^{-\ri \beta t\vH_{B_a(r)}}$)
\begin{align*}
    \vT^{a,\alpha}_{\beta,r} &= \int_{-\infty}^\infty \frac{1}{\cosh(2\pi t)}  \vA^{a,\alpha}_{\beta,r}(t) \rd t, 
    \\
    \vM^{a,\alpha}_{\beta,r} &= \iint_{-\infty}^\infty \frac{1}{4\pi\cosh(2\pi t)\sinhc(2\pi t')}\int_0^1[[\vA^{a,\alpha}_{\beta,r}(t+st'), \beta\vH_{B_a(r)}],\vA^{a,\alpha}_{\beta,r}(t)]\rd s \rd t \rd t',
    \\
    \vN^{a,\alpha}_{\beta,r} &= \frac{\vI}{4} - \iint_{-\infty}^\infty \frac{1}{2\cosh(2\pi t)\cosh(2\pi t')} \vA^{a,\alpha}_{\beta,r}(t)  \vA^{a,\alpha}_{\beta,r}(t') \rd t\rd t'.
\end{align*}
We now bound the norms $\|\CD_{\beta,r+1}^{a,\alpha} - \CD_{\beta,r}^{a,\alpha}\|_{2\to 2}$.
\decayingfunctions*
\begin{proof}
We start by bounding the norms $\|\CD_{\beta,r+1}^{a,\alpha} - \CD_{\beta,r}^{a,\alpha}\|_{2\to 2}$. For such, we bound the terms associated with $\vT^{a,\alpha}_{\beta,r}[\cdot]\vT^{a,\alpha}_{\beta,r}$, $\vM^{a,\alpha}_{\beta,r}[\cdot] + [\cdot]\vM^{a,\alpha}_{\beta,r}$, and $\vN^{a,\alpha}_{\beta,r}[\cdot] + [\cdot]\vN^{a,\alpha}_{\beta,r}$ separately.

\paragraph*{The $\vM^{a,\alpha}_{\beta,r}[\cdot] + [\cdot]\vM^{a,\alpha}_{\beta,r}$ term.} We start by bounding
\begin{align*}
    \|(\vM^{a,\alpha}_{\beta,r+1}[\cdot] + [\cdot]\vM^{a,\alpha}_{\beta,r+1}) - (\vM^{a,\alpha}_{\beta,r}[\cdot] + [\cdot]\vM^{a,\alpha}_{\beta,r}) \|_{2\to 2} \leq 2\|\vM^{a,\alpha}_{\beta,r+1} - \vM^{a,\alpha}_{\beta,r}\|,
\end{align*}
or, more simply, the quantity
\begin{align}
    \|\vM^{a,\alpha}_{\beta,r+1} - \vM^{a,\alpha}_{\beta,r}\| &\leq \iint_{-\infty}^\infty \frac{1}{4\pi\cosh(2\pi t)\sinhc(2\pi t')}\int_0^1 \|\vM^{a,\alpha}_{r+1,r}( t, st')  \| \rd s \rd t \rd t', \label{eq:M_term}\\
    \text{where}\quad \vM^{a,\alpha}_{r+1,r}( t, st') &:= [[\vA^{a,\alpha}_{\beta,r+1}( t +  st'), \beta\vH_{B_a(r+1)}],\vA^{a,\alpha}_{\beta,r+1}( t)] \nonumber \\
    &~~~~+ [[\vA^{a,\alpha}_{\beta,r}( t +  st'), \beta\vH_{B_a(r)}],\vA^{a,\alpha}_{\beta,r}( t)].\nonumber
\end{align}
First observe that
\begin{align*}
    \|\vM^{a,\alpha}_{r+1,r}( t,  st') \| &\leq 2\|[\vA^{a,\alpha}_{\beta,r+1}( t +  st'), \beta\vH_{B_a(r+1)}] - [\vA^{a,\alpha}( t +  st'), \beta\vH_{B_a(r)}]\| \|\vA^{a,\alpha}_{\beta,r+1}( t)\| \\
    &~~~~+ 2\|\vA^{a,\alpha}_{\beta,r+1}( t) - \vA^{a,\alpha}_{\beta,r}( t) \| \|[\vA^{a,\alpha}_{\beta,r}( t +  st'), \beta\vH_{B_a(r)}]\| \\
    &= 2\|[\vA^{a,\alpha}_{\beta,r+1}( t +  st'), \beta\vH_{B_a(r+1)}] - [\vA^{a,\alpha}_{\beta,r}( t +  st'), \beta\vH_{B_a(r)}]\| \\
    &~~~~+ 2\|\vA^{a,\alpha}_{\beta,r+1}( t) - \vA^{a,\alpha}_{\beta,r}( t) \| \|[\vA^{a,\alpha}_{\beta,r}, \beta\vH_{B_a(r)}]\|. \tag{$\|\vA^{a,\alpha}\| = 1$}
\end{align*}
In order to proceed, note that $\|[\vA^{a,\alpha}, \beta\vH_{B_a(r)}]\|  \leq 2\beta h\ell$ for all $r\in\mathbb{N}$. Regarding the other terms, we can employ a Lieb-Robinson bound (\Cref{thm:lieb_robinson}) to obtain
\begin{align}
    \|\vA^{a,\alpha}_{\beta,r+1}( t) - \vA^{a,\alpha}_{\beta,r}( t) \| &\leq 2\frac{(2 J\beta |t|)^r}{r!} \label{eq:lieb_robinson_1}
\end{align}
using that $\vA^{a,\alpha}$ has support on $1$ site and $\|\vA^{a,\alpha}\| = 1$, and
\begin{align}
    \|[\vA^{a,\alpha}_{\beta,r+1}( t +  st'), \beta\vH_{B_a(r+1)}] - [\vA^{a,\alpha}_{\beta,r}( t +  st'), \beta\vH_{B_a(r)}]\| &\leq 4\beta hk\ell^2\frac{(2 J \beta(|t| + |st'|))^r}{r!} \label{eq:lieb_robinson_2}
\end{align}
using that, for all $r\in\mathbb{N}$, $[\vA^{a,\alpha}, \vH_{B_a(r)}]$ has support on at most $k\ell$ sites and $\|[\vA^{a,\alpha}, \beta\vH_{B_a(r)}]\|\leq 2\beta h\ell$. Alternatively,
\begin{align}
    \|\vA^{a,\alpha}_{\beta,r+1}( t) - \vA^{a,\alpha}_{\beta,r}( t)\| &= \|\e^{\ri \beta t \vH_{B_a(r+1)}}\vA^{a,\alpha} \e^{-\ri \beta t \vH_{B_a(r+1)}} - \e^{\ri \beta t \vH_{B_a(r)}}\vA^{a,\alpha} \e^{-\ri \beta t \vH_{B_a(r)}} \| \nonumber\\
    &\leq \| \e^{\ri \beta t \vH_{B_a(r+1)}} - \e^{\ri \beta t \vH_{B_a(r)}}\| \|\vA^{a,\alpha} \e^{-\ri \beta t \vH_{B_a(r+1)}}\| \nonumber\\
    &~~~~+ \|\e^{-\ri \beta t \vH_{B_a(r+1)}} - \e^{-\ri \beta t \vH_{B_a(r)}}\| \|\e^{\ri \beta t \vH_{B_a(r)}}\vA^{a,\alpha}\| \nonumber\\
    &\leq 2\beta|t| \|\vH_{B_a(r+1)} - \vH_{B_a(r)}\| \tag{$\e^{\ri x}$ is $1$-Lipschitz and $\|\vA^{a,\alpha}\| = 1$} \nonumber\\
    &\leq 4 \beta h |t| \frac{\pi^{D/2}}{\Gamma(D/2)}(r+1)^{D-1}, \label{eq:large_times_bound_1}
\end{align}
where the last inequality follows from the fact that $\vH_{B_a(r+1)} - \vH_{B_a(r)}$ has at most many $\mathbf{h}_Z$ terms as the volume of the spherical shell $B_a(r+1)\setminus B_a(r)$, which is at most the area of $B_a(r+1)$, $2\pi^{D/2}(r+1)^{D-1}/\Gamma(D/2)$, times the distance between both spheres, which is $1$. 
Similarly,
\begin{align}
    \|[\vA^{a,\alpha}_{\beta,r+1}( t +  st'),& \beta\vH_{B_a(r+1)}] - [\vA^{a,\alpha}_{\beta,r}( t +  st'), \beta\vH_{B_a(r)}]\| \nonumber\\
    &\leq \beta |t + st'|\|\vH_{B_a(r+1)} - \vH_{B_a(r)}\| \big( \|[\vA^{a,\alpha},\beta\vH_{B_a(r+1)}]\| +  \|[\vA^{a,\alpha},\beta\vH_{B_a(r)}]\| \big)  \nonumber\\
    &\leq 8 \beta^2 h^2 \ell |t+st'| \frac{\pi^{D/2}}{\Gamma(D/2)}(r+1)^{D-1}. \label{eq:large_times_bound_2}
\end{align}
Depending on $(t,t')$, we employ either \eqref{eq:lieb_robinson_1}-\eqref{eq:lieb_robinson_2} or \eqref{eq:large_times_bound_1}-\eqref{eq:large_times_bound_2}. To be more precise, for some $T>0$ to be determined later, we have that (remember that $s\in[0,1]$)
\begin{align*}
    \|\vM^{a,\alpha}_{r+1,r}( t,  st') \| \leq \begin{cases}
         16\beta hk\ell^2 \frac{(4 J \beta T)^r}{r!} &\text{if}~(t,t')\in [-T,T]^2,\\
         16 \beta^2 h^2 \ell (2|t|+|t'|) \frac{\pi^{D/2}}{\Gamma(D/2)}(r+1)^{D-1} &\text{if}~(t,t')\notin [-T,T]^2.
    \end{cases}
\end{align*}
We now bound \eqref{eq:M_term} by splitting the integral into $(t,t')\in [-T,T]^2$, i.e.,
\begin{align*}
    \iint_{-T}^{T} \frac{\int_0^1 \|\vM^{a,\alpha}_{r+1,r}( t,  st')  \| \rd s}{4\pi\cosh(2\pi t)\sinhc(2\pi t')} \rd t \rd t'
    \leq \iint_{-\infty}^{\infty} \frac{4 \beta h k \ell^2 \frac{(4J\beta T)^r}{r!}}{\pi\cosh(2\pi t)\sinhc(2\pi t')} \rd t \rd t' =  \frac{\beta h k \ell^2}{2} \frac{(4J\beta T)^r}{r!},
\end{align*}
using that $\int_{-\infty}^\infty \frac{1}{\cosh(2\pi x)}\rd x = \frac{1}{2}$ and $\int_{-\infty}^\infty \frac{1}{\sinhc(2\pi x)}\rd x = \frac{\pi}{4}$, and into $(t,t')\notin [-T,T]^2$, i.e.,
\begin{align*}
    &\left(\int_{-T}^{T}\int_{(-\infty,\infty)\setminus[-T,T]} + \int_{-\infty}^{-T}\int_{-\infty}^{\infty} + \int_{T}^{\infty}\int_{-\infty}^{\infty} \right)\frac{\int_0^1 \|\vM^{a,\alpha}_{r+1,r}( t,  st')  \| \rd s}{4\pi\cosh(2\pi t)\sinhc(2\pi t')} \rd t \rd t' \\
    &\leq \left(\int_{-T}^{T}\int_{T}^\infty + \int_{T}^{\infty}\int_{-\infty}^{\infty} \right) \frac{8 \beta^2 h^2 \ell (2|t|+|t'|) \frac{\pi^{D/2-1}}{\Gamma(D/2)}(r+1)^{D-1}}{\cosh(2\pi t)\sinhc(2\pi t')}\rd t \rd t'\\
    &\leq 8\beta^2 h^2\ell \frac{\pi^{D/2-1}}{\Gamma(D/2)}(r+1)^{D-1} \left(\int_{T}^\infty \frac{1/5 + t'/2}{\sinhc(2\pi t')} \rd t' + \int_{T}^\infty \frac{\pi t/2 + \pi/10}{\cosh(2\pi t)} \rd t \right)\\
    &\leq 8\beta^2 h^2\ell \frac{\pi^{D/2}}{\Gamma(D/2)}(r+1)^{D-1} \int_{T}^\infty \left(\frac{3}{10} + \frac{5}{4}t\right)\e^{-\frac{2\pi}{3} t}\rd t \tag{$\cosh{x}\geq \frac{1}{2}\e^{x}$, $\sinhc{x}\geq \frac{2}{\pi}\e^{x/3}$}\\
    &= 8\beta^2 h^2\ell \frac{\pi^{D/2}}{\Gamma(D/2)}(r+1)^{D-1}\left(\frac{9}{20\pi} + \frac{15(2\pi T + 3)}{16\pi^2}\right)\e^{-\frac{2\pi}{3} T},
\end{align*}
using that $\int_{-\infty}^\infty \frac{|x|}{\cosh(2\pi x)}\rd x = \frac{G}{\pi^2} < \frac{1}{10}$ and $\int_{-\infty}^\infty \frac{|x|}{\sinhc(2\pi x)}\rd x = \frac{7\zeta(3)}{4\pi^2} < \frac{\pi}{10}$, where $G$ is Catalan's constant and $\zeta(3)$ is Ap\'ery's constant. All in all,
\begin{align*}
    \|\vM^{a,\alpha}_{\beta,r+1} - \vM^{a,\alpha}_{\beta,r}\| \leq \frac{\beta J  \ell}{2} \frac{(4J\beta T)^r}{r!} + \beta^2 h^2\ell\frac{\pi^{D/2}}{\Gamma(D/2)}(r+1)^{D-1} \left(\frac{7}{2} + 5T\right)\e^{-\frac{2\pi}{3} T}.
\end{align*}


\paragraph*{The $\vN^{a,\alpha}_{\beta,r}[\cdot] + [\cdot]\vN^{a,\alpha}_{\beta,r}$ term.} Similarly to the previous term, we bound
\begin{align*}
    \|(\vN^{a,\alpha}_{\beta,r+1}[\cdot] &+ [\cdot]\vN^{a,\alpha}_{\beta,r+1}) - (\vN^{a,\alpha}_{\beta,r}[\cdot] + [\cdot]\vN^{a,\alpha}_{\beta,r}) \|_{2\to 2} \\
    &\leq 2 \|\vN^{a,\alpha}_{\beta,r+1} - \vN^{a,\alpha}_{\beta,r}\|\\
    &\leq \iint_{-\infty}^\infty \frac{\|\vA^{a,\alpha}_{\beta,r+1}( t)\vA^{a,\alpha}_{\beta,r+1}( t') - \vA^{a,\alpha}_{\beta,r}( t)\vA^{a,\alpha}_{\beta,r}( t')\|}{\cosh(2\pi t)\cosh(2\pi t')} \rd t \rd t'\\
    &\leq \iint_{-\infty}^\infty \frac{2\|\vA^{a,\alpha}_{\beta,r+1}( t) - \vA^{a,\alpha}_{\beta,r}( t)\|}{\cosh(2\pi t)\cosh(2\pi t')} \rd t \rd t' \tag{triangle inequality and $\|\vA^{a,\alpha}_{\beta,r}( t')\| = 1$}\\
    &\leq \int_0^{T} \frac{4\frac{(2J\beta T)^r}{r!}}{\cosh(2\pi t)}\rd t + \int_{T}^\infty \frac{8\beta h t \frac{\pi^{D/2}}{\Gamma(D/2)}(r+1)^{D-1}}{\cosh(2\pi t)} \rd t \tag{by \eqref{eq:lieb_robinson_1} and \eqref{eq:large_times_bound_1}}\\
    &\leq \frac{(2J\beta T)^r}{r!} + \frac{4\beta h}{\pi^2} \frac{\pi^{D/2}}{\Gamma(D/2)}(r+1)^{D-1} (2\pi T + 1)\e^{-2\pi T}.
\end{align*}

\paragraph{The $\vT^{a,\alpha}_{\beta,r}[\cdot]\vT^{a,\alpha}_{\beta,r}$ term.} Finally, we bound
\begin{align*}
    \|\vT^{a,\alpha}_{\beta,r+1}[\cdot] \vT^{a,\alpha}_{\beta,r+1} -\vT^{a,\alpha}_{\beta,r}[\cdot] \vT^{a,\alpha}_{\beta,r}\|_{2\to 2} &\leq (\|\vT^{a,\alpha}_{\beta,r+1}\| + \|\vT^{a,\alpha}_{\beta,r}\|)\|\vT^{a,\alpha}_{\beta,r+1} - \vT^{a,\alpha}_{\beta,r}\|\\
    &\leq \|\vT^{a,\alpha}_{\beta,r+1} - \vT^{a,\alpha}_{\beta,r}\| \tag{$\|\vT^{a,\alpha}_{\beta,r+1}\|, \|\vT^{a,\alpha}_{\beta,r}\| \leq \frac{1}{2}$}\\
    &\leq \int_{-\infty}^\infty \frac{1}{\cosh(2\pi t)}\|\vA^{a,\alpha}_{\beta,r+1}( t) - \vA^{a,\alpha}_{\beta,r}( t)\| \rd t\\
    &\leq \int_{0}^{T}\frac{2\frac{(2\beta J T)^r}{r!}}{\cosh(2\pi t)}\rd t +  \int_{T}^\infty \frac{4 \beta h t \frac{\pi^{D/2}}{\Gamma(D/2)}(r+1)^{D-1}}{\cosh(2\pi t)} \rd t \tag{by \eqref{eq:lieb_robinson_1} and \eqref{eq:large_times_bound_1}}\\
    &\leq \frac{1}{2}\frac{(2\beta J T)^r}{r!} + \frac{2\beta h}{\pi^2}\frac{\pi^{D/2}}{\Gamma(D/2)}(r+1)^{D-1}(2\pi T + 1)\e^{-2\pi T} .
\end{align*}
%

For each $r\in\mathbb{N}$, we shall pick $T = \frac{r \chi(\beta J)}{4\e \beta J}$ where $\chi(x) := \frac{x}{1+x}$. By gathering all the terms computed above with $T = \frac{r \chi(\beta J)}{4\e \beta J}$ and using Stirling's approximation $r! \geq \sqrt{2\pi r}(r/\e)^r$ and $\Gamma(D/2) \geq \sqrt{4\pi /D}(D/2\e)^{D/2}$,  we finally conclude that
\begin{align*}
    \|\CD_{\beta,r+1}^{a,\alpha} - \CD_{\beta,r}^{a,\alpha}\|_{2\to 2} \leq c_1 (1 + \beta J \ell )\frac{\chi(\beta J)^r}{\sqrt{r}} + (\beta h + \beta^2 h^2\ell)\left(\frac{c_2 r}{\sqrt{D}}\right)^{D+1}\e^{-\Omega\big(r \frac{\chi(\beta J)}{\beta J}\big)}
\end{align*}
for some constants $c_1,c_2 > 0$.

\paragraph{The $\|\CD_{\beta,0}^{a,\alpha} - \CD_{0,0}^{a,\alpha}\|_{2\to 2}$ norm.} We now move on to bounding the term $\CD_{\beta,0}^{a,\alpha} - \CD_{0,0}^{a,\alpha}$ between the $1$-qubit perturbation at $\beta$ and that at $\beta = 0$. We have that
\begin{align*}
    \| \vM_{\beta,0}^{a,\alpha} - \vM_{0,0}^{a,\alpha}\| = \| \vM_{\beta,0}^{a,\alpha}\| &\leq \iint_{-\infty}^\infty \frac{\int_0^1 \|[[\vA_{\beta,0}^{a,\alpha}( t+ st'), \beta\vH_{B_a(0)}],\vA_{\beta,0}^{a,\alpha}( t)]\|}{4\pi\cosh(2\pi t)\sinhc(2\pi t')}\rd s \rd t \rd t' \\
    &\leq \iint_{-\infty}^\infty \frac{\int_0^1 \|[\vA_{\beta,0}^{a,\alpha}( t+ st'), \beta\vH_{B_a(0)}]\|}{2\pi\cosh(2\pi t)\sinhc(2\pi t')}\rd s \rd t \rd t' \tag{$\|\vA^{a,\alpha}\| = 1$}\\
    &\leq \iint_{-\infty}^\infty \frac{\beta h\ell}{\pi\cosh(2\pi t)\sinhc(2\pi t')}\rd t \rd t' \tag{$\|[\vA^{a,\alpha}, \vH_{B_a(0)}]\|\leq 2h\ell$}\\
    &= \frac{\beta h\ell}{8}.
\end{align*}
Moreover,
\begin{align*}
    \| \vN_{\beta,0}^{a,\alpha} - \vN_{0,0}^{a,\alpha}\| &\leq \iint_{-\infty}^\infty \frac{\|\vA^{a,\alpha}_{\beta,0}( t)  \vA^{a,\alpha}_{\beta,0}( t') - \vI\|}{2\cosh(2\pi t)\cosh(2\pi t')}  \rd t\rd t' \\
    &\leq \iint_{-\infty}^\infty \frac{\|\vA^{a,\alpha}_{\beta,0}( t -  t')   - \vA^{a,\alpha}\|}{2\cosh(2\pi t)\cosh(2\pi t')}  \rd t \rd t' \\
    &\leq \iint_{-\infty}^\infty \frac{|t-t'|\|[\vA^{a,\alpha}, \beta\vH_{B_a(0)}]\|}{2\cosh(2\pi t)\cosh(2\pi t')}  \rd t \rd t \tag{$\|\vA(t) - \vA\| \leq |t|\|[\vA,\beta\vH]\|$}\\
    &\leq \beta h \ell \iint_{-\infty}^\infty \frac{|t-t'|}{\cosh(2\pi t)\cosh(2\pi t')} \rd t \rd t' \tag{$\|[\vA^{a,\alpha},\vH_{B_a(0)}]\| \leq 2h \ell$}\\
    &\leq 0.068\beta h\ell. \tag{$\iint_{-\infty}^\infty \frac{|t-t'|}{\cosh(2\pi t)\cosh(2\pi t')} \rd t \rd t' < 0.0679$}
\end{align*}
Finally, 
\begin{align*}
    \|\vT_{\beta, 0}^{a,\alpha} - \vT_{0, 0}^{a,\alpha}\| \leq \int_{-\infty}^\infty \frac{\|\vA_{\beta,0}^{a,\alpha}( t) - \vA^{a,\alpha}\|}{\cosh(2\pi t)}\rd t \leq 2 \beta h \ell \int_{-\infty}^\infty \frac{|t|}{\cosh(2\pi t)}\rd t \leq 0.186 \beta h \ell,
\end{align*}
since $\int_{-\infty}^\infty \frac{|t|}{\cosh(2\pi t)}\rd t = \frac{G}{\pi^2} < 0.0929$, where $G$ is Catalan's constant. Gathering all the above inequalities, we conclude that
\[
    \|\CD_{\beta,0}^{a,\alpha} - \CD_{0,0}^{a,\alpha}\|_{2\to 2} \leq \|\vT_{\beta, 0}^{a,\alpha} - \vT_{0, 0}^{a,\alpha}\| + 2\| \vN_{\beta,0}^{a,\alpha} - \vN_{0,0}^{a,\alpha}\| + 2\| \vM_{\beta,0}^{a,\alpha} - \vM_{0,0}^{a,\alpha}\| \leq 0.38 \beta h \ell. \qedhere
\]
\end{proof}

\section{A recursive discrete ``rejection map'' construction}\label{apx:RecDisc}

Starting with a detailed-balanced transition part $\CT'$, we have designed a rejection part $\CR$ such that the sum is detailed balanced and trace preserving. However, the trace-preserving requirement leads to a nontrivial series expansion due to direct Taylor expansion for the matrix square root. Here, we consider a different strategy: What if we only fix the leading-order term? 
\begin{align}
    \vK = \sqrt{\sqrt{\vrho}(\vI-\CT^{\dagger}[\vI])\sqrt{\vrho}}\vrho^{-\frac12} &= \vI - \sum_{k=1}^{\infty} (\CS^{-}\circ \nabla^{(k)})[\CT'^{\dagger}[\vI]]
    = \vI  - \CS^{-}[\CT'^{\dagger}[\vI]] + \dots
\end{align}
Remarkably, even the first-order term by itself is already detailed balanced. Recall that the reweighting superoperator $\CS^{\pm} := \frac{\CI}{2} \pm \CS$ comes from
\begin{align}\label{eq:SDefIntro}
	\CS[\vA] = \lim_{\theta\rightarrow 0+} \int_{\BR\setminus[-\theta,\theta]} \frac{\ri}{\sinh(2\pi t)}\e^{\ri\vH t} \vA \e^{-\ri\vH t} \rd t = \int_{\BR} \frac{\ri}{\sinh(2\pi t)}\left(\e^{\ri\vH t} \vA \e^{-\ri\vH t} -\vA\right) \rd t.
\end{align}

\begin{lemma}[First-order rejection term is detailed balanced]
Consider a Hamiltonian $\vH$ with the Gibbs state $\vrho\propto \e^{-\vH}$. Then, for any Hermitian $\vB = \vB^{\dagger}$ and real $b\in \BR$, the completely positive map $\widetilde{\CR}$ with one Kraus operator
\begin{align}
    \widetilde{\CR}[\cdot]:=\left(\vI - b \CS^{-}[\vB] \right) [\cdot ]\left(\vI - b \CS^{-}[\vB] \right)^{\dagger}
\end{align}  
is $\vrho$-detailed balanced. Further, its action on the identity gives
\begin{align}\label{eq:baseTelescope}
    \widetilde{\CR}^{\dagger}[\vI]  = \vI - b \vB + b^2 \CS^{+}[\vB]\CS^{-}[\vB].
\end{align}
\end{lemma}
\begin{proof}
The first equation is due to~\autoref{thm:CoherentRule} with function $f_G(\nu) = 1/(1+\exp(\nu/2)) = \frac{1}{2} - \frac{1}{2}\tanh(\frac{\nu}{4})$ by observing that $\vI - b \CS^{-}[\vB] = \CS^{-}[2\vI - b\vB]$ since $\CS^{-}[\vI] = \frac{\vI}{2}$ and that $\CS^- = \CS^{[f_G]}$. For the second equation, we expand
\begin{align}
    \widetilde{\CR}^{\dagger}[\vI] = \left(\vI - b \CS^{-}[\vB] \right)^{\dagger}\left(\vI - b \CS^{-}[\vB] \right)&= \vI - b(\CS^{-}[\vB])^{\dagger} - b\CS^{-}[\vB] + b^2 (\CS^{-}[\vB])^{\dagger}\CS^{-}[\vB].
\end{align}To further simplify the $\CO(b)$ terms, recall the identity $\CS^{\pm} = \frac{\CI}{2} \pm \CS$ and symmetry $\CS^{+}[\vM]=(\CS^{-}[\vM^\dagger])^\dagger$ for each $\vM\in\BC^{d\times d}$.
\end{proof}

Although the first-order fix does not give full trace preservation $\widetilde{\CR}^{\dagger}[\vI] \ne \vI$, detailed balance (and complete positivity) is always maintained. This observation allows us to iteratively improve the trace, by recursively correcting the trace operator $\CR^{\dagger}[\vI]$ up to a second-order error that is Hermitian. Consider the sequence of Hermitian matrices $\vB_1,\vB_2,\dots$ by
\begin{align}
    \vB_1 &:= \CT'^{\dagger}[\vI] \qquad\text{and}\qquad \vB_{k+1} := \CS^+[\vB_k]\CS^-[\vB_k] \quad \text{for each}\quad k\geq 1\label{eq:Bk}
\end{align}
and consider the ansatz 
\begin{align}\label{eq:ansatz}
    \CR'[\cdot] &:=\sum_{k=1}^{\infty} \vK_k[\cdot] \vK_k^{\dagger} \quad \text{where}\quad \vK_k := c_k \left(\vI - b_k \CS^-[\vB_k]\right)\quad \text{for each}\quad k\geq 1
\end{align}
depending on some $c_k,b_k \in \BR$ to be determined. The recursion~\eqref{eq:Bk} is defined in accordance with~\eqref{eq:baseTelescope} such that the trace operator $\CR'^\dagger[\vI]$ telescopes 
\begin{align}
    \CR'^\dagger[\vI] &= \sum_{k=1}^{\infty} c_k^2\left(\vI - b_k\vB_k + b_k^2\vB_{k+1}\right) = \vI - \vB_1,\label{eq:telescope}\\
     &\text{as long as}\quad \sum_{k=1}^{\infty} c_k^2 =1, \quad b_1=1/c_1^2, \quad  b_{k+1}=b_k^2 c_k^2/c_{k+1}^2, \quad \text{and}\quad	\lim_{k\rightarrow \infty}c_k^2b_k^2\vB_{k+1} = 0.\nonumber
\end{align}
A natural choice of parameters is
\begin{align}\label{eq:akbkchoice}
        c_k = 2^{-\frac{k}{2}} \quad \text{and}\quad b_k = 2^{2^{k}-1}\quad \text{for each}\quad k\ge 1.
\end{align}

In order to show convergence, it suffices to bound the terms $\nrm{\vB_{k+1}}$ as $k$ grows as follows.
\begin{lemma}[Recursive norm bounds]\label{lem:RecKrausNormBnd}
 Consider the sequence of Hermitian $\vB_k$ recursively defined as~\eqref{eq:Bk}. Suppose that\footnote{Right now, the factor $\|\CS^-\|_{\infty\to\infty}^2$ may scale logarithmically with the system size for general jumps $\vA$, which means that the transition rate may have to decrease mildly with the system size. For local jumps, this can likely be improved using Lieb-Robinson-type arguments.}
    \begin{align}
        \|\CT^{\dagger}[\vI]\| =\nrm{\vB_1} &\leq \frac{\lambda}{\|\CS^-\|_{\infty\to\infty}^2} \quad\text{for some}\quad\lambda >0.
    \end{align}
    Then
    \begin{align}
\nrm{\vB_{k+1}} \leq \frac{\lambda^{2^{k}}}{\|\CS^-\|_{\infty\to\infty}^2} \quad \text{for each}\quad k\ge 1.    
    \end{align}
\end{lemma}
\begin{proof}
By induction, assume
\begin{align}
    \nrm{\vB_k} &\leq \frac{\lambda^{2^{k-1}}}{\|\CS^-\|_{\infty\to\infty}^2},
\end{align}
which holds at $k=1$. Then, it holds for $k+1$ that
\[
    \nrm{\vB_{k+1}} \leq (\|\CS^-\|_{\infty\to\infty}\nrm{\vB_k})^2
\leq\left(\frac{\lambda^{2^{k-1}}}{\|\CS^-\|_{\infty\to\infty}}\right)^{\!\!2} 
=\frac{\lambda^{2^{k}}}{\|\CS^-\|_{\infty\to\infty}^2}. \qedhere
\]
\end{proof}

\subsection{Finite-order errors}

To implement the rejection map to a finite order $\ell\geq 1$ (which may depend on the target precision), consider the completely positive (and nearly trace-presering) map
\begin{align}\label{eq:truncated_ansatz}
\CR_{\ell}'[\cdot] := \sum_{k=1}^{\ell}\vK_k[\cdot] \vK_k^{\dagger}.
\end{align}
We claim that $\CR_{\ell}' + \left(\sum_{j=\ell+1}^\infty c_j^2\right)\CI$ is close to $\CR'$ from \eqref{eq:ansatz}. \AC{Actually, if we want to make the truncated map trace-preserving, we have to add a 4-th order term to achieve 4-th order error.}
In order to analyse the truncation error we introduce the following lemma and corollary:

\begin{lemma}\label{lem:DilationToDiamond}
Let $\CH_1$, $\CH_2$, and $\CH_3$ be finite-dimensional Hilbert spaces.
Let $\vD,\widetilde{\vD}\colon \CH_1 \rightarrow \CH_2 \otimes\CH_3$ be bounded dilation operators, and let $\CM[\cdot]:=\tr_{\CH_3}\vD[\cdot]\vD^\dagger$ and $\widetilde{\CM}[\cdot]:=\tr_{\CH_3}\widetilde{\vD}[\cdot]\widetilde{\vD}^\dagger$ be the respective induced completely positive maps. Then 
\begin{align*}
\|\CM-\widetilde{\CM}\|_\Diamond\leq(\nrm{\vD}+\lVert\widetilde{\vD}\rVert) \|\vD-\widetilde{\vD}\|.
\end{align*}
\end{lemma}
 \noindent\anote{Check that all steps (trace -- first ineq., Hölder -- last ineq.) work in infinite dimensions.}
 \AC{I don't think we will ever actually need infinite dimensional spaces.}
 \anote{Technically, in \Cref{cor:InfDiscRejTrunc} the dilation that belongs to the infinitely many Kraus operators, defines an infinite dimensional auxiliary space. If we have infinite variants, we can deal with these infinitely many Kraus operators more directly.}\jnote{I think the steps look fine.}
\begin{proof}
For any Hilbert space $\CH_0$ and $\vsigma\colon \CH_1\otimes\CH_0 \rightarrow \CH_1\otimes\CH_0$ of trace norm $1$ we have 
\begin{align*}
    \|{(\CM\otimes\CI)[\vsigma] - (\widetilde{\CM}\otimes\CI)[\vsigma]}\|_1&
    =\nrm{\tr_{\CH_3\otimes\CH_0}\left[(\vD\otimes \vI)\vsigma(\vD^\dagger\otimes\vI)-(\widetilde{\vD}\otimes \vI)\vsigma(\widetilde{\vD}^\dagger\otimes\vI)\right]}_1\\&
    \leq\|{(\vD\otimes \vI)\vsigma(\vD^\dagger\otimes\vI)-(\widetilde{\vD}\otimes \vI)\vsigma(\widetilde{\vD}^\dagger\otimes\vI)}\|_1\\&
    \leq \|{(\vD\otimes \vI)\vsigma(\vD^\dagger\otimes\vI)-(\widetilde{\vD}\otimes \vI)\vsigma(\vD^\dagger\otimes\vI)}\|_1\\&
    \quad+\|{(\widetilde{\vD}\otimes \vI)\vsigma(\vD^\dagger\otimes\vI)-(\widetilde{\vD}\otimes \vI)\vsigma(\widetilde{\vD}^\dagger\otimes\vI)}\|_1\\&
    \leq \|{((\vD-\widetilde{\vD})\otimes \vI)\vsigma(\vD^\dagger\otimes\vI)}\|_1
    +\|{(\widetilde{\vD}\otimes \vI)\vsigma((\vD-\widetilde{\vD})^\dagger\otimes\vI)}\|_1\\&
    \leq \|{(\vD-\widetilde{\vD})\otimes \vI}\|\nrm{\vsigma}_1\|{\vD^\dagger\!\otimes\vI}\|
    + \|{\widetilde{\vD}\otimes \vI}\|\nrm{\vsigma}_1\|{(\vD-\widetilde{\vD})^\dagger\!\otimes\vI}\|,
\end{align*}
where the last inequality follows from H\"older's inequality and the fact that the infinity Schatten norm equals the spectral norm.
\end{proof}

From the above the next two corollaries follow.
\begin{corollary}\label{cor:KrausDiamond}
For any matrices $\vA,\vB$, we have that $\norm{ \vA[\cdot ]\vA^{\dagger}- \vB[\cdot ]\vB^{\dagger}}_{\Diamond} \le (\nrm{\vA}+\nrm{\vB})\nrm{\vA-\vB}$.
\end{corollary}

\begin{corollary}\label{cor:InfDiscRejTrunc}
    Consider the maps $\CR'$ and $\CR_{\ell}'$ from \eqref{eq:ansatz} and \eqref{eq:truncated_ansatz}. For all integer $\ell \geq 1$,
\begin{align}
    \nrm{ \CR_{\ell}' + \left(\sum_{j=\ell+1}^\infty c_j^2\right)\CI - \CR'}_{\Diamond}
\leq \sum_{k=\ell+1}^{\infty}2b_k\nrm{\CS^-[\vB_k]}+b_k^2\nrm{\CS^-[\vB_k]}^2.
\end{align}
\end{corollary}
\begin{proof}
\begin{align*}
    \nrm{ \CR_{\ell}' + \left(\sum_{j=\ell+1}^\infty c_j^2\right)\CI - \CR'}_{\Diamond} &=
    \nrm{ \sum_{k=\ell+1}^{\infty} c_k \vI[\cdot] \vI^{\dagger}c_k - \sum_{k=\ell+1}^{\infty}\vK_k[\cdot] \vK_k^{\dagger}}_{\Diamond}\\&
	\leq\sum_{k=\ell+1}^{\infty}\| c_k \vI[\cdot] \vI^{\dagger}c_k - \vK_k[\cdot] \vK_k^{\dagger} \|_{\Diamond}\tag{by triangle inequality}\\&
	\leq\sum_{k=\ell+1}^{\infty}(\nrm{c_k \vI}+\nrm{\vK_k})\nrm{c_k \vI-\vK_k}\tag{by \Cref{cor:KrausDiamond}}\\&
	\leq\sum_{k=\ell+1}^{\infty}(2c_k + c_k b_k\nrm{\CS^-[\vB_k]})\nrm{c_k b_k\CS^-[\vB_k]}\tag{by triangle inequality}\\&
	=\sum_{k=\ell+1}^{\infty}c_k^2(2+ b_k\nrm{\CS^-[\vB_k]})\nrm{b_k\CS^-[\vB_k]}\\&
	\leq\sum_{k=\ell+1}^{\infty}(2+ b_k\nrm{\CS^-[\vB_k]})\nrm{b_k\CS^-[\vB_k]}.\tag*{(since $\sum_{k=1}^{\infty} c_k^2 =1$)\kern3mm\qedhere}
\end{align*}
\end{proof}
From now on let us use the choices of \eqref{eq:akbkchoice}, so that $c_k = 2^{-\frac{k}{2}}$ and $b_k = 2^{2^{k}-1}$. If also 
\begin{align}
    \nrm{\vB_1} &\leq \frac{\lambda}{\|\CS^-\|_{\infty\to\infty}^2},
\end{align}
then by \Cref{lem:RecKrausNormBnd} we get
\begin{align}
    b_{\ell+1}\nrm{\CS^-[\vB_{\ell+1}]} \leq 2^{2^{\ell+1} - 1}\|\CS^-\|_{\infty\to\infty}\|\vB_{\ell+1}\| \leq \frac{1}{2}\frac{(4\lambda)^{2^{\ell}}}{\|\CS^-\|_{\infty\to\infty}},
\end{align}
which for $\lambda\in[0,\frac{1}{8}]$ implies that 
\begin{align}\label{eq:RestRClose}
	\| \CR_{\ell}' + 2^{-\ell}\CI - \CR' \|_{\Diamond} =
    \nrm{ \CR_{\ell}' + \left(\sum_{j=\ell+1}^\infty c_j^2\right)\CI - \CR'}_{\Diamond} 
	\leq 3\frac{(4\lambda)^{2^{\ell}}}{\|\CS^-\|_{\infty\to\infty}}.
\end{align}
Lastly, we introduce two approaches to retain trace preservation. The first approach, which does \emph{not} preserve detailed balance, augments the map $\mathcal{T}' + \mathcal{R}'_\ell$ with $\widetilde{\vK}_{\ell+1}[\cdot]\widetilde{\vK}_{\ell+1}^\dagger$, where
\begin{align*}
\widetilde{\vK}_{\ell+1}:=\sqrt{\sum_{j=\ell+1}^\infty c_j^2}\sqrt{\vI-\frac{c_\ell^2b_\ell^2\vB_{\ell+1}}{\sum_{j=\ell+1}^\infty c_j^2}} =2^{-\frac{\ell}{2}}\sqrt{\vI-b_\ell^2\vB_{\ell+1}},
\end{align*}
with the simplification due to $c_k = 2^{-\frac{k}{2}}$ in \eqref{eq:akbkchoice}. The second approach, which \emph{does} preserve detailed balance, is to simply use the discrete construction from \Cref{lem:FindingDiscDecayTerm}, which augments the map $\mathcal{T}' + \mathcal{R}'_\ell$ with $\vK'_{\ell+1}[\cdot]\vK'^{\dagger}_{\ell+1}$, where $\vK'_{\ell+1} := \sqrt{\sqrt{\vrho}(\vI - \CT'^{\dagger}[\vI] - \CR'^\dagger_{\ell}[\vI])\sqrt{\vrho}}\vrho^{-\frac{1}{2}}$. Since
\begin{align*}
    \CT'^{\dagger}[\vI] + \CR'^\dagger_{\ell}[\vI] &= \vB_1 + \sum_{k=1}^\ell 2^{-k}\left(\vI - b_k\vB_k + b^2_k\vB_{k+1}\right) = (1-2^{-\ell})\vI + 2^{-\ell} b_\ell^2\vB_{\ell+1},
\end{align*}
the Kraus operator $\vK'_{\ell+1}$ can be written as $\vK'_{\ell+1} = 2^{-\frac{\ell}{2}}\sqrt{\sqrt{\vrho}(\vI - b^2_\ell \vB_{\ell+1})\sqrt{\vrho}}\vrho^{-\frac{1}{2}}$. We note the similarity to $\widetilde{\vK}_{\ell+1}$. \jnote{Can we relate them more closely?}

\begin{corollary}[Trace-preserving truncation schemes]
    The maps $\widetilde{\CQ}[\cdot]:=\CT'[\cdot]+\CR_{\ell}'[\cdot]+\widetilde{\vK}_{\ell+1}[\cdot]\widetilde{\vK}_{\ell+1}^\dagger$ and $\CQ[\cdot]:=\CT'[\cdot]+\CR_{\ell}'[\cdot]+\vK'_{\ell+1}[\cdot]\vK'^\dagger_{\ell+1}$ are quantum channels which, for $\lambda\in[0,\frac{1}{8}]$, satisfy
\begin{align}\label{eq:ChannelClose}
    \norm{\widetilde{\CQ} - (\CT'+\CR')}_{\Diamond} 
	\leq 4\frac{(4\lambda)^{2^{\ell}}}{\|\CS^-\|_{\infty\to\infty}} \qquad\text{and}\qquad \norm{\CQ - (\CT'+\CR')}_{\Diamond} 
	\leq 4\frac{(4\lambda)^{2^{\ell}}}{\|\CS^-\|_{\infty\to\infty}}.
\end{align}
\end{corollary}
\begin{proof}
Regarding the map $\widetilde{\CQ}$, first observe that 
\begin{align*}
    \nrm{\widetilde{\CQ} - \CT' - \CR_{\ell}' - \left(\sum_{j=\ell+1}^\infty c_j^2\right)\CI }_{\Diamond} 
	&= \|{\widetilde{\vK}_{\ell+1}[\cdot]\widetilde{\vK}_{\ell+1}^\dagger - 2^{-\ell}\vI[\cdot]\vI^\dagger }\|_{\Diamond}\\&
	\leq\left(\|{2^{-\frac{\ell}{2}}\vI}\|+\|{\widetilde{\vK}_{\ell+1}}\|\right)\|{2^{-\frac{\ell}{2}}\vI-\widetilde{\vK}_{\ell+1}}\|\tag{by \Cref{cor:KrausDiamond}}\\&
	\leq2\|{2^{-\frac{\ell}{2}}\vI-\widetilde{\vK}_{\ell+1}}\|\\&
	= 2^{1-\frac{\ell}{2}}\nrm{\vI-\sqrt{\vI-b_\ell^2\vB_{\ell+1}}}\\&
	= 2^{1-\frac{\ell}{2}}\Big(1-\sqrt{1-b_\ell^2\nrm{\vB_{\ell+1}}}\Big) \tag{since $\vB_{\ell+1}\succeq 0$}\\&
	\leq2^{1-\frac{\ell}{2}}b_\ell^2\nrm{\vB_{\ell+1}}\tag{since $1-\sqrt{1-x}\le x$ for $0\le x\le 1$}\\&
	=2^{1-\frac{\ell}{2}} 2^{2^{\ell+1}-2}\nrm{\vB_{\ell+1}}\\&
	\leq 2^{1-\frac{\ell}{2}} 4^{2^{\ell}-1}\frac{\lambda^{2^{\ell}}}{\|\CS^-\|_{\infty\to\infty}^2}\tag{by \Cref{lem:RecKrausNormBnd}}\\&
	= 2^{-1-\frac{\ell}{2}}\frac{(4\lambda)^{2^{\ell}}}{\|\CS^-\|_{\infty\to\infty}^2}\\&
	\leq 2^{-\frac{\ell}{2}}\frac{(4\lambda)^{2^{\ell}}}{\|\CS^-\|_{\infty\to\infty}}.
\end{align*}
Combining this with \eqref{eq:RestRClose} yields the first inequality in \eqref{eq:ChannelClose} via the triangle inequality.

Regarding the map $\CQ$, first observe, similarly to map $\widetilde{\CQ}$, that \jnote{Easier proof?}
\begin{align*}
    \nrm{\CQ - \CT' - \CR_{\ell}' - \left(\sum_{j=\ell+1}^\infty c_j^2\right)\CI }_{\Diamond} 
	&= \|{\vK_{\ell+1}[\cdot]\vK'^\dagger_{\ell+1} - 2^{-\ell}\vI[\cdot]\vI^\dagger }\|_{\Diamond}\\
    &\leq 2^{1-\frac{\ell}{2}}\nrm{\vI - \sqrt{\sqrt{\vrho}(\vI - b^2_\ell \vB_{\ell+1})\sqrt{\vrho}}\vrho^{-\frac{1}{2}}}\tag{by \Cref{cor:KrausDiamond}} \\
    &= 2^{1-\frac{\ell}{2}}\nrm{\sum_{k=1}^\infty (\CS^- \circ \nabla^{(k)})[b_\ell^2 \vB_{\ell+1}]} \tag{by \Cref{lem:disTaylor}}\\
    &\leq 2^{1-\frac{\ell}{2}}\sum_{k=1}^\infty \|\CS^-\|_{\infty\to\infty} (-1)^k\binom{\frac{1}{2}}{k}\frac{(4b_\ell^2\lambda^{2^\ell})^k}{2\|\CS^-\|^2_{\infty\to\infty}} \tag{by \Cref{lem:recursive_norm} and \Cref{lem:RecKrausNormBnd}}\\
    &\leq 2^{-\frac{\ell}{2}} \frac{(4\lambda)^{2^\ell}}{\|\CS^-\|_{\infty\to\infty}}. \tag{by \eqref{eq:SqrtTaylorCoeffSum}}
\end{align*}
Combining this with \eqref{eq:RestRClose} yields the second inequality in \eqref{eq:ChannelClose} via the triangle inequality.
\end{proof}

\subsection{Implementation on a quantum computer}
\label{apx:RecDiscImpl}

We now describe how to efficiently implement the map $\CR'_\ell + 2^{-\ell}\CI$ via QSVT given a block-encoding of $\vB_1 = \CT'^\dagger[\vI]$. We first describe how to prepare block-encodings of the Kraus operators $\vK_k$ defined in \eqref{eq:ansatz}.
\begin{lemma}\label{thr:kraus_operators_preparation}
    Let $\varepsilon>0$, $\vB_1$ such that $\|\vB_1\| < \frac{1}{\|\CS^-\|_{\infty\to\infty}^2}$, and $\Delta := \frac{1}{2} + \frac{1}{\pi}\ln\Big(1+\frac{2\nrm{\vH}}{\pi \eps}\Big)$. It is possible to implement an $\delta_k$-approximate block-encoding of $\vK_k$ defined in \eqref{eq:ansatz}, where
    \begin{align*}
        \delta_{k} = \bigO{(\Delta\|\CS^-\|_{\infty\to\infty}/2)^k \varepsilon},
    \end{align*}
    with $\mathcal{O}(\Delta^k)$ uses of a block-encoding of $4\vB_1$, $\widetilde{\mathcal{O}}(\Delta^k)$ (controlled) Hamiltonian simulation time, and $\bigO{k}$ additional ancillary qubits.
\end{lemma}
\begin{proof}
    We first prove by induction on $k$ that it is possible to implement an $\varepsilon_{k+1}$-approximate block-encoding of $\vB_{k+1}' := 2^{2^{k+1}}\vB_{k+1} = 2b_{k+1}\vB_{k+1}$, where $b_k = 2^{2^k - 1}$ and
    \begin{align*}
        \varepsilon_{k+1} = 2\pi \Delta \frac{(2\pi \Delta \|\CS^-\|_{\infty\to\infty})^k - 1}{2\pi \Delta \|\CS^-\|_{\infty\to\infty} - 1}\varepsilon = \bigO{(\Delta \|\CS^-\|_{\infty\to\infty})^k \varepsilon},
    \end{align*}
    with $\mathcal{O}(\Delta^{k})$ uses of $\vU_1$ and $\widetilde{\mathcal{O}}(\Delta^k)$ (controlled) Hamiltonian simulation time. The statement is obviously true for $k=0$. For $k>0$, according to \Cref{thm:SImp} and since $\CS^{\pm} = \frac{\CI}{2} \pm \CS$, given an $\varepsilon_k$-approximate block-encoding $\vU_k$ of $\vB_{k}' := 2b_k\vB_k$ where $\widetilde{\vB}'_k := (\langle 0|^{\otimes d_k}\otimes \vI)\vU_k(|0\rangle^{\otimes d_k}\otimes \vI)$ for some $d_k$ and $\|\vB_k' - \widetilde{\vB}_k'\| \leq \varepsilon_k$, it is possible to implement a block-encoding of
    \begin{align*}
        \frac{\widetilde{\CS}^{\pm}[\widetilde{\vB}_k]}{\Delta} \quad &\text{such that}\quad \|\CS^{\pm}[\vB_k']-\widetilde{\CS}^{\pm}[\widetilde{\vB}_k']\| \leq \varepsilon_k\|\CS^-\|_{\infty\to\infty} + \varepsilon
    \end{align*}
    with a single use of $\vU_k$ and $\bigOt{1}$ (controlled) Hamiltonian simulation time. Now note that $\|\widetilde{\CS}^{\pm}[\widetilde{\vB}_k']\| \leq (\varepsilon_k + 2b_k\|\vB_k\|)\|\CS^-\|_{\infty\to\infty} + \varepsilon \leq 1$ for small enough $\varepsilon$ and using \Cref{lem:RecKrausNormBnd}. Hence we can employ robust oblivious amplitude amplification (see, e.g.,~\cite[Theorem~15]{gilyen2018QSingValTransf}) in order to get rid of the normalization $\Delta$ and obtain a block-encoding of 
    \begin{align*}
        \vC_k^{\pm} \quad\text{such that}\quad \|\vC_k^{\pm} - \CS^{\pm}[\vB_k']\| \leq \pi \Delta(\varepsilon_k\|\CS^-\|_{\infty\to\infty} + \varepsilon)
    \end{align*}
    with $\bigO{\Delta}$ uses of $\vU_k$ and $\bigOt{\Delta}$ (controlled) Hamiltonian simulation time. Since $\vB_{k+1}' :=\CS^+[\vB_k']\CS^-[\vB_k'] = 2b_{k+1}\CS^+[\vB_k]\CS^-[\vB_k]$, by the product of two block-encoded matrices we get an approximate block-encoding of $\vB_{k+1}'$ with error $\varepsilon_{k+1} = 2\pi\Delta(\varepsilon_k\|\CS^-\|_{\infty\to\infty} + \varepsilon)$ with two uses of $\vU_k$ and $\bigOt{1}$ (controlled) Hamiltonian simulation time. By the induction hypothesis,
    \begin{align*}
        \varepsilon_{k+1} &= 2\pi\Delta(\varepsilon_k\|\CS^-\|_{\infty\to\infty} + \varepsilon) \\
        &= (2\pi \Delta)^2 \|\CS^-\|_{\infty\to\infty} \frac{(2\pi \Delta \|\CS^-\|_{\infty\to\infty})^{k-1} - 1}{2\pi \Delta \|\CS^-\|_{\infty\to\infty} - 1}\varepsilon + 2\pi \Delta \varepsilon\\
        &= 2\pi \Delta \frac{(2\pi \Delta \|\CS^-\|_{\infty\to\infty})^k - 1}{2\pi \Delta \|\CS^-\|_{\infty\to\infty} - 1}\varepsilon,
    \end{align*}
    while $\bigO{\Delta}$ uses of $\vU_k$ equal $\mathcal{O}(\Delta^{k})$ uses of $\vU_1$, and similarly for (controlled) Hamiltonian simulation time.

    Given an $\varepsilon_k$-approximate block-encoding of $\vB_k' = 2b_k\vB_k$, we now create a block-encoding for the Kraus operator $\vK_k := c_k(\vI - b_k\CS^-[\vB_k]) = c_k(\vI - \frac{1}{2}\CS^-[\vB_k'])$ where $c_k = 2^{-\frac{k}{2}}$. Once again, an approximate block-encoding for $\CS^-[\vB_k']$ can be implemented via \Cref{thm:SImp}, while the linear combination $c_k(\vI - \frac{1}{2}\CS^-[\vB_k'])$ is done with LCU. By already employing robust oblivious amplitude amplification to remove the normalization, the result is a block-encoding of
    \begin{align*}
        \widetilde{\vK}_k \quad \text{such that}\quad \|\vK_k - \widetilde{\vK}_k\| &\leq \pi c_k^2(1+1/2)^2\pi\Delta(\varepsilon_k\|\CS^-\|_{\infty\to\infty} + \varepsilon)\\
        &= \frac{9\pi}{8} 2^{-k} \varepsilon_{k+1}.
    \end{align*}
    The result now follows by setting $\delta_k = \frac{9\pi}{8} 2^{-k}\varepsilon_{k+1}$.
\end{proof}

\begin{theorem}
    Let $\varepsilon > 0$ and assume access to a block-encoding of $4\CT^\dagger[\vI] = 4\vB_1$, where $\|\vB_1\| \leq \frac{\lambda}{\|\CS^-\|_{\infty\to\infty}^2}$ for $\lambda\in[0,\frac{1}{8}]$. For all $\ell\in\mathbb{N}$, it is possible to implement a map $\widetilde{\CR}'_\ell + 2^{-\ell}\CI$ such that
    \begin{align*}
        \|\widetilde{\CR}'_\ell + 2^{-\ell}\CI - \CR'\|_{\Diamond} \leq \varepsilon + 3\frac{(4\lambda)^{2^{\ell}}}{\|\CS^-\|_{\infty\to\infty}}
    \end{align*}
    with
    \begin{align}\label{eq:complexity_final_map}
        \bigOt{\log^\ell\left(1 + \frac{\|\vH\|\|\CS^-\|_{\infty\to\infty}^\ell}{\varepsilon}\right)}
    \end{align}
    uses of the block-encoding of $4\vB_1$ and (controlled) Hamiltonian simulation time, where $\bigOt{\cdot}$ hides $\log\log$ terms in $\varepsilon$ and $\|\vH\|$.
\end{theorem}
\begin{proof}
    According to \Cref{thr:kraus_operators_preparation}, we have access to an $\varepsilon_k$-approximate block-encoding for each Kraus operator $\vK_k$ by using the block-encoding of $4\vB_1$ a number of $\mathcal{O}(\log^k(1+\|\vH\|/\varepsilon_0))$ times, plus $\widetilde{\mathcal{O}}(\log^k(1+\|\vH\|/\varepsilon_0))$ (controlled) Hamiltonian simulation time, where
    \begin{align*}
        \varepsilon_k = \bigO{\varepsilon_0 \log^k(1+\|\vH\|/\varepsilon_0) \|\CS^-\|^k_{\infty\to\infty} }.
    \end{align*}
    By setting $\varepsilon_0 = \mathcal{O}(\varepsilon \log^{-\ell-1}(1+\|\vH\|/\varepsilon) \|\CS^-\|^{-\ell}_{\infty\to\infty} )$, we can make $\varepsilon_k = \bigO{\varepsilon}$ for all $k\leq \ell$.
    We can block-encode the map
    \begin{align*}
        \widetilde{\CR}'_\ell[\cdot] + 2^{-\ell}\CI[\cdot] = \sum_{k=1}^\ell \widetilde{\vK}_k[\cdot]\widetilde{\vK}_k^\dagger + 2^{-\ell}\vI[\cdot]\vI
    \end{align*}
    by taking the product of the block-encoded matrices of $\widetilde{\vK}_k$ or $2^{-\frac{\ell}{2}}\vI$ and the input density matrix. We now show that the approximate map above is close to $\CR'_\ell + 2^{-\ell}\CI$:
    \begin{align*}
        \|\widetilde{\CR}'_\ell - \CR'_\ell\|_{\Diamond} &= \left\|\sum_{k=1}^\ell (\widetilde{\vK}_k[\cdot]\widetilde{\vK}_k^\dagger - \vK_k[\cdot]\vK_k^\dagger) \right\|_{\Diamond}\\
        &\leq \sum_{k=1}^\ell (\|\widetilde{\vK}_k\| + \|\vK_k\|)\|\widetilde{\vK}_k- \vK_k\| \tag{by \Cref{cor:KrausDiamond}}\\
        &\leq \sum_{k=1}^\ell (2\|\vK_k\| + \varepsilon_k) \varepsilon_k\\
        &\leq \sum_{k=1}^\ell (2c_k + 2c_kb_k\|\CS^-\|_{\infty\to\infty}\|\vB_k\| + \varepsilon_k)\varepsilon_k \tag{by \eqref{eq:telescope}}\\
        &\leq \sum_{k=1}^\ell \left(2^{1-\frac{k}{2}} + 2^{-\frac{k}{2}}\frac{(4\lambda)^{2^{k-1}}}{\|\CS^-\|_{\infty\to\infty}} + \varepsilon_k\right)\varepsilon_k \tag{by \Cref{lem:RecKrausNormBnd}}\\
        &= \sum_{k=1}^\ell \bigO{\varepsilon 2^{-\frac{k}{2}} + \varepsilon^2}\\
        &= \bigO{\varepsilon}.
    \end{align*}
    Therefore, it is possible to implement a map $\widetilde{\CR}'_\ell + 2^{-\ell}\CI$ which is $\varepsilon$-away in the diamond norm from the map $\CR'_\ell + 2^{-\ell}\CI$. The complexity in \eqref{eq:complexity_final_map} comes from setting $\varepsilon_0 = \mathcal{O}(\varepsilon \log^{-\ell-1}(1+\|\vH\|/\varepsilon) \|\CS^-\|^{-\ell}_{\infty\to\infty} )$ in $\widetilde{\mathcal{O}}(\log^k(1+\|\vH\|/\varepsilon_0))$. Finally, the distance between $\widetilde{\CR}'_\ell + 2^{-\ell}\CI$ and the true map $\CR'$ given in \eqref{eq:telescope} follows from \eqref{eq:RestRClose}.
\end{proof}







\section{Mathematica code for symbolically verifying our calculations}\label{apx:Mathematica}

\mmaDefineMathReplacement[≤]{<=}{\leq}
\mmaDefineMathReplacement[≥]{>=}{\geq}
\mmaDefineMathReplacement[≠]{!=}{\neq}
\mmaDefineMathReplacement[→]{->}{\to}[2]
\mmaDefineMathReplacement[⧴]{:>}{:\hspace{-.2em}\to}[2]
\mmaDefineMathReplacement{∈}{\in}
\mmaDefineMathReplacement{∉}{\notin}
\mmaDefineMathReplacement{ℝ}{\mathbb{R}}
\mmaDefineMathReplacement{∞}{\infty}
\mmaDefineMathReplacement{±}{\pm}
\mmaDefineMathReplacement{α}{\alpha}
\mmaDefineMathReplacement{β}{\beta}
\mmaDefineMathReplacement{γ}{\gamma}
\mmaDefineMathReplacement{ν}{\nu}
\mmaDefineMathReplacement{σ}{\sigma}
\mmaDefineMathReplacement{ω}{\omega} 
\mmaDefineMathReplacement{π}{\pi} 
\phantom{text}\\[-3mm]
\begin{itemize}[leftmargin=-1mm]
 	\item Setting up proper parameters for the Fourier Transforms:
\begin{mmaCell}[functionlocal=x,mathreplacements=bold]{Code}
		SetOptions[FourierTransform, FourierParameters -> {0, -1}];
		SetOptions[InverseFourierTransform, FourierParameters -> {0, -1}];
\end{mmaCell}
	\item Verifying Fourier transforms in \Cref{cor:concreteBounds}:
\begin{mmaCell}[functionlocal=x,mathreplacements=bold]{Code}
  InverseFourierTransform[I Sqrt[2π] / Sinh[2π t], t, \mmaFnc{ω}]
  FourierTransform[Piecewise[{{1/2, \mmaFnc{ω} < 0}, {Exp[-\mmaFnc{ω}] - 1/2, \mmaFnc{ω} >= 0}}], \mmaFnc{ω}, t]
\end{mmaCell}
\begin{mmaCell}{Output}
  -(1/2) Tanh[ω/4]
  1/(Sqrt[2π] t (t - I))
\end{mmaCell} 

	\item Verifying the time domain function for the superoperator $\CS_c$:
\begin{mmaCell}[functionlocal=n,mathreplacements=bold]{Code}
  FourierTransform[1 / (2 Cosh[\mmaFnc{ω}/4]), \mmaFnc{ω}, t]
\end{mmaCell}
\begin{mmaCell}{Output}
  Sqrt[2π] / Cosh[2π t]
\end{mmaCell}
 	\item Verifying classical detailed balance with uncertain estimates \eqref{eq:ClassicalUncertainBalance}:
\begin{mmaCell}[functionlocal=n,mathreplacements=bold]{Code}
  Integrate[Exp[-\mmaFnc{n}^2/(2 \mmaUnd{σ}^2)]/(\mmaUnd{σ} Sqrt[2π]) Exp[-Max[0,n + \mmaUnd{ν} + \mmaUnd{σ}^2/2]], 
  			{n,-∞,∞}, Assumptions -> \mmaUnd{σ} > 0] // FullSimplify
\end{mmaCell}
\begin{mmaCell}{Output}
  1/2 (Exp[-\mmaUnd{ν}] Erfc[(-2\mmaUnd{ν} + \mmaUnd{σ}^2)/(Sqrt[8]\mmaUnd{σ})] + Erfc[(2\mmaUnd{ν} + \mmaUnd{σ}^2)/(Sqrt[8]\mmaUnd{σ})])
\end{mmaCell}
\end{itemize}

\end{document}